\documentclass[11pt]{article}

\pdfoutput=1 

\usepackage{mathtools} 

\usepackage{enumitem} 
\setlist{noitemsep,topsep=0pt,parsep=0pt,partopsep=0pt,listparindent=\parindent}

\usepackage[margin=0.8cm]{caption} 
\usepackage{subcaption}
\usepackage{rotating}
\usepackage{floatpag} 
\usepackage{cases}
\usepackage{tikz}
\usetikzlibrary{arrows,backgrounds,calc,fit,decorations.pathreplacing,decorations.markings,patterns,positioning,shapes.geometric}

\makeatletter
\pgfarrowsdeclare{new open triangle 60}{new open triangle 60}
{
  \pgfutil@tempdima=0.5pt%
  \advance\pgfutil@tempdima by.25\pgflinewidth%
  \pgfutil@tempdimb=7.29\pgfutil@tempdima\advance\pgfutil@tempdimb by.5\pgflinewidth%
  \pgfarrowsleftextend{+-\pgfutil@tempdimb}
  \pgfutil@tempdimb=.5\pgfutil@tempdima\advance\pgfutil@tempdimb by\pgflinewidth%
  \pgfarrowsrightextend{+\pgfutil@tempdimb}6
}
{
  \pgfutil@tempdima=0.5pt%
  \advance\pgfutil@tempdima by.25\pgflinewidth%
  \pgfsetdash{}{+0pt}
  \pgfsetmiterjoin
    \pgfsetfillcolor{white}
  \pgfpathmoveto{\pgfpointadd{\pgfqpoint{0.5\pgfutil@tempdima}{0pt}}{\pgfqpointpolar{150}{9\pgfutil@tempdima}}}
  \pgfpathlineto{\pgfqpoint{0.5\pgfutil@tempdima}{0\pgfutil@tempdima}}
  \pgfpathlineto{\pgfpointadd{\pgfqpoint{0.5\pgfutil@tempdima}{0pt}}{\pgfqpointpolar{-150}{9\pgfutil@tempdima}}}
  \pgfpathclose
  \pgfusepathqfillstroke
}

\tikzset{every fit/.append style=text badly centered}

\usepackage{tyson}

\usepackage[bookmarks=true,hypertexnames=false]{hyperref}

\newcommand{\Holant}{\operatorname{Holant}}
\newcommand{\PlHolant}{\operatorname{Pl-Holant}}
\newcommand{\holant}[2]{\Holant(#1\mid #2)}
\newcommand{\plholant}[2]{\PlHolant(#1\mid #2)}

\newcommand{\AllDistinct}{\textsc{All-Distinct}}
\newcommand{\AD}{\operatorname{AD}}
\newcommand{\subMat}[4]{\mbox{\scriptsize $\begin{smallmatrix} #1 & #4 \\ #2 & #3 \end{smallmatrix}$}}

\newenvironment{remark}{\medskip{\bfseries \noindent Remark:}}{\par\medskip}{\par\medskip}

\def\borderColor{blue!60}
\def\arrowType{new open triangle 60}
\def\scale{0.6}
\def\nodeDist{1.4cm}
\tikzstyle{internal} = [draw, fill, shape=circle]
\tikzstyle{external} = [shape=circle]
\tikzstyle{square}   = [draw, fill, rectangle, inner sep=5pt]
\tikzstyle{oval}     = [draw, fill, ellipse, minimum width=15pt, inner sep=2.5pt]
\tikzstyle{triangle} = [draw, fill, regular polygon, regular polygon sides=3, inner sep=2.5pt] 

\title{The Complexity of Counting Edge Colorings\\and a Dichotomy for Some Higher Domain Holant Problems}

\author{
 Jin-Yi Cai\\
 \scriptsize University of Wisconsin--Madison\\
 \footnotesize \texttt{jyc@cs.wisc.edu}
 \and
 Heng Guo\\
 \scriptsize University of Wisconsin--Madison\\
 \footnotesize \texttt{hguo@cs.wisc.edu}
 \and
 Tyson Williams\\
 \scriptsize University of Wisconsin--Madison\\
 \footnotesize \texttt{tdw@cs.wisc.edu}
}

\date{} 

\begin{document}
\maketitle

\begin{abstract}
 We show that an effective version of Siegel's Theorem on finiteness of integer solutions and an application of elementary Galois theory
 are key ingredients in a complexity classification of some Holant problems.
 These Holant problems, denoted by $\Holant(f)$,
 are defined by a symmetric ternary function $f$ that is invariant under any permutation of the $\kappa \ge 3$ domain elements.
 We prove that $\Holant(f)$ exhibits a complexity dichotomy.
 This dichotomy holds even when restricted to planar graphs.
 A special case of this result is that counting edge $\kappa$-colorings is $\SHARPP$-hard over planar $3$-regular graphs for $\kappa \ge 3$.
 In fact, we prove that counting edge $\kappa$-colorings is $\SHARPP$-hard over planar $r$-regular graphs for all $\kappa \ge r \ge 3$.
 The problem is polynomial-time computable in all other parameter settings.
 The proof of the dichotomy theorem for $\Holant(f)$ depends on the fact that a specific polynomial $p(x, y)$ has an explicitly listed finite set of integer solutions,
 and the determination of the Galois groups of some specific polynomials.
 In the process,
 we also encounter the Tutte polynomial, medial graphs, Eulerian partitions, Puiseux series,
 and a certain lattice condition on the (logarithm of) the roots of polynomials.
\end{abstract}

\section{Introduction}

What do Siegel's Theorem and Galois theory have to do with complexity theory?
In this paper,
we show that an effective version of Siegel's Theorem on finiteness of integer solutions and an application of elementary Galois theory
are key ingredients in a chain of steps that lead to a complexity classification of some counting problems.
More specifically,
we consider a certain class of counting problems that are expressible as Holant problems
with an arbitrary domain of size $\kappa$ over $3$-regular graphs,
and prove a dichotomy theorem for this class of problems.
This dichotomy also holds when restricted to planar graphs.
Among other things, the proof of the dichotomy theorem depends on 
the following: (A) a specific polynomial
$p(x,y) = x^5 - 2 x^3 y - x^2 y^2 - x^3 + x y^2 + y^3 - 2 x^2 - x y$
has only the following integer solutions
$(x,y) = (-1,1), (0,0), (1,-1), (1,2), (3,3)$,
and (B) the determination of the Galois groups of some specific polynomials.
In the process,
we also encounter the Tutte polynomial, medial graphs, Eulerian partitions, Puiseux series,
and a certain lattice condition on the (logarithm of) the roots of polynomials such as $p(x,y)$.

A special case of this dichotomy theorem is the problem of counting edge colorings over planar $3$-regular graphs using $\kappa$ colors.
In this case,
the corresponding constraint function is the $\AllDistinct_{3,\kappa}$ function,
which takes value~$1$ when all three inputs from $[\kappa]$ are distinct and~$0$ otherwise.
We further prove that the problem using $\kappa$ colors over $r$-regular graphs is $\SHARPP$-hard for all $\kappa \ge r \ge 3$,
even when restricted to planar graphs.
The problem is polynomial-time computable in all other parameter settings.
This solves a long-standing open problem.

We give a brief description of the framework of Holant problems~\cite{CLX11c, CLX12, CLX09a, CLX11d}.
The problem $\Holant(\mathcal{F})$,
defined by a set of functions $\mathcal{F}$,
takes as input a \emph{signature grid} $\Omega = (G, \pi)$,
where $G = (V,E)$ is a graph,
$\pi$ assigns each $v \in V$ a function $f_v \in \mathcal{F}$,
and $f_v$ maps $[\kappa]^{\deg(v)}$ to $\mathbb{C}$ for some integer $\kappa \ge 2$.
An edge $\kappa$-labeling $\sigma : E \to [\kappa]$ gives an evaluation $\prod_{v \in V} f_v(\sigma \mid_{E(v)})$,
where $E(v)$ denotes the incident edges of $v$ and $\sigma \mid_{E(v)}$ denotes the restriction of $\sigma$ to $E(v)$.
The counting problem on the instance $\Omega$ is to compute
\[
 \Holant(\Omega, \mathcal{F}) = \sum_{\sigma : E \to [\kappa]} \prod_{v \in V} f_v\left(\sigma \mid_{E(v)}\right).
\]
Counting edge $\kappa$-colorings on $r$-regular graphs amounts to setting $f_v = \AllDistinct_{r,\kappa}$ for all $v$.

An edge $\kappa$-coloring of a graph $G$ is an edge $\kappa$-labeling of $G$ such that any two incident edges have different colors.
A fundamental problem in graph theory is to determine how many colors are required to edge color $G$.
The obvious lower bound is $\Delta(G)$,
the maximum degree of the graph.
By Vizing's Theorem~\cite{Viz65},
an edge coloring using just $\Delta(G) + 1$ colors always exists.
Whether $\Delta(G)$ colors suffice depends on the graph $G$.

Consider the edge coloring problem over $3$-regular graphs.
It follows from the parity condition (Lemma~\ref{lem:k=r:parity_condition})
that any graph containing a bridge does not have an edge $3$-coloring.
For bridgeless planar graphs,
Tait~\cite{Tai80} showed that the existence of an edge $3$-coloring is equivalent to the Four-Color Theorem.
Thus, the answer for the decision problem over planar $3$-regular graphs is that there is an edge $3$-coloring iff the graph is bridgeless.

Without the planarity restriction,
determining if a $3$-regular graph has an edge $3$-coloring is $\NP$-complete~\cite{Hol81}.
This hardness extends to finding an edge $\kappa$-coloring over $\kappa$-regular graphs for all $\kappa \ge 3$~\cite{LG83}.
However, these reductions are not parsimonious,
and, in fact, it is claimed that no parsimonious reduction exists unless $\P = \NP$~\cite[p.~118]{Wel93}.
The counting complexity of this problem has remained open.

We prove that counting edge colorings over planar regular graphs is $\SHARPP$-hard.%
\footnote{Vizing's Theorem is for simple graphs.
In Holant problems as well as counting complexity such as Graph Homomorphism or counting CSP,
one typically considers multigraphs (i.e.~self-loops and parallel edges are allowed).
However, our hardness result for counting edge $3$-colorings over planar $3$-regular graphs also holds for simple graphs (Theorem~\ref{thm:edge_coloring:k=r=3_simple}).}

\begin{theorem} \label{thm:edge_coloring}
 \#$\kappa$-\textsc{EdgeColoring} is $\SHARPP$-hard over planar $r$-regular graphs for all $\kappa \ge r \ge 3$.
\end{theorem}

\noindent
This theorem is proved in Theorem~\ref{thm:edge_coloring:k=r} for $\kappa = r$
and Theorem~\ref{thm:edge_coloring:k>r} for $\kappa > r$.

The techniques we develop to prove Theorem~\ref{thm:edge_coloring} naturally extend to a class of Holant problems with domain size $\kappa \ge 3$ over planar $3$-regular graphs.
Functions such as $\AllDistinct_{3,\kappa}$ are symmetric,
which means that they are invariant under any permutation of its~$3$ inputs.
But $\AllDistinct_{3,\kappa}$ has another invariance---it is invariant under any permutation of the $\kappa$ domain elements.
We call the second property \emph{domain invariance}.

A ternary function that is both symmetric and domain invariant is specified by three values,
which we denote by $\langle a,b,c \rangle$.
The output is $a$ when all inputs are the same,
the output is $c$ when all inputs are distinct,
and the output is $b$ when two inputs are the same but the third input is different.

We prove a dichotomy theorem for such functions with complex weights.

\begin{theorem} \label{thm:dichotomy:simple}
 Suppose $\kappa \ge 3$ is the domain size and $a,b,c \in \mathbb{C}$.
 Then $\PlHolant(\langle a,b,c \rangle)$ is either computable in polynomial time or $\SHARPP$-hard.
 Furthermore, given $a,b,c$, there is a polynomial-time algorithm that decides whether $\langle a,b,c \rangle$ is in polynomial time or $\SHARPP$-hard.
\end{theorem}

\noindent
See Theorem~\ref{thm:dichotomy} for an explicit listing of the tractable cases.
Note that counting edge $\kappa$-colorings over $3$-regular graphs is the special case when $\langle a,b,c \rangle = \langle 0,0,1 \rangle$.

There is only one previous dichotomy theorem for higher domain Holant problems~\cite{CLX13} (see Theorem~\ref{thm:tractable:holant-star}).
The important difference is that the present work is for general domain size $\kappa \ge 3$ while the previous result is for domain size $\kappa = 3$.
When restricted to domain size~$3$,
the result in~\cite{CLX13} assumes that all unary functions are available,
while this dichotomy does not assume that; however it does assume domain invariance.
Dichotomy theorems for an arbitrary domain size are generally difficult to prove.
The Feder-Vardi Conjecture for decision Constraint Satisfaction Problems (CSP) is still open~\cite{FV99}.
It was a major achievement to prove this conjecture for domain size~$3$~\cite{Bul06}.
The counting CSP dichotomy was proved after a long series of work~\cite{BD07, Bul08, BG05, DGJ09, BDGJR09, CLX09a, CCL11, CHL12, DR10, GHLX11, CK12, CC12}.

Our proof of Theorem~\ref{thm:dichotomy:simple} has many components,
and a number of new ideas are introduced in this proof.
We discuss some of these ideas and give an outline of our proof in Section~\ref{sec:outline}.
In Section~\ref{sec:preliminaries},
we review basic terminology and define the notation of a \emph{succinct signature}.
Section~\ref{sec:coloring} contains our proof of Theorem~\ref{thm:edge_coloring} about edge coloring.
In Section~\ref{sec:tractable},
we discuss the tractable cases of Theorem~\ref{thm:dichotomy:simple}.
In Section~\ref{sec:interpolation},
we extend our main proof technique of polynomial interpolation.
Then in Sections~\ref{sec:ternary},~\ref{sec:unary}, and~\ref{sec:binary},
we develop our hardness proof
and tie everything together in Section~\ref{sec:dichotomy},

\section{Proof Outline and Techniques} \label{sec:outline}

As usual,
the difficult part of a dichotomy theorem is to carve out \emph{exactly} the tractable problems in the class,
and prove all the rest $\SHARPP$-hard.
A dichotomy theorem for Holant problems has the additional difficulty that
some tractable problems are only shown to be tractable under a holographic transformation,
which can make the appearance of the problem rather unexpected.
For example,
we show in Section~\ref{sec:tractable} that the problem
$\Holant(\langle -3 - 4 i, 1, -1 +  2 i \rangle)$ on domain size~$4$ is tractable.
Despite its appearance,
this problem is intimately connected with a tractable graph homomorphism problem defined by the Hadamard matrix
$\left[\begin{smallmatrix*}[r] 1 & -1 & -1 & -1 \\ -1 & 1 & -1 & -1 \\ -1 & -1 & 1 & -1 \\ -1 & -1 & -1 & 1 \end{smallmatrix*}\right]$.
In order to understand all problems in a Holant problem class,
we must deal with such problems.
Dichotomy theorems for graph homomorphisms and for counting CSP do not have to deal with as varied a class of such problems,
since they implicitly assume all \textsc{Equality} functions are available and must be preserved.
This restricts the possible transformations.

After isolating a set of tractable problems,
our $\SHARPP$-hardness results in both Theorem~\ref{thm:edge_coloring} and Theorem~\ref{thm:dichotomy:simple}
are obtained by reducing from evaluations of the Tutte polynomial over planar graphs.
A dichotomy is known for such problems (Theorem~\ref{thm:tutte}).

The chromatic polynomial, 
a specialization of the Tutte polynomial (Proposition~\ref{prop:k>r:chromatic_tutte}),
is concerned with vertex colorings.
On domain size $\kappa$,
one starting point of our hardness proofs is the chromatic polynomial,
for the problem of counting vertex colorings using at most $\kappa$ colors.
By the planar dichotomy for the Tutte polynomial,
this problem is $\SHARPP$-hard for all $\kappa \ge 3$.

Another starting point for our hardness reductions is the evaluation of the Tutte polynomial at an integer diagonal point $(x,x)$,
which is $\SHARPP$-hard for all $x \ge 3$ by the same planar Tutte dichotomy.
These are new starting places for reductions involving Holant problems.
These problems were known to have a so-called state-sum expression (Lemma~\ref{lem:tutte_connection}),
which is a sum over weighted Eulerian partitions.
This sum is not over the original planar graph but over its directed medial graph,
which is always a planar $4$-regular graph (Figure~\ref{fig:directed_medial_graph_example}).
We show that this state-sum expression is naturally expressed as a Holant problem with a particular quaternary constraint function (Lemma~\ref{lem:holant_connection}).

To reduce from these two problems,
we execute the following strategy.
First, we attempt to construct the unary constraint function $\langle 1 \rangle$,
which takes value~$1$ on all $\kappa$ inputs (Lemma~\ref{lem:unary:construct_<1>}).
Second, we attempt to interpolate all succinct binary signatures assuming that we have $\langle 1 \rangle$ (Section~\ref{sec:binary}).
(See Section~\ref{sec:preliminaries} for the definition of a succinct signature.)
Lastly, we attempt to construct a ternary signature with a special property assuming that all these binary signatures are available (Lemma~\ref{lem:ternary:construct_abb}).
At each step,
there are some problems specified by certain signatures $\langle a,b,c \rangle$ for which our attempts fail.
In such cases,
we directly obtain a dichotomy without the help of additional signatures.
See Figure~\ref{fig:outline} for a flow chart of hardness reductions.

\begin{figure}[p]
 \tikzstyle{block} = [rectangle, draw, fill=blue!20, text centered, rounded corners, minimum height=2em]
 \centering
 \begin{tikzpicture}[scale=1, transform shape, node distance=2.5cm, semithick]
  \node [block, text width=7em]                                         (main)    {$\Holant(\langle a,b,c \rangle)$};
  \node [block, text width=8em,  below of=main]                         (unary)   {Attempts~1 and~2\\Lemma~\ref{lem:unary:construct_<1>}};
  \node [block, text width=5em,  below of=unary]                        (binary1) {Attempt~1\\Lemma~\ref{lem:binary:general}};
  \node [block, text width=14em, below of=binary1, node distance=1.7cm] (binary2)
        {Attempt~2\\
        Cases~1,~2,~3,~4,~5\\
        Lemmas~\ref{lem:binary:2b+(k-2)c=0},
        \ref{lem:binary:2a+2(2k-3)b+(k-2)^2c=0},
        \ref{lem:binary:a+3(k-1)b+(k-2)(k-1)c=0},
        \ref{lem:binary:k=3_and_b=0},
        \ref{lem:binary:k=3_and_2a+3b+4c=0}};
  \node [block, text width=8em,  below of=binary2, node distance=1.9cm] (binary3) {Attempts~3 and~4\\All Cases\\Lemma~\ref{lem:appendix:binary}};
  \node [block, text width=6em,  below of=binary3] (ternary) {Attempt~1\\Lemma~\ref{lem:ternary:construct_abb}};
  \node [block, text width=10em, below of=ternary] (Fischer) {Bobby Fischer Gadget\\Lemma~\ref{lem:k>r:interpolate_equality4}};
  \node [block, text width=13em, below of=Fischer] (vertex)  {Counting Vertex $\kappa$-Colorings\\Corollary~\ref{cor:k>r:abb_unary_binaries}};
  \node [draw=\borderColor,thick,dashed,rounded corners,inner xsep=6pt,inner ysep=6pt,fit = (unary), label=right:{Fail}]    (uborder) {};
  \node [draw=\borderColor,thick,dashed,rounded corners,inner xsep=6pt,inner ysep=6pt,fit = (binary1)  (binary2) (binary3)] (bborder) {};
  \node [draw=\borderColor,thick,dashed,rounded corners,inner xsep=6pt,inner ysep=6pt,fit = (ternary)]                      (tborder) {};
  \node [text centered, minimum height=2em, text width=7em, left of=bborder, node distance=4.6cm] (btext)
        {\large Interpolate\\all $\langle x,y \rangle$\\Corollary~\ref{cor:binary:interpolate}};
  \path let
         \p1 = (btext),
         \p2 = (unary)
        in
         node [text centered, minimum height=2em, text width=7em] (utext) at (\x1, \y2) {\large Construct $\langle 1 \rangle$};
  \path let
         \p1 = (btext),
         \p2 = (ternary)
        in
         node [text centered, minimum height=2em, text width=9em] (ttext) at (\x1, \y2) {\large Construct $\langle a,b,b \rangle$\\with \quad $a \ne b$};
  \node [      right of=uborder] (uright)   {};
  \node [above right of=uright]  (urightAR) {};
  \node [below right of=uright]  (urightBR) {};
  \node [above of=urightBR] (ADlikeY)  {};
  \node [below of=urightAR] (DiffDomY) {};
  \node [      right of=bborder] (bright)   {};
  \node [block, text width=6em, right of=bright, node distance=3.3cm] (Bfail) {Corollary~\ref{cor:unary:dichotomy:a+(k-3)b-(k-2)c=0}};
  \path let
         \p1 = (Bfail),
         \p2 = (ADlikeY)
        in
         node [block, text width=5em] (ADlike) at (\x1, \y2) {Lemma~\ref{lem:unary:dichotomy:AD-like}};
  \path let
         \p1 = (Bfail),
         \p2 = (DiffDomY)
        in
         node [block, text width=5em] (DiffDom) at (\x1, \y2) {Lemma~\ref{lem:unary:<(k-1)(k-2),-(k-2),2>}};
  \node [block, text width=8em, right of=ternary, node distance=4.5cm] (Afail)
        {Construct\\$\langle 3 (\kappa - 1), \kappa - 3, -3 \rangle$\\Lemma~\ref{lem:ternary:fixed_point}};
  \node [block, text width=6.5em, right of=Bfail, node distance=2.8cm] (CWEP) {Counting\\Weighted\\Eulerian\\Partitions\\Corollary~\ref{cor:ternary:201000100_hard}};
  \path let
         \p1 = (CWEP),
         \p2 = (Afail)
        in
         node [block, text width=6em, node distance=4cm] (Ahard) at (\x1, \y2) {Lemmas~\ref{lem:ternary:3k-1k-3-3_hard_k>3} and~\ref{lem:ternary:-201}};
  \path (main)    edge[ultra thick,->] (uborder)
        (uborder) edge[ultra thick,->] node[label=right:{Succeed}] {} (bborder)
        (bborder) edge[ultra thick,->] node[label=right:{Succeed}] {} (tborder)
        (tborder) edge[ultra thick,->] node[label=right:{Succeed}] {} (Fischer)
        (Fischer) edge[ultra thick,->] (vertex)
        (DiffDom) edge[ultra thick,->] (ADlike)
        (uborder) edge[ultra thick,->] (ADlike)
        (uborder) edge[ultra thick,->] (DiffDom)
        (bborder) edge[ultra thick,->] node[label=above:{Fail},label=below:{$\mathfrak{B}=0$}] {} (Bfail)
        (tborder) edge[ultra thick,->] node[label=above:{Fail},label=below:{$\mathfrak{A}=0$}] {} (Afail)
        (Afail)   edge[ultra thick,->] (Ahard)
        (Bfail)   edge[ultra thick,->] (DiffDom)
        (Ahard)   edge[ultra thick,->] (CWEP)
        (ADlike)  edge[ultra thick,->,out=0,in=90] (CWEP);
 \end{tikzpicture}
 \caption{Flow chart of hardness reductions in our proof of Theorem~\ref{thm:dichotomy:simple} going back to our two starting points of hardness.}
 \label{fig:outline}
\end{figure}
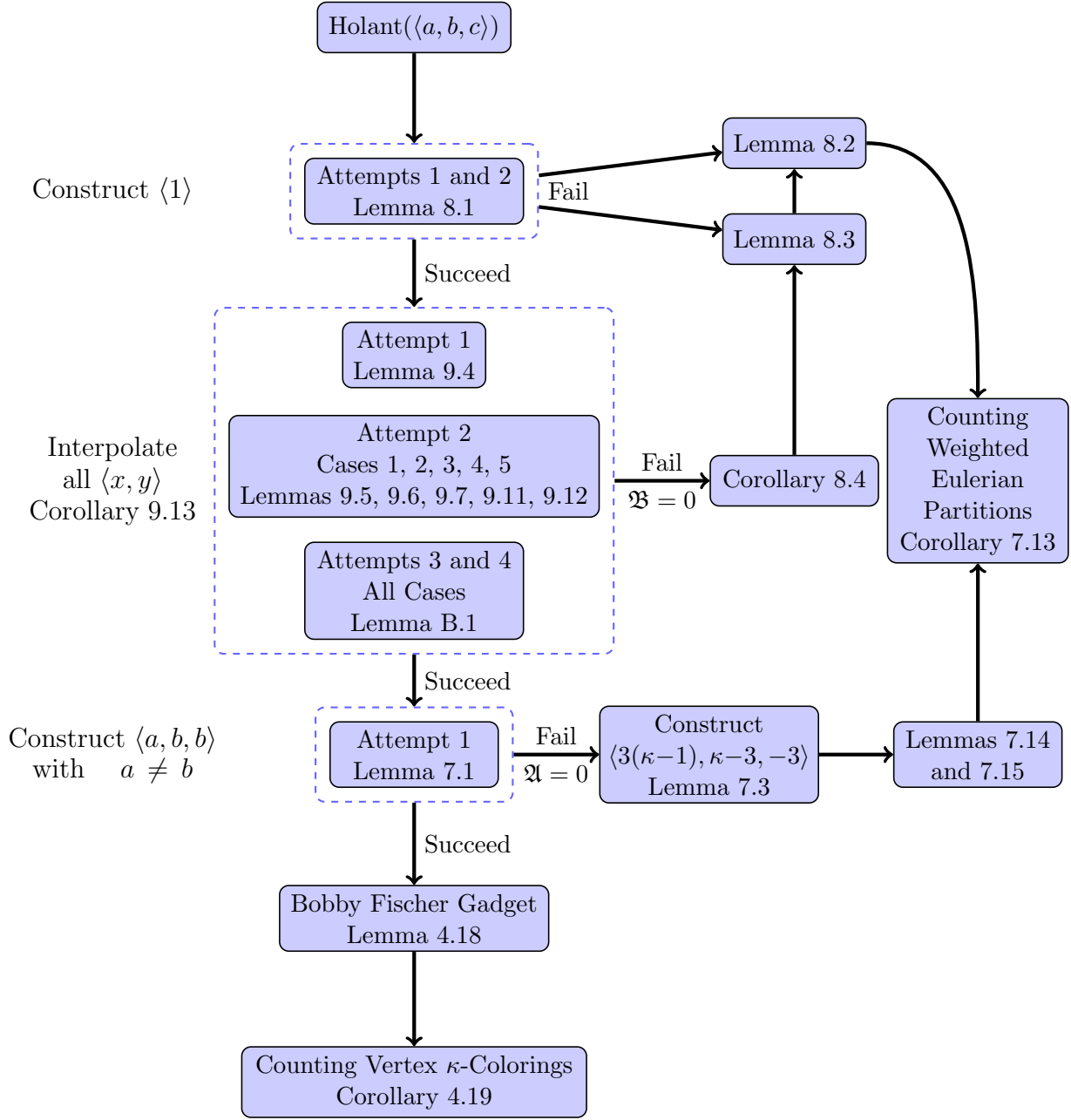

Below we highlight some of our proof techniques.

\paragraph{Interpolation within an orthogonal subspace}
We develop the ability to interpolate when faced with some nontrivial null spaces inherently present in interpolation constructions.
In any construction involving an initial signature and a recurrence matrix,
it is possible that the initial signature is orthogonal to some row eigenvectors of the recurrence matrix.
Previous interpolation results always attempt to find a construction that avoids this.
In the present work,
this avoidance seems impossible.
In Section~\ref{sec:interpolation},
we prove an interpolation result that can succeed in this situation to the greatest extent possible. 
We prove that one can interpolate any signature provided that it is orthogonal to the same set of row eigenvectors,
and the relevant eigenvalues satisfy a lattice condition (Lemma~\ref{lem:interpolate_all_not_orthogonal}).

\paragraph{Satisfy lattice condition via Galois theory}
A key requirement for this interpolation to succeed is the lattice condition (Definition~\ref{def:interpolation:lattice_condition}),
which involves the roots of the characteristic polynomial of the recurrence matrix.
We use Galois theory to prove that our constructions satisfy this condition.
If a polynomial has a large Galois group,
such as $S_n$ or $A_n$,
and its roots do not all have the same complex norm,
then we show that its roots satisfy the lattice condition (Lemma~\ref{lem:lattice_condition:Sn_An}).

\paragraph{Effective Siegel's Theorem via Puiseux series}
We need to determine the Galois groups for an infinite family of polynomials,
one for each domain size.
If these polynomials are irreducible,
then we can show they all have the full symmetric group as their Galois group,
and hence fulfill the lattice condition.
We suspect that these polynomials are all irreducible but are unable to prove it.

A necessary condition for irreducibility is the absence of any linear factor.
This infinite family of polynomials,
as a single bivariate polynomial in $(x, \kappa)$,
defines an algebraic curve,
which has genus~$3$.
By a well-known theorem of Siegel~\cite{Sie29},
there are only a finite number of integer values of $\kappa$ for which the corresponding polynomial has a linear factor.
However this theorem and others like it are not \emph{effective} in general.
There are some effective versions of Siegel's Theorem that can be applied to the algebraic curve,
but the best general effective bound is over $10^{20,000}$~\cite{Wal92} and hence cannot be checked in practice.
Instead, we use Puiseux series to show that this algebraic curve has exactly five explicitly listed integer solutions (Lemma~\ref{lem:ternary:lattice:no_linear}).

\paragraph{Eigenvalue Shifted Triples}
For a pair of eigenvalues,
the lattice condition is equivalent to the statement that the ratio of these eigenvalues is not a root of unity.
A sufficient condition is that the eigenvalues have distinct complex norms.
We prove three results,
each of which is a different way to satisfy this sufficient condition.
Chief among them is the technique we call an \emph{Eigenvalue Shifted Triple} (EST).
These generalize the technique of Eigenvalue Shifted Pairs from~\cite{KC10}.
In an EST,
we have three recurrence matrices,
each of which differs from the other two by a nonzero additive multiple of the identity matrix.
Provided these two multiples are linearly independent over $\R$,
we show at least one of these matrices has eigenvalues with distinct complex norms (Lemma~\ref{lem:binary:EST}).
(However determining which one succeeds is a difficult task;
but we need not know that).

\paragraph{E Pluribus Unum}
When the ratio of a pair of eigenvalues is a root of unity,
it is a challenge to effectively use this failure condition.
Direct application of this cyclotomic condition is often of limited use.
We introduce an approach that uses this cyclotomic condition effectively.
A direct recursive construction involving these two eigenvalues only creates a finite number of different signatures.
We reuse all of these signatures in a multitude of new interpolation constructions (Lemma~\ref{lem:binary:general:root_of_unity}),
one of which we hope will succeed.
If the eigenvalues in all of these constructions also satisfy a cyclotomic condition,
then we obtain a more useful condition than any of the previous cyclotomic conditions.
This idea generalizes the anti-gadget technique~\cite{CKW12},
which only reuses the ``last'' of these signatures.

\paragraph{Local holographic transformation}
One reason to obtain all succinct binary signatures is for use in the gadget construction known as
a local holographic transformation (Figure~\ref{fig:gadget:ternary:local_holographic_transformation}).
This construction mimics the effect of a holographic transformation applied on a single signature.
In particular, 
using this construction,
we attempt to obtain a succinct ternary signature of the form $\langle a,b,b \rangle$,
where $a \not = b$ (Lemma~\ref{lem:ternary:construct_abb}).
This signature turns out to have some magical properties in the Bobby Fischer gadget,
which we discuss next.

\paragraph{Bobby Fischer gadget}
Typically,
any combinatorial construction for higher domain Holant problems produces very intimidating looking expressions that are nearly impossible to analyze.
In our case,
it seems necessary to consider a construction that has to satisfy multiple requirements involving at least nine polynomials.
However,
we are able to combine the signature $\langle a,b,b \rangle$,
where $a \not = b$,
with a succinct binary signature of our choice in a special construction that we call the \emph{Bobby Fischer gadget} (Figure~\ref{fig:gadget:k>r:fischer}).
This gadget is able to satisfy seven conditions using just one degree of freedom (Lemma~\ref{lem:k>r:interpolate_equality4}).
This ability to satisfy a multitude of constraints simultaneously in one magic stroke reminds us of some unfathomably brilliant moves by Bobby Fischer,
the chess genius extraordinaire.

\section{Preliminaries} \label{sec:preliminaries}

\subsection{Problems and Definitions}

The framework of Holant problems is defined for functions mapping any $[\kappa]^n \to R$ for a finite $\kappa$ and some commutative semiring $R$.
In this paper, we investigate some complex-weighted $\Holant$ problems on domain size $\kappa \ge 3$.
A constraint function, or \emph{signature}, of arity $n$, maps from $[\kappa]^n \to \mathbb{C}$.
For consideration of models of computation,
functions take complex algebraic numbers.

A \emph{signature grid} $\Omega = (G, \pi)$ of $\Holant(\mathcal{F})$ consists of a graph $G = (V,E)$,
where $\pi$ assigns each vertex $v \in V$ and its incident edges with some $f_v \in \mathcal{F}$ and its input variables.
We say $\Omega$ is a \emph{planar signature grid} if $G$ is planar,
where the variables of $f_v$ are ordered counterclockwise.
The Holant problem on instance $\Omega$ is to evaluate $\Holant(\Omega; \mathcal{F}) = \sum_{\sigma} \prod_{v \in V} f_v(\sigma \mid_{E(v)})$,
a sum over all edge labelings $\sigma: E \to [\kappa]$,
where $E(v)$ denotes the incident edges of $v$ and $\sigma \mid_{E(v)}$ denotes the restriction of $\sigma$ to $E(v)$.

A function $f_v$ can be represented by listing its values in lexicographical order as in a truth table,
which is a vector in $\mathbb{C}^{\kappa^{\deg(v)}}$,
or as a tensor in $(\mathbb{C}^{\kappa})^{\otimes \deg(v)}$.
In this paper, we consider symmetric signatures.
An example of which is the \textsc{Equality} signature $=_r$ of arity $r$.
Sometimes we represent $f$ as a matrix $M_f$ that we call its \emph{signature matrix},
which has row index $(x_1, \dotsc, x_t)$ and column index $(x_k, \dotsc, x_{t+1})$ (in reverse order) for some $t$ that will be clear from context.

A Holant problem is parametrized by a set of signatures.

\begin{definition}
 Given a set of signatures $\mathcal{F}$,
 we define the counting problem $\Holant(\mathcal{F})$ as:

 Input: A \emph{signature grid} $\Omega = (G, \pi)$;

 Output: $\Holant(\Omega; \mathcal{F})$.
\end{definition}

\noindent
The problem $\PlHolant(\mathcal{F})$ is defined similarly using a planar signature grid.

A signature $f$ of arity $n$ is \emph{degenerate} if there exist unary signatures $u_j \in \mathbb{C}^\kappa$ ($1 \le j \le n$)
such that $f = u_1 \otimes \cdots \otimes u_n$.
A symmetric degenerate signature has the from $u^{\otimes n}$.
For such signatures, it is equivalent to replace it by $n$ copies of the corresponding unary signature.
Replacing a signature $f \in \mathcal{F}$ by a constant multiple $c f$,
where $c \ne 0$, does not change the complexity of $\Holant(\mathcal{F})$.
It introduces a global nonzero factor to $\Holant(\Omega; \mathcal{F})$.

We allow $\mathcal{F}$ to be an infinite set.
For $\Holant(\mathcal{F})$ to be tractable,
the problem must be computable in polynomial time even when the description of the signatures in the input $\Omega$ are included in the input size.
In contrast,
we say $\Holant(\mathcal{F})$ is $\SHARPP$-hard if there exists a finite subset of $\mathcal{F}$ for which the problem is $\SHARPP$-hard.
The same definitions apply for $\PlHolant(\mathcal{F})$ when $\Omega$ is a planar signature grid.
We say a signature set $\mathcal{F}$ is tractable (resp.~$\SHARPP$-hard)
if the corresponding counting problem $\Holant(\mathcal{F})$ is tractable (resp.~$\SHARPP$-hard).
We say $\mathcal{F}$ is tractable (resp.~$\SHARPP$-hard) for planar problems
if $\PlHolant(\mathcal{F})$ tractable (resp.~$\SHARPP$-hard).
Similarly for a signature $f$,
we say $f$ is tractable (resp.~$\SHARPP$-hard) if $\{f\}$ is.

We follow the usual conventions about polynomial time Turing reduction $\le_T$ and polynomial time Turing equivalence $\equiv_T$.
We use $I_n$ and $J_n$ to denote the $n$-by-$n$ identity matrix and $n$-by-$n$ matrix of all $1$'s respectively.

\subsection{Holographic Reduction}

To introduce the idea of holographic reductions, it is convenient to consider bipartite graphs.
For a general graph, we can always transform it into a bipartite graph while preserving the Holant value, as follows.
For each edge in the graph, we replace it by a path of length two.
(This operation is called the \emph{2-stretch} of the graph and yields the edge-vertex incidence graph.)
Each new vertex is assigned the binary \textsc{Equality} signature $=_2$.

We use $\holant{\mathcal{F}}{\mathcal{G}}$ to denote the Holant problem on bipartite graphs $H = (U,V,E)$,
where each vertex in $U$ or $V$ is assigned a signature in $\mathcal{F}$ or $\mathcal{G}$, respectively.
Signatures in $\mathcal{F}$ are considered as row vectors (or covariant tensors);
signatures in $\mathcal{G}$ are considered as column vectors (or contravariant tensors)~\cite{DP91}.
Similarly, $\plholant{\mathcal{F}}{\mathcal{G}}$ denotes the Holant problem using a planar bipartite signature grid.

For a $\kappa$-by-$\kappa$ matrix $T$ and a signature set $\mathcal{F}$,
define $T \mathcal{F} = \{g \mid \exists f \in \mathcal{F}$ of arity $n,~g = T^{\otimes n} f\}$, similarly for $\mathcal{F} T$.
Whenever we write $T^{\otimes n} f$ or $T \mathcal{F}$,
we view the signatures as column vectors;
similarly for $f T^{\otimes n} $ or $\mathcal{F} T$ as row vectors.

Let $T$ be an invertible $\kappa$-by-$\kappa$ matrix.
The holographic transformation defined by $T$ is the following operation:
given a signature grid $\Omega = (H, \pi)$ of $\holant{\mathcal{F}}{\mathcal{G}}$,
for the same bipartite graph $H$,
we get a new grid $\Omega' = (H, \pi')$ of $\holant{\mathcal{F} T}{T^{-1} \mathcal{G}}$ by replacing each signature in
$\mathcal{F}$ or $\mathcal{G}$ with the corresponding signature in $\mathcal{F} T$ or $T^{-1} \mathcal{G}$.
Valiant's Holant Theorem~\cite{Val08} (see also~\cite{CC07a}) is easily generalized to domain size $\kappa \ge 3$.

\begin{theorem}
 Suppose $\kappa \ge 3$ is the domain size.
 If $T \in \mathbb{C}^{\kappa \times \kappa}$ is an invertible matrix,
 then $\Holant(\Omega; \mathcal{F} \mid \mathcal{G}) = \Holant(\Omega'; \mathcal{F} T \mid T^{-1} \mathcal{G})$.
\end{theorem}

Therefore, an invertible holographic transformation does not change the complexity of the Holant problem in the bipartite setting.
Furthermore, there is a special kind of holographic transformation, the orthogonal transformation,
that preserves the binary equality and thus can be used freely in the standard setting.
For $\kappa = 2$, this first appeared in~\cite{CLX09a} as Theorem~2.2.

\begin{theorem} \label{thm:ortho_holo_trans}
 Suppose $\kappa \ge 3$ is the domain size.
 If $T \in \mathbb{C}^{\kappa \times \kappa}$ is an orthogonal matrix (i.e.~$T \transpose{T} = I_\kappa$),
 then $\Holant(\Omega; \mathcal{F}) = \Holant(\Omega'; T \mathcal{F})$.
\end{theorem}

Since the complexity of a signature is unchanged by a nonzero constant multiple,
we also call a transformation $T$ such that $T \transpose{T} = \lambda I$ for some $\lambda \neq 0$ an orthogonal transformation.
Such transformations do not change the complexity of a problem.

\subsection{Realization}

One basic notion used throughout the paper is realization.
We say a signature $f$ is \emph{realizable} or \emph{constructible} from a signature set $\mathcal{F}$
if there is a gadget with some dangling edges such that each vertex is assigned a signature from $\mathcal{F}$,
and the resulting graph, when viewed as a black-box signature with inputs on the dangling edges, is exactly $f$.
If $f$ is realizable from a set $\mathcal{F}$, then we can freely add $f$ into $\mathcal{F}$ preserving the complexity.

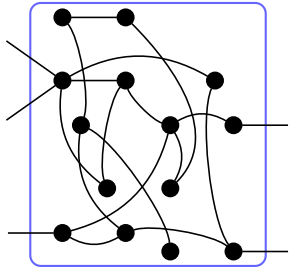
\begin{figure}[t]
 \centering
 \begin{tikzpicture}[scale=\scale,transform shape,node distance=\nodeDist,semithick]
  \node[external]  (0)                     {};
  \node[internal]  (1) [below right of=0]  {};
  \node[external]  (2) [below left  of=1]  {};
  \node[internal]  (3) [above       of=1]  {};
  \node[internal]  (4) [right       of=3]  {};
  \node[internal]  (5) [below       of=4]  {};
  \node[internal]  (6) [below right of=5]  {};
  \node[internal]  (7) [right       of=6]  {};
  \node[internal]  (8) [below       of=6]  {};
  \node[internal]  (9) [below       of=8]  {};
  \node[internal] (10) [right       of=9]  {};
  \node[internal] (11) [above right of=6]  {};
  \node[internal] (12) [below left  of=8]  {};
  \node[internal] (13) [left        of=8]  {};
  \node[internal] (14) [below left  of=13] {};
  \node[external] (15) [left        of=14] {};
  \node[internal] (16) [below left  of=5]  {};
  \path let
         \p1 = (15),
         \p2 = (0)
        in
         node[external] (17) at (\x1, \y2) {};
  \path let
         \p1 = (15),
         \p2 = (2)
        in
         node[external] (18) at (\x1, \y2) {};
  \node[external] (19) [right of=7]  {};
  \node[external] (20) [right of=10] {};
  \path (1) edge                             (5)
            edge[bend left]                 (11)
            edge[bend right]                (13)
            edge node[near start] (e1) {}   (17)
            edge node[near start] (e2) {}   (18)
        (3) edge                             (4)
        (4) edge[out=-45,in=45]              (8)
        (5) edge[bend right, looseness=0.5] (13)
            edge[bend right, looseness=0.5]  (6)
        (6) edge[bend left]                  (8)
            edge[bend left]                  (7)
            edge[bend left]                 (14)
        (7) edge node[near start] (e3) {}   (19)
       (10) edge[bend right, looseness=0.5] (12)
            edge[bend left,  looseness=0.5] (11)
            edge node[near start] (e4) {}   (20)
       (12) edge[bend left]                 (16)
       (14) edge node[near start] (e5) {}   (15)
            edge[bend right]                (12)
       (16) edge[bend left,  looseness=0.5]  (9)
            edge[bend right, looseness=0.5]  (3);
  \begin{pgfonlayer}{background}
   \node[draw=\borderColor,thick,rounded corners,fit = (3) (4) (9) (e1) (e2) (e3) (e4) (e5)] {};
  \end{pgfonlayer}
 \end{tikzpicture}
 \caption{An $\mathcal{F}$-gate with 5 dangling edges.}
 \label{fig:Fgate}
\end{figure}

Formally, such a notion is defined by an $\mathcal{F}$-gate~\cite{CLX09a, CLX10}.
An $\mathcal{F}$-gate is similar to a signature grid $(G, \pi)$ for $\Holant(\mathcal{F})$ except that $G = (V,E,D)$ is a graph with some dangling edges $D$.
The dangling edges define external variables for the $\mathcal{F}$-gate.
(See Figure~\ref{fig:Fgate} for an example.)
We denote the regular edges in $E$ by $1, 2, \dotsc, m$, and denote the dangling edges in $D$ by $m+1, \dotsc, m+n$.
Then we can define a function $\Gamma$ for this $\mathcal{F}$-gate as
\[
 \Gamma(y_1, y_2, \dotsc, y_n) = \sum_{x_1, x_2, \dotsc, x_m \in [\kappa]} H(x_1, x_2, \dotsc, x_m, y_1, y_2, \dotsc, y_n),
\]
where $(y_1, \dotsc, y_n) \in [\kappa]^n$ denotes a labeling on the dangling edges
and $H(x_1, \dotsc, x_m,$ $y_1, \dotsc, y_n)$ denotes the value of the signature grid on a labeling of all edges in $G$,
which is the product of evaluations at all internal vertices.
We also call this function $\Gamma$ the signature of the $\mathcal{F}$-gate.

An $\mathcal{F}$-gate is planar if the underlying graph $G$ is a planar graph,
and the dangling edges,
ordered counterclockwise corresponding to the order of the input variables,
are in the outer face in a planar embedding.
A planar $\mathcal{F}$-gate can be used in a planar signature grid as if it is just a single vertex with the particular signature.

Using the idea of planar $\mathcal{F}$-gates,
we can reduce one planar Holant problem to another.
Suppose $g$ is the signature of some planar $\mathcal{F}$-gate.
Then $\PlHolant(\mathcal{F} \cup \{g\}) \leq_T \PlHolant(\mathcal{F})$.
The reduction is simple.
Given an instance of $\PlHolant(\mathcal{F} \cup \{g\})$,
by replacing every appearance of $g$ by the $\mathcal{F}$-gate,
we get an instance of $\PlHolant(\mathcal{F})$.
Since the signature of the $\mathcal{F}$-gate is $g$,
the Holant values for these two signature grids are identical.

Our main results are about symmetric signatures
(i.e.~signatures that are invariant under any permutation of inputs).
However, we also need some asymmetric signatures in our proofs.
When a gadget has an asymmetric signature,
we place a diamond on the edge corresponding to the first input.
The remaining inputs are ordered counterclockwise around the vertex.
(See Figure~\ref{fig:gadget:k=r:arity_reduction} for an example.)

We note that even for a very simple signature set $\mathcal{F}$,
the signatures for all $\mathcal{F}$-gates can be quite complicated and expressive.

\subsection{Succinct Signatures} \label{subsec:succinct_signature}

An arity $r$ signature on domain size $\kappa$ is fully specified by $\kappa^r$ values.
However, some special cases can be defined using far fewer values.
Consider the signature $\AllDistinct_{r,\kappa}$ of arity $r$ on domain size $\kappa$ that outputs~$1$ when all inputs are distinct and~$0$ otherwise.
We also denote this signature by $\AD_{r,\kappa}$.
In addition to being symmetric,
it is also invariant under any permutation of the $\kappa$ domain elements.
We call the second property \emph{domain invariance}.
The signature of an $\mathcal{F}$-gate in which all signatures in $\mathcal{F}$ are domain invariant is itself domain invariant.

\begin{definition}[Succinct signature] \label{def:succinct_signature}
 Let $\tau = (P_1, P_2, \dotsc, P_\ell)$ be a partition of $[\kappa]^r$ listed in some order.
 We say that $f$ is a \emph{succinct signature} of type $\tau$ if $f$ is constant on each $P_i$.
 A set $\mathcal{F}$ of signatures is of type $\tau$ if every $f \in \mathcal{F}$ has type $\tau$.
 We denote a succinct signature $f$ of type $\tau$ by $\langle f(P_1), \dotsc, f(P_\ell) \rangle$,
 where $f(P) = f(x)$ for any $x \in P$.
 
 Furthermore, we may omit~$0$ entries.
 If $f$ is a succinct signature of type $\tau$,
 we also say $f$ is a \emph{succinct signature} of type $\tau'$ with length $\ell'$,
 where $\tau'$ lists $\ell'$ parts of the partition $\tau$
 and we write $f$ as $\langle f_1, f_2, \dotsc, f_{\ell'} \rangle$,
 provided all nonzero values $f(P_i)$ are listed.
 When using this notation,
 we will make it clear which zero entries have been omitted.
\end{definition}

For example,
a symmetric signature in the Boolean domain (i.e.~$\kappa = 2$) has been denoted in previous work~\cite{CGW13} by $[f_0, f_1, \dotsc, f_r]$,
where $f_w$ is the output on inputs of Hamming weight $w$.
This corresponds to the succinct signature type $(P_0, P_1, \dotsc, P_r)$,
where $P_w$ is the set of inputs of Hamming weight $w$.
A similar succinct signature notation was used for symmetric signatures on domain size~$3$~\cite[p.~1282]{CLX13}.

We prove a dichotomy theorem for $\PlHolant(f)$ when $f$ is a succinct ternary signature of type $\tau_3$ on domain size $\kappa \ge 3$.
For $\kappa \ge 3$,
the succinct signature of type $\tau_3 = (P_1, P_2, P_3)$ is a partition of $[\kappa]^3$ with $P_i = \{(x,y,z) \in [\kappa]^3 : |\{x, y, z\}| = i\}$ for $1 \le i \le 3$.
The notation $\{x, y, z\}$ denotes a multiset and $|\{x, y, z\}|$ denotes the number of distinct elements in it.
Succinct signatures of type $\tau_3$ are exactly the symmetric and domain invariant ternary signatures.
In particular,
the succinct ternary signature for $\AD_{3,\kappa}$ is $\langle 0,0,1 \rangle$.

We use several other succinct signature types as well.
For domain invariant unary signatures,
there are only two signatures up to a nonzero scalar.
Using the trivial partition that contains all inputs,
we denote these two succinct unary signatures as $\langle 0 \rangle$ and $\langle 1 \rangle$ and say that they have succinct type $\tau_1$.
We also need a succinct signature type for domain invariant binary signatures.
Such signatures are necessarily symmetric.
We call their succinct signature type $\tau_2 = (P_1, P_2)$,
where $P_i = \{(x,y) \in [\kappa]^2 : |\{x, y\}| = i\}$ for $1 \le i \le 2$.

We note that the number of succinct signature types for arity $r$ signatures on domain size $\kappa$
that are both symmetric and domain invariant is the number of partitions of $r$ into at most $\kappa$ parts.
This is related to the partition function from number theory,
which is not to be confused with the partition function with its origins in statistical mechanics
and has been intensively studied in complexity theory of counting problems.

While there are some other succinct signature types that we define later as needed,
there is one more important type that we define here.
Any quaternary signature $f$ that is domain invariant has a succinct signature of length at most~$15$.
When a signature has both vertical and horizontal symmetry,
there is a shorter succinct signature that has only length~$9$.
We say a signature $f$ has vertical symmetry if $f(w,x,y,z) = f(x,w,z,y)$ and horizontal symmetry if $f(w,x,y,z) = f(z,y,x,w)$.
For example,
the signature of the gadget in Figure~\ref{fig:gadget:k>r:fischer} has both vertical and horizontal symmetry.
Accordingly, let
$\tau_4 = (
P_{\subMat{1}{1}{1}{1}},
P_{\subMat{1}{1}{1}{2}},
P_{\subMat{1}{1}{2}{2}},
P_{\subMat{1}{1}{2}{3}},
P_{\subMat{1}{2}{1}{2}},
P_{\subMat{1}{2}{1}{3}},
P_{\subMat{1}{2}{2}{1}},
P_{\subMat{1}{2}{3}{1}},
P_{\subMat{1}{2}{3}{4}})$
be a type of succinct quaternary signature with partitions
\begin{align*}
 P_{\subMat{1}{1}{1}{1}} &= \{(w,x,y,z) \in [\kappa]^4 \st w = x = y = z\},\\
 P_{\subMat{1}{1}{1}{2}} &= \left\{(w,x,y,z) \in [\kappa]^4 \middle|
 \begin{array}{r}
  (w = x = y \ne z) \lor (w = x = z \ne y)\\
  {} \lor (w = y = z \ne x) \lor (x = y = z \ne w)
 \end{array}
 \right\},\\
 P_{\subMat{1}{1}{2}{2}} &= \{(w,x,y,z) \in [\kappa]^4 \st w = x \ne y = z\},\\
 P_{\subMat{1}{1}{2}{3}} &= \{(w,x,y,z) \in [\kappa]^4 \st (w = x \ne y \ne z \ne x) \lor (y = z \ne w \ne x \ne z)\},\\
 P_{\subMat{1}{2}{1}{2}} &= \{(w,x,y,z) \in [\kappa]^4 \st w = y \ne x = z\},\\
 P_{\subMat{1}{2}{1}{3}} &= \{(w,x,y,z) \in [\kappa]^4 \st (w = y \ne x \ne z \ne y) \lor (x = z \ne w \ne y \ne z)\},\\
 P_{\subMat{1}{2}{2}{1}} &= \{(w,x,y,z) \in [\kappa]^4 \st w = z \ne x = y\},\\
 P_{\subMat{1}{2}{3}{1}} &= \{(w,x,y,z) \in [\kappa]^4 \st (w = z \ne x \ne y \ne z) \lor (x = y \ne w \ne z \ne y)\}, \text{ and}\\
 P_{\subMat{1}{2}{3}{4}} &= \{(w,x,y,z) \in [\kappa]^4 \st w, x, y, z \text{ are all distinct}\}.
\end{align*}

\section{Counting Edge \texorpdfstring{$\kappa$}{kappa}-Colorings over Planar \texorpdfstring{$r$}{r}-Regular Graphs} \label{sec:coloring}

In this section,
we show that counting edge $\kappa$-colorings over planar $r$-regular graphs is $\SHARPP$-hard provided $\kappa \ge r \ge 3$.
When this condition fails to hold,
the problem is trivially tractable.
There are two cases depending on whether $\kappa = r$ or not.

\subsection{The Case \texorpdfstring{$\kappa = r$}{kappa-equal-r}}

When $\kappa = r$,
we reduce from evaluating the Tutte polynomial of a planar graph at the positive integer points on the diagonal $x = y$.
For $x \ge 3$, evaluating the Tutte polynomial of a planar graph at $(x,x)$ is $\SHARPP$-hard.

\begin{theorem}[Theorem~5.1 in~\cite{Ver05}] \label{thm:tutte}
 For $x, y \in \mathbb{C}$,
 evaluating the Tutte polynomial at $(x,y)$ is $\SHARPP$-hard over planar graphs
 unless $(x - 1) (y - 1) \in \{1, 2\}$ or $(x,y) \in \{(1,1), (-1, -1), (\omega, \omega^2), (\omega^2, \omega)\}$,
 where $\omega = e^{2 \pi i / 3}$.
 In each exceptional case,
 the computation can be done in polynomial time.
\end{theorem}

To state the connection with the diagonal of the Tutte polynomial,
we need to consider Eulerian subgraphs in directed medial graphs.
We say a graph is Eulerian (di)graph if every vertex has even degree (resp.~in-degree equal to out-degree),
but connectedness is not required.
Now recall the definition of a medial graph and its directed variant.

\begin{definition}[cf. Section~4 in~\cite{Ell04}]
 For a connected plane graph $G$ (i.e.~a planar embedding of a connected planar graph),
 its \emph{medial graph} $G_m$ has a vertex on each edge of $G$
 and two vertices in $G_m$ are joined by an edge for each face of $G$ in which their corresponding edges occur consecutively.
 
 The \emph{directed medial graph} $\vec{G}_m$ of $G$ colors the faces of $G_m$ black or white depending on whether they contain or do not contain, respectively,
 a vertex of $G$.
 Then the edges of the medial graph are directed so that the black face is on the left.
\end{definition}

\begin{figure}[t]
 \centering
 \def\medialNodeDist{2.5cm}
 \tikzstyle{open}   = [draw, black, fill=white, shape=circle]
 \tikzstyle{closed} = [draw,        fill,       shape=circle]
 \subcaptionbox{\label{subfig:planar_graph}}{
  \begin{tikzpicture}[scale=\scale,transform shape,node distance=\medialNodeDist,semithick]
   \node[closed] (0)              {};
   \node[closed] (1) [right of=0] {};
   \node[closed] (2) [above of=0] {};
   \node[closed] (3) [above of=1] {};
   \node[closed] (4) [above of=2] {};
   \path (0) edge[out=-45, in=-135]               node[external] (m0) {} (1)
             edge[out= 45, in= 135]               node[external] (m1) {} (1)
             edge                                 node[external] (m2) {} (2)
         (1) edge                                 node[external] (m3) {} (3)
         (2) edge                                 node[external] (m4) {} (3)
             edge                                 node[external] (m5) {} (4)
         (3) edge[out=125, in=  55, looseness=30] node[external] (m6) {} (3);
   \path (m0) edge[white, densely dashed, out= 135, in=-135]                (m1)
              edge[white, densely dashed, out=  45, in= -45]                (m1)
              edge[white, densely dashed, out=-145, in=-135, looseness=1.7] (m2)
              edge[white, densely dashed, out= -35, in= -45, looseness=1.7] (m3)
         (m1) edge[white, densely dashed]                                   (m2)
              edge[white, densely dashed]                                   (m3)
         (m2) edge[white, densely dashed]                                   (m4)
              edge[white, densely dashed, out= 135, in=-135]                (m5)
         (m3) edge[white, densely dashed]                                   (m4)
              edge[white, densely dashed, out=  45, in=  15]                (m6)
         (m4) edge[white, densely dashed]                                   (m5)
              edge[white, densely dashed, out=  90, in= 165]                (m6)
         (m5) edge[white, densely dashed, out= 125, in=  55, looseness=30]  (m5)
         (m6) edge[white, densely dashed, out=-125, in= -55, looseness=15]  (m6);
  \end{tikzpicture}}
 \qquad
 \qquad
 \subcaptionbox{\label{subfig:superimposed}}{
  \begin{tikzpicture}[scale=\scale,transform shape,node distance=\medialNodeDist,semithick]
   \node[closed] (0)              {};
   \node[closed] (1) [right of=0] {};
   \node[closed] (2) [above of=0] {};
   \node[closed] (3) [above of=1] {};
   \node[closed] (4) [above of=2] {};
   \path (0) edge[out=-45, in=-135]               node[open] (m0) {} (1)
             edge[out= 45, in= 135]               node[open] (m1) {} (1)
             edge                                 node[open] (m2) {} (2)
         (1) edge                                 node[open] (m3) {} (3)
         (2) edge                                 node[open] (m4) {} (3)
             edge                                 node[open] (m5) {} (4)
         (3) edge[out=125, in=  55, looseness=30] node[open] (m6) {} (3);
   \path (m0) edge[densely dashed, out= 135, in=-135]                (m1)
              edge[densely dashed, out=  45, in= -45]                (m1)
              edge[densely dashed, out=-145, in=-135, looseness=1.7] (m2)
              edge[densely dashed, out= -35, in= -45, looseness=1.7] (m3)
         (m1) edge[densely dashed]                                   (m2)
              edge[densely dashed]                                   (m3)
         (m2) edge[densely dashed]                                   (m4)
              edge[densely dashed, out= 135, in=-135]                (m5)
         (m3) edge[densely dashed]                                   (m4)
              edge[densely dashed, out=  45, in=  15]                (m6)
         (m4) edge[densely dashed]                                   (m5)
              edge[densely dashed, out=  90, in= 165]                (m6)
         (m5) edge[densely dashed, out= 125, in=  55, looseness=30]  (m5)
         (m6) edge[densely dashed, out=-125, in= -55, looseness=15]  (m6);
  \end{tikzpicture}}
 \qquad
 \qquad
 \subcaptionbox{\label{subfig:medial_graph}}{
  \begin{tikzpicture}[scale=\scale,transform shape,node distance=\medialNodeDist,semithick]
   \node[external] (0)              {};
   \node[external] (1) [right of=0] {};
   \node[external] (2) [above of=0] {};
   \node[external] (3) [above of=1] {};
   \node[external] (4) [above of=2] {};
   \path (0) edge[white, out=-45, in=-135]               node[open] (m0) {} (1)
             edge[white, out= 45, in= 135]               node[open] (m1) {} (1)
             edge[white]                                 node[open] (m2) {} (2)
         (1) edge[white]                                 node[open] (m3) {} (3)
         (2) edge[white]                                 node[open] (m4) {} (3)
             edge[white]                                 node[open] (m5) {} (4)
         (3) edge[white, out=125, in=  55, looseness=30] node[open] (m6) {} (3);
   \path (m0) edge[densely dashed, out= 135, in=-135]                (m1)
              edge[densely dashed, out=  45, in= -45]                (m1)
              edge[densely dashed, out=-145, in=-135, looseness=1.7] (m2)
              edge[densely dashed, out= -35, in= -45, looseness=1.7] (m3)
         (m1) edge[densely dashed]                                   (m2)
              edge[densely dashed]                                   (m3)
         (m2) edge[densely dashed]                                   (m4)
              edge[densely dashed, out= 135, in=-135]                (m5)
         (m3) edge[densely dashed]                                   (m4)
              edge[densely dashed, out=  45, in=  15]                (m6)
         (m4) edge[densely dashed]                                   (m5)
              edge[densely dashed, out=  90, in= 165]                (m6)
         (m5) edge[densely dashed, out= 125, in=  55, looseness=30]  (m5)
         (m6) edge[densely dashed, out=-125, in= -55, looseness=15]  (m6);
  \end{tikzpicture}}
 \caption{A plane graph~(\protect\subref{subfig:planar_graph}), its medial graph~(\protect\subref{subfig:medial_graph}),
 and the two graphs superimposed~(\protect\subref{subfig:superimposed}).}
 \label{fig:medial_graph_example}
\end{figure}
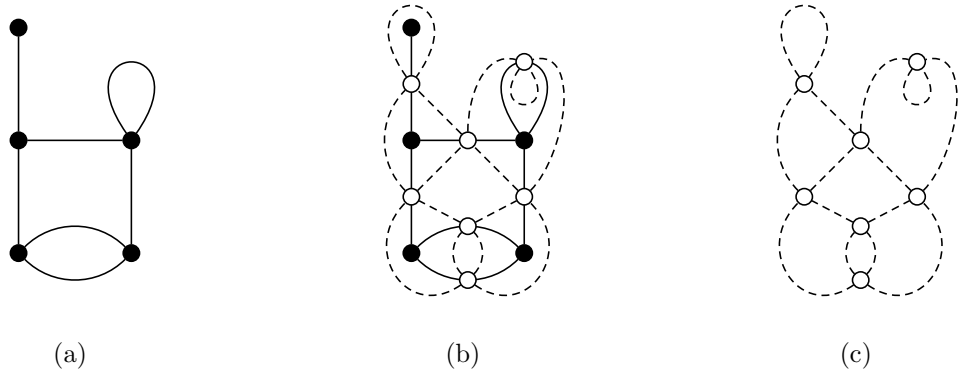

\begin{figure}[t]
 \centering
 \def\medialNodeDist{2.5cm}
 \tikzstyle{open}   = [draw, black, fill=white, shape=circle]
 \tikzstyle{closed} = [draw,        fill,       shape=circle]
 \tikzstyle{invisibleVertex} = [shape=circle]
 \tikzstyle{invisibleEdge} = [draw opacity=0]
 \subcaptionbox{\label{subfig:planar_graph2}}{
  \begin{tikzpicture}[scale=\scale,transform shape,node distance=\medialNodeDist,>=\arrowType,semithick]
   \node[closed] (0)              {};
   \node[closed] (1) [right of=0] {};
   \node[closed] (2) [below of=1] {};
   \node[closed] (3) [left  of=2] {};
   \path (0) edge node[invisibleVertex] (m0) {} (1)
         (1) edge node[invisibleVertex] (m1) {} (2)
         (2) edge node[invisibleVertex] (m2) {} (3)
         (3) edge node[invisibleVertex] (m3) {} (0);
   \path (m0) edge[invisibleEdge, ->, densely dashed]                                                                        (m1)
              edge[invisibleEdge, <-, densely dashed, out=  45, in=  45, looseness=3, overlay] node[invisibleVertex] (e0) {} (m1)
         (m1) edge[invisibleEdge, ->, densely dashed]                                                                        (m2)
              edge[invisibleEdge, <-, densely dashed, out= -45, in= -45, looseness=3, overlay] node[invisibleVertex] (e1) {} (m2)
         (m2) edge[invisibleEdge, ->, densely dashed]                                                                        (m3)
              edge[invisibleEdge, <-, densely dashed, out=-135, in=-135, looseness=3, overlay] node[invisibleVertex] (e2) {} (m3)
         (m3) edge[invisibleEdge, ->, densely dashed]                                                                        (m0)
              edge[invisibleEdge, <-, densely dashed, out= 135, in= 135, looseness=3, overlay] node[invisibleVertex] (e3) {} (m0);
   \node[invisibleVertex, below=0cm of e0] {};
   \node[invisibleVertex, left =0cm of e1] {};
   \node[invisibleVertex, above=0cm of e2] {};
   \node[invisibleVertex, right=0cm of e3] {};
  \end{tikzpicture}}
 \qquad
 \qquad
 \subcaptionbox{\label{subfig:superimposed2}}{
  \begin{tikzpicture}[scale=\scale,transform shape,node distance=\medialNodeDist,>=\arrowType,semithick]
   \node[closed] (0)              {};
   \node[closed] (1) [right of=0] {};
   \node[closed] (2) [below of=1] {};
   \node[closed] (3) [left  of=2] {};
   \path (0) edge node[open] (m0) {} (1)
         (1) edge node[open] (m1) {} (2)
         (2) edge node[open] (m2) {} (3)
         (3) edge node[open] (m3) {} (0);
   \path (m0) edge[->, densely dashed]                                                                        (m1)
              edge[<-, densely dashed, out=  45, in=  45, looseness=3, overlay] node[invisibleVertex] (e0) {} (m1)
         (m1) edge[->, densely dashed]                                                                        (m2)
              edge[<-, densely dashed, out= -45, in= -45, looseness=3, overlay] node[invisibleVertex] (e1) {} (m2)
         (m2) edge[->, densely dashed]                                                                        (m3)
              edge[<-, densely dashed, out=-135, in=-135, looseness=3, overlay] node[invisibleVertex] (e2) {} (m3)
         (m3) edge[->, densely dashed]                                                                        (m0)
              edge[<-, densely dashed, out= 135, in= 135, looseness=3, overlay] node[invisibleVertex] (e3) {} (m0);
   \node[invisibleVertex, below=0cm of e0] {};
   \node[invisibleVertex, left =0cm of e1] {};
   \node[invisibleVertex, above=0cm of e2] {};
   \node[invisibleVertex, right=0cm of e3] {};
   \begin{scope}[on background layer, overlay]
    \fill[fill=black!20] (m0.  45) to [out=  45, in=  45, looseness=3] (m1.  45) to (m1.center) to (m0.center) to (m0.  45);
    \fill[fill=black!20] (m1. -45) to [out= -45, in= -45, looseness=3] (m2. -45) to (m2.center) to (m1.center) to (m1. -45);
    \fill[fill=black!20] (m2.-135) to [out=-135, in=-135, looseness=3] (m3.-135) to (m3.center) to (m2.center) to (m2.-135);
    \fill[fill=black!20] (m3. 135) to [out= 135, in= 135, looseness=3] (m0. 135) to (m0.center) to (m3.center) to (m3. 135);
   \end{scope}
  \end{tikzpicture}}
 \qquad
 \qquad
 \subcaptionbox{\label{subfig:directed_medial_graph}}{
  \begin{tikzpicture}[scale=\scale,transform shape,node distance=\medialNodeDist,>=\arrowType,semithick]
   \node[invisibleVertex] (0)              {};
   \node[invisibleVertex] (1) [right of=0] {};
   \node[invisibleVertex] (2) [below of=1] {};
   \node[invisibleVertex] (3) [left  of=2] {};
   \path (0) edge[invisibleEdge] node[draw opacity=100, open] (m0) {} (1)
         (1) edge[invisibleEdge] node[draw opacity=100, open] (m1) {} (2)
         (2) edge[invisibleEdge] node[draw opacity=100, open] (m2) {} (3)
         (3) edge[invisibleEdge] node[draw opacity=100, open] (m3) {} (0);
   \path (m0) edge[->, densely dashed]                                                                        (m1)
              edge[<-, densely dashed, out=  45, in=  45, looseness=3, overlay] node[invisibleVertex] (e0) {} (m1)
         (m1) edge[->, densely dashed]                                                                        (m2)
              edge[<-, densely dashed, out= -45, in= -45, looseness=3, overlay] node[invisibleVertex] (e1) {} (m2)
         (m2) edge[->, densely dashed]                                                                        (m3)
              edge[<-, densely dashed, out=-135, in=-135, looseness=3, overlay] node[invisibleVertex] (e2) {} (m3)
         (m3) edge[->, densely dashed]                                                                        (m0)
              edge[<-, densely dashed, out= 135, in= 135, looseness=3, overlay] node[invisibleVertex] (e3) {} (m0);
   \node[invisibleVertex, below=0cm of e0] {};
   \node[invisibleVertex, left =0cm of e1] {};
   \node[invisibleVertex, above=0cm of e2] {};
   \node[invisibleVertex, right=0cm of e3] {};
   \begin{scope}[on background layer, overlay]
    \fill[fill=black!20] (m0.  45) to [out=  45, in=  45, looseness=3] (m1.  45) to (m1.center) to (m0.center) to (m0.  45);
    \fill[fill=black!20] (m1. -45) to [out= -45, in= -45, looseness=3] (m2. -45) to (m2.center) to (m1.center) to (m1. -45);
    \fill[fill=black!20] (m2.-135) to [out=-135, in=-135, looseness=3] (m3.-135) to (m3.center) to (m2.center) to (m2.-135);
    \fill[fill=black!20] (m3. 135) to [out= 135, in= 135, looseness=3] (m0. 135) to (m0.center) to (m3.center) to (m3. 135);
   \end{scope}
  \end{tikzpicture}}
 \caption{A plane graph~(\protect\subref{subfig:planar_graph2}), its directed medial graph~(\protect\subref{subfig:directed_medial_graph}),
 and both superimposed~(\protect\subref{subfig:superimposed2}).}
 \label{fig:directed_medial_graph_example}
\end{figure}
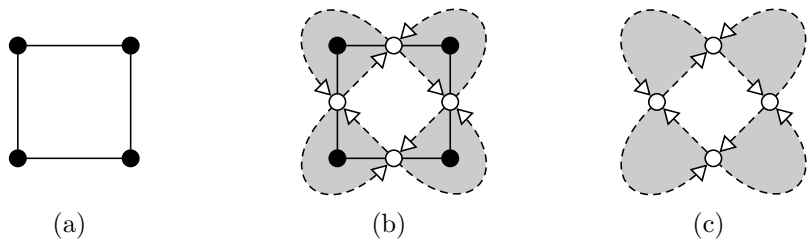

Figures~\ref{fig:medial_graph_example} and~\ref{fig:directed_medial_graph_example} give examples of a medial graph and a directed medial graph respectively.
Notice that the (directed) medial graph is always a planar $4$-regular graph.

The connection with the diagonal of the Tutte polynomial is due to Ellis-Monaghan.
A monochromatic vertex is a vertex with all its incident edges having the same color.

\begin{lemma}[Equation~(17) in~\cite{Ell04}] \label{lem:tutte_connection}
 Suppose $G$ is a connected plane graph and $\vec{G}_m$ is its directed medial graph.
 For $\kappa \in \N$,
 let $\mathcal{C}(\vec{G}_m)$ be the set of all edge $\kappa$-labelings of $\vec{G}_m$
 so that each (possibly empty) set of monochromatic edges forms an Eulerian digraph.
 Then
 \begin{equation}
  \kappa \operatorname{T}(G; \kappa+1, \kappa+1)
  = \sum_{c\; \in\; \mathcal{C}(\vec{G}_m)} 2^{m(c)}, \label{equ:Tutte}
 \end{equation}
 where $m(c)$ is the number of monochromatic vertices in the coloring $c$.
\end{lemma}

The Eulerian partitions in $\mathcal{C}(\vec{G}_m)$ have the property that the subgraphs induced by each partition do not intersect (or crossover) each other
due to the orientation of the edges in the medial graph.
We call the counting problem defined by the sum on the right-hand side of~(\ref{equ:Tutte})
counting weighted Eulerian partitions over planar $4$-regular graphs.
This problem also has an expression as a Holant problem using a succinct signature.
To define this succinct signature,
it helps to know the following basic result about edge colorings.

When the number of available colors coincides with the regularity parameter of the graph,
the cuts in any coloring satisfy a well-known parity condition.
This parity condition follows from a more general parity argument (see (1.2) and the Parity Argument on page~95 in~\cite{SSTF12}).
We state this simpler parity condition and provide a short proof for completeness.

\begin{lemma}[Parity Condition] \label{lem:k=r:parity_condition}
 Let $G$ be a $\kappa$-regular graph and consider a cut $C$ in $G$.
 For any edge $\kappa$-coloring of $G$,
 \[c_1 \equiv c_2 \equiv \dotsb \equiv c_\kappa \pmod{2},\]
 where $c_i$ is the number of edges in $C$ colored $i$ for $1 \le i \le \kappa$.
\end{lemma}

\begin{proof}
 Consider two distinct colors $i$ and $j$.
 Remove from $G$ all edges not colored $i$ or $j$.
 The resulting graph is a disjoint union of cycles consisting of alternating colors $i$ and $j$.
 Each cycle in this graph must cross the cut $C$ an even number of times.
 Therefore, $c_i \equiv c_j \pmod{2}$.
\end{proof}

Consider all quaternary $\{\AD_{\kappa,\kappa}\}$-gates on domain size $\kappa \ge 3$.
These gadgets have a succinct signature of type $\tau_\text{color} =
(
 P_{\subMat{1}{1}{1}{1}},
 P_{\subMat{1}{1}{2}{2}},
 P_{\subMat{1}{2}{1}{2}},
 P_{\subMat{1}{2}{2}{1}},
 P_{\subMat{1}{2}{3}{4}},
 P_0
)$,
where
\begin{align*}
 P_{\subMat{1}{1}{1}{1}} &= \{(w,x,y,z) \in [\kappa]^4 \st w = x = y = z\},\\
 P_{\subMat{1}{1}{2}{2}} &= \{(w,x,y,z) \in [\kappa]^4 \st w = x \ne y = z\},\\
 P_{\subMat{1}{2}{1}{2}} &= \{(w,x,y,z) \in [\kappa]^4 \st w = y \ne x = z\},\\
 P_{\subMat{1}{2}{2}{1}} &= \{(w,x,y,z) \in [\kappa]^4 \st w = z \ne x = y\},\\
 P_{\subMat{1}{2}{3}{4}} &= \{(w,x,y,z) \in [\kappa]^4 \st w, x, y, z \text{ are distinct}\}, \text{ and}\\
 P_0 &= [\kappa]^4
 - P_{\subMat{1}{1}{1}{1}}
 - P_{\subMat{1}{1}{2}{2}}
 - P_{\subMat{1}{2}{1}{2}}
 - P_{\subMat{1}{2}{2}{1}}
 - P_{\subMat{1}{2}{3}{4}}.
\end{align*}
Any quaternary signature of an $\{\AD_{\kappa,\kappa}\}$-gate is constant on the first five parts of $\tau_\text{color}$ since $\AD_{\kappa,\kappa}$ is domain invariant.
Using Lemma~\ref{lem:k=r:parity_condition},
we can show that the entry corresponding to $P_0$ is~$0$.

\begin{lemma} \label{lem:coloring:k=r:f(P_0)=0}
 Suppose $\kappa \ge 3$ is the domain size
 and let $F$ be a quaternary $\{\AD_{\kappa,\kappa}\}$-gate with succinct signature $f$ of type $\tau_\text{color}$.
 Then $f(P_0) = 0$.
\end{lemma}

\begin{proof}
 Let $\sigma_0 \in  P_0$ be an edge $\kappa$-labeling of the external edges of $F$.
 Assume for a contradiction that $\sigma_0$ can be extended to an edge $\kappa$-coloring of $F$.
 We form a graph $G$ from two copies of $F$ (namely, $F_1$ and $F_2$) by connecting their corresponding external edges.
 Then $G$ has a coloring $\sigma$ that extends $\sigma_0$.
 Consider the cut $C = (F_1, F_2)$ in $G$ whose cut set contains exactly those edges labeled by $\sigma_0$.
 By Lemma~\ref{lem:k=r:parity_condition},
 the counts of the colors assigned by $\sigma_0$ must satisfy the parity condition.
 However, this is a contradiction since no edge $\kappa$-labeling in $P_0$ satisfies the parity condition.
\end{proof}

By Lemma~\ref{lem:coloring:k=r:f(P_0)=0},
we denote a quaternary signature $f$ of an $\{\AD_{\kappa,\kappa}\}$-gate by the succinct signature
$\langle
 f(P_{\subMat{1}{1}{1}{1}}),
 f(P_{\subMat{1}{1}{2}{2}}),
 f(P_{\subMat{1}{2}{1}{2}}),
 f(P_{\subMat{1}{2}{2}{1}}),
 f(P_{\subMat{1}{2}{3}{4}})
\rangle$ of type $\tau_\text{color}$,
which has the entry for $P_0$ omitted.%
\footnote{If $\kappa > 4$,
then Lemma~\ref{lem:k=r:parity_condition} further implies that these signatures are also~0 on $P_{\subMat{1}{2}{3}{4}}$.
However, when $\kappa = 4$, this value might be nonzero. The $\AD_{4,4}$ signature is an example of this.
Instead of using this observation that depends on $\kappa$ in our proof,
we only construct gadgets such that their signatures are~$0$ on $P_{\subMat{1}{2}{3}{4}}$ for any value of $\kappa$.}
When $\kappa = 3$,
$P_{\subMat{1}{2}{3}{4}}$ is empty and we define its entry in the succinct signature to be~$0$.

\begin{lemma} \label{lem:holant_connection}
 Let $G$ be a connected plane graph and let $G_m$ be the medial graph of $G$.
 Then
 \[
  \kappa \operatorname{T}(G; \kappa+1, \kappa+1) = \PlHolant(G_m; \langle 2,1,0,1,0 \rangle),
 \]
 where the Holant problem has domain size $\kappa$ and $\langle 2,1,0,1,0 \rangle$ is a succinct signature of type $\tau_\text{color}$.
\end{lemma}

\begin{proof}
 Let $f = \langle 2,1,0,1,0 \rangle$.
 By Lemma~\ref{lem:tutte_connection}, we only need to prove that
 \begin{equation} \label{equ:coloring:Holant}
  \sum_{c\; \in\; \mathcal{C}(\vec{G}_m)} 2^{m(c)} = \PlHolant(G_m; f),
 \end{equation}
 where the notation is from Lemma~\ref{lem:tutte_connection}.

 Each $c \in \mathcal{C}(\vec{G}_m)$ is also an edge $\kappa$-labeling of $G_m$.
 At each vertex $v \in V(\vec{G}_m)$,
 the four incident edges are assigned at most two distinct colors by $c$.
 If all four edges are assigned the same color,
 then the constraint $f$ on $v$ contributes a factor of~$2$ to the total weight.
 This is given by the value in the first entry of $f$.
 Otherwise, there are two different colors, say $x$ and $y$.
 Because the orientation at $v$ in $\vec{G}_m$ is cyclically ``in, out, in, out'',
 the coloring around $v$ can only be of the form $xxyy$ or $xyyx$.
 These correspond to the second and fourth entries of $f$.
 Therefore, every term in the summation on the left-hand side of~$(\ref{equ:coloring:Holant})$ appears (with the same nonzero weight) in the summation $\PlHolant(G_m; f)$.

 It is also easy to see that every nonzero term in $\PlHolant(G_m; f)$ appears in the sum on the left-hand side of~(\ref{equ:coloring:Holant})
 with the same weight of~$2$ to the power of the number of monochromatic vertices.
 In particular,
 any coloring with a vertex that is cyclically colored $xyxy$ for different colors $x$ and $y$ does not contribute because $f(P_{\subMat{1}{2}{1}{2}}) = 0$.
\end{proof}

\begin{remark}
 This result shows that this planar Holant problem is at least as hard as computing the Tutte polynomial at the point $(\kappa+1, \kappa+1)$ over planar graphs,
 which implies $\SHARPP$-hardness.
 Of course they are equally difficult in the sense that both are $\SHARPP$-complete.
 In fact,
 they are more directly related since every $4$-regular plane graph is the medial graph of some plane graph.
\end{remark}

By Theorem~\ref{thm:tutte} and Lemma~\ref{lem:holant_connection},
the problem $\PlHolant(\langle 2,1,0,1,0 \rangle)$ is $\SHARPP$-hard.
We state this as a corollary.

\begin{corollary} \label{cor:k=r:21010_hard}
 Suppose $\kappa \ge 3$ is the domain size.
 Let $\langle 2,1,0,1,0 \rangle$ be a succinct quaternary signature of type $\tau_\text{color}$.
 Then $\PlHolant(\langle 2,1,0,1,0 \rangle)$ is $\SHARPP$-hard.
\end{corollary}

With this connection established,
we can now show that counting edge colorings is $\SHARPP$-hard over planar regular graphs when the number of colors and the regularity parameter coincide.

\begin{figure}[t]
 \centering
 \begin{tikzpicture}[scale=\scale,transform shape,node distance=\nodeDist,semithick]
  \node[external] (0)              {};
  \node[external] (1) [below of=0] {};
  \node[internal] (2) [right of=0] {};
  \node[internal] (3) [right of=1] {};
  \node[external] (4) [right of=2] {};
  \node[external] (5) [right of=3] {};
  \path (0.west) edge[postaction={decorate, decoration={
                                             markings,
                                             mark=at position 0.70 with {\arrow[>=diamond,white] {>}; },
                                             mark=at position 0.70 with {\arrow[>=open diamond]  {>}; } } }] (2)
        (1.west) edge             (3)
        (2)      edge[very thick] (3)
                 edge             (4.east)
        (3)      edge             (5.east);
  \begin{pgfonlayer}{background}
   \node[draw=\borderColor,thick,rounded corners,inner xsep=12pt,inner ysep=8pt,fit = (2) (3)] {};
  \end{pgfonlayer}
 \end{tikzpicture}
 \caption{Quaternary gadget used in the interpolation construction below.
 All vertices are assigned $\AD_{\kappa,\kappa}$.
 The bold edge represents $\kappa - 2$ parallel edges.}
 \label{fig:gadget:k=r:arity_reduction}
\end{figure}
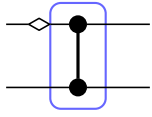

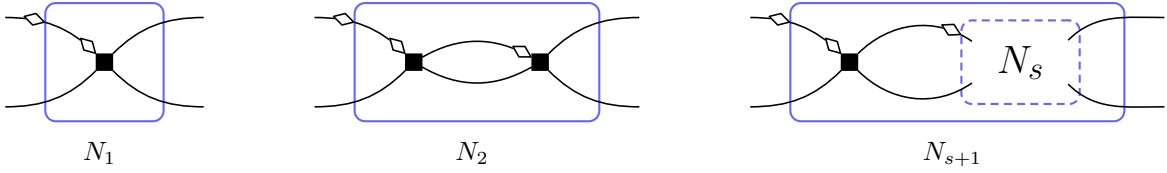
\begin{figure}[t]
 \centering
 \captionsetup[subfigure]{labelformat=empty}
 \subcaptionbox{$N_1$}{
  \begin{tikzpicture}[scale=\scale,transform shape,node distance=\nodeDist,semithick]
   \node[external] (0)                    {};
   \node[external] (1) [right       of=0] {};
   \node[square]   (2) [below right of=1] {};
   \node[external] (3) [below left  of=2] {};
   \node[external] (4) [left        of=3] {};
   \node[external] (5) [above right of=2] {};
   \node[external] (6) [right       of=5] {};
   \node[external] (7) [below right of=2] {};
   \node[external] (8) [right       of=7] {};
   \path (0) edge[out=   0, in=135, postaction={decorate, decoration={
                                                           markings,
                                                           mark=at position 0.4   with {\arrow[>=diamond, white] {>}; },
                                                           mark=at position 0.4   with {\arrow[>=open diamond]   {>}; },
                                                           mark=at position 0.999 with {\arrow[>=diamond, white] {>}; },
                                                           mark=at position 1.0   with {\arrow[>=open diamond]   {>}; } } }] (2)
         (2) edge[out=-135, in=  0] (4)
             edge[out=  45, in=180] (6)
             edge[out= -45, in=180] (8);
   \begin{pgfonlayer}{background}
    \node[draw=\borderColor,thick,rounded corners,fit = (1) (3) (5) (7)] {};
   \end{pgfonlayer}
  \end{tikzpicture}}
 \qquad
 \subcaptionbox{$N_2$}{
  \begin{tikzpicture}[scale=\scale,transform shape,node distance=\nodeDist,semithick]
   \node[external]  (0)                    {};
   \node[external]  (1) [right       of=0] {};
   \node[square]    (2) [below right of=1] {};
   \node[external]  (3) [below left  of=2] {};
   \node[external]  (4) [left        of=3] {};
   \node[external]  (5) [right       of=2] {};
   \node[square]    (6) [right       of=5] {};
   \node[external]  (7) [above right of=6] {};
   \node[external]  (8) [right       of=7] {};
   \node[external]  (9) [below right of=6] {};
   \node[external] (10) [right       of=9] {};
   \path (0) edge[out=   0, in=135, postaction={decorate, decoration={
                                                           markings,
                                                           mark=at position 0.4   with {\arrow[>=diamond, white] {>}; },
                                                           mark=at position 0.4   with {\arrow[>=open diamond]   {>}; },
                                                           mark=at position 0.999 with {\arrow[>=diamond, white] {>}; },
                                                           mark=at position 1.0   with {\arrow[>=open diamond]   {>}; } } }] (2)
         (2) edge[out=-135, in=  0]  (4)
             edge[bend left,        postaction={decorate, decoration={
                                                           markings,
                                                           mark=at position 0.999 with {\arrow[>=diamond, white] {>}; },
                                                           mark=at position 1.0   with {\arrow[>=open diamond]   {>}; } } }] (6)
             edge[bend right]        (6)
         (6) edge[out=  45, in=180]  (8)
             edge[out= -45, in=180] (10);
   \begin{pgfonlayer}{background}
    \node[draw=\borderColor,thick,rounded corners,fit = (1) (3) (7) (9)] {};
   \end{pgfonlayer}
  \end{tikzpicture}}
 \qquad
 \subcaptionbox{$N_{s+1}$}{
  \begin{tikzpicture}[scale=\scale,transform shape,node distance=\nodeDist,semithick]
   \node[external]  (0)                     {};
   \node[external]  (1) [above right of=0]  {};
   \node[external]  (2) [below right of=0]  {};
   \node[external]  (3) [below right of=1]  {};
   \node[external]  (4) [below right of=3]  {};
   \node[external]  (5) [above right of=3]  {};
   \node[external]  (6) [right       of=4]  {};
   \node[external]  (7) [right       of=5]  {};
   \node[external]  (8) [left        of=0]  {};
   \node[square]    (9) [left        of=8]  {};
   \node[external] (10) [above left  of=9]  {};
   \node[external] (11) [below left  of=9]  {};
   \node[external] (12) [left        of=10] {};
   \node[external] (13) [left        of=11] {};
   \path let
          \p1 = (1),
          \p2 = (2)
         in
          node[external] at (\x1, \y1 / 2 + \y2 / 2) {\Huge $N_s$};
   \path let
          \p1 = (0)
         in
          node[external] (14) at (\x1 + 2, \y1 + 10) {};
   \path let
          \p1 = (0)
         in
          node[external] (15) at (\x1 + 2, \y1 - 10) {};
   \path let
          \p1 = (3)
         in
          node[external] (16) at (\x1 - 2, \y1 + 10) {};
   \path let
          \p1 = (3)
         in
          node[external] (17) at (\x1 - 2, \y1 - 10) {};
   \path (12) edge[out=   0, in= 135, postaction={decorate, decoration={
                                                           markings,
                                                           mark=at position 0.4   with {\arrow[>=diamond, white] {>}; },
                                                           mark=at position 0.4   with {\arrow[>=open diamond]   {>}; },
                                                           mark=at position 0.999 with {\arrow[>=diamond, white] {>}; },
                                                           mark=at position 1.0   with {\arrow[>=open diamond]   {>}; } } }] (9)
          (6) edge[out= 180, in= -45] (17)
          (7) edge[out= 180, in=  45] (16)
          (9) edge[out=  45, in= 145, postaction={decorate, decoration={
                                                           markings,
                                                           mark=at position 0.9 with {\arrow[>=diamond, white] {>}; },
                                                           mark=at position 0.9   with {\arrow[>=open diamond]   {>}; } } }] (14)
              edge[out= -45, in=-145] (15)
              edge[out=-135, in=   0] (13);
   \begin{pgfonlayer}{background}
    \node[draw=\borderColor,thick,densely dashed,rounded corners,fit = (0) (1.south) (2.north) (3)] {};
    \node[draw=\borderColor,thick,rounded corners,fit = (4) (5) (10) (11)] {};
   \end{pgfonlayer}
  \end{tikzpicture}}
 \caption{Recursive construction to interpolate $\langle 2,1,0,1,0 \rangle$.
 The vertices are assigned the signature of the gadget in Figure~\ref{fig:gadget:k=r:arity_reduction}.}
 \label{fig:gadget:linear_interpolation:quaternary}
\end{figure}

\begin{theorem} \label{thm:edge_coloring:k=r}
 \#$\kappa$-\textsc{EdgeColoring} is $\SHARPP$-hard over planar $\kappa$-regular graphs for all $\kappa \ge 3$.
\end{theorem}

\begin{proof}
 Let $\kappa$ be the domain size of all Holant problems in this proof
 and let $\langle 2,1,0,1,0 \rangle$ be a succinct quaternary signature of type $\tau_\text{color}$.
 We reduce from $\PlHolant(\langle 2,1,0,1,0 \rangle)$ to $\PlHolant(\AD_{\kappa,\kappa})$,
 which denotes the problem of counting edge $\kappa$-colorings in planar $\kappa$-regular graphs as a Holant problem.
 Then by Corollary~\ref{cor:k=r:21010_hard},
 we conclude that $\PlHolant(\AD_{\kappa,\kappa})$ is $\SHARPP$-hard.

 Consider the gadget in Figure~\ref{fig:gadget:k=r:arity_reduction},
 where the bold edge represents $\kappa - 2$ parallel edges.
 We assign $\AD_{\kappa,\kappa}$ to both vertices.
 Up to a nonzero factor of $(\kappa-2)!$,
 this gadget has the succinct quaternary signature $f = \langle 0,1,1,0,0 \rangle$ of type $\tau_\text{color}$.
 Now consider the recursive construction in Figure~\ref{fig:gadget:linear_interpolation:quaternary}.
 All vertices are assigned the signature $f$.
 Let $f_s$ be the succinct quaternary signature of type $\tau_\text{color}$ for the $s$th gadget of the recursive construction.
 Then $f_1 = f$ and $f_s = M^s f_0$,
 where
 \[
  M =
  \begin{bmatrix}
   0 & \kappa - 1 & 0 & 0 & 0 \\
   1 & \kappa - 2 & 0 & 0 & 0 \\
   0 &          0 & 0 & 1 & 0 \\
   0 &          0 & 1 & 0 & 0 \\
   0 &          0 & 0 & 0 & 1
  \end{bmatrix}
  \qquad
  \text{and}
  \qquad
  f_0 =
  \begin{bmatrix}
   1 \\
   0 \\
   0 \\
   1 \\
   0
  \end{bmatrix}.
 \]
 The signature $f_0$ is simply the succinct quaternary signature of type $\tau_\text{color}$ for two parallel edges.
 We can express $M$ via the Jordan decomposition $M = P \Lambda P^{-1}$,
 where
 \[
  P =
  \begin{bmatrix}
   1 & 1 - \kappa & 0 &  0 & 0 \\
   1 &          1 & 0 &  0 & 0 \\
   0 &          0 & 1 &  1 & 0 \\
   0 &          0 & 1 & -1 & 0 \\
   0 &          0 & 0 &  0 & 1
  \end{bmatrix}
 \]
 and $\Lambda = \diag(\kappa - 1, -1, 1, -1, 1)$.
 Then for $t = 2 s$, we have
 \[
  f_t
  =
  P
  \begin{bmatrix}
   \kappa - 1 &  0 & 0 &  0 & 0 \\
            0 & -1 & 0 &  0 & 0 \\
            0 &  0 & 1 &  0 & 0 \\
            0 &  0 & 0 & -1 & 0 \\
            0 &  0 & 0 &  0 & 1
  \end{bmatrix}^t
  P^{-1}
  \begin{bmatrix}
   1 \\
   0 \\
   0 \\
   1 \\
   0
  \end{bmatrix}
  =
  P
  \begin{bmatrix}
   x & 0 & 0 & 0 & 0 \\
   0 & 1 & 0 & 0 & 0 \\
   0 & 0 & 1 & 0 & 0 \\
   0 & 0 & 0 & 1 & 0 \\
   0 & 0 & 0 & 0 & 1
  \end{bmatrix}
  P^{-1}
  \begin{bmatrix}
   1 \\
   0 \\
   0 \\
   1 \\
   0
  \end{bmatrix}
  =
  \begin{bmatrix}
   y+1 \\
   y \\
   0 \\
   1 \\
   0
  \end{bmatrix},
 \]
 where $x = (\kappa - 1)^t$ and $y = \frac{x-1}{\kappa}$.
 
 Consider an instance $\Omega$ of $\PlHolant(\langle 2,1,0,1,0 \rangle)$ on domain size $\kappa$.
 Suppose $\langle 2,1,0,1,0 \rangle$ appears $n$ times in $\Omega$.
 We construct from $\Omega$ a sequence of instances $\Omega_t$ of $\PlHolant(f)$ indexed by $t$,
 where $t = 2 s$ with $s \ge 0$.
 We obtain $\Omega_t$ from $\Omega$ by replacing each occurrence of $\langle 2,1,0,1,0 \rangle$ with the gadget $f_t$.
 
 As a polynomial in $x = (\kappa - 1)^t$, $\PlHolant(\Omega_t; f)$ is independent of $t$ and has degree at most $n$ with integer coefficients.
 Using our oracle for $\PlHolant(f)$,
 we can evaluate this polynomial at $n+1$ distinct points $x = (\kappa - 1)^{2s}$ for $0 \le s \le n$.
 Then via polynomial interpolation,
 we can recover the coefficients of this polynomial efficiently.
 Evaluating this polynomial at $x = \kappa + 1$ (so that $y = 1$) gives the value of $\PlHolant(\Omega; \langle 2,1,0,1,0 \rangle)$, as desired.
\end{proof}

\begin{remark}
 For $\kappa = 3$,
 the interpolation step is actually unnecessary since the succinct signature of $f_2$ happens to be $\langle 2,1,0,1,0 \rangle$.
\end{remark}

When $\kappa = 3$,
it is easy to extend Theorem~\ref{thm:edge_coloring:k=r} by further restricting to simple graphs,
i.e.~graphs without self-loops or parallel edges.

\begin{theorem} \label{thm:edge_coloring:k=r=3_simple}
 \#$3$-\textsc{EdgeColoring} is $\SHARPP$-hard over simple planar $3$-regular graphs.
\end{theorem}

\begin{proof}
 By Theorem~\ref{thm:edge_coloring:k=r},
 it suffices to efficiently compute the number of edge $3$-colorings of a planar $3$-regular graph $G$ that might have self-loops and parallel edges.
 Furthermore, we can assume that $G$ is connected since the number of edge colorings is multiplicative over connected components.
 If $G$ contains a self-loop,
 then there are no edge colorings in $G$,
 so assume $G$ contains no self-loops.
 If $G$ also contains no parallel edges,
 then $G$ is simple and we are done.
 
 Thus, assume that $G$ contains $n$ vertices and parallel edges between some distinct vertices $u$ and $v$.
 If $u$ and $v$ are connected by three edges,
 then this constitutes the whole graph,
 which has six edge $3$-colorings.
 Otherwise, $u$ and $v$ are connected by two edges.
 Then there exist vertices $u'$ and $v'$ such that
 $u$ and $u'$ are connected by a single edge,
 $v$ and $v'$ are connected by a single edge,
 and $u' \ne v'$.
 In any edge $3$-coloring of $G$,
 it is easy to see that the edges $(u,u')$ and $(v,v')$ must be assigned the same color.
 By removing $u$, $v$, and their incident edges while adding an edge between $u'$ and $v'$,
 we have a planar $3$-regular graph $G'$ on $n - 2$ vertices with half as many edge colorings as $G$.
 Then by induction,
 we can efficiently compute the number of edge $3$-colorings in $G'$.
\end{proof}

In Appendix~\ref{sec:appendix:invariant},
we give an alternative proof of Theorem~\ref{thm:edge_coloring:k=r},
which uses the interpolation techniques we develop in Section~\ref{sec:interpolation}.
The purpose of Appendix~\ref{sec:appendix:invariant} is to show how a recursive construction in an interpolation proof
can be used to form a hypothesis about possible invariance properties.
One example of an invariance property is that any planar $\{\AD_{\kappa,\kappa}\}$-gate
with a succinct quaternary signature $\langle a,b,c,d,e \rangle$ of type $\tau_\text{color}$ must satisfy $a + c = b + d$ (Lemma~\ref{lem:k=r:invariant}).

\subsection{The Case \texorpdfstring{$\kappa > r$}{kappa-greater-than-r}} \label{subsec:k>r}

Now we consider $\kappa > r \ge 3$.
This time, we reduce from the problem of counting vertex $\kappa$-colorings over planar graphs.
This problem is also $\SHARPP$-hard by the same dichotomy for the Tutte polynomial (Theorem~\ref{thm:tutte})
since the chromatic polynomial is a specialization the Tutte polynomial.

\begin{proposition}[Proposition~6.3.1 in~\cite{BO92}] \label{prop:k>r:chromatic_tutte}
 Let $G = (V,E)$ be a graph.
 Then $\chi(G; \lambda)$,
 the chromatic polynomial of $G$,
 is expressed as a specialization of the Tutte polynomial via the relation
 \[
  \chi(G; \lambda) = (-1)^{|V| - k(G)} \lambda^{k(G)} \operatorname{T}(G; 1 - \lambda, 0),
 \]
 where $k(G)$ is the number of connected components of the graph $G$.
\end{proposition}

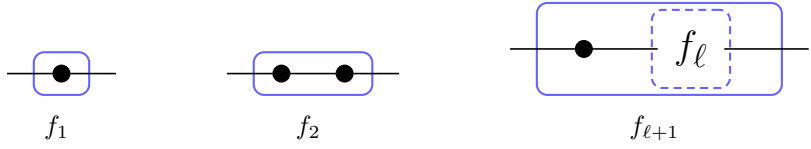
\begin{figure}[t]
 \centering
 \captionsetup[subfigure]{labelformat=empty}
 \subcaptionbox{$f_1$}{
  \begin{tikzpicture}[scale=\scale,transform shape,node distance=\nodeDist,semithick]
   \node[external] (0)              {};
   \node[internal] (1) [right of=0] {};
   \node[external] (2) [right of=1] {};
   \path (0) edge (1)
         (1) edge (2);
   \begin{pgfonlayer}{background}
    \node[draw=\borderColor,thick,rounded corners,inner xsep=12pt,inner ysep=8pt,fit = (1)] {};
   \end{pgfonlayer}
  \end{tikzpicture}}
 \qquad
 \subcaptionbox{$f_2$}{
  \begin{tikzpicture}[scale=\scale,transform shape,node distance=\nodeDist,semithick]
   \node[external] (0)              {};
   \node[internal] (1) [right of=0] {};
   \node[internal] (2) [right of=1] {};
   \node[external] (3) [right of=2] {};
   \path (0) edge (1)
         (1) edge (2)
         (2) edge (3);
   \begin{pgfonlayer}{background}
    \node[draw=\borderColor,thick,rounded corners,inner xsep=12pt,inner ysep=8pt,fit = (1) (2)] {};
   \end{pgfonlayer}
  \end{tikzpicture}}
 \qquad
 \subcaptionbox{$f_{\ell+1}$}{
  \begin{tikzpicture}[scale=\scale,transform shape,node distance=\nodeDist,semithick]
   \node[external] (0)              {};
   \node[external] (1) [right of=0] {};
   \node[external] (2) [right of=1] {};
   \node[external] (3) [right of=2] {\Huge $f_\ell$};
   \node[external] (4) [right of=3] {};
   \node[external] (5) [right of=4] {};
   \path (0) edge node[internal] (e1) {} (3)
         (3) edge (5);
   \begin{pgfonlayer}{background}
    \node[draw=\borderColor,thick,densely dashed,rounded corners,fit = (3)] {};
    \node[draw=\borderColor,thick,rounded corners,inner xsep=12pt,inner ysep=8pt,fit = (1) (3) (4)] {};
   \end{pgfonlayer}
  \end{tikzpicture}}
 \caption{Recursive construction to interpolate any succinct binary signature of type $\tau_2$.
 All vertices are assigned the same succinct binary signature of type $\tau_2$.}
 \label{fig:gadget:k>r:binary:interpolation}
\end{figure}

The first step in the proof is to interpolate every possible binary signature that is domain invariant,
which are necessarily symmetric.
These signatures have the succinct signature type $\tau_2$.

\begin{lemma} \label{lem:k>r:binary:interpolation:eigenvalues}
 Suppose $\kappa \ge 3$ is the domain size and let $x,y \in \mathbb{C}$.
 If we assign the succinct binary signature $\langle x,y \rangle$ of type $\tau_2$
 to every vertex of the recursive construction in Figure~\ref{fig:gadget:k>r:binary:interpolation},
 then the corresponding recurrence matrix is $\left[\begin{smallmatrix} x & (\kappa - 1) y \\ y & x + (\kappa - 2) y \end{smallmatrix}\right]$
 with eigenvalues $x + (\kappa - 1) y$ and $x - y$.
\end{lemma}

\begin{proof}
 Let $f_\ell$ be the signature of the $\ell$th gadget in this construction.
 To obtain $f_{\ell+1}$ from $f_\ell$,
 we view $f_\ell$ as a column vector and multiply it by the recurrence matrix $M = \left[\begin{smallmatrix} x & (\kappa - 1) y \\ y & x + (\kappa - 2) y \end{smallmatrix}\right]$.
 In general, we have $f_\ell = M^\ell f_0$,
 where $f_0$ is the initial signature,
 which is just a single edge and has the succinct binary signature $\langle 1,0 \rangle$ of type $\tau_2$.
 The (column) eigenvectors of $M$ are $\left[\begin{smallmatrix} 1 \\ 1 \end{smallmatrix}\right]$
 and $\left[\begin{smallmatrix} 1 - \kappa \\ 1 \end{smallmatrix}\right]$ with eigenvalues $x + (\kappa - 1) y$ and $x - y$ respectively,
 as claimed.
\end{proof}

The success of interpolation depends on these eigenvalues and the relationship between the recurrence matrix and the initial signature of the construction.
To show that the interpolation succeeds,
we use a result from~\cite{GW12},
the full version of~\cite{GW13}.
This result is about interpolating unary signatures on a Boolean domain for planar Holant problems,
but the same proof applies equally well for higher domains using a binary recursive construction (like that in Figure~\ref{fig:gadget:k>r:binary:interpolation})
and a succinct signature type with length~$2$.

\begin{lemma}[Lemma~4.4 in~\cite{GW12}] \label{lem:k>r:binary:interpolate}
 Suppose $\mathcal{F}$ is a set of signatures and $\tau$ is a succinct signature type with length~$2$.
 If there exists an infinite sequence of planar $\mathcal{F}$-gates defined by an initial succinct signature $s \in \mathbb{C}^{2 \times 1}$ of type $\tau$
 and recurrence matrix $M \in \mathbb{C}^{2 \times 2}$ satisfying the following conditions,
 \begin{enumerate}
  \item $\det(M) \ne 0$;
  \item $\det([s\ M s]) \ne 0$;
  \item $M$ has infinite order modulo a scalar;
 \end{enumerate}
 then
 \[
  \PlHolant(\mathcal{F} \union \{\langle x,y \rangle\}) \le_T \PlHolant(\mathcal{F}),
 \]
 for any $x,y \in \mathbb{C}$,
 where $\langle x,y \rangle$ is a succinct binary signature of type $\tau$.
\end{lemma}

Consider the recursive construction in Figure~\ref{fig:gadget:k>r:binary:interpolation}.
To every vertex,
we assign the succinct binary signature $\langle x,y \rangle$.
Since the initial signature is $s = \langle 1,0 \rangle$,
the determinant of the matrix $[s\ M s]$ is simply $y$.
In order to interpolate all binary succinct signatures of type $\tau_2$,
we need to satisfy the second condition of Lemma~\ref{lem:k>r:binary:interpolate},
which is $y \ne 0$.
However, when $y = 0$, the recurrence matrix is a scalar multiple of the identity matrix,
which implies that the eigenvalues are the same.
For two dimensional interpolation using a matrix with a full basis of eigenvectors, as is the case here,
the third condition of Lemma~\ref{lem:k>r:binary:interpolate} is equivalent to the ratio of the eigenvalues not being a root of unity.
In particular,
the eigenvalues cannot be the same.
Therefore, when using the recursive construction in Figure~\ref{fig:gadget:k>r:binary:interpolation},
it suffices to satisfy the first and third conditions of Lemma~\ref{lem:k>r:binary:interpolate}.
We state this as a corollary.

\begin{corollary} \label{cor:k>r:binary:interpolate}
 Suppose $\mathcal{F}$ is a set of signatures.
 Let $s = \langle 1,0 \rangle$ of type $\tau_2$ be the initial succinct binary signature
 and let $M \in \mathbb{C}^{2 \times 2}$ be the recurrence matrix
 for some infinite sequence of planar $\mathcal{F}$-gates defined by the recursive construction in Figure~\ref{fig:gadget:k>r:binary:interpolation}.
 If $M$ satisfies the following conditions,
 \begin{enumerate}
  \item $\det(M) \ne 0$;
  \item $M$ has infinite order modulo a scalar;
 \end{enumerate}
 then
 \[
  \PlHolant(\mathcal{F} \union \{\langle x,y \rangle\}) \le_T \PlHolant(\mathcal{F}),
 \]
 for any $x,y \in \mathbb{C}$,
 where $\langle x,y \rangle$ is a succinct binary signature of type $\tau_2$.
\end{corollary}

\begin{figure}[t]
 \centering
 \begin{tikzpicture}[scale=\scale,transform shape,node distance=\nodeDist,semithick]
  \node[external] (0)              {};
  \node[internal] (1) [right of=0] {};
  \node[internal] (2) [right of=1] {};
  \node[external] (3) [right of=2] {};
  \path (0) edge              (1)
        (1) edge[ultra thick] (2)
        (2) edge              (3);
  \begin{pgfonlayer}{background}
   \node[draw=\borderColor,thick,rounded corners,inner xsep=12pt,inner ysep=8pt,fit = (1) (2)] {};
  \end{pgfonlayer}
 \end{tikzpicture}
 \caption{Binary gadget used in the interpolation construction of Figure~\ref{fig:gadget:k>r:binary:interpolation}.
 Both vertices are assigned $\AD_{r,\kappa}$ and the bold edge represents $r - 1$ parallel edges.}
 \label{fig:gadget:k>r:binary}
\end{figure}
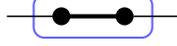

\begin{lemma} \label{lem:k>r:edge_coloring:binary:interpolation}
 Suppose $\kappa$ is the domain size with $\kappa > r$ for any integer $r \ge 3$, and $x,y \in \mathbb{C}$.
 Let $\mathcal{F}$ be a signature set containing $\AD_{r,\kappa}$.
 Then
 \[
  \PlHolant(\mathcal{F} \union \{\langle x,y \rangle\}) \le_T \PlHolant(\mathcal{F}),
 \]
 where $\langle x,y \rangle$ is a succinct binary signature of type $\tau_2$.
\end{lemma}

\begin{proof}
 Let $(n)_k = n (n-1) \dotsm (n-k+1)$ be the $k$th falling power of $n$.
 Consider the gadget in Figure~\ref{fig:gadget:k>r:binary}.
 We assign $\AD_{r,\kappa}$ to both vertices.
 The succinct binary signature of type $\tau_2$ for this gadget is
 $f = \langle (\kappa - 1)_{r - 1}, (\kappa - 2)_{r - 1} \rangle$.
 Up to a nonzero factor of $(\kappa - 2)_{r - 2}$,
 we have the signature $f' = \frac{1}{(\kappa - 2)_{r - 2}} f = \langle \kappa - 1, \kappa - r \rangle$.
 
 Consider the recursive construction in Figure~\ref{fig:gadget:k>r:binary:interpolation}.
 We assign $f'$ to all vertices.
 By Lemma~\ref{lem:k>r:binary:interpolation:eigenvalues},
 the eigenvalues of corresponding recurrence matrix are $(r - 1) > 0$ and $(\kappa - 1) (\kappa - r + 1) > 0$.
 Thus, $M$ is nonsingular.
 Furthermore, the eigenvalues are not equal since $\kappa \not\in \{0,r\}$.
 Therefore,
 we are done by Corollary~\ref{cor:k>r:binary:interpolate}.
\end{proof}

Next we show that $\PlHolant(\AD_{r,\kappa})$ is at least as hard as $\PlHolant(\AD_{3,\kappa})$.
To overcome a difficulty when $r$ is even,
we apply the following result,
which uses the notion of a planar pairing.

\begin{definition}[Definition~11 in~\cite{GW13}]
 A \emph{planar pairing} in a graph $G = (V, E)$ is a set of edges $P \subset V \times V$
 such that $P$ is a perfect matching in the graph $(V, V \times V)$,
 and the graph $(V, E \union P)$ is planar.
\end{definition}

\begin{lemma}[Lemma~12 in~\cite{GW13}] \label{lem:planar_pairing}
 For any planar 3-regular graph $G$,
 there exists a planar pairing that can be computed in polynomial time.
\end{lemma}

\begin{lemma} \label{lem:reduction_to_AD3}
 Suppose $\kappa$ is the domain size with $\kappa > r$ for any integer $r \ge 3$.
 Then
 \[
  \PlHolant(\AD_{3,\kappa}) \le_T \PlHolant(\AD_{r,\kappa}).
 \]
\end{lemma}

\begin{proof}
 By Lemma~\ref{lem:k>r:edge_coloring:binary:interpolation},
 we can assume that $\langle 1,1 \rangle$ is available.
 Take $\AD_{r,\kappa}$ and first form $t = \ceil{\frac{r-4}{2}}$ self-loops.
 Then add a new vertex on each self-loop and assign $\langle 1,1 \rangle$ to each of these new vertices.
 Up to a nonzero factor of $(\kappa - 3)_{2t}$,
 the resulting signature is $\AD_{3,\kappa}$ if $r$ is odd and $\AD_{4,\kappa}$ if $r$ is even.
 To reduce from $r = 3$ to $r = 4$,
 we use a planar pairing,
 which can be efficiently computed by Lemma~\ref{lem:planar_pairing}.
 We add a new vertex on each edge in a planar pairing and assign $\langle 1,1 \rangle$ to each of these new vertices.
 Then up to a nonzero factor of $\kappa - 3$,
 the signature at each vertex of the initial graph is effectively $\AD_{3,\kappa}$.
\end{proof}

The succinct binary signature $\langle 1 - \kappa, 1 \rangle$ of type $\tau_2$ has a special property.
Let $u$ be any constant unary signature,
which has a succinct signature of type $\tau_1$.
If $\langle 1 - \kappa, 1 \rangle$ is connected to $u$,
then the resulting unary signature is identically~$0$.

This observation is the key to what follows.
We use it in the next lemma to achieve what would appear to be an impossible task.
The requirements, if duly specified,
would result in multiple conditions to be satisfied by nine separate polynomials pertaining to some construction in place of the gadget in Figure~\ref{fig:gadget:k>r:fischer}.
And yet we are able to use just one degree of freedom to cause seven of the polynomials to vanish,
satisfying most of these conditions.
In addition, the other two polynomials are not forgotten.
They are nonzero and their ratio is not a root of unity,
which allows interpolation to succeed.

This ability to satisfy a multitude of constraints simultaneously in one magic stroke
reminds us of some unfathomably brilliant moves by Bobby Fischer,
the chess genius extraordinaire,
and so we name this gadget (Figure~\ref{fig:gadget:k>r:fischer}) the \emph{Bobby Fischer gadget}.

This gadget is the new idea that allows us to prove Theorem~\ref{thm:edge_coloring:k>r}.

\begin{figure}[t]
 \centering
 \begin{tikzpicture}[scale=\scale,transform shape,node distance=\nodeDist,semithick]
  \node[square]   (0)              {};
  \node[internal] (1) [above of=0] {};
  \node[internal] (2) [below of=0] {};
  \node[external] (3) [left  of=1] {};
  \node[external] (4) [left  of=2] {};
  \node[external] (5) [right of=1] {};
  \node[external] (6) [right of=2] {};
  \path (0) edge (1)
            edge (2)
        (1) edge[postaction={decorate, decoration={
                                             markings,
                                             mark=at position 0.70 with {\arrow[>=diamond,white] {>}; },
                                             mark=at position 0.70 with {\arrow[>=open diamond]  {>}; } } }] node[near start] (e1) {} (3)
            edge node[near start] (e2) {} (5)
        (2) edge node[near start] (e3) {} (4)
            edge node[near start] (e4) {} (6);
  \begin{pgfonlayer}{background}
   \node[draw=\borderColor,thick,rounded corners,inner xsep=8pt,inner ysep=8pt,fit = (0) (1) (2)] {};
  \end{pgfonlayer}
 \end{tikzpicture}
 \caption{The Bobby Fischer gadget, which achieves many objectives using only a single degree of freedom.}
 \label{fig:gadget:k>r:fischer}
\end{figure}
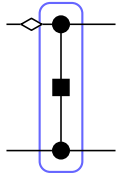

\begin{lemma} \label{lem:k>r:interpolate_equality4}
 Suppose $\kappa \ge 3$ is the domain size and $a, b \in \mathbb{C}$.
 Let $\mathcal{F}$ be a signature set containing the succinct ternary signature $\langle a,b,b \rangle$ of type $\tau_3$
 and the succinct binary signature $\langle 1 - \kappa, 1 \rangle$ of type $\tau_2$.
 If $a \ne b$,
 then
 \[
  \PlHolant(\mathcal{F} \union \{=_4\}) \le_T \PlHolant(\mathcal{F}).
 \]
\end{lemma}

\begin{proof}
 Consider the gadget in Figure~\ref{fig:gadget:k>r:fischer}.
 We assign $\langle a,b,b \rangle$ to the circle vertices and $\langle 1 - \kappa, 1 \rangle$ to the square vertex.
 This gadget has a succinct quaternary signature of type $\tau_4$,
 which has length~$9$.
 However, all but two of the entries in this succinct signature must be~$0$.
 
 To see this,
 consider an assignment that assigns different values to the two edges incident to the circle vertex on top.
 Since the assignment to these two edges differ,
 the signature $\langle a,b,b \rangle$ contributes a factor of $b$ regardless of the value of its third edge,
 which is connected to the square vertex assigned $\langle 1 - \kappa, 1 \rangle$.
 From the perspective of $\langle 1 - \kappa, 1 \rangle$,
 this behavior is equivalent to connecting it to the succinct unary signature $b \langle 1 \rangle$ of type $\tau_1$.
 Thus, the sum over the possible assignments to this third edge is~$0$.
 The same argument shows that the two edges incident to the circle vertex on the bottom
 do not contribute anything to the Holant sum when assigned different values.
 
 Thus, any nonzero contribution to the Holant sum comes from assignments
 where the top two dangling edges are assigned the same value and the bottom two dangling edges are assigned the same value.
 There are only two entries that satisfy this requirement in the succinct quaternary signature of type $\tau_4$ for this gadget,
 which are the entries for $P_{\subMat{1}{1}{1}{1}}$ and $P_{\subMat{1}{2}{2}{1}}$.
 To compute those two entries,
 we use the following trick.
 Since the two external edges of each circle vertex must be assigned the same value,
 we think of them as just a single edge.
 Then the effective succinct binary signature of type $\tau_2$ for the circle vertices is just $\langle a,b \rangle$.
 By connecting the first $\langle a,b \rangle$ with $\langle 1-\kappa, 1 \rangle$,
 the result is $(a-b) \langle 1-\kappa, 1 \rangle$ like it is an eigenvector.
 Connecting the other copy of $\langle a,b \rangle$ to $(a-b) \langle 1-\kappa, 1 \rangle$ gives $(a-b)^2 \langle 1-\kappa, 1 \rangle$.
 This computation is expressed via the matrix multiplication
 $[b J_\kappa + (a - b) I_\kappa] [J_\kappa - \kappa I_\kappa] [b J_\kappa + (a - b) I_\kappa]
 = (a - b) [J_\kappa - \kappa I_\kappa] [b J_\kappa + (a - b) I_\kappa]
 = (a - b)^2 [J_\kappa - \kappa I_\kappa]$.
 Thus up to a nonzero factor of $(a-b)^2$,
 the corresponding succinct quaternary signature of type $\tau_4$ for this gadget is $f = \langle 1-\kappa,0,0,0,0,0,1,0,0 \rangle$.
 
 Consider the recursive construction in Figure~\ref{fig:gadget:linear_interpolation:quaternary}.
 We assign $f$ to all vertices.
 Let $f_s$ be the signature of the $s$th gadget in this construction.
 The seven entries that are~$0$ in the succinct signature of type $\tau_4$ for $f$ are also~$0$ in the succinct signature of type $\tau_4$ for $f_s$.
 Thus, we can express $f_s$ via a succinct signature of type $\tau_4'$ with length~$2$, defined as follows.
 The first two parts in $\tau_4'$ are $P_{\subMat{1}{1}{1}{1}}$
 and $P_{\subMat{1}{2}{2}{1}}$ from the succinct signature type $\tau_4$.
 The last part contains all the remaining assignments.
 Then the succinct signature for $f_s$ of type $\tau_4'$ is $M^s f_0$,
 where $M = \left[\begin{smallmatrix} 1-\kappa & 0 \\ 0 & 1 \end{smallmatrix}\right]$ and $f_0 = \langle 1,1 \rangle$,
 which is just the succinct signature of type $\tau_4'$ for two parallel edges.
 
 Clearly the conditions in Lemma~\ref{lem:k>r:binary:interpolate} hold,
 so we can interpolate any succinct signature of type $\tau_4'$.
 In particular,
 we can interpolate our target signature $=_4$,
 which is $\langle 1,0 \rangle$ when expressed as a succinct signature of type $\tau_4'$.
\end{proof}

\begin{remark}
 The nine polynomials mentioned before Lemma~\ref{lem:k>r:interpolate_equality4}
 correspond to the nine entries of some quaternary gadget with a succinct signature of type $\tau_4$.
 In light of Lemma~\ref{lem:k>r:edge_coloring:binary:interpolation},
 this gadget might involve many succinct binary signatures $\langle x,y \rangle$ of type $\tau_2$ for various choices of $x,y \in \mathbb{C}$.
 Each distinct binary signature provides an additional degree of freedom to these polynomials.
 Our construction in Figure~\ref{fig:gadget:k>r:fischer} only requires one binary signature $\langle x,y \rangle$
 and we use our one degree of freedom to set $\frac{x}{y} = 1 - \kappa$.
\end{remark}

With the aid of the succinct unary signature $\langle 1 \rangle$ of type $\tau_1$
and the succinct binary signature $\langle 0,1 \rangle$ of type $\tau_2$,
the assumptions in the previous lemma are sufficient to prove $\SHARPP$-hardness.

\begin{corollary} \label{cor:k>r:abb_unary_binaries}
 Suppose $\kappa \ge 3$ is the domain size and $a,b \in \mathbb{C}$.
 Let $\mathcal{F}$ be a signature set containing the succinct ternary signature $\langle a, b, b \rangle$ of type $\tau_3$,
 the succinct unary signature $\langle 1 \rangle$ of type $\tau_1$,
 and the succinct binary signatures $\langle 1 - \kappa, 1 \rangle$ and $\langle 0,1 \rangle$ of type $\tau_2$.
 If $a \ne b$,
 then $\PlHolant(\mathcal{F})$ is $\SHARPP$-hard.
\end{corollary}

\begin{proof}
 By Lemma~\ref{lem:k>r:interpolate_equality4},
 we have $=_4$.
 Connecting $\langle 1 \rangle$ to $=_4$ gives $=_3$.
 With $=_3$,
 we can construct the equality signatures of every arity.
 Along with the binary disequality signature $\ne_2$,
 which is the succinct binary signature $\langle 0,1 \rangle$ of type $\tau_2$,
 we can express the problem of counting the number of vertex $\kappa$-colorings over planar graphs.
 By Proposition~\ref{prop:k>r:chromatic_tutte},
 this is, up to a nonzero factor, the problem of evaluating the Tutte polynomial at $(1 - \kappa, 0)$,
 which is $\SHARPP$-hard by Theorem~\ref{thm:tutte}.
\end{proof}

Now we can show that counting edge colorings is $\SHARPP$-hard over planar regular graphs when there are more colors than the regularity parameter.

\begin{figure}[t]
 \centering
 \begin{tikzpicture}[scale=\scale,transform shape,node distance=\nodeDist,semithick]
  \node [external] (0)              {};
  \node [square]   (1) [below of=0] {};
  \node [internal] (2) [below of=1] {};
  \path (2) ++(-150:\nodeDist) node [square] (3) {} ++(-150:\nodeDist) node [external] (4) {};
  \path (2) ++( -30:\nodeDist) node [square] (5) {} ++( -30:\nodeDist) node [external] (6) {};
  \path (0) edge (1)
        (1) edge (2)
        (2) edge (3)
            edge (5)
        (3) edge (4)
        (5) edge (6);
  \begin{pgfonlayer}{background}
   \node[draw=\borderColor,thick,rounded corners,inner xsep=12pt,inner ysep=12pt,fit = (1) (3) (5)] {};
  \end{pgfonlayer}
 \end{tikzpicture}
 \caption{Local holographic transformation gadget construction for a ternary signature.}
 \label{fig:gadget:k>r:local_holographic_transformation}
\end{figure}
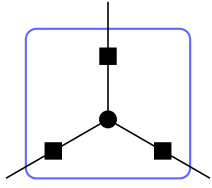

\begin{theorem} \label{thm:edge_coloring:k>r}
 \#$\kappa$-\textsc{EdgeColoring} is $\SHARPP$-hard over planar $r$-regular graphs for all $\kappa > r \ge 3$.
\end{theorem}

\begin{proof}
 By Lemma~\ref{lem:reduction_to_AD3},
 it suffices to consider $r = 3$.
 By Lemma~\ref{lem:k>r:edge_coloring:binary:interpolation},
 we can assume that any succinct binary signature of type $\tau_2$ is available.
 
 Consider the gadget in Figure~\ref{fig:gadget:k>r:local_holographic_transformation}.
 We assign $\AD_{3,\kappa}$ to the circle vertex and $\langle 3 - \kappa, 1 \rangle$ to the square vertices.
 By Lemma~\ref{lem:compute:ternary:holographic_transformation},
 the succinct ternary signature of type $\tau_3$ for this gadget is $f = 2 (\kappa-2) \langle -(\kappa-3)(\kappa-1), 1, 1 \rangle$.
 
 Now take two edges of $\AD_{3,\kappa}$ and connect them to the two edges of $\langle 1,1 \rangle$.
 Up to a nonzero factor of $(\kappa - 1) (\kappa - 2)$,
 this gadget has the succinct unary signature $\langle 1 \rangle$ of type $\tau_1$.
 Then we are done by Corollary~\ref{cor:k>r:abb_unary_binaries}.
\end{proof}

\section{Tractable Problems} \label{sec:tractable}

In the rest of the paper,
we adapt and extend our previous proof techniques to obtain a dichotomy for $\PlHolant(\langle a,b,c \rangle)$,
where $\langle a,b,c \rangle$ is a succinct ternary signature of type $\tau_3$ on domain size $\kappa \ge 3$,
for any $a,b,c \in \mathbb{C}$.
In this section,
we show how to compute a few of these problems in polynomial time.

\subsection{Previous Dichotomy Theorem}

There is only one previous dichotomy theorem for higher domain Holant problems.
It is a dichotomy for a single symmetric ternary signature on domain size $\kappa = 3$ in the framework of Holant$^*$ problems,
which means that all unary signatures are assumed to be freely available.

In Theorem~\ref{thm:tractable:holant-star},
the notation $f ^\frown g$ denotes the signature that results from
connecting one edge incident to a vertex assigned the signature $f$ to one edge incident to a vertex assigned the signature $g$.
When $f$ and $g$ are both unary signatures,
which are represented by vectors,
then the resulting $0$-ary signature is just a scalar.

\begin{theorem}[Theorem~3.1 in~\cite{CLX13}] \label{thm:tractable:holant-star}
 Let $f$ be a symmetric ternary signature on domain size~$3$.
 Then $\Holant^*(f)$ is $\SHARPP$-hard unless $f$ is of one of the following forms,
 in which case, the problem is computable in polynomial time.
 \begin{enumerate}
  \item There exists $\alpha, \beta, \gamma \in \mathbb{C}^3$ that are mutually orthogonal
  (i.e.~$\alpha ^\frown \beta = \alpha ^\frown \gamma = \beta ^\frown \gamma = 0$)
  and \[f = \alpha^{\otimes 3} + \beta^{\otimes 3} + \gamma^{\otimes 3};\]
  \item There exists $\alpha, \beta_1, \beta_2 \in \mathbb{C}^3$
  such that $\alpha ^\frown \beta_1 = \alpha ^\frown \beta_2 = \beta_1 ^\frown \beta_1 = \beta_2 ^\frown \beta_2 = 0$ and
  \[f = \alpha^{\otimes 3} + \beta_1^{\otimes 3} + \beta_2^{\otimes 3};\]
  \item There exists $\beta, \gamma \in \mathbb{C}^3$ and $f_\beta \in (\mathbb{C}^{3})^{\otimes 3}$
  such that $\beta \ne \mathbf{0}$, $\beta ^\frown \beta = 0$, $f_\beta ^\frown \beta = \mathbf{0}$,
  and \[f = f_\beta + \beta^{\otimes 2} \otimes \gamma + \beta \otimes \gamma \otimes \beta + \gamma \otimes \beta^{\otimes 2}.\]
 \end{enumerate}
\end{theorem}

Some domain invariant signatures are tractable by Theorem~\ref{thm:tractable:holant-star}.

\begin{corollary} \label{cor:tractable:holant-star}
 Suppose the domain size is~$3$ and $a,b,\lambda \in \mathbb{C}$.
 Let $f$ be a succinct ternary signature of type $\tau_3$.
 Then $\Holant(f)$ is computable in polynomial time when $f$ has one of the following forms:
 \begin{enumerate}
  \item $f = \hphantom{3} \lambda \langle 1, 0,  0 \rangle \hphantom{-}
  = \lambda \left[(1,0, 0)^{\otimes 3} + (0, 1,0)^{\otimes 3} + ( 0,0,1)^{\otimes 3}\right]$; \label{case:cor:tractable:holant-star:equality}
  \item $f =           3  \lambda \langle -5, -2, 4 \rangle
  = \lambda \left[(1,-2,-2)^{\otimes 3} + (-2,1,-2)^{\otimes 3} + (-2,-2,1)^{\otimes 3}\right]$; \label{case:cor:tractable:holant-star:-5-24}
  \item $f = \hphantom{3 \lambda} \langle a, b,  a \rangle \hphantom{-}
  = \hphantom{\lambda \big[}\frac{a + 2 b}{3}  (1,1, 1)^{\otimes 3} + \frac{a - b}{3} \left[(1, \omega, \omega^2)^{\otimes 3} + (1, \omega^2, \omega)^{\otimes 3}\right]$, \label{case:cor:tractable:holant-star:aba}
 \end{enumerate}
 where $\omega$ is a primitive third root of unity.
\end{corollary}

In Corollary~\ref{cor:tractable:holant-star},
form~\ref{case:cor:tractable:holant-star:equality} is the ternary equality signature $=_3$,
which is trivially tractable for any domain size.
Then form~\ref{case:cor:tractable:holant-star:-5-24} is just form~\ref{case:cor:tractable:holant-star:equality}
after a holographic transformation by the matrix $T = \left[\begin{smallmatrix*}[r] 1 & -2 & -2 \\ -2 & 1 & -2 \\ -2 & -2 & 1 \end{smallmatrix*}\right]$,
which is orthogonal after scaling by $\frac{1}{3}$.
This is an example of two problems that must have the same complexity by Theorem~\ref{thm:ortho_holo_trans}.

The tractability of these two problems for higher domain sizes is stated in the following corollary.

\begin{corollary} \label{cor:tractable:100}
 Suppose $\kappa \ge 3$ is the domain size and $\lambda \in \mathbb{C}$.
 Let $f$ be a succinct ternary signature of type $\tau_3$.
 Then $\Holant(f)$ is computable in polynomial time if $f$ has one of the following forms:
 \begin{enumerate}
  \item $f = \lambda               \langle 1,0,0 \rangle$; \label{case:cor:tractable:equality:equality}
  \item $f = \lambda T^{\otimes 3} \langle 1,0,0 \rangle = \lambda \kappa \langle \kappa^2 - 6 \kappa + 4, -2 (\kappa - 2), 4 \rangle$, \label{case:cor:tractable:equality:trans}
 \end{enumerate}
 where $T = \kappa I_\kappa - 2 J_\kappa$.
\end{corollary}

Note that $T = \kappa I_\kappa - 2 J_\kappa$ is an orthogonal matrix after scaling by $\frac{1}{\kappa}$.

\subsection{Affine Signatures}

Let $\omega$ be a primitive third root of unity.
Consider the ternary signature $f(x,y,z)$ with succinct ternary signature $\langle 1,0,c \rangle$ of type $\tau_3$ on domain size~$3$,
where $c^3 = 1$.
Its support is an affine subspace of $\mathbb{Z}_3$ defined by $x + y + z = 0$.
Furthermore, consider the quadratic polynomial $q_c(x, y, z) = \lambda_c (x y + x z + y z)$,
where $\lambda_1 = 0$, $\lambda_\omega = 2$, and $\lambda_{\omega^2} = 1$.
Then $\omega^{q_c(x, y, z)}$ agrees with $f$ when $x + y + z = 0$.
This function $f$ is an example of a ternary domain affine signature.

\begin{definition}
 A $k$-ary function $f(x_1, \dotsc, x_k)$ is \emph{affine} on domain size~$3$ if it has the form
 \[
  \lambda \cdot \chi_{A x = 0} \cdot e^{\frac{2 \pi i}{3} q(x)},
 \]
 where $\lambda \in \mathbb{C}$,
 $x = \transpose{(x_1, x_2, \dotsc, x_k, 1)}$,
 $A$ is a matrix over $\Z_3$,
 $q(x) \in \Z_3$ is a quadratic polynomial,
 and $\chi$ is a 0-1 indicator function such that $\chi_{A x = 0}$ is~$1$ iff $A x = 0$.
 We use $\mathscr{A}$ to denote the set of all affine functions.
\end{definition}

Like the Boolean domain affine signatures~\cite{CCLL10},
the ternary domain affine signatures are tractable.

\begin{lemma} \label{lem:tractable:k=3_affine}
 Suppose the domain size is~$3$.
 Then $\Holant(\mathscr{A})$ is computable in polynomial time.
\end{lemma}

\begin{proof}
 Given an instance of $\Holant(\mathscr{A})$,
 the output can be expressed as the summation of a single function $F = \chi_{Ax = 0} \cdot e^{\frac{2 \pi i}{3} q(x_1, x_2, \dotsc, x_k)}$
 since $\mathscr{A}$ is closed under multiplication.
 In polynomial time,
 we can solve the linear system $A x = 0$ over $\Z_3$ and decide if it is feasible.
 If the linear system is infeasible,
 then the function is the identically~$0$ function,
 so the output is just~$0$.
 
 Otherwise, the linear system is feasible (including possibly vacuous).
 Without loss of generality,
 we can assume that $y_1, y_2, \dotsc, y_s$ are independent variables over $\Z_3$ while all others are dependent variables,
 where $0 \le s \le k$.
 Each dependent variable can be expressed by an affine linear form of $y_1, y_2, \dotsc, y_s$.
 We can substitute for all of the dependent variables in $q(x_1, x_2, \dotsc, x_k)$,
 which gives a new quadratic polynomial $q'(y_1, y_2, \dotsc, y_s)$.
 Thus, we have
 \begin{equation} \label{eqn:tractable:affine:remove_support}
  \sum_{x_1, \dotsc, x_k \in \Z_3} \chi_{Ax = 0} \cdot e^{\frac{2 \pi i}{3} q (x_1, x_2, \dotsc, x_k)}
  =
  \sum_{y_1, \dotsc, y_s \in \Z_3}                     e^{\frac{2 \pi i}{3} q'(y_1, y_2, \dotsc, y_s)}.
 \end{equation}
 Then the right-hand side of~(\ref{eqn:tractable:affine:remove_support}) is computable in polynomial time by Theorem~1 in~\cite{CLX14}.
\end{proof}

After multiplying the function $\langle 1,0,c \rangle$ by a scalar,
we obtain the succinct ternary signature $\langle a,0,c \rangle$ of type $\tau_3$ such that $a^3 = c^3$.
Since undergoing an orthogonal transformation does not change the complexity of the problem by Theorem~\ref{thm:ortho_holo_trans},
we obtain the following corollary of the previous result.

\begin{corollary} \label{cor:tractable:k=3_a^3=c^3}
 Suppose the domain size is~$3$ and $a,c \in \mathbb{C}$.
 Let $T \in \mathbf{O}_3(\mathbb{C})$ and let $\langle a,0,c \rangle$ be a succinct ternary signature of type $\tau_3$.
 If $a^3 = c^3$,
 then $\Holant(T^{\otimes 3} \langle a,0,c \rangle)$ is computable in polynomial time.
\end{corollary}

For domain size~$3$,
the only orthogonal matrix $T$ such that $T^{\otimes 3} \langle a,b,c \rangle$
is still a succinct ternary signature of type $\tau_3$ is
$\pm \frac{1}{3} \left[\begin{smallmatrix*}[r] 1 & -2 & -2 \\ -2 & 1 & -2 \\ -2 & -2 & 1 \end{smallmatrix*}\right]$.
However, the tractability in Corollary~\ref{cor:tractable:k=3_a^3=c^3} holds for any orthogonal matrix $T$.

We introduce another affine signature.
It can be considered as a signature of arity~$4$ on the Boolean domain.
When placed in a planar signature grid,
its input variables are listed in a cyclic order $x_1, x_2, y_2, y_1$ counterclockwise.
We then consider it as a binary signature on domain size~$4$,
where the two variables $(x_1, x_2)$ and $(y_1, y_2)$ range over the four values in $\{0,1\}^2$.
Notice the reversal of the order $y_2, y_1$.
This is to allow a planar connection between these signatures.
We list its values as the matrix
$H_4 = \left[\begin{smallmatrix*}[r] 1 & -1 & -1 & -1 \\ -1 & 1 & -1 & -1 \\ -1 & -1 & 1 & -1 \\ -1 & -1 & -1 & 1 \end{smallmatrix*}\right]$,
which is an Hadamard matrix,
where the row index is $(x_1, x_2)$ and the column index is $(y_1, y_2)$,
both ordered lexicographically.
A closed form expression showing that this is an affine signature on the Boolean domain is
$f(x_1, x_2, y_2, y_1) = (-1)^{q(x_1, x_2, y_1, y_2)}$,
where $q$ is the quadratic polynomial
\begin{equation} \label{equ:prelim:affine:quad_poly}
 q(x_1, x_2, y_1, y_2) = x_1 + x_2 + x_1 x_2 + y_1 + y_2 + y_1 y_2 +x_1 y_2 + x_2 y_1.
\end{equation}

As a binary signature on domain size~$4$,
$f$ has the succinct signature $\langle 1, -1 \rangle$ of type $\tau_2$.
The fact that $f$ is an affine signature on the Boolean domain shows that the Holant problem defined by $f$ on domain size~$4$ is tractable.
This follows from Theorem~1.4 in~\cite{CLX14},
or the more general graph homomorphism dichotomy theorems~\cite{GGJT10,CCL13}.

We are interested in this problem because its tractability implies the tractability of a set of problems defined by a succinct ternary signature of type $\tau_3$.

\begin{lemma} \label{lem:tractable:u^21u}
 Suppose the domain size is~$4$ and $\lambda, \mu \in \mathbb{C}$.
 Let $\langle \mu^2, 1, \mu \rangle$ be a succinct ternary signature of type $\tau_3$.
 If $\mu = -1 + \varepsilon 2 i$ with $\varepsilon = \pm 1$,
 then $\Holant(\lambda \langle \mu^2, 1, \mu \rangle)$ is computable in polynomial time.
\end{lemma}

\begin{proof}
 Let $T = \left[\begin{smallmatrix*}[r] x & y & y & y \\ y & x & y & y \\ y & y & x & y \\ y & y & y & x \end{smallmatrix*}\right]$,
 where $x = -\frac{3 + \varepsilon i}{2}$ and $y = \frac{1 - \varepsilon i}{2}$.
 Then up to a factor of $\lambda^n$ on graphs with $n$ vertices,
 the output of $\Holant(\lambda \langle \mu^2, 1, \mu \rangle)$ is the same as the output for
 \begin{align*}
  \Holant(\langle \mu^2, 1, \mu \rangle)
  &= \Holant(\langle -3 - \varepsilon 4 i, 1, -1 + \varepsilon 2 i \rangle)\\
  &\equiv_T \holant{{=}_2}{T^{\otimes 3} ({=}_3)}\\
  &= \holant{({=}_2) T^{\otimes 2}}{{=}_3}\\
  &= \holant{2 \langle 1,-1 \rangle}{{=}_3}\\
  &\le_T \holant{\langle 1,-1 \rangle}{\{=_k \st k \in \Z^+\}}.
 \end{align*}
 Since $\holant{\langle 1,-1 \rangle}{\{=_k \st k \in \Z^+\}}$ is the Holant expression
 for the graph homomorphism problem defined by the Hadamard matrix
 $\left[\begin{smallmatrix*}[r] 1 & -1 & -1 & -1 \\ -1 & 1 & -1 & -1 \\ -1 & -1 & 1 & -1 \\ -1 & -1 & -1 & 1 \end{smallmatrix*}\right]$,
 we can finish the proof by applying the dichotomy theorems for symmetric matrices in~\cite{GGJT10,CCL13}.
 For example,
 this problem is tractable by Theorem~1.2 in~\cite{GGJT10} (see also~\cite{CLX14}),
 where the quadratic representation is $q(x_1, x_2, y_1, y_2)$ from~(\ref{equ:prelim:affine:quad_poly}).
\end{proof}

We restate this lemma as a simple corollary for later convenience.

\begin{corollary} \label{cor:tractable:u^21u_a+5b+2c=5b^2+2bc+c^2=0}
 Suppose the domain size is~$4$ and $a,b,c \in \mathbb{C}$.
 Let $\langle a,b,c \rangle$ be a succinct ternary signature of type $\tau_3$.
 If $a + 5 b + 2 c = 0$ and $5 b^2 + 2 b c + c^2 = 0$,
 then $\Holant(\langle a,b,c \rangle)$ is computable in polynomial time.
\end{corollary}

\begin{proof}
 Since $a = -5 b - 2 c$ and $b = \frac{1}{5} (-1 \pm 2 i) c$,
 after scaling by $\mu = -1 \mp 2 i$,
 we have $\mu \langle a,b,c \rangle = c \langle \mu^2, 1, \mu \rangle$
 and are done by Lemma~\ref{lem:tractable:u^21u}.
\end{proof}

\section{An Interpolation Result} \label{sec:interpolation}

The goal of this section is to generalize an interpolation result from~\cite{CLX12},
which we rephrase using our notion of a succinct signature~(cf.~Lemma~\ref{lem:k>r:binary:interpolate}).

\begin{theorem}[Theorem~3.5 in~\cite{CLX12}] \label{thm:interpolation:pervious}
 Suppose $\mathcal{F}$ is a set of signatures and $\tau$ is a succinct signature type with length~$3$.
 If there exists an infinite sequence of planar $\mathcal{F}$-gates defined by an initial succinct signature $s \in \mathbb{C}^{3 \times 1}$ of type $\tau$
 and a recurrence matrix $M \in \mathbb{C}^{3 \times 3}$ with eigenvalues $\alpha$, $\beta$, and $\gamma$ satisfying the following conditions:
 \begin{enumerate}
  \item $\det(M) \ne 0$;
  \item $s$ is not orthogonal to any row eigenvector of $M$;
  \item for all $(i,j,k) \in \Z^3 - \{(0,0,0)\}$ with $i+j+k=0$,
  we have $\alpha^i \beta^j \gamma^k \ne 1$;
 \end{enumerate}
 then
 \[
  \PlHolant(\mathcal{F} \union \{f\}) \le_T \PlHolant(\mathcal{F}),
 \]
 for any succinct ternary signature $f$ of type $\tau$.
\end{theorem}

Our generalization of this result is designed to relax the second condition so that $s$ can be orthogonal to some row eigenvectors of $M$.
Suppose $r$ is a row eigenvector of $M$, with eigenvalue $\lambda$, that is orthogonal to $s$ (i.e.~the dot product $r \cdot s$ is~$0$).
Consider $M^k s$, the $k$th signature in the infinite sequence defined by $M$ and $s$.
This signature is also orthogonal to $r$ since $r \cdot M^k s = \lambda^k r \cdot s = 0$.
We do not know of any way of interpolating a signature using this infinite sequence that is not also orthogonal to $r$.
On the other hand,
we would like to interpolate those signatures that do satisfy this orthogonality condition.
Our interpolation result gives a sufficient condition to achieve this.

We assume our $n$-by-$n$ matrix $M$ is diagonalizable,
i.e., it has $n$ linearly independent (row and column) eigenvectors.
We do not assume that $M$ necessarily has $n$ distinct eigenvalues (although this would be a sufficient condition for it to be diagonalizable).
The relaxation of the second condition is that, for some positive integer $\ell$,
the initial signature $s$ is \emph{not} orthogonal to exactly $\ell$ of these linearly independent row eigenvectors of $M$.
To satisfy this condition,
we use a two-step approach.
First, we explicitly exhibit $n - \ell$ linearly independent row eigenvectors of $M$ that are orthogonal to $s$.
Then we use the following lemma to show that the remaining row eigenvectors of $M$ are not orthogonal to $s$.
The justification for this approach is that the eigenvectors orthogonal to $s$ are often simple to express
while the eigenvectors not orthogonal to $s$ tend to be more complicated.

\begin{lemma} \label{lem:2nd_condition_implication}
 For $n \in \Z^+$,
 let $s \in \mathbb{C}^{n \times 1}$ be a vector and let $M \in \mathbb{C}^{n \times n}$ be a diagonalizable matrix.
 If $\rank([s\ M s\ \ldots\ M^{n-1} s]) \ge \ell$,
 then for any set of $n$ linearly independent row eigenvectors,
 $s$ is not orthogonal to at least $\ell$ of them.
\end{lemma}

\begin{proof}
 Since $M$ is diagonalizable,
 it has $n$ linearly independent eigenvectors.
 Suppose for a contradiction that there exists a set of $n$ linearly independent row eigenvectors of $M$
 such that $s$ is orthogonal to $t > n - \ell$ of them.
 Let $N \in \mathbb{C}^{t \times n}$ be the matrix whose $t$ rows are the row eigenvectors of $M$ that are orthogonal to $s$.
 Then $N [s\ M s\ \ldots\ M^{n-1} s]$ is the zero matrix.
 From this,
 it follows that $\rank([s\ M s\ \ldots\ M^{n-1} s]) < \ell$, a contradiction.
\end{proof}

The third condition of Theorem~\ref{thm:interpolation:pervious} is also known as the lattice condition.

\begin{definition} \label{def:interpolation:lattice_condition}
 Fix some $\ell \in \N$.
 We say that $\lambda_1, \lambda_2, \dotsc, \lambda_\ell \in \mathbb{C} - \{0\}$ satisfy the \emph{lattice condition}
 if for all $x \in \Z^\ell - \{\mathbf{0}\}$ with $\sum_{i=1}^\ell x_i = 0$,
 we have $\prod_{i=1}^\ell \lambda_i^{x_i} \ne 1$.
\end{definition}

When $\ell \ge 3$,
we use Galois theory to show that the lattice condition is satisfied.
The idea is that the lattice condition must hold if the Galois group of the polynomial,
whose roots are the $\lambda_i$'s,
is large enough.
In~\cite{CLX12}, for the special case $n = \ell = 3$,
it was shown that the roots of most cubic polynomials satisfy the lattice condition using this technique.

\begin{lemma}[Lemma~5.2 in~\cite{CLX12}] \label{lem:interpolation:lattice_condition:cubic}
 Let $f(x) \in \Q[x]$ be an irreducible cubic polynomial.
 Then the roots of $f(x)$ satisfy the lattice condition iff $f(x)$ is not of the form $a x^3 + b$ for some $a,b \in \Q$.
\end{lemma}

In the following lemma, 
we show that if the Galois group for a polynomial of degree $n$ is one of the two largest possible groups, $S_n$ or $A_n$,
then its roots satisfy the lattice condition provided these roots do not all have the same complex norm.

\begin{lemma} \label{lem:lattice_condition:Sn_An}
 Let $f$ be a polynomial of degree $n \ge 2$ with rational coefficients.
 If the Galois group of $f$ over $\Q$ is $S_n$ or $A_n$ and the roots of $f$ do not all have the same complex norm,
 then the roots of $f$ satisfy the lattice condition.
\end{lemma}

\begin{proof}
 We consider $A_n$ since the same argument applies to $S_n \supset A_n$.
 For $1 \le i \le n$,
 let $a_i$ be the roots of $f$ such that $|a_1| \le \dotsb \le |a_n|$.
 By assumption, as least one of these inequalities is strict.
 Suppose for a contradiction that these roots fail to satisfy the lattice condition.
 This means there exists $x \in \Z^n - \{\mathbf{0}\}$ satisfying $\sum_{i=1}^n x_i = 0$ such that
 \begin{equation} \label{eqn:lattice_condition_fails}
  a_1^{x_1} \dotsb a_n^{x_n} = 1.
 \end{equation}
 
 Since $x$ is not all~$0$,
 it must contain some positive entries and some negative entries.
 We can rewrite~(\ref{eqn:lattice_condition_fails}) as $b_1^{y_1} \dotsb b_s^{y_s} = c_1^{z_1} \dotsb c_t^{z_t}$,
 where $s, t \ge 1$, $b_1, \ldots, b_s, c_1, \ldots, c_t$ are $s+t$ distinct members from $\{a_1, \ldots, a_n\}$,
 $y_i > 0$ for $1 \le i \le s$, $z_i > 0$ for $1 \le i \le t$, and
 $y_1 + \dotsb + y_s = z_1 + \dotsb + z_t$.
 We omit factors in~(\ref{eqn:lattice_condition_fails}) with exponent~$0$.
 
 If $n=2$,
 then $s=t=1$ and $|b_1| = |c_1|$. 
 This is a contradiction to the assumption that roots of $f$ do not all have the same complex norm.
 Otherwise, assume $n \ge 3$.
 If $s = t = 1$,
 then $|b_1| = |c_1|$ again. 
 We apply 3-cycles from $A_n$ to conclude that all roots of $f$ have the same complex norm, a contradiction.
 Otherwise $s + t > 2$.
 Without loss of generality,
 suppose $s \ge t$,
 which implies $s \ge 2$.
 Pick $j \in \{0, \dotsc, n-s-t\}$ such that $|a_{j+1}| \le \dotsb \le |a_{j+s+t}|$ contains a strict inequality.
 We permute the roots so that $b_i = a_{j+i}$ for $1 \le i \le s$ and $c_i = a_{j+s+i}$ for $1 \le i \le t$
 (or possibly swapping $b_1$ and $b_2$ if necessary to ensure the permutation is in $A_n$).
 Then taking the complex norm of both sides gives a contradiction.
\end{proof}

\begin{remark}
 This result can simplify the interpolation arguments in~\cite{CLX12}.
 Since each of their cubic polynomials is irreducible,
 the corresponding Galois groups are transitive subgroups of $S_3$,
 namely $S_3$ or $A_3$ (and in fact  by inspection, they are all $S_3$).
 Then Lemma~4.5 from~\cite{KC10_arXiv} (the full version of~\cite{KC10}) shows that the eigenvalues of these polynomials do not all have the same complex norm.
 Thus,
 the roots of all polynomials exhibited in~\cite{CLX12} satisfy the lattice condition by Lemma~\ref{lem:lattice_condition:Sn_An}.
\end{remark}

In the current paper,
we apply Lemma~\ref{lem:lattice_condition:Sn_An} to an infinite family of quintic polynomials that we encounter in Section~\ref{sec:ternary}.
If the polynomials are irreducible,
then we are able to apply this lemma.
Unfortunately,
we are unable to show that all these polynomials are irreducible and thus also have to consider the possible ways that they could factor.
Nevertheless, we are still able to show that all these polynomials satisfy the lattice condition.

To conclude,
we state and prove our new interpolation result.

\begin{lemma} \label{lem:interpolate_all_not_orthogonal}
 Suppose $\mathcal{F}$ is a set of signatures and $\tau$ is a succinct signature type with length $n \in \Z^+$.
 If there exists an infinite sequence of planar $\mathcal{F}$-gates defined by an initial succinct signature $s \in \mathbb{C}^{n \times 1}$ of type $\tau$
 and a recurrence matrix $M \in \mathbb{C}^{n \times n}$ satisfying the following conditions,
 \begin{enumerate}
  \item $M$ is diagonalizable with $n$ linearly independent eigenvectors;
  \item $s$ is not orthogonal to exactly $\ell$ of these linearly independent row eigenvectors of $M$ with eigenvalues $\lambda_1, \dotsc, \lambda_\ell$;
  \item $\lambda_1, \dotsc, \lambda_\ell$ satisfy the lattice condition;
 \end{enumerate}
 then
 \[
  \PlHolant(\mathcal{F} \union \{f\}) \le_T \PlHolant(\mathcal{F})
 \]
 for any succinct signature $f$ of type $\tau$ that is orthogonal to the $n - \ell$ of these linearly independent eigenvectors of $M$ to which $s$ is also orthogonal.
\end{lemma}

\begin{proof}
 Let $\lambda_{1}, \dotsc, \lambda_n$ be the $n$ eigenvalues of $M$,
 with possible repetition.
 Since $M$ is diagonalizable,
 we can write $M$  as $T \Lambda T^{-1}$,
 where $\Lambda$ is the diagonal matrix $\left[\begin{smallmatrix} B_1 & \mathbf{0} \\ \mathbf{0} & B_2 \end{smallmatrix}\right]$
 with $B_1 = \diag(\lambda_1, \dotsc, \lambda_\ell)$ and
 $B_2 = \diag(\lambda_{\ell+1}, \dotsc, \lambda_n)$.
 Notice that the columns of $T$ are the column eigenvectors of $M$ and the rows of $T^{-1}$ are the row eigenvectors of $M$.
 Let $\mathbf{t}_i$ be the $i$th column $T$ and let $T^{-1} s = \transpose{[\alpha_1\ \ldots\ \alpha_n]}$.
 Then $\alpha_i \ne 0$ for $1 \le i \le \ell$ and $\alpha_i = 0$ for $\ell < i \le n$,
 since $s$ is not orthogonal to exactly the first $\ell$ row eigenvectors of $M$.
 
 Now we can write
 \begin{align*}
  M^k s
  &= T \begin{bmatrix} B_1^k & \mathbf{0} \\ \mathbf{0} & B_2^k \end{bmatrix} T^{-1} s
  = T \begin{bmatrix} B_1^k & \mathbf{0} \\ \mathbf{0} & B_2^k \end{bmatrix} \left[\begin{smallmatrix} \alpha_1 \\ \vdots \\ \alpha_\ell \\ 0 \\ \vdots \\ 0 \end{smallmatrix}\right]
  = T \diag(\alpha_1 \lambda_1^k, \dotsc, \alpha_\ell \lambda_\ell^k, 0, \dotsc, 0)\\
  &= T \diag(\alpha_1, \dotsc, \alpha_\ell, 0, \dotsc, 0) \left[\begin{smallmatrix} \lambda_1^k \\ \vdots \\ \lambda_\ell^k \\ 0 \\ \vdots \\ 0 \end{smallmatrix}\right]
  = [\alpha_1 \mathbf{t}_1, \dotsc, \alpha_\ell \mathbf{t}_\ell, \mathbf{0}, \dotsc, \mathbf{0}] \left[\begin{smallmatrix}  \lambda_1^k \\ \vdots \\ \lambda_\ell^k \\ 0 \\ \vdots \\ 0 \end{smallmatrix}\right].
 \end{align*}
 For $1 \le i \le \ell$,
 let $\mathbf{t}_i' = \alpha_i \mathbf{t}_i$.
 Both the columns of $T$ and the rows of $T^{-1}$ are linearly independent.
 From $T^{-1} T = I_m$,
 we see that $\mathbf{t}_i$ for $1 \le i \le \ell$ is orthogonal to the last $n - \ell$ rows of $T^{-1}$.
 Thus $\operatorname{span}\{\mathbf{t}_1, \dotsc, \mathbf{t}_\ell\} = \operatorname{span}\{\mathbf{t}_1', \dotsc, \mathbf{t}_\ell'\}$
 is precisely the space of vectors orthogonal to the last $n - \ell$ rows of $T^{-1}$.

 Consider an instance $\Omega$ of $\PlHolant(\mathcal{F} \union \{f\})$.
 Let $V_f$ be the subset of vertices assigned $f$ with $n_f = |V_f|$.
 Since $f$ is orthogonal to any row eigenvector of $M$ to which $s$ is also orthogonal,
 we have $f \in \operatorname{span}\{\mathbf{t}_1', \dotsc, \mathbf{t}_\ell'\}$.
 Let $f = \beta_1 \mathbf{t}_1' + \dotsb + \beta_\ell \mathbf{t}_\ell'$.
 Then $\PlHolant(\Omega; \mathcal{F} \union \{f\})$ is a homogeneous polynomial in the $\beta_i$'s of total degree $n_f$.
 For $y = (y_1, \ldots, y_\ell) \in \N^\ell$,
 let $c_y$ be the coefficient of $\beta_1^{y_1} \dotsb \beta_\ell^{y_\ell}$ in $\PlHolant(\Omega; \mathcal{F} \union \{f\})$ so that
 \[
  \PlHolant(\Omega; \mathcal{F} \union \{f\}) = \sum_{y_1 + \dotsb + y_\ell = n_f} c_y \beta_1^{y_1} \dotsb \beta_\ell^{y_\ell}.
 \]
 
 We construct from $\Omega$ a sequence of instances $\Omega_k$ of $\PlHolant(\mathcal{F})$ indexed by $k \in \N$.
 We obtain $\Omega_k$ from $\Omega$ by replacing each occurrence of $f$ with $M^k s$, for $k \ge 0$.
 Then
 \[
  \PlHolant(\Omega_k; \mathcal{F}) = \sum_{y_1 + \dotsb + y_\ell = n_f} c_y \left(\lambda_1^{y_1} \dotsb \lambda_\ell^{y_\ell}\right)^k.
 \]
 Note that, crucially, the same $c_y$ coefficients appear.
 We treat this as a linear system with the $c_y$'s as the unknowns.
 The coefficient matrix is  a Vandermonde matrix of order $\binom{n_f + \ell - 1}{\ell - 1}$,
 which is polynomial in $n_f$ and thus polynomial in the size of $\Omega$.
 It is nonsingular if every $\lambda_1^{y_1} \dotsb \lambda_\ell^{y_\ell}$ is distinct,
 which is indeed the case since $\lambda_1, \dotsc, \lambda_\ell$ satisfy the lattice condition.
 
 Therefore, we can solve for the $c_y$'s in polynomial time and compute $\PlHolant(\Omega; \mathcal{F} \union \{f\})$.
\end{proof}

\begin{remark}
 When restricted to $n = \ell = 3$,
 this proof is simpler than the one given in~\cite{CLX12} for Theorem~\ref{thm:interpolation:pervious}
 due to our implicit use of a local holographic transformation
 (i.e.~the writing of $f$ as a linear combination of $\mathbf{t}_1', \dotsc, \mathbf{t}_\ell'$
 and expressing $\PlHolant(\Omega; \mathcal{F} \union \{f\})$ in terms of this).
\end{remark}

\section{Puiseux series, Siegel's Theorem, and Galois theory} \label{sec:ternary}

We prove our main dichotomy theorem in three stages.
This section covers the last stage,
which assumes that all succinct binary signatures of type $\tau_2$ are available.
Among the ways we utilize this assumption is
to build the gadget known as a local holographic transformation (see Figure~\ref{fig:gadget:ternary:local_holographic_transformation}),
which is the focus of Section~\ref{subsec:ternary:construct}.
Then in Section~\ref{subsec:ternary:<(k-1),k-3,-3>},
our efforts are largely spent proving that a certain interpolation succeeds.
To that end,
we employ Galois theory aided by an effective version of Siegel's theorem,
which is made possible by analyzing Puiseux series expansions.

We define the following expressions which appear throughout the rest of the paper:
\begin{alignat}{2}
 \mathfrak{A} &= a - 3b + 2 c; \label{eqn:ternary:frakA}\\
 \mathfrak{B} &= \mathfrak{A} + \kappa (b - c)                &&~~ = a +   (\kappa - 3) b -              (\kappa - 2) c; \qquad \qquad \text{and} \label{eqn:ternary:frakB}\\
 \mathfrak{C} &= \mathfrak{B} + \kappa [2 b + (\kappa - 2) c] &&~~ = a + 3 (\kappa - 1) b + (\kappa - 1) (\kappa - 2) c. \label{eqn:ternary:frakC}
\end{alignat}

\subsection{Constructing a Special Ternary Signature} \label{subsec:ternary:construct}

We construct one of two special ternary signatures.
Either we construct a signature of the form $\langle a,b,b \rangle$ with $a \ne b$ and can finish the proof with Corollary~\ref{cor:k>r:abb_unary_binaries}
or we construct $\langle 3 (\kappa - 1), \kappa - 3, -3 \rangle$.
With this latter signature,
we can interpolate the weighted Eulerian partition signature.

A key step in our dichotomy theorem occurred back in Section~\ref{subsec:k>r} through Lemma~\ref{lem:k>r:interpolate_equality4} with the Bobby Fischer gadget.
To apply this lemma,
we need to construct a gadget with a succinct ternary signature of type $\tau_3$ such that the last two entries are equal and different from the first.
This is the goal of the next lemma,
which assumes $\mathfrak{B} \ne 0$.
We will determine the complexity of the case $\mathfrak{B} = 0$ in Corollary~\ref{cor:unary:dichotomy:a+(k-3)b-(k-2)c=0} without using the results from this section.

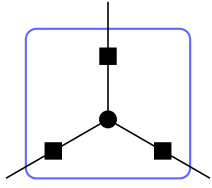
\begin{figure}[t]
 \centering
 \begin{tikzpicture}[scale=\scale,transform shape,node distance=\nodeDist,semithick]
  \node [external] (0)              {};
  \node [square]   (1) [below of=0] {};
  \node [internal] (2) [below of=1] {};
  \path (2) ++(-150:\nodeDist) node [square] (3) {} ++(-150:\nodeDist) node [external] (4) {};
  \path (2) ++( -30:\nodeDist) node [square] (5) {} ++( -30:\nodeDist) node [external] (6) {};
  \path (0) edge (1)
        (1) edge (2)
        (2) edge (3)
            edge (5)
        (3) edge (4)
        (5) edge (6);
  \begin{pgfonlayer}{background}
   \node[draw=\borderColor,thick,rounded corners,inner xsep=12pt,inner ysep=12pt,fit = (1) (3) (5)] {};
  \end{pgfonlayer}
 \end{tikzpicture}
 \caption{Local holographic transformation gadget construction for a ternary signature.}
 \label{fig:gadget:ternary:local_holographic_transformation}
\end{figure}

\begin{lemma} \label{lem:ternary:construct_abb}
 Suppose $\kappa \ge 3$ is the domain size and $a,b,c \in \mathbb{C}$.
 Let $\mathcal{F}$ be a signature set containing the succinct ternary signature $\langle a,b,c \rangle$ of type $\tau_3$
 and the succinct binary signature $\langle x,y \rangle$ of type $\tau_2$ for all $x,y \in \mathbb{C}$.
 If $\mathfrak{A} \mathfrak{B} \ne 0$,
 then there exist $a', b' \in \mathbb{C}$ satisfying $a' \ne b'$ such that
 \[
  \PlHolant(\mathcal{F} \union \{\langle a',b',b' \rangle\}) \le_T \PlHolant(\mathcal{F})
 \]
 where $\langle a',b',b' \rangle$ is a succinct ternary signature of type $\tau_3$.
\end{lemma}

\begin{proof}
 Consider the gadget in Figure~\ref{fig:gadget:ternary:local_holographic_transformation}.
 We assign $\langle a,b,c \rangle$ to the circle vertex and $\langle x,y \rangle$ to the square vertices for some $x,y \in \mathbb{C}$ of our choice,
 to be determined shortly.
 By Lemma~\ref{lem:compute:ternary:holographic_transformation},
 the succinct ternary signature of type $\tau_3$ for the resulting gadget is $\langle a',b',c' \rangle$,
 where
 \begin{align*}
  a' - b' = (x - y)^2 [2 \mathfrak{D} + \mathfrak{A} (x - y)]
            \qquad \text{and} \qquad
  b' - c' = (x - y)^2 \mathfrak{D}
 \end{align*}
 with $\mathfrak{D} = (b - c) (x - y) + \mathfrak{B} y$.
 We pick $x = \mathfrak{B} + y$ and $y = -(b - c)$ so that $\mathfrak{D} = 0$ and thus $b' - c' = 0$.
 Then the first difference simplifies to $a' - b' = \mathfrak{A} \mathfrak{B}^3 \ne 0$.
 This signature has the desired properties,
 so we are done.
\end{proof}

The previous proof fails when $\mathfrak{A} = 0$
because such signatures are invariant set-wise under this type of local holographic transformation.
With the exception of a single point,
we can use this same gadget construction to reduce between any two of these points.

\begin{lemma} \label{lem:ternary:normalize_preparation}
 Suppose $\kappa \ge 3$ is the domain size and $b,c,s,t \in \mathbb{C}$.
 Let $\mathcal{F}$ be a signature set containing the succinct ternary signature $\langle 3 b - 2 c, b, c \rangle$ of type $\tau_3$
 and the succinct binary signature $\langle x,y \rangle$ of type $\tau_2$ for all $x,y \in \mathbb{C}$.
 If $b \ne c$, $3 b + (\kappa - 3) c \ne 0$, and $3 s + (\kappa - 3) t \ne 0$,
 then
 \[
  \PlHolant(\mathcal{F} \union \{\langle 3 s - 2 t, s, t \rangle\}) \le_T \PlHolant(\mathcal{F}),
 \]
 where $\langle 3 s - 2 t, s, t \rangle$ is a succinct ternary signature of type $\tau_3$.
\end{lemma}

\begin{proof}
 Consider the gadget in Figure~\ref{fig:gadget:ternary:local_holographic_transformation}.
 We assign $\langle 3 b - 2 c, b, c \rangle$ to the circle vertex and $\langle x,y \rangle$ to the square vertices for some $x,y \in \mathbb{C}$ of our choice,
 to be determined shortly.
 By Lemma~\ref{lem:compute:ternary:holographic_transformation},
 the signature of this gadget is $f = [x + (\kappa - 1) y] \langle 3 \hat{b} - 2 \hat{c}, \hat{b}, \hat{c} \rangle$,
 where
 \begin{align*}
  \hat{b} &=   b x^2
        + 2 [2 b + (\kappa - 3) c] x y
        + [(3 \kappa - 5) b + (\kappa^2 - 5 \kappa + 6) c] y^2
        \qquad \text{and}\\
  \hat{c} &=   c x^2
        + 2 [3 b + (\kappa - 4) c] x y
        + [(3 \kappa - 6) b + (\kappa^2 - 5 \kappa + 7) c] y^2.
 \end{align*}
 
 We note that the difference $\hat{b} - \hat{c}$ nicely factors as
 \[
  \hat{b} - \hat{c} = (b - c) (x - y)^2.
 \]
 We pick $x = y + \sqrt{s - t}$ so that $\hat{b} - \hat{c} = (b - c) (s - t)$ is the desired difference $s - t$ up to a nonzero factor of $b - c$.
 Then we want to set $\hat{c}$ to be $(b - c) t$.
 With $x = y + \sqrt{s - t}$,
 we can simplify $(b - c) t - \hat{c}$ to
 \begin{equation} \label{eqn:ternary:normalize_preparation}
  (b - c) t - \hat{c}
  = -\kappa [3 b + (\kappa - 3) c] y^2 - 2 \sqrt{s-t} [3 b + (\kappa - 3) c] y + b t - c s.
 \end{equation}
 Since $\kappa [3 b + (\kappa - 3) c] \ne 0$,~(\ref{eqn:ternary:normalize_preparation}) is a nontrivial quadratic polynomial in $y$,
 so we can set $y$ such that this expression vanishes.
 Then the signature is $f = (b - c) [x + (\kappa - 1) y] \langle 3 s - 2 t, s, t \rangle$.
 It remains to check that $x + (\kappa - 1) y \ne 0$.
 
 If $x + (\kappa - 1) y = 0$,
 then $y = -\frac{\sqrt{s-t}}{\kappa}$.
 However, plugging this into~(\ref{eqn:ternary:normalize_preparation}) gives $\frac{(b-c) [3 s + (\kappa - 3) t]}{k} \ne 0$,
 so $x + (\kappa - 1) y$ is indeed nonzero.
\end{proof}

If $\mathfrak{A} = 0$ and $3 b + (\kappa - 3) c = 0$,
then $-3 \langle a,b,c \rangle$ simplifies to $c \langle 3 (\kappa - 1), \kappa - 3, -3 \rangle$,
which is a failure condition of the previous lemma.
The reason is that this signature is pointwise invariant under such local holographic transformations.
However, a different ternary construction can reach this point.

\begin{figure}[t]
 \centering
 \begin{tikzpicture}[scale=\scale,transform shape,node distance=\nodeDist,semithick]
  \node [external] (0)              {};
  \node [internal] (1) [below of=0] {};
  \path (1) ++(-120:\nodeDist) node [internal] (2) {} ++(-150:\nodeDist) node [external] (3) {};
  \path (1) ++( -60:\nodeDist) node [internal] (4) {} ++( -30:\nodeDist) node [external] (5) {};
  \path (0) edge (1)
        (1) edge (2)
            edge (4)
        (2) edge (3)
            edge (4)
        (4) edge (5);
  \begin{pgfonlayer}{background}
   \node[draw=\borderColor,thick,rounded corners,inner xsep=12pt,inner ysep=12pt,fit = (1) (2) (4)] {};
  \end{pgfonlayer}
 \end{tikzpicture}
 \caption{Triangle gadget used to construct $\langle 3 (\kappa - 1), \kappa - 3, -3 \rangle$.}
 \label{fig:gadget:ternary:triangle}
\end{figure}
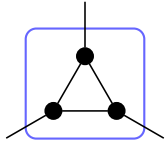

\begin{lemma} \label{lem:ternary:fixed_point}
 Suppose $\kappa \ge 3$ is the domain size and $b,c \in \mathbb{C}$.
 Let $\mathcal{F}$ be a signature set containing the succinct ternary signature $\langle 3 b - 2 c, b, c \rangle$ of type $\tau_3$
 and the succinct binary signature $\langle x,y \rangle$ of type $\tau_2$ for every $x,y \in \mathbb{C}$.
 If $b \ne c$,
 then
 \[
  \PlHolant(\mathcal{F} \union \{\langle 3 (\kappa - 1), \kappa - 3, -3 \rangle\}) \le_T \PlHolant(\mathcal{F}),
 \]
 where $\langle 3 (\kappa - 1), \kappa - 3, -3 \rangle$ is a succinct ternary signature of type $\tau_3$.
\end{lemma}

\begin{proof}
 If $3 b + (\kappa - 3) c = 0$,
 then up to a nonzero factor of $\frac{-c}{3}$,
 $\langle 3 b - 2 c, b, c \rangle$ is already the desired signature.
 Otherwise, $3 b + (\kappa - 3) c \ne 0$.
 By Lemma~\ref{lem:ternary:normalize_preparation},
 we have $\langle 3 s - 2 t, s, t \rangle$ for any $s,t \in \mathbb{C}$ satisfying $3 s + (\kappa - 3) t \ne 0$.
 
 Consider the gadget in Figure~\ref{fig:gadget:ternary:triangle}.
 We assign $\langle 3 s - 2 t, s, t \rangle$ to vertices for some $s, t \in \mathbb{C}$ satisfying $3 s + (\kappa - 3) t \ne 0$ of our choice,
 to be determined shortly.
 By Lemma~\ref{lem:compute:ternary:triangle},
 the signature of this gadget is $\langle 3 s' - 2 t', s', t' \rangle$,
 where
 \begin{align*}
  s' &=   (5 \kappa + 14) s^3
        + (\kappa^2 + 9 \kappa - 42) s^2 t
        + (7 \kappa^2 - 33 \kappa + 42) s t^2
        + (\kappa - 2) (\kappa^2 - 6 \kappa + 7) t^3,
        \qquad \text{and}\\
  t' &=   (\kappa + 14) s^3
        + 21 (\kappa - 2) s^2 t
        + 3 (3 \kappa^2 - 15 \kappa + 14) s t^2
        + (\kappa^3 - 9 \kappa^2 + 23 \kappa - 14) t^3.
 \end{align*}
 It suffices to pick $s$ and $t$ satisfying $3 s + (\kappa - 3) t \ne 0$ such that $s' = \kappa - 3$ and $t' = -3$ up to a common nonzero factor.
 
 We note that the difference $s' - t'$ factors as
 \[
  s' - t' = \kappa (s - t)^2 [4 s + (\kappa - 4) t].
 \]
 We pick $s = \frac{-(\kappa - 4) t + 1}{4}$ so that $s' - t' = \kappa (s - t)^2$ is the desired difference $\kappa$ up to a factor of $(s - t)^2$.
 Then we want to set $t'$ to be $-3 (s - t)^2$.
 With $s = \frac{-(\kappa - 4) t + 1}{4}$,
 we can simplify $-3 (s - t)^2 - t'$ to
 \begin{equation} \label{eqn:ternary:fixed_point}
  -3 (s - t)^2 - t'
  = \frac{1}{64} \left[\kappa^3 (\kappa - 2) t^3 - 3 \kappa^2 (\kappa + 2) t^2 + 3 \kappa (\kappa - 10) t - (\kappa + 26)\right].
 \end{equation}
 Since $\kappa \ge 3$,~(\ref{eqn:ternary:fixed_point}) is a nontrivial cubic polynomial in $t$,
 so we can set $t$ such that this expression vanishes.
 Then $\langle 3 s' - 2 t', s', t' \rangle = (s - t)^2 \langle 3 (\kappa - 1), \kappa - 3, -3 \rangle$.
 It remains to check that $s \ne t$ and $3 s + (\kappa - 3) t \ne 0$.
 
 If $s = t$,
 then $t = \frac{1}{\kappa}$.
 Plugging this into~(\ref{eqn:ternary:fixed_point}) gives~$-1$,
 so $s \ne t$.
 If $3 s + (\kappa - 3) t = 0$,
 then $t = -\frac{3}{\kappa}$.
 Plugging this into~(\ref{eqn:ternary:fixed_point}) gives $1 - \kappa \ne 0$,
 so $3 s + (\kappa - 3) t \ne 0$.
\end{proof}

\subsection{Dose of an effective Siegel's Theorem and Galois theory} \label{subsec:ternary:<(k-1),k-3,-3>}

It suffices to show that $\langle 3 (\kappa - 1), \kappa - 3, -3 \rangle$ is $\SHARPP$-hard for all $\kappa \ge 3$.
The general strategy is to use interpolation.
However, proving that this interpolation succeeds presents a significant challenge.

Consider the polynomial $p(x,y) \in \Z[x,y]$ defined by
\begin{align*}
 p(x,y)
 &= x^5 - 2 x^3 y - x^2 y^2 - x^3 +  x y^2 + y^3  - 2 x^2 - x y\\
 &= x^5 - (2 y + 1) x^3 - (y^2 + 2) x^2 + y (y-1) x + y^3.
\end{align*}
We consider $y$ as an integer parameter $y \ge 4$,
and treat $p(x,y)$ as an infinite family of quintic polynomials in $x$ with integer coefficients.
We want to show that the roots of all these quintic polynomials satisfy the lattice condition.
First, we determine the number of real and nonreal roots.

\begin{lemma} \label{lem:ternary:3r2c_roots}
 For any integer $y \ge 1$,
 the polynomial $p(x,y)$ in $x$ has three distinct real roots and two nonreal complex conjugate roots.
\end{lemma}

\begin{proof}
 Up to a factor of $-4 y^2$,
 the discriminant of $p(x,y)$ (with respect to $x$) is
 \[
  27 y^{11} - 4 y^{10} + 726 y^9 - 493 y^8 + 2712 y^7 - 400 y^6 - 2503 y^5 + 475 y^4 + 956 y^3 - 904 y^2 + 460 y + 104.
 \]
 By replacing $y$ with $z + 1$,
 we get
 \begin{gather*}
    27 z^{11} + 293 z^{10} + 2171 z^9 + 10316 z^8 + 33334 z^7 + 77398 z^6 + 127383 z^5 \\
  + 141916 z^4 + 102097 z^3 + 44373 z^2 + 10336 z + 1156,
 \end{gather*}
 which is positive for any $z \ge 0$.
 Thus the discriminant is negative.
 
 Therefore, $p(x,y)$ has distinct roots in $x$ for all $y \ge 1$.
 Furthermore, with a negative discriminant,
 $p(x,y)$ has $2 s$ nonreal complex conjugate roots for some odd integer $s$.
 Since $p(x,y)$ is a quintic polynomial (in $x$),
 the only possibility is $s = 1$.
\end{proof}

We suspect that for any integer $y \ge 4$,
$p(x,y)$ is in fact irreducible over $\Q$ as a polynomial in $x$.
When considering $y$ as an indeterminate,
the bivariate polynomial $p(x,y)$ is irreducible over $\Q$ and the algebraic curve it defines has genus~$3$,
so by Theorem~1.2 in~\cite{Mul99},
$p(x,y)$ is reducible over $\Q$ for at most a finite number of $y \in \Z$.
For any integer $y \ge 4$, if $p(x,y)$ is irreducible over $\Q$ 
as a polynomial in $x$, then
its Galois group is $S_5$ and its roots satisfy the lattice condition.

\begin{lemma} \label{lem:ternary:lattice:y>=4:irreducible}
 For any integer $y \ge 4$,
 if $p(x,y)$ is irreducible in $\Q[x]$,
 then the roots of $p(x,y)$ satisfy the lattice condition.
\end{lemma}

\begin{proof}
 By Lemma~\ref{lem:ternary:3r2c_roots},
 $p(x,y)$ has three distinct real roots and two nonreal complex conjugate roots.
 With three distinct real roots,
 we know that not all the roots have the same complex norm.
 It is well-known that an irreducible polynomial of prime degree $n$ with exactly two nonreal roots has $S_n$ as a Galois group over $\Q$
 (for example Theorem~10.15 in~\cite{Ste03}).
 Then we are done by Lemma~\ref{lem:lattice_condition:Sn_An}.
\end{proof}

We know of just five values of $y \in \Z$ for which $p(x,y)$ is reducible as a polynomial in $x$:
\[
 p(x,y) =
 \begin{cases}
  (x - 1) (x^4 + x^3 + 2 x^2 - x + 1)     & y = -1\\
  x^2 (x^3 - x - 2)                       & y =  0\\
  (x + 1) (x^4 - x^3 - 2 x^2 - x + 1)     & y =  1\\
  (x - 1) (x^2 - x - 4) (x^2 + 2 x + 2)   & y =  2\\
  (x - 3) (x^4 + 3 x^3 + 2 x^2 - 5 x - 9) & y =  3.
 \end{cases}
\]
These five factorizations also give five integer solutions to $p(x,y) = 0$.
It is a well-known theorem of Siegel~\cite{Sie29} that an algebraic curve of genus at least~$1$ has only a finite number of integral points.
For this curve of genus~$3$,
Faltings' Theorem~\cite{Fal83} says that there can be only a finite number of rational points.
However these theorems are not \emph{effective} in general.
There are some effective versions of Siegel's Theorem that can be applied to our polynomial,
but the best effective bound that we can find is over $10^{20,000}$~\cite{Wal92} and hence cannot be checked in practice.

However, it is shown in the next lemma that in fact these five are the only integer solutions.
In particular, for any integer $y \ge 4$,
$p(x,y)$ does not have a linear factor in $\Z[x]$, and hence by Gauss's Lemma, also no linear factor in $\Q[x]$.
The following proof is essentially due to Aaron Levin~\cite{Lev13}.
We thank Aaron for suggesting the key auxiliary function $g_2(x,y) = \frac{y^2}{x} + y - x^2 + 1$,
as well as for his permission to include the proof here.
We also thank Bjorn Poonen~\cite{Poo13} who suggested a similar proof.
After the proof, we will explain certain complications in the proof.

\begin{lemma} \label{lem:ternary:lattice:no_linear}
 The only integer solutions to $p(x,y) = 0$ are $(1,-1)$, $(0,0)$, $(-1,1)$, $(1,2)$, and $(3,3)$.
\end{lemma}

\begin{proof}
 Clearly these five points are solutions to $p(x,y) = 0$.
 For $a \in \Z$ with $-3 < a < 17$,
 one can directly check that $p(a,y) = 0$ has no other integer solutions in $y$.

 Let $(a,b) \in \Z^2$ be a solution to $p(x,y) = 0$ with $a \ne 0$.
 We claim $a \divides b^2$.
 By definition of $p(x,y)$,
 clearly $a \divides b^3$.
 If $p$ is a prime that divides $a$,
 then let $\operatorname{ord}_p(a) = e$ and $\operatorname{ord}_p(b) = f$ be the exact orders with which $p$ divides $a$ and $b$ respectively.
 Then $f \ge 1$ since $3 f \ge e$ and our claim is that $2 f \ge e$.
 Suppose for a contradiction that $2 f < e$.
 From $p(a,b) = 0$,
 we have
 \[
  a^2 (a^3 - 2 a b - a - b^2 - 2) = -b^3 - a b (b - 1).
 \]
 The order with respect to $p$ of the left-hand side is
 \[
  \operatorname{ord}_p\left(a^2 (a^3 - 2 a b - a - b^2 - 2)\right)
  \ge \operatorname{ord}_p\left(a^2\right)
  = 2 e.
 \]
 Since $p$ is relatively prime to $b-1$,
 $\operatorname{ord}_p\left(a b (b - 1) \right) = e +f > 3f$,
 and therefore the order of the right-hand side with respect to $p$ is
 \[
  \operatorname{ord}_p\left(-b^3 - a b (b - 1)\right)
  = \operatorname{ord}_p(b^3) 
  = 3 f.
 \]
 However, $2 e > 3 f$, a contradiction.
 This proves the claim.
 
 Now consider the functions $g_1(x,y) = y - x^2$ and
 $g_2(x,y) = \frac{y^2}{x}  + y - x^2 + 1$.
 Whenever $(a,b) \in \Z^2$ is a solution to $p(x,y) = 0$ with $a \neq 0$,
 $g_1(a,b)$ and $g_2(a,b)$ are integers.
 However,
 we show that if $a \le -3$ or $a \ge 17$,
 then either $g_1(a,b)$ or $g_2(a,b)$ is not an integer.
 
 Let $c_2 = -(x-1) x$, $c_1 = -x (2 x^2 + 1)$, and $c_0 = x^2 (x^3 - x - 2)$ so that $p(x,y) = y^3 + c_2 y^2 + c_1 y + c_0$.
 Then the discriminant of $p(x,y)$ with respect to $y$ is
 \begin{align}
  \operatorname{disc}_y(p(x,y))
  &= c_2^2 c_1^2 - 4 c_1^3 - 4 c_2^3 c_0 - 27 c_0^2 + 18 c_2 c_1 c_0 \notag\\
  &= (x-1) x^3 (4 x^7 + 5 x^6 + x^5 + 45 x^4 + 151 x^3 + 163 x^2 + 67 x - 4). \label{eqn:ternary:lattice:no_linear}
 \end{align}
 Suppose $x \le -3$.
 Replacing $x$ with $-z - 1$ in~(\ref{eqn:ternary:lattice:no_linear}) gives
 \[
  -(z+1)^3 (z+2) (4 z^7 + 23 z^6 + 55 z^5 + 25 z^4 + 21 z^3 + 39 z^2 + 17 z + 14).
 \]
 This is clearly negative (for $z \ge 0$),
 so~(\ref{eqn:ternary:lattice:no_linear}) is negative.
 Thus $p(x,y)$ only has one real root as a polynomial in $y$.
 Let $y_1(x)$ be that root and consider $y_1^-(x) = x^2 + 2 x^{-1}$ and $y_1^+(x) = x^2 + 2 x^{-1} + 2 x^{-2}$.
 We have $p(x, y_1^-(x)) = -2 x^2 + 6 + 4 x^{-1} + 8 x^{-3} < 0$.
 Also $p(x, y_1^+(x)) = 6 + 18 x^{-1} + 16 x^{-2} + 12 x^{-3} + 24 x^{-4} + 24 x^{-5} + 8 x^{-6} > 0$.
 Hence $y_1^-(x) < y_1(x) < y_1^+(x)$,
 and all three are positive since $y_1^-(x)$ is positive.
 Then for $x \le -3$,
 \[
  -1
  < 2 x^{-1}
  = g_1(x, y_1^-(x))
  < g_1(x, y_1(x))
  < g_1(x, y_1^+(x))
  = 2 x^{-1} + 2 x^{-2}
  < 0,
 \]
 so $g_1(x,y_1(x))$ is not an integer.
 Therefore, $y_1(x)$, the only real root for any integer $x \le -3$, is not an integer.
 
 Now suppose $x \ge 17$.
 Then~(\ref{eqn:ternary:lattice:no_linear}) is positive and there are three distinct real roots.
 Similar to the previous argument,
 we have $p(x, y_1^-(x)) < 0$ and $p(x, y_1^+(x)) > 0$.
 Hence there is some root $y_1(x)$ in the open interval $(y_1^-(x), y_1^+(x))$.
 All three terms $y_1^-(x) < y_1(x) < y_1^+(x)$ are positive because $y_1^-(x) > 0$.
 Then
 \[
  0
  < 2 x^{-1}
  = g_1(x, y_1^-(x))
  < g_1(x, y_1(x))
  < g_1(x, y_1^+(x))
  = 2 x^{-1} + 2 x^{-2}
  < 1,
 \]
 so $g_1(x, y_1(x))$ is not an integer.
 
 There are two more real roots.
 Consider
 \begin{align*}
  y_2^-(x) &= x^{3/2} - \frac{1}{2} x + \frac{1}{8} x^{1/2} - \frac{65}{128} x^{-1/2} - 2 x^{-1}
  \qquad \text{and}\\
  y_2^+(x) &= x^{3/2} - \frac{1}{2} x + \frac{1}{8} x^{1/2} - \frac{65}{128} x^{-1/2}.
 \end{align*}
 Replacing $x$ with $(z + 2)^2$ in
 \begin{align*}
  p(x, y_2^-(x)) = {}
  &2 x^{5/2}
  - \frac{2495}{512} x^2
  + \frac{1087}{512} x^{3/2}
  - \frac{19569}{16384} x
  - \frac{8579}{16384} x^{1/2}
  + \frac{126847}{32768}
  + \frac{1452419}{131072} x^{-1/2}\\
  &- \frac{317}{256} x^{-1}
  + \frac{2871103}{2097152} x^{-3/2}
  - \frac{12675}{8192} x^{-2}
  - \frac{195}{32} x^{-5/2}
  - 8 x^{-3}
 \end{align*}
 gives
 \[
  \frac{1}{2097152 (z+2)^6}
  \left(
   \begin{array}{lll}
    4194304 z^{11}
    + 82055168 z^{10}
    + 722808832 z^9
    + 3774605184 z^8\\
    {}+ 12935149184 z^7
    + 30375187136 z^6
    + 49489164080 z^5
    + 55372934880 z^4\\
    {}+ 41238374079 z^3
    + 19431701370 z^2
    + 5465401844 z
    + 812262392
   \end{array}
  \right),
 \]
 which is clearly positive ($z \ge 0$).
 Thus, $p(x, y_2^-(x)) > 0$.
 Also
 \begin{align*}
  p(x, y_2^+(x)) &=\\
  -2 x^{5/2}
  &- \frac{447}{512} x^2
  - \frac{193}{512} x^{3/2}
  - \frac{3185}{16384} x
  + \frac{20605}{16384} x^{1/2}
  - \frac{4225}{32768}
  + \frac{12675}{131072} x^{-1/2}
  - \frac{274625}{2097152} x^{-3/2}\\
  &< 0.
 \end{align*}
 Hence there is some root $y_2(x)$ in the open interval $(y_2^-(x), y_2^+(x))$.
 All three terms  $y_2^-(x) < y_2(x) < y_2^+(x)$ are positive because $y_2^-(x) > 0$.
 Hence, for $x \ge 17$,
 \begin{align*}
  -1
  < {}& -4 x^{-1/2} - \frac{65}{512} x^{-1} - \frac{1}{2} x^{-3/2} + \frac{4225}{16384} x^{-2} + \frac{65}{32} x^{-5/2} + 4 x^{-3}\\
  = {}&g_2(x, y_2^-(x))
  <    g_2(x, y_2(x))
  <    g_2(x, y_2^+(x))
  =    -\frac{65}{512} x^{-1} + \frac{4225}{16384} x^{-2}
  <    0,
 \end{align*}
 so $g_2(x, y_2(x))$ is not an integer.
 
 Finally, consider
 \begin{align*}
  y_3^-(x) &= -x^{3/2} - \frac{1}{2} x - \frac{1}{8} x^{1/2} + \frac{65}{128} x^{-1/2} - x^{-1}
  \qquad \text{and}\\
  y_3^+(x) &= -x^{3/2} - \frac{1}{2} x - \frac{1}{8} x^{1/2} + \frac{65}{128} x^{-1/2} - \frac{1}{2} x^{-1}.
 \end{align*}
 We have
 \begin{align*}
  p(x, y_3^-(x)) = {}
  &-\frac{1471}{512} x^2
   - \frac{447}{512} x^{3/2}
   - \frac{11377}{16384} x
   - \frac{6013}{16384} x^{1/2}
   + \frac{94079}{32768}
   - \frac{339331}{131072} x^{-1/2}
   - \frac{61}{512} x^{-1}\\
   &- \frac{511807}{2097152} x^{-3/2}
   - \frac{12675}{16384} x^{-2}
   + \frac{195}{128} x^{-5/2}
   - x^{-3}\\
   < {}&0.
 \end{align*}
 Replacing $x$ with $(z + 3)^2$ in
 \begin{align*}
  p(x, y_3^+(x)) = {}
  &x^{5/2}
  - \frac{959}{512} x^2
  - \frac{127}{512} x^{3/2}
  - \frac{7281}{16384} x
  - \frac{13309}{16384} x^{1/2}
  + \frac{53119}{32768}
  - \frac{77699}{131072} x^{-1/2}\\
  &+ \frac{67}{1024} x^{-1}
  + \frac{78017}{2097152} x^{-3/2}
  - \frac{12675}{32768} x^{-2}
  + \frac{195}{512} x^{-5/2}
  - \frac{1}{8} x^{-3}
 \end{align*}
 gives
 \[
  \frac{1}{2097152 (z+3)^6}
  \left(
   \begin{array}{lll}
    2097152 z^{11}
    + 65277952 z^{10}
    + 919728128 z^9
    + 7736969088 z^8\\
    {}+ 43137332608 z^7
    + 167175471424 z^6
    + 458797435600 z^5
    + 889807335920 z^4\\
    {}+ 1191781601633 z^3
    + 1045691960361 z^2
    + 537771428331 z
    + 121660965323
   \end{array}
  \right),
 \]
 which is clearly positive  ($z \ge 0$).
 Thus, $p(x, y_3^+(x)) > 0$.
 Hence there is some root $y_3(x)$ in the open interval $(y_3^-(x), y_3^+(x))$.
 All three terms $y_3^-(x) < y_3(x) < y_3^+(x)$ are negative because $y_3^+(x) < 0$.
 Furthermore, the partial derivative $\frac{\partial g_2(x,y)}{\partial y} = 2 x^{-1} y + 1$
 and $\frac{\partial^2 g_2(x,y)}{\partial y^2} = 2 x^{-1} > 0$.
 Thus $\frac{\partial g_2(x,y)}{\partial y} \le \frac{\partial g_2(x,y)}{\partial y} \mid_{y=y_3^+(x)} = -2 x^{1/2} - \frac{1}{4} x^{-1/2} + \frac{65}{64} x^{-3/2} - x^{-2} < 0$, 
 for all $y \in (-\infty, y_3^+(x)]$.
 Thus, $g_2(x,y)$ is decreasing monotonically in $y$ over the interval $(-\infty, y_3^+(x)]$.
 Then
 \begin{align*}
  0
  < {}& x^{-1/2} - \frac{65}{512} x^{-1} + \frac{1}{8} x^{-3/2} + \frac{4225}{16384} x^{-2} - \frac{65}{128} x^{-5/2} + \frac{1}{4} x^{-3}\\
  = {}&g_2(x, y_3^+(x))
  < g_2(x, y_3(x))
  < g_2(x, y_3^-(x))\\
  = {}& 2 x^{-1/2} - \frac{65}{512} x^{-1} + \frac{1}{4} x^{-3/2} + \frac{4225}{16384} x^{-2} - \frac{65}{64} x^{-5/2} + x^{-3}
  < 1,
 \end{align*}
 so $g_2(x, y_3(x))$ is not an integer.
 To complete the proof,
 notice that the intervals $(y_1^-(x), y_1^+(x))$, $(y_2^-(x), y_2^+(x))$, and $(y_3^-(x), y_3^+(x))$ are disjoint.
 Therefore, we have shown that none of the three roots is an integer for any integer $x \ge 17$.
\end{proof}

\begin{remark}
 One can obtain the Puiseux series expansions for $p(x,y)$,
 which are
 \begin{alignat*}{2}
  y_1(x) &= x^2 + 2 x^{-1} + 2 x^{-2} - 6 x^{-4} - 18 x^{-5} + O(x^{-6}) \qquad &&\text{for $x \in \R$,}\\
  y_2(x) &=  x^{3/2} - \frac{1}{2} x + \frac{1}{8} x^{1/2} - \frac{65}{128} x^{-1/2} - x^{-1} - \frac{1471}{1024} x^{-3/2} - x^{-2} + O(x^{-5/2}) \qquad &&\text{for $x>0$, and}\\
  y_3(x) &= -x^{3/2} - \frac{1}{2} x - \frac{1}{8} x^{1/2} + \frac{65}{128} x^{-1/2} - x^{-1} + \frac{1471}{1024} x^{-3/2} - x^{-2} + O(x^{-5/2}) \qquad &&\text{for $x>0$}.
 \end{alignat*}
 These series converge to the actual roots of $p(x,y)$ for large $x$.
 The basic idea of the proof---called Runge's method---is that,
 for example, when we substitute $y_2(x)$ in $g_2(x, y)$,
 we get $g_2(x, y_2(x)) = O(x^{-1/2})$,
 where the multiplier in the $O$-notation is bounded both above and below by a nonzero constant in absolute value.
 Thus for large $x$,
 this cannot be an integer.
 However, for integer solutions $(x,y)$ of $p(x,y)$,
 this must be an integer.

 We note that the expressions for the $y_i^+(x)$ and $y_i^-(x)$ are the truncated or rounded Puiseux series expansions.
 The reason we discuss $y_i^+(x)$ and $y_i^-(x)$ is because we want to prove an absolute bound,
 instead of the asymptotic bound implied by the $O$-notation.
\end{remark}

By Lemma~\ref{lem:ternary:lattice:no_linear},
if $p(x,y)$ is reducible over $\Q$ as a polynomial in $x$ for any integer $y \ge 4$,
then the only way it can factor is as a product of an irreducible quadratic and an irreducible cubic. 
The next lemma handles this possibility.

\begin{lemma} \label{lem:ternary:lattice:y>=4:reducible}
 For any integer $y_0 \ge 4$,
 if $p(x,y_0)$ is reducible over $\Q$,
 then the roots of $p(x,y_0)$ satisfy the lattice condition.
\end{lemma}

\begin{proof}
 Let $q(x) = p(x,y_0)$ for a fixed integer $y_0 \ge 4$.
 Suppose that $q(x) = f(x) g(x)$,
 where $f(x), g(x) \in \Q[x]$ are monic polynomials of degree at least~$1$.
 By Lemma~\ref{lem:ternary:lattice:no_linear},
 the degree of each factor must be at least~$2$.
 Then without loss of generality,
 let $f(x)$ and $g(x)$ be quadratic and cubic polynomials respectively,
 both of which are irreducible over $\Q$.
 By Gauss' Lemma,
 we can further assume $f(x), g(x) \in \Z[x]$.
 
 Let $\Q_f$ and $\Q_g$ denote the splitting fields over $\Q$ of $f$ and $g$ respectively.
 Suppose $\alpha, \beta$ are the roots of $f(x)$ and $\gamma, \delta, \epsilon$ are the roots of $g(x)$.
 Of course none of these roots are~$0$.
 Suppose there exist $i,j,k,m,n \in \Z$ such that 
 \begin{equation} \label{eqn:ternary:lattice:y>=4:reducible}
  \alpha^i \beta^j = \gamma^k \delta^m \epsilon^n 
  \qquad \text{and} \qquad
  i+j = k+m+n.
 \end{equation}
 We want to show that $i=j=k=m=n=0$.

 We first show that if $i=j$ and $k=m=n$,
 then $i=j=k=m=n=0$.
 By~(\ref{eqn:ternary:lattice:y>=4:reducible}),
 we have $(\alpha \beta)^i = (\gamma \delta \epsilon)^k$ and $2 i = 3 k$.
 Suppose $i \neq 0$,
 then also $k \neq 0$.
 We can write $i = 3 t$ and $k = 2 t$ for some nonzero $t \in \Z$.
 Let $A = \alpha \beta$ and $B = \gamma \delta \epsilon$.
 Then both $A$ and $B$ are integers and $A B = y_0^3$.
 From $A^{3t} = B^{2t}$,
 we have $A^3 = \pm B^2$.
 Then $y_0^6 = A^2 B^2 = \pm A^5$,
 and since $y_0 > 3$,
 there is a nonzero integer $s > 1$ such that $y_0 = s^5$.
 This implies $A = \pm s^6$ and $B = \pm s^9$ (with the same $\pm$ sign).
 Then $f(x) = x^2 + c_1 x \pm s^6$, $g(x) = x^3 + c_2' x^2 + c_1' x \pm s^9$,
 and $q(x) = x^5 - (2 s^5 + 1) x^3 - (s^{10} + 2) x^2 + s^5 (s^5-1) x + s^{15}$.
 We consider the coefficient of $x$ in $q(x) = f(x) g(x)$.
 This is $s^{10} - s^5 = \pm c_1' s^6 \pm c_1 s^9$.
 Since $s > 1$,
 there is a prime $p$ such that $p^u \divides s$ and $p^{u+1} \not |\; s$,
 for some $u \ge 1$.
 But then $p^{6u}$ divides $s^5 = s^{10} \pm c_1' s^6 \pm c_1 s^9$.
 This is a contradiction.
 Hence $i=j$ and $k=m=n$ imply $i=j=k=m=n=0$.

 Now we claim that $\omega = \alpha / \beta$ is not a root of unity.
 For a contradiction,
 suppose that $\omega$ is a primitive $d$th root of unity.
 Since $\omega \in \Q_f$,
 which is a degree~$2$ extension over $\Q$,
 we have $\phi(d) \divides 2$,
 where $\phi(\cdot)$ is Euler's totient function.
 Hence $d \in \{1,2,3,4,6\}$.
 The quadratic polynomial $f(x)$ has the form $x^2 - (1 + \omega) \beta x + \omega \beta^2 \in \Z[x]$.
 Hence $r = \frac{(1 + \omega)}{\omega \beta} \in \Q$.
 We prove the claim separately according to whether $r=0$ or not.

 If $r = 0$,
 then  $\omega = -1$ and $d = 2$.
 In this case,
 $f(x)$ has the form $x^2 + a$ for some $a \in \Z$.
 It is easy to check that $q(x)$ has no such polynomial factor in $\Z[x]$ unless $y_0 = 0$.
 In fact,
 suppose $x^2 + a \divides q(x)$ in $\Z[x]$.
 Then $q(x) = (x^2 + a)(x^3 + bx + c)$ since the coefficient of $x^4$ in $q(x)$ is~$0$.
 Also $a + b = -(2 y_0 + 1)$,
 $c= -(y_0^2 +2)$,
 $a b = y_0 (y_0 - 1)$ and $a c = y_0^3$.
 It follows that $a$ and $b$ are the two roots of the quadratic polynomial $X^2 + (2 y_0 + 1) X + y_0^2 - y_0 \in \Z[X]$.
 Since $a, b \in \Z$, 
 the discriminant $8 y_0 + 1$ must be a perfect square,
 and in fact an odd perfect square $(2 z - 1)^2$ for some $z \in \Z$.
 Thus $y_0 = z (z - 1) / 2$.
 By the quadratic formula,
 $a = -y_0 + z -1$ or $-y_0 - z$.
 On the other hand,
 $a = a c / c = -y_0^3 / (y_0^2 + 2)$.
 In both cases,
 this leads to a polynomial in $z$ in $\Z[z]$ that has no integer solutions other than $z = 0$,
 which gives $y_0 = 0$.

 Now suppose $r \neq 0$.
 Plugging $r$ back in $f(x)$,
 we have $f(x) = x^2 - (2 + \omega + \omega^{-1}) r^{-1} x + (2 + \omega + \omega^{-1}) r^{-2}$.
 The quantity $2 + \omega + \omega^{-1} = 4,1,2,3$ when $d = 1,3,4,6$ respectively.
 Since $(2 + \omega + \omega^{-1}) r^{-2} \in \Z$,
 the rational number $r^{-1}$ must be an integer when $d = 3,4,6$ and half an integer when $d = 1$.
 In all cases,
 it is easy to check that a polynomial $f(x)$ of the specified form does not divide $q(x)$ unless $y = 0$ or $y = -1$.
 Thus, we have proved the claim that $\omega = \alpha / \beta$ is not a root of unity.




 Next consider the case that $f(x)$ is irreducible over $\Q_g$.
 Let $E$ be the splitting field of $f$ over $\Q_g$.
 Then $[E:\Q_g] = 2$.
 Therefore,
 there exists an automorphism $\tau \in \Gal(E / \Q_g)$ that swaps $\alpha$ and $\beta$ but fixes $\Q_g$ and thus fixes $\gamma,\delta,\epsilon$ pointwise.
 By applying $\tau$ to~(\ref{eqn:ternary:lattice:y>=4:reducible}),
 we have $\alpha^j \beta^i = \gamma^k \delta^m \epsilon^n$.
 Dividing by~(\ref{eqn:ternary:lattice:y>=4:reducible}) gives $(\alpha / \beta)^{j-i} = 1$.
 Since $\alpha / \beta$ is not a root of unity,
 we get $i=j$.
 Hence we have $(\alpha \beta)^i = \gamma^k \delta^m \epsilon^n$.
 The order of $\Gal(\Q_g / \Q)$ is $[\Q_g:\Q]$,
 which is divisible by~$3$.
 Thus $\Gal(\Q_g / \Q) \subseteq S_3$ contains an element of order~$3$,
 which must act as a 3-cycle on $\gamma,\delta,\epsilon$.
 Since $\alpha \beta \in \Q$,
 applying this cyclic permutation
 gives $(\alpha \beta)^i = \gamma^m \delta^n \epsilon^k$.
 Therefore $\gamma^{k-m} \delta^{m-n} \epsilon^{n-k} = 1$.
 Notice that $(k-m) + (m-n) + (n-k) = 0$.

 It can be directly checked that $q(x)$ is not divisible by any $x^3 + c \in \Z[x]$,
 and therefore by Lemma~\ref{lem:interpolation:lattice_condition:cubic},
 the roots $\gamma, \delta, \epsilon$ of the cubic polynomial $g(x)$ satisfy the lattice condition.
 Therefore, $k=m=n$.
 Again, we have shown that $i=j$ and $k=m=n$ imply $i=j=k=m=n=0$. 

 The last case is when $f(x)$ splits in $\Q_g[x]$.
 Then $\Q_f$ is a subfield of $\Q_g$, and $2 = [\Q_f: \Q] | [\Q_g: \Q]$.
 Therefore $[\Q_g:\Q] = 6$ and $\Gal(\Q_g / \Q) = S_3$.
 Since $\Q_f$ is normal over $\Q$,
 being a splitting field of a separable polynomial in characteristic~$0$, 
 by the fundamental theorem of Galois theory,
 the corresponding subgroup for $\Q_f$ is $\Gal(\Q_g / \Q_f)$,
 which is a normal subgroup of $S_3$ with index~$2$.
 Such a subgroup of $S_3$ is unique, namely $A_3$.
 In particular,
 the transposition $\tau'$ that swaps $\gamma$ and $\delta$ but fixes $\epsilon$ is an element in $\Gal(\Q_g / \Q) = S_3$ but not in $\Gal(\Q_g / \Q_f) = A_3$.
 This transposition must fix $\alpha$ and $\beta$ setwise but not pointwise.
 Hence, it must swap $\alpha$ and $\beta$.

 By applying $\tau'$ to~(\ref{eqn:ternary:lattice:y>=4:reducible}),
 we have $\alpha^j \beta^i = \gamma^m \delta^k \epsilon^n$.
 Then dividing these two equations gives $(\alpha / \beta)^{i-j} = (\delta / \gamma)^{m-k}$.
 Similarly,
 by considering the transposition that switches $\gamma$ and $\epsilon$ and fixes $\delta$,
 we get $(\alpha / \beta)^{i-j} = (\gamma / \epsilon)^{k-n}$.
 By combining these two equations,
 we have $\gamma^{n-m} \delta^{m-k} \epsilon^{k-n} = 1$.
 Note that $(n-m) + (m-k) + (k-n) = 0$.

 As we noted above,
 the roots of the irreducible $g(x)$ satisfy the lattice condition,
 so we conclude that $k=n=m$.
 From $(\alpha / \beta)^{i-j} = (\delta / \gamma)^{m-k} = 1$,
 we get $i=j$ since $\alpha / \beta$ is not a root of unity.
 We conclude that $i=j=k=m=n=0$,
 so the roots of $q(x)$ satisfy the lattice condition.
\end{proof}

Even though $p(x,3) = (x - 3) (x^4 + 3 x^3 + 2 x^2 - 5 x - 9)$ is reducible,
its roots still satisfy the lattice condition.
To show this,
we utilize a few results, Theorem~\ref{thm:ternary:lattice:dedekind},
Lemma~\ref{lem:ternary:lattice:transitive_transpositiion_p-cycle_Sn},
and Lemma~\ref{lem:ternary:lattice:quartic_root_norms}.

The first is a well-known theorem of Dedekind.

\begin{theorem}[Theorem~4.37~\cite{Jac85}] \label{thm:ternary:lattice:dedekind}
 Suppose $f(x) \in \Z[x]$ is a monic polynomial of degree $n$.
 For a prime $p$,
 let $f_p(x)$ be the corresponding polynomial in $\Z_p[x]$.
 If $f_p(x)$ has distinct roots and factors over $\Z_p[x]$ as a product of irreducible factors with degrees $d_1, d_2, \dotsc, d_r$,
 then the Galois group of $f$ over $\Q$ contains an element with cycle type $(d_1, d_2, \dotsc, d_r)$.
\end{theorem}

With the second result,
we can show that $x^4 + 3 x^3 + 2 x^2 - 5 x - 9$ has Galois group $S_4$ over $\Q$.

\begin{lemma}[Lemma on page~98 in~\cite{Gal73}] \label{lem:ternary:lattice:transitive_transpositiion_p-cycle_Sn}
 For $n \ge 2$,
 let $G$ be a subgroup of $S_n$.
 If $G$ is transitive, contains a transposition,
 and contains a $p$-cycle for some prime $p > n / 2$,
 then $G = S_n$.
\end{lemma}

In the contrapositive,
the third result shows that the roots of $x^4 + 3 x^3 + 2 x^2 - 5 x - 9$ do not all have the same complex norm.

\begin{lemma}[Lemma~D.2 in~\cite{CKW12}] \label{lem:ternary:lattice:quartic_root_norms}
 If all roots of $x^4 + a_3 x^3 + a_2 x^2 + a_1 x + a_0 \in \mathbb{C}[x]$ have the same complex norm,
 then $a_2 |a_1|^2 = |a_3|^2 \overline{a_2} a_0$.
\end{lemma}

\begin{theorem} \label{thm:ternary:lattice:y=3}
 The roots of $p(x,3) = (x - 3) (x^4 + 3 x^3 + 2 x^2 - 5 x - 9)$ satisfy the lattice condition.
\end{theorem}

\begin{proof}
 Let $f(x) = x^4 + 3 x^3 + 2 x^2 - 5 x - 9$ and let $G_f$ be the Galois group of $f$ over $\Q$.
 We claim that $G_f = S_4$.
 As a polynomial over $\Z_5$,
 $f(x) \equiv x^4 + 3 x^3 + 2 x^2 + 1$ is irreducible,
 so $f(x)$ is also irreducible over $\Z$.
 By Gauss' Lemma,
 this implies irreducibility over $\Q$.
 Over $\Z_{13}$,
 $f(x)$ factors into the product of irreducibles $(x^2 + 7) (x + 6) (x + 10)$ and clearly has distinct roots,
 so by Theorem~\ref{thm:ternary:lattice:dedekind},
 $G_f$ contains a transposition.
 Over $\Z_3$,
 $f(x)$ factors into the product of irreducibles $x (x^3 + 2 x + 1)$ and has distinct roots because its discriminant is $1 \not\equiv 0 \pmod{3}$,
 so by Theorem~\ref{thm:ternary:lattice:dedekind},
 $G_f$ contains a 3-cycle.
 Then by Lemma~\ref{lem:ternary:lattice:transitive_transpositiion_p-cycle_Sn},
 $G_f = S_4$.
 
 Let $\alpha, \beta, \gamma, \delta$ be the roots of $f(x)$.
 Suppose there exist $i,j,k,\ell,n \in \Z$ satisfying $n = i + j + k + \ell$ such that $3^n = \alpha^i \beta^j \gamma^k \delta^\ell$.
 Now $G_f = S_4$ contains the 4-cycle $(1\ 2\ 3\ 4)$ that cyclically permutes the roots of $f(x)$ but fixes $\Q$.
 We apply it zero, one, two, and three times to get
 \begin{align*}
  3^n
  &= \alpha^i \beta^j  \gamma^k \delta^\ell, \\
  &= \beta^i  \gamma^j \delta^k \alpha^\ell, \\
  &= \gamma^i \delta^j \alpha^k \beta^\ell, \text{ and}  \\
  &= \delta^i \alpha^j \beta^k  \gamma^\ell.
 \end{align*}
 Then $3^{4n} = (\alpha \beta \gamma \delta)^{i+j+k+\ell} = (-9)^{i+j+k+\ell}$.
 Since $n = i+j+k+\ell$,
 this can only hold when $n = 0$.
 
 Thus, it suffices to show that the roots of $f(x)$ satisfy the lattice condition.
 By the contrapositive of Lemma~\ref{lem:ternary:lattice:quartic_root_norms},
 the roots of $f(x)$ do not all have the same complex norm.
 Then we are done by Lemma~\ref{lem:lattice_condition:Sn_An}.
\end{proof}

From Lemma~\ref{lem:ternary:lattice:y>=4:irreducible},
Lemma~\ref{lem:ternary:lattice:y>=4:reducible},
and Theorem~\ref{thm:ternary:lattice:y=3},
we obtain the following Theorem.

\begin{theorem} \label{thm:ternary:lattice}
 For any integer $y_0 \ge 3$,
 the roots of $p(x,y_0)$ satisfy the lattice condition.
\end{theorem}

We use Theorem~\ref{thm:ternary:lattice} to prove Lemma~\ref{lem:ternary:3k-1k-3-3_hard_k>3}.
We note that the succinct signature type $\tau_4$ is a refinement of $\tau_\text{color}$,
so any succinct signature of type $\tau_\text{color}$ can also be expressed as a succinct signature of type $\tau_4$.
In particular,
the succinct signature $\langle 2,1,0,1,0 \rangle$ of type $\tau_\text{color}$ is written $\langle 2,0,1,0,0,0,1,0,0 \rangle$ of type $\tau_4$.
Then the following is a restatement of Corollary~\ref{cor:k=r:21010_hard}.

\begin{corollary} \label{cor:ternary:201000100_hard}
 Suppose $\kappa \ge 3$ is the domain size.
 Let $\langle 2,0,1,0,0,0,1,0,0 \rangle$ be a succinct quaternary signature of type $\tau_4$.
 Then $\PlHolant(\langle 2,0,1,0,0,0,1,0,0 \rangle)$ is $\SHARPP$-hard.
\end{corollary}

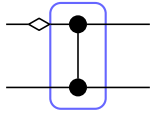
\begin{figure}[t]
 \centering
 \begin{tikzpicture}[scale=\scale,transform shape,node distance=\nodeDist,semithick]
  \node[external] (0)              {};
  \node[external] (1) [below of=0] {};
  \node[internal] (2) [right of=0] {};
  \node[internal] (3) [right of=1] {};
  \node[external] (4) [right of=2] {};
  \node[external] (5) [right of=3] {};
  \path (0.west) edge[postaction={decorate, decoration={
                                             markings,
                                             mark=at position 0.70 with {\arrow[>=diamond,white] {>}; },
                                             mark=at position 0.70 with {\arrow[>=open diamond]  {>}; } } }] (2)
        (1.west) edge (3)
        (2)      edge (3)
                 edge (4.east)
        (3)      edge (5.east);
  \begin{pgfonlayer}{background}
   \node[draw=\borderColor,thick,rounded corners,inner xsep=12pt,inner ysep=8pt,fit = (2) (3)] {};
  \end{pgfonlayer}
 \end{tikzpicture}
 \caption{Quaternary gadget used in the interpolation construction below.
 All vertices are assigned $\langle 3 (\kappa - 1), \kappa - 3, -3 \rangle$.}
 \label{fig:gadget:ternary:quaternary:I}
\end{figure}

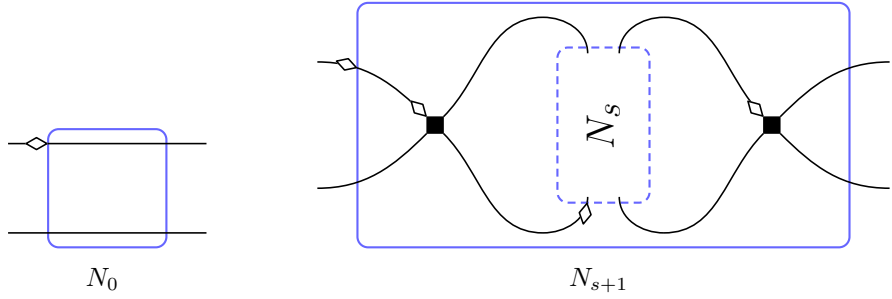
\begin{figure}[t]
 \centering
 \captionsetup[subfigure]{labelformat=empty}
 \subcaptionbox{$N_0$}{
  \begin{tikzpicture}[scale=\scale,transform shape,node distance=\nodeDist,semithick]
   \node[external] (0)                    {};
   \node[external] (1) [right       of=0] {};
   \node[external] (2) [below right of=1] {};
   \node[external] (3) [below left  of=2] {};
   \node[external] (4) [left        of=3] {};
   \node[external] (5) [above right of=2] {};
   \node[external] (6) [right       of=5] {};
   \node[external] (7) [below right of=2] {};
   \node[external] (8) [right       of=7] {};
   \path (0) edge[postaction={decorate, decoration={
                                         markings,
                                         mark=at position 0.2   with {\arrow[>=diamond, white] {>}; },
                                         mark=at position 0.2   with {\arrow[>=open diamond]   {>}; } } }] (6)
         (4) edge (8);
   \begin{pgfonlayer}{background}
    \node[draw=\borderColor,thick,rounded corners,fit = (1) (3) (5) (7)] {};
   \end{pgfonlayer}
  \end{tikzpicture}}
 \qquad
 \subcaptionbox{$N_{s+1}$}{
  \begin{tikzpicture}[scale=\scale,transform shape,node distance=\nodeDist,semithick]
   \node[external] (1)              {};
   \node[external] (2) [right of=1] {};
   \path let
          \p1 = (1),
          \p2 = (2)
         in
          node[external] (0) at (\x1 / 2 + \x2 / 2, \y1) {\Huge \begin{sideways}$N_s$\end{sideways}};
   \node[external] (3) [above of=1] {};
   \node[external] (4) [above of=2] {};
   \node[external] (5) [below of=1] {};
   \node[external] (6) [below of=2] {};
   \path let
          \p1 = (3),
          \p2 = (4)
         in
          node[external]  (7) at (3 * \x1 / 4 + \x2 / 4, \y1) {};
   \path let
          \p1 = (3),
          \p2 = (4)
         in
          node[external]  (8) at (\x1 / 4 + 3 * \x2 / 4, \y1) {};
   \path let
          \p1 = (5),
          \p2 = (6)
         in
          node[external]  (9) at (3 * \x1 / 4 + \x2 / 4, \y1) {};
   \path let
          \p1 = (5),
          \p2 = (6)
         in
          node[external] (10) at (\x1 / 4 + 3 * \x2 / 4, \y1) {};
   \node[external] (11) [above left  of=7]  {};
   \node[external] (12) [above right of=8]  {};
   \node[external] (13) [below left  of=9]  {};
   \node[external] (14) [below right of=10] {};
   \node[external] (15) [below left  of=11] {};
   \node[external] (16) [below right of=12] {};
   \node[external] (17) [      left  of=15] {};
   \node[external] (18) [      right of=16] {};
   \path let
          \p1 = (17),
          \p2 = (0)
         in
          node[square] (19) at (\x1, \y2) {};
   \path let
          \p1 = (18),
          \p2 = (0)
         in
          node[square] (20) at (\x1, \y2) {};
   \node[external] (21) [      left  of=19] {};
   \node[external] (22) [      right of=20] {};
   \node[external] (23) [      left  of=21] {};
   \node[external] (24) [      right of=22] {};
   \path let
          \p1 = (23),
          \p2 = (3)
         in
          node[external] (25) at (\x1, \y2) {};
   \path let
          \p1 = (24),
          \p2 = (4)
         in
          node[external] (26) at (\x1, \y2) {};
   \path let
          \p1 = (23),
          \p2 = (5)
         in
          node[external] (27) at (\x1, \y2) {};
   \path let
          \p1 = (24),
          \p2 = (6)
         in
          node[external] (28) at (\x1, \y2) {};
   \path (19) edge[out=  45, in=180] (11.center)
              edge[out= -45, in=180] (13.center)
              edge[out=-135, in=  0] (27)
         (25) edge[out=   0, in=135, postaction={decorate, decoration={
                                                            markings,
                                                            mark=at position 0.32 with {\arrow[>=diamond,white] {>}; },
                                                            mark=at position 0.32 with {\arrow[>=open diamond]  {>}; },
                                                            mark=at position 0.999 with {\arrow[>=diamond,white] {>}; },
                                                            mark=at position 0.999 with {\arrow[>=open diamond]  {>}; } } }] (19)
         (20) edge[out=-135, in=  0] (14.center)
              edge[out=  45, in=180] (26)
              edge[out= -45, in=180] (28)
         (12.center) edge[out=  0, in=135, postaction={decorate, decoration={
                                                            markings,
                                                            mark=at position 0.999 with {\arrow[>=diamond,white] {>}; },
                                                            mark=at position 0.999 with {\arrow[>=open diamond]  {>}; } } }] (20)
         (11.center) edge[out=  0, in= 90]  (7)
         (12.center) edge[out=180, in= 90]  (8)
         (13.center) edge[out=  0, in=-90, postaction={decorate, decoration={
                                                                  markings,
                                                                  mark=at position 0.91 with {\arrow[>=diamond,white] {>}; },
                                                                  mark=at position 0.91 with {\arrow[>=open diamond]  {>}; } } }] (9)
         (14.center) edge[out=180, in=-90] (10);
   \begin{pgfonlayer}{background}
    \node[draw=\borderColor,thick,densely dashed,rounded corners,fit = (3) (4) (5) (6)] {};
    \node[draw=\borderColor,thick,rounded corners,fit = (11) (13) (21) (22)] {};
   \end{pgfonlayer}
  \end{tikzpicture}}
 \caption{Recursive construction to interpolate the weighted Eulerian partition signature.
 The vertices are assigned the signature of the gadget in Figure~\ref{fig:gadget:ternary:quaternary:I}.}
 \label{fig:gadget:ternary:symmetric_weave_interpolation}
\end{figure}

\begin{sidewaystable}[p]
 \centering
 $\mathclap{
  \left[
   \begin{smallmatrix}
    (\kappa-1) \left(\kappa^2+9 \kappa-9\right) & 12 (\kappa-3) (\kappa-1)^2 & (\kappa-3)^2 (\kappa-1) & 2 (\kappa-3)^2 (\kappa-2) (\kappa-1) & (\kappa-3)^2 (\kappa-1) & 2 (\kappa-3)^2 (\kappa-2) (\kappa-1) & (\kappa-1) (2 \kappa-3) (4 \kappa-3) & 6 (\kappa-3) (\kappa-2) (\kappa-1)^2 & (\kappa-3)^3 (\kappa-2) (\kappa-1) \\
    3 (\kappa-3) (\kappa-1) & 3 \kappa^3-28 \kappa^2+60 \kappa-36 & -(\kappa-3) (2 \kappa-3) & -2 (\kappa-3) (\kappa-2) (2 \kappa-3) & -(\kappa-3) (2 \kappa-3) & -2 (\kappa-3) (\kappa-2) (2 \kappa-3) & 3 (\kappa-3) (\kappa-1)^2 & (\kappa-2) \left(\kappa^3-14 \kappa^2+30 \kappa-18\right) & -(\kappa-3)^2 (\kappa-2) (2 \kappa-3) \\
    (2 \kappa-3) (4 \kappa-3) & 12 (\kappa-3) (\kappa-1)^2 & (\kappa-3)^2 (\kappa-1) & 2 (\kappa-3)^2 (\kappa-2) (\kappa-1) & (\kappa-3)^2 (\kappa-1) & 2 (\kappa-3)^2 (\kappa-2) (\kappa-1) & 9 \kappa^3-26 \kappa^2+27 \kappa-9 & 6 (\kappa-3) (\kappa-2) (\kappa-1)^2 & (\kappa-3)^3 (\kappa-2) (\kappa-1) \\
    3 (\kappa-3) (\kappa-1) & 2 \left(\kappa^3-14 \kappa^2+30 \kappa-18\right) & -(\kappa-3) (2 \kappa-3) & -2 (\kappa-3) (\kappa-2) (2 \kappa-3) & -(\kappa-3) (2 \kappa-3) & -2 (\kappa-3) (\kappa-2) (2 \kappa-3) & 3 (\kappa-3) (\kappa-1)^2 & (\kappa-3) \left(\kappa^3-12 \kappa^2+22 \kappa-12\right) & -(\kappa-3)^2 (\kappa-2) (2 \kappa-3) \\
    (\kappa-3)^2 & -4 (\kappa-3) (2 \kappa-3) & 3 (\kappa-3) & 6 (\kappa-3) (\kappa-2) & \kappa^3+3 \kappa-9 & 6 (\kappa-3) (\kappa-2) & (\kappa-3)^2 (\kappa-1) & -2 (\kappa-3) (\kappa-2) (2 \kappa-3) & 3 (\kappa-3)^2 (\kappa-2) \\
    (\kappa-3)^2 & -4 (\kappa-3) (2 \kappa-3) & 3 (\kappa-3) & 6 (\kappa-3) (\kappa-2) & 3 (\kappa-3) & \kappa^3+6 \kappa^2-30 \kappa+36 & (\kappa-3)^2 (\kappa-1) & -2 (\kappa-3) (\kappa-2) (2 \kappa-3) & 3 (\kappa-3)^2 (\kappa-2) \\
    (\kappa-3)^2 & -4 (\kappa-3) (2 \kappa-3) & \kappa^3+3 \kappa-9 & 6 (\kappa-3) (\kappa-2) & 3 (\kappa-3) & 6 (\kappa-3) (\kappa-2) & (\kappa-3)^2 (\kappa-1) & -2 (\kappa-3) (\kappa-2) (2 \kappa-3) & 3 (\kappa-3)^2 (\kappa-2) \\
    (\kappa-3)^2 & -4 (\kappa-3) (2 \kappa-3) & 3 (\kappa-3) & \kappa^3+6 \kappa^2-30 \kappa+36 & 3 (\kappa-3) & 6 (\kappa-3) (\kappa-2) & (\kappa-3)^2 (\kappa-1) & -2 (\kappa-3) (\kappa-2) (2 \kappa-3) & 3 (\kappa-3)^2 (\kappa-2) \\
    (\kappa-3)^2 & -4 (\kappa-3) (2 \kappa-3) & 3 (\kappa-3) & 6 (\kappa-3) (\kappa-2) & 3 (\kappa-3) & 6 (\kappa-3) (\kappa-2) & (\kappa-3)^2 (\kappa-1) & -2 (\kappa-3) (\kappa-2) (2 \kappa-3) & (2 \kappa-3) \left(2 \kappa^2-9 \kappa+18\right)
   \end{smallmatrix}
  \right]
 }$
 \caption{Recurrence matrix for the recursive construction in the proof of Lemma~\ref{lem:ternary:3k-1k-3-3_hard_k>3}.}
 \label{tbl:ternary:symmetric_weave_interpolation}
 \thisfloatpagestyle{empty}
\end{sidewaystable}

\begin{lemma} \label{lem:ternary:3k-1k-3-3_hard_k>3}
 Suppose $\kappa \ge 4$ is the domain size.
 Then $\PlHolant(\langle 3 (\kappa - 1), \kappa - 3, -3 \rangle)$ is $\SHARPP$-hard.
\end{lemma}

\begin{proof}
 Let $\langle 2,0,1,0,0,0,1,0,0 \rangle$ be a succinct quaternary signature of type $\tau_4$.
 We reduce from $\PlHolant(\langle 2,0,1,0,0,0,1,0,0 \rangle)$,
 which is $\SHARPP$-hard by Corollary~\ref{cor:ternary:201000100_hard}.
 
 Consider the gadget in Figure~\ref{fig:gadget:ternary:quaternary:I}.
 We assign $\langle 3 (\kappa - 1), \kappa - 3, -3 \rangle$ to the vertices.
 By Lemma~\ref{lem:compute:quaternary:I},
 the signature of this gadget is
 $f =
 \langle
  f_{\subMat{1}{1}{1}{1}},
  f_{\subMat{1}{1}{1}{2}},
  f_{\subMat{1}{1}{2}{2}},
  f_{\subMat{1}{1}{2}{3}},
  f_{\subMat{1}{2}{1}{2}},
  f_{\subMat{1}{2}{1}{3}},
  f_{\subMat{1}{2}{2}{1}},
  f_{\subMat{1}{2}{3}{1}},
  f_{\subMat{1}{2}{3}{4}}
 \rangle$ up to a nonzero factor of $\kappa$,
 where
 \begin{align*}
  f_{\subMat{1}{1}{1}{1}} &= (\kappa - 1) (\kappa + 3),\\
  f_{\subMat{1}{1}{1}{2}} &= \kappa - 3,\\
  f_{\subMat{1}{1}{2}{2}} &= 2 \kappa - 3,\\
  f_{\subMat{1}{1}{2}{3}} &= \kappa - 3,\\
  f_{\subMat{1}{2}{1}{2}} &= 2 \kappa - 3,\\
  f_{\subMat{1}{2}{1}{3}} &= \kappa - 3,\\
  f_{\subMat{1}{2}{2}{1}} &= (\kappa - 3) (\kappa + 1),\\
  f_{\subMat{1}{2}{3}{1}} &= \kappa - 3, \text{ and}\\
  f_{\subMat{1}{2}{3}{4}} &= -3.
 \end{align*}
 
 Now consider the recursive construction in Figure~\ref{fig:gadget:ternary:symmetric_weave_interpolation}.
 We assign $f$ to every vertex.
 Up to a nonzero factor of $\kappa^s$,
 let $g_s$ be the succinct signature of type $\tau_4$ for the $s$th gadget in this construction.
 Then $g_0 = \langle 1,0,0,0,0,0,1,0,0 \rangle$ and $g_s = M^s g_0$,
 where $M$ is the matrix in Table~\ref{tbl:ternary:symmetric_weave_interpolation}.
 
 The row vectors
 \begin{align*}
  ( 0, 0,0,0,-1,0,0,0,1)&,\\
  ( 0,-1,0,1,-1,0,0,1,0)&,\\
  (-1, 0,1,0,-1,0,1,0,0)&, \text{ and}\\
  ( 0, 0,0,0,-1,1,0,0,0)&
 \end{align*}
 are linearly independent row eigenvectors of $M$,
 all with eigenvalue $\kappa^3$,
 that are orthogonal to the initial signature $g_0$.
 Note that our target signature $\langle 2, 0, 1, 0, 0, 0, 1, 0, 0 \rangle$ is also orthogonal to these four row eigenvectors.
 
 Up to a factor of $(x - \kappa^3)^4$,
 the characteristic polynomial of $M$ is
 \[
  h(x,\kappa) = x^5 - \kappa^6 (2 \kappa - 1) x^3 - \kappa^9 (\kappa^2 - 2 \kappa + 3) x^2 + (\kappa-2) (\kappa-1) \kappa^{12} x + (\kappa-1)^3 \kappa^{15}.
 \]
 Since $h(\kappa^3, \kappa) = (\kappa - 3) \kappa^{17}$ and $\kappa \ge 4$,
 we know that $\kappa^3$ is not a root of $h(x,\kappa)$ as a polynomial in $x$.
 Thus,
 none of the remaining eigenvalues are $\kappa^3$.
 The roots of $h(x,\kappa)$ satisfy the lattice condition iff the roots of
 \[
  \tilde{h}(x,\kappa)
  = \frac{1}{\kappa^{15}} h(\kappa^3 x, \kappa)
  = x^5 - (2 \kappa - 1) x^3 - (\kappa^2 - 2 \kappa + 3) x^2 + (\kappa-2) (\kappa-1) x + (\kappa-1)^3
 \]
 satisfy the lattice condition.
 In $\tilde{h}(x,\kappa)$,
 we replace $\kappa$ by $y+1$ to get $p(x,y) = x^5 - (2 y + 1) x^3 - (y^2 + 2) x^2 + (y-1) y x + y^3$.
 By Theorem~\ref{thm:ternary:lattice},
 the roots $p(x,y_0)$ satisfy the lattice condition for any integer $y_0 \ge 3$.
 Thus, the roots of $\tilde{h}(x,\kappa)$ satisfy the lattice for any $\kappa \ge 4$.
 In particular,
 this means that the five eigenvalues of $M$ different from $\kappa^3$ are distinct,
 so $M$ is diagonalizable.
 
 The 5-by-5 matrix in the upper-left corner of $[g_0\ M g_0\ \ldots\ M^8 g_0]$ is
 \[
  \left[
   \begin{smallmatrix}
    1 & 9 (\kappa-1)^2 \kappa & (\kappa-1) \kappa^4 \left(\kappa^3-3 \kappa^2+11 \kappa+3\right) & (\kappa-1) \kappa^7 \left(\kappa^3+12 \kappa^2-11 \kappa+6\right) & (\kappa-1) \kappa^{10} \left(\kappa^4+4 \kappa^3-4 \kappa^2+44 \kappa-33\right) \\
    0 & 3 (\kappa-3) (\kappa-1) \kappa & -(\kappa-3) \kappa^4 \left(\kappa^2-2 \kappa-1\right) & (\kappa-3) \kappa^7 \left(3 \kappa^2-3 \kappa+2\right) & (\kappa-3) \kappa^{10} \left(\kappa^3-4 \kappa^2+16 \kappa-11\right) \\
    0 & 9 (\kappa-1)^2 \kappa & \kappa^4 \left(\kappa^4-4 \kappa^3+6 \kappa^2+4 \kappa-3\right) & \kappa^7 \left(15 \kappa^3-28 \kappa^2+11 \kappa-6\right) & \kappa^{10} \left(\kappa^5+3 \kappa^4-22 \kappa^3+72 \kappa^2-83 \kappa+33\right) \\
    0 & 3 (\kappa-3) (\kappa-1) \kappa & -(\kappa-3) (\kappa-1) \kappa^4 (\kappa+1) & 2 (\kappa-3) \kappa^7 \left(2 \kappa^2-\kappa+1\right) & (\kappa-3) (\kappa-1) \kappa^{10} \left(\kappa^2-6 \kappa+11\right) \\
    0 & (\kappa-3)^2 \kappa & (\kappa-3) \kappa^4 (\kappa+1) & (\kappa-3) \kappa^7 \left(\kappa^2-\kappa+2\right) & (\kappa-3) \kappa^{10} \left(\kappa^3-2 \kappa^2+10 \kappa-11\right)
   \end{smallmatrix}
  \right].
 \]
 Its determinant is $(\kappa-3)^3 (\kappa-1)^2 \kappa^{26} (\kappa^4+\kappa^3+17 \kappa^2+3 \kappa+2)$,
 which is nonzero since $\kappa \ge 4$.
 Thus $[g_0\ M g_0\ \ldots\ M^8 g_0]$ has rank at least~$5$,
 so by Lemma~\ref{lem:2nd_condition_implication},
 $g_0$ is not orthogonal to the five remaining row eigenvectors of $M$.
 
 Therefore, by Lemma~\ref{lem:interpolate_all_not_orthogonal},
 we can interpolate $\langle 2, 0, 1, 0, 0, 0, 1, 0, 0 \rangle$,
 which completes the proof.
\end{proof}

When $\kappa = 3$,
$\langle 3 (\kappa - 1), \kappa - 3, -3 \rangle$ simplifies to $-3 \langle -2,0,1 \rangle$.
We have a much simpler proof that this signature is $\SHARPP$-hard.

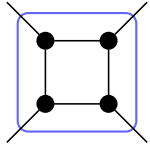
\begin{figure}[t]
 \centering
 \begin{tikzpicture}[scale=\scale,transform shape,node distance=\nodeDist,semithick]
  \node [internal] (0)                    {};
  \node [external] (1) [above left  of=0] {};
  \node [internal] (2) [      right of=0] {};
  \node [external] (3) [above right of=2] {};
  \node [internal] (4) [below       of=2] {};
  \node [external] (5) [below right of=4] {};
  \node [internal] (6) [      left  of=4] {};
  \node [external] (7) [below left  of=6] {};
  \path (0) edge (1)
            edge (2)
            edge (6)
        (2) edge (3)
            edge (4)
        (4) edge (5)
            edge (6)
        (6) edge (7);
  \begin{pgfonlayer}{background}
   \node[draw=\borderColor,thick,rounded corners,inner xsep=12pt,inner ysep=12pt,fit = (0) (4)] {};
  \end{pgfonlayer}
 \end{tikzpicture}
 \caption{Square gadget used to construct the weighted Eulerian partition signature.}
 \label{fig:gadget:ternary:square}
\end{figure}

\begin{lemma} \label{lem:ternary:-201}
 Suppose the domain size is~$3$.
 Then $\PlHolant(\langle -2,0,1 \rangle)$ is $\SHARPP$-hard.
\end{lemma}

\begin{proof}
 Let $g = \langle 2,0,1,0,0,0,1,0 \rangle$ be a succinct quaternary signature of type $\tau_4$.
 We reduce from $\PlHolant(g)$,
 which is $\SHARPP$-hard by Corollary~\ref{cor:ternary:201000100_hard}.

 Consider the gadget in Figure~\ref{fig:gadget:ternary:square}.
 The vertices are assigned $\langle -2,0,1 \rangle$.
 Up to a factor of~$9$,
 the signature of this gadget is $g$, as desired.
\end{proof}

We summarize this section with the following result.
With all succinct binary signatures of type $\tau_2$ available
as well as the succinct unary signature $\langle 1 \rangle$ of type $\tau_1$,
any succinct ternary signature $\langle a,b,c \rangle$ of type $\tau_3$ satisfying $\mathfrak{B} \ne 0$ is $\SHARPP$-hard.

\begin{lemma} \label{lem:ternary}
 Suppose $\kappa \ge 3$ is the domain size and $a,b,c \in \mathbb{C}$.
 Let $\mathcal{F}$ be a signature set containing the succinct ternary signature $\langle a, b, c \rangle$ of type $\tau_3$,
 the succinct unary signature $\langle 1 \rangle$ of type $\tau_1$,
 and the succinct binary signature $\langle x, y \rangle$ of type $\tau_2$ for all $x,y \in \mathbb{C}$.
 If $\mathfrak{B} \ne 0$,
 then $\PlHolant(\mathcal{F})$ is $\SHARPP$-hard.
\end{lemma}

\begin{proof}
 Suppose $\mathfrak{A} \ne 0$.
 By Lemma~\ref{lem:ternary:construct_abb},
 we have a succinct ternary signature $\langle a', b', b' \rangle$ of type $\tau_3$ with $a' \ne b'$.
 Then we are done by Corollary~\ref{cor:k>r:abb_unary_binaries}.
 
 Otherwise, $\mathfrak{A} = 0$.
 Since $\mathfrak{B} \ne 0$,
 we have $b \ne c$.
 By Lemma~\ref{lem:ternary:fixed_point},
 we have $\langle 3 (\kappa - 1), \kappa - 3, -3 \rangle$.
 If $\kappa \ge 4$,
 then we are done by Lemma~\ref{lem:ternary:3k-1k-3-3_hard_k>3}.
 Otherwise, $\kappa = 3$ and we are done by Lemma~\ref{lem:ternary:-201}.
\end{proof}

\section{Constructing a Nonzero Unary Signature} \label{sec:unary}

The primary goal of this section is to construct the succinct unary signature $\langle 1 \rangle$ of type $\tau_1$.
However, this is not always possible.
For example,
the succinct ternary signature $\langle 0,0,1 \rangle = \AD_{3,3}$ of type $\tau_3$ (on domain size~$3$) cannot construct $\langle 1 \rangle$.
This follows from the parity condition (Lemma~\ref{lem:k=r:parity_condition}).
In such cases,
we show that the problem is either computable in polynomial time or $\SHARPP$-hard without the help of additional signatures.

Lemma~\ref{lem:unary:construct_<1>} handles two easy cases for which it is possible to construct $\langle 1 \rangle$.

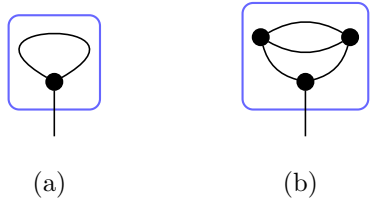
\begin{figure}[t]
 \centering
 \subcaptionbox
 {\label{subfig:gadget:unary:unary:self-loop}}{
  \begin{tikzpicture}[scale=\scale,transform shape,node distance=\nodeDist,semithick]
   \node[external] (0)              {};
   \node[internal] (1) [above of=0] {};
   \path (0) edge (1)
         (1) edge[out=150, in=30, looseness=20] node[pos=0.3] (e1) {} node[pos=0.7] (e2) {} (1);
   \begin{pgfonlayer}{background}
    \node[draw=\borderColor,thick,rounded corners,inner xsep=6pt,inner ysep=12pt,fit = (1) (e1) (e2)] {};
   \end{pgfonlayer}
  \end{tikzpicture}}
 \qquad
 \subcaptionbox
 {\label{subfig:gadget:unary:unary:parallel_edges}}{
  \begin{tikzpicture}[scale=\scale,transform shape,node distance=\nodeDist,semithick]
   \node[external] (0)                    {};
   \node[internal] (1) [above       of=0] {};
   \node[internal] (2) [above left  of=1] {};
   \node[internal] (3) [above right of=1] {};
   \path (0) edge (1)
         (1) edge[bend left ] (2)
             edge[bend right] (3)
         (2) edge[bend left ] coordinate (c1) (3)
             edge[bend right] (3);
   \begin{pgfonlayer}{background}
    \node[draw=\borderColor,thick,rounded corners,inner xsep=6pt,inner ysep=12pt,fit = (1) (2) (3) (c1)] {};
   \end{pgfonlayer}
  \end{tikzpicture}}
 \caption{Two simple unary gadgets.}
 \label{fig:gadget:unary}
\end{figure}

\begin{lemma} \label{lem:unary:construct_<1>}
 Suppose $\kappa \ge 3$ is the domain size and $a,b,c \in \mathbb{C}$.
 Let $\mathcal{F}$ be a signature set containing the succinct ternary signature $\langle a,b,c \rangle$ of type $\tau_3$.
 If $a + (\kappa - 1) b \ne 0$ or $[2 b + (\kappa - 2) c] [b^2 - 4 b c - (\kappa - 3) c^2] \ne 0$,
 then
 \[
  \PlHolant(\mathcal{F} \union \{\langle 1 \rangle\}) \le_T \PlHolant(\mathcal{F}),
 \]
 where $\langle 1 \rangle$ is a succinct unary signature of type $\tau_1$.
\end{lemma}

\begin{proof}
 Suppose $a + (\kappa - 1) b \ne 0$.
 Consider the gadget in Figure~\ref{subfig:gadget:unary:unary:self-loop}.
 We assign $\langle a,b,c \rangle$ to its vertex.
 By Lemma~\ref{lem:unary:construct_<1>},
 this gadget has the succinct unary signature $\langle u \rangle$ of type $\tau_1$,
 where $u = a + (\kappa - 1) b$.
 Since $u \ne 0$,
 this signature is equivalent to $\langle 1 \rangle$.
 
 Otherwise, $a + (\kappa - 1) b = 0$,
 and $[2 b + (\kappa - 2) c] [b^2 - 4 b c - (\kappa - 3) c^2] \ne 0$.
 Consider the gadget in Figure~\ref{subfig:gadget:unary:unary:parallel_edges}.
 We assign $\langle a,b,c \rangle$ to all three vertices.
 By Lemma~\ref{lem:unary:construct_<1>},
 this gadget has the succinct unary signature $\langle u' \rangle$ of type $\tau_1$,
 where $u' = -(\kappa - 1) (\kappa - 2) [2 b + (\kappa - 2) c] [b^2 - 4 b c - (\kappa - 3) c^2]$.
 Since $u' \ne 0$,
 this signature is equivalent to $\langle 1 \rangle$.
\end{proof}

One of the failure conditions of Lemma~\ref{lem:unary:construct_<1>} is when both $a + (\kappa - 1) b = 0$ and $b^2 - 4 b c - (\kappa - 3) c^2 = 0$ hold.
In this case, $\langle a,b,c \rangle = c \langle -(\kappa - 1) (2 \pm \sqrt{\kappa + 1}), 2 \pm \sqrt{\kappa + 1}, 1 \rangle$.
If $c = 0$,
then $a = b = c = 0$ and the signature is trivial.
Otherwise, $c \ne 0$.
Then up to a nonzero factor of $c$,
this signature further simplifies to $\AD_{3,3}$ by taking the minus sign when $\kappa = 3$.
Just like $\AD_{3,3}$,
we show (in Lemma~\ref{lem:unary:dichotomy:AD-like}) that all of these signatures are $\SHARPP$-hard.

Similar to the proof of Theorem~\ref{thm:edge_coloring:k=r},
we prove the hardness in Lemma~\ref{lem:unary:dichotomy:AD-like} by reducing from counting weighted Eulerian partitions.

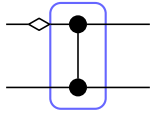
\begin{figure}[t]
 \centering
 \begin{tikzpicture}[scale=\scale,transform shape,node distance=\nodeDist,semithick]
  \node[external] (0)              {};
  \node[external] (1) [below of=0] {};
  \node[internal] (2) [right of=0] {};
  \node[internal] (3) [right of=1] {};
  \node[external] (4) [right of=2] {};
  \node[external] (5) [right of=3] {};
  \path (0.west) edge[postaction={decorate, decoration={
                                             markings,
                                             mark=at position 0.70 with {\arrow[>=diamond,white] {>}; },
                                             mark=at position 0.70 with {\arrow[>=open diamond]  {>}; } } }] (2)
        (1.west) edge (3)
        (2)      edge (3)
                 edge (4.east)
        (3)      edge (5.east);
  \begin{pgfonlayer}{background}
   \node[draw=\borderColor,thick,rounded corners,inner xsep=12pt,inner ysep=8pt,fit = (2) (3)] {};
  \end{pgfonlayer}
 \end{tikzpicture}
 \caption{Quaternary gadget used in the interpolation construction below.
 All vertices are assigned $\langle -(\kappa - 1) \gamma, \gamma, 1 \rangle$.}
 \label{fig:gadget:unary:quaternary:I}
\end{figure}

\begin{sidewaystable}[p]
 \centering
 $
  \mathclap{
  \left[
   \begin{smallmatrix}
    (\kappa -1) (\gamma -3) \gamma ^2 & -2 (\kappa -2) (\kappa -1) \gamma & (\kappa -1) (3 \gamma -1) & 2 (\kappa -2) (\kappa -1) \gamma & 0 & 0 & 0 & 0 & 0 \\
    0 & \kappa^2 (\gamma +1) - 4 \kappa \gamma + 2 (\gamma +1) & 0 & (\kappa -2) (3 \gamma -1) & -(\kappa -2) \gamma & -(\kappa -4) (\kappa -2) \gamma & -(\kappa -2) \gamma & -(\kappa -4) (\kappa -2) \gamma & 2 (\kappa -2) (\gamma -4) \gamma ^2 \\
    3 \gamma -1 & 2 (\kappa -2) \gamma & \kappa^2 (\gamma +1) + \kappa (3 \gamma -5) - 7 \gamma + 5 & -2 (\kappa -2) \gamma & 0 & 0 & 0 & 0 & 0 \\
    0 & 2 (3 \gamma -1) & 0 & (\kappa -2) \gamma (\kappa +\gamma +1) & 2 \gamma & 2 (\kappa -4) \gamma & 2 \gamma & 2 (\kappa -4) \gamma & -4 (\gamma -4) \gamma ^2 \\
    0 & -2 (\kappa -2) \gamma & 0 & 2 (\kappa -2) \gamma & -(\gamma -3) \gamma ^2 & 4 (\kappa -2) \gamma & 3 \gamma -1 & 4 (\kappa -2) \gamma & (\kappa -2) (\gamma -4) \gamma (\gamma +1) \\
    0 & -(\kappa -4) \gamma & 0 & (\kappa -4) \gamma & 2 \gamma & 2 (\kappa -4) \gamma & 2 \gamma & \kappa (3 \gamma +1) - 4 (\gamma +1) & (\gamma -4) \gamma (\gamma \kappa +\kappa -4) \\
    0 & -2 (\kappa -2) \gamma & 0 & 2 (\kappa -2) \gamma & 3 \gamma -1 & 4 (\kappa -2) \gamma & -(\gamma -3) \gamma ^2 & 4 (\kappa -2) \gamma & (\kappa -2) (\gamma -4) \gamma (\gamma +1) \\
    0 & -(\kappa -4) \gamma & 0 & (\kappa -4) \gamma & 2 \gamma & \kappa (3 \gamma +1) - 4 (\gamma +1) & 2 \gamma & 2 (\kappa -4) \gamma & (\gamma -4) \gamma (\gamma \kappa +\kappa -4) \\
    0 & 4 \gamma & 0 & -4 \gamma & \gamma +1 & 2 (\gamma \kappa +\kappa -4) & \gamma +1 & 2 (\gamma \kappa +\kappa -4) & \kappa^2 (\gamma +1) - 2 (\gamma +5) - 2 (5 \gamma -11)
   \end{smallmatrix}
  \right]
  }
 $
 \caption{The recurrence matrix $M$, up to a factor of $(\gamma + 1)$, for the recursive construction in the proof of Lemma~\ref{lem:unary:dichotomy:AD-like}.}
 \label{tbl:unary:matrix:recurrence}
 
 \vspace*{1in}
 
 $
  \left[
   \begin{smallmatrix}
    0 & 0 & 0 & 0 & 1 & -2 & 0 & 0 & 1 \\
    0 & 0 & 0 & 0 & 0 & -1 & 0 & 1 & 0 \\
    0 & -1 & 0 & 1 & -(\gamma -3) \gamma & (\gamma -3) \gamma & 0 & 0 & 0 \\
    0 & 0 & 0 & 0 & -1 & 0 & 1 & 0 & 0 \\
    0 & (\kappa -2) \gamma & 0 & -(\kappa -2) \gamma & 0 & (\kappa -2) (\gamma -1) & 0 & (\kappa -2) (\gamma -1) & (\kappa -2) (\gamma -4) (\gamma -1) \gamma \\
    0 & -(\kappa -2) \gamma & 0 & (\kappa -2) \gamma & \gamma -1 & (\kappa -2) (\gamma -1) & \gamma -1 & (\kappa -2) (\gamma -1) & 0 \\
    0 & 2 & 0 & \kappa -2 & 0 & 0 & 0 & 0 & 0 \\
    1 & 0 & \kappa -1 & 0 & 0 & 0 & 0 & 0 & 0 \\
    (\gamma -3) \gamma & \kappa^2 + \kappa (2 \gamma -7) - 2 (\gamma -5) & -(\gamma -3) \gamma & -\kappa^2 - \kappa (2 \gamma  -7) + 2 (\gamma -5) & -(\gamma -3) \gamma & -(\kappa -4) (\gamma -3) \gamma & -(\gamma -3) \gamma  & -(\kappa -4) (\gamma -3) \gamma & 2 (\gamma -4) (\gamma -3) \gamma^2
   \end{smallmatrix}
  \right]
 $
 \caption{The matrix $P$ whose rows are the row eigenvectors of the matrix in Table~\ref{tbl:unary:matrix:recurrence}.}
 \label{tbl:unary:matrix:eigenvectors}
\end{sidewaystable}

\begin{lemma} \label{lem:unary:dichotomy:AD-like}
 Suppose $\kappa \ge 3$ is the domain size and $a,b,c \in \mathbb{C}$.
 Let $\langle a,b,c \rangle$ be a succinct ternary signature of type $\tau_3$.
 If $a + (\kappa - 1) b = 0$ and $b^2 - 4 b c - (\kappa - 3) c^2 = 0$,
 then
 \[
  \langle a,b,c \rangle = c \langle -(\kappa - 1) (2 + \varepsilon \sqrt{\kappa + 1}), 2 + \varepsilon \sqrt{\kappa + 1}, 1 \rangle,
 \]
 where $\varepsilon = \pm 1$,
 and $\PlHolant(\langle a,b,c \rangle)$ is $\SHARPP$-hard unless $c = 0$,
 in which case,
 the problem is computable in polynomial time.
\end{lemma}

\begin{proof}
 If $c = 0$,
 then $a = b = c = 0$ so the output is always~$0$.
 Otherwise, $c \ne 0$.
 Up to a nonzero factor of $c$,
 $\langle a,b,c \rangle$ can be written as $\langle -(\kappa - 1) (2 + \varepsilon \sqrt{\kappa + 1}), 2 + \varepsilon \sqrt{\kappa + 1}, 1 \rangle$ under the given assumptions,
 where $\varepsilon = \pm 1$.
 
 Suppose $\kappa = 3$.
 If $\varepsilon = -1$,
 then we have $\langle 0,0,1 \rangle = \AD_{3,3}$ and we are done by Theorem~\ref{thm:edge_coloring:k=r}.
 Otherwise, $\varepsilon = 1$ and we have $\langle 8,-4,-1 \rangle$.
 Let $T = \frac{1}{3} \left[\begin{smallmatrix*}[r] 1 & -2 & -2 \\ -2 & 1 & -2 \\ -2 & -2 & 1 \end{smallmatrix*}\right]$,
 which is an orthogonal matrix.
 It follows from Theorem~\ref{thm:ortho_holo_trans} and Lemma~\ref{lem:compute:ternary:holographic_transformation} that
 \[
  \Holant(\langle 8,-4,-1 \rangle)
  \equiv_T \Holant(T^{\otimes 3} \langle 8,-4,-1 \rangle)
  \equiv_T \Holant(\langle 0,0,1 \rangle),
 \]
  so again we are done by Theorem~\ref{thm:edge_coloring:k=r}.
 
 Now we suppose $\kappa \ge 4$.
 Let $g = \langle 2,0,1,0,0,0,1,0,0 \rangle$ be a succinct quaternary signature of type $\tau_4$.
 We reduce from $\PlHolant(g)$ to $\PlHolant(\langle a,b,c \rangle)$.
 Then by Corollary~\ref{cor:ternary:201000100_hard},
 $\PlHolant(\langle a,b,c \rangle)$ is $\SHARPP$-hard.
 We write this signature as $\langle -(\kappa - 1) \gamma, \gamma, 1 \rangle$,
 where $\gamma = 2 + \varepsilon \sqrt{\kappa + 1}$.
 Consider the gadget in Figure~\ref{fig:gadget:unary:quaternary:I}.
 We assign $\langle -(\kappa - 1) \gamma, \gamma, 1 \rangle$ to both vertices.
 By Lemma~\ref{lem:compute:quaternary:I},
 up to a nonzero factor of $\gamma - 1$,
 this gadget has the succinct quaternary signature $f$ of type $\tau_4$,
 where
 \[
 f =
  \left\langle
   (\kappa - 1) (\gamma - 3) \gamma^2,\quad
   -(\kappa - 2) \gamma,\quad
   3 \gamma - 1,\quad
   2 \gamma,\quad
   3 \gamma - 1,\quad
   2 \gamma,\quad
   -(\gamma - 3) \gamma^2,\quad
   2 \gamma,\quad
   \gamma + 1
  \right\rangle.
 \]
 Now consider the recursive construction in Figure~\ref{fig:gadget:linear_interpolation:quaternary}.
 We assign $f$ to all vertices.
 Let $f_s$ be the succinct signature of type $\tau_4$ for the $s$th gadget in this recursive construction.
 The initial signature,
 which is just two parallel edges,
 has the succinct signature $f_0 = \langle 1,0,0,0,0,0,1,0,0 \rangle$ of type $\tau_4$.
 We can express $f_s$ as $M^s f_0$,
 where $M$ is the matrix in Table~\ref{tbl:unary:matrix:recurrence}.
 
 Consider an instance $\Omega$ of $\PlHolant(g)$.
 Suppose $g$ appears $n$ times in $\Omega$.
 We construct from $\Omega$ a sequence of instances $\Omega_s$ of $\PlHolant(f)$ indexed by $s \ge 0$.
 We obtain $\Omega_s$ from $\Omega$ by replacing each occurrence of $g$ with the gadget $f_s$.
 
 We can express $M$ as $(\gamma - 1)^3 P^{-1} \Lambda P$,
 where $P$ is the matrix in Table~\ref{tbl:unary:matrix:eigenvectors},
 \[
  \Lambda = \diag(-1, -1, -1, -1, \kappa - 2, \kappa - 2, \kappa - 1, \kappa - 1, \lambda),
 \]
 and $\lambda = \frac{(\kappa - 2) (\kappa + 2 \gamma - 4)}{(\gamma - 1)^2}$.
 The rows of $P$ are linearly independent since
 \[
  \det(P) = (\kappa - 1) (\kappa - 2)^2 (\gamma - 1)^6 (\gamma - 3)^3 \gamma \ne 0.
 \]
 For $1 \le i \le 9$,
 let $r_i$ be the $i$th row of $P$.
 Notice that the initial signature $f_0$ and the target signature $g$ are orthogonal to the same set of row eigenvectors of $M$,
 namely $\{r_1, r_2, r_3, r_5, r_7, r_9\}$.
 Up to a common factor of $(\gamma - 1)^3$,
 the eigenvalues for $M$ corresponding to $r_4$, $r_6$, and $r_8$
 (the three row eigenvectors of $M$ \emph{not} orthogonal to $f_0$)
 are~$-1$, $\kappa - 2$, and $\kappa - 1$ respectively.
 Since $\kappa \ge 4$,
 $\kappa - 2$ and $\kappa - 1$ are relatively prime and greater than~$1$,
 so these three eigenvalues satisfy the lattice condition.
 Thus by Lemma~\ref{lem:interpolate_all_not_orthogonal},
 we can interpolate $g$ as desired.
\end{proof}

\begin{remark}
 Although the matrices in Table~\ref{tbl:unary:matrix:recurrence} and Table~\ref{tbl:unary:matrix:eigenvectors} seem large,
 they are probably the smallest possible to succeed in this recursive quaternary construction.
 In fact, for quaternary signatures one would normally expect these matrices to be even larger since there are typically fifteen different entries in a domain invariant signature of arity~$4$.
\end{remark}

The other failure condition of Lemma~\ref{lem:unary:construct_<1>} is when both $a + (\kappa - 1) b = 0$ and $2 b + (\kappa - 2) c = 0$ hold.
In this case,
$\langle a,b,c \rangle = c \langle (\kappa - 1) (\kappa - 2), -(\kappa - 2), 2 \rangle$.
If this signature is connected to $\langle 1 \rangle$,
then the first entry of the resulting succinct binary signature of type $\tau_2$ is $(\kappa - 1) (\kappa - 2) \cdot 1 - (\kappa - 2) \cdot (\kappa - 1) = 0$
while the second entry is $-(\kappa - 2) \cdot 2 + 2 \cdot (\kappa - 2) = 0$.
That is, the resulting binary signature is identically~$0$.
This suggests we apply a holographic transformation such that the support of the resulting signature is only on $\kappa - 1$ of the domain elements.

If $c = 0$, then $a = b = c = 0$ and the signature is trivial.
Otherwise, $c \ne 0$.
If $\kappa = 3$,
then up to a nonzero factor of $c$,
this signature further simplifies to $\langle 2,-1,2 \rangle$,
which is tractable by case~\ref{case:cor:tractable:holant-star:aba} of Corollary~\ref{cor:tractable:holant-star}.
Otherwise $\kappa \ge 4$,
and we show the problem is $\SHARPP$-hard.

\begin{lemma} \label{lem:unary:<(k-1)(k-2),-(k-2),2>}
 Suppose $\kappa \ge 4$ is the domain size.
 Let $f = \langle (\kappa - 1) (\kappa - 2), -(\kappa - 2), 2 \rangle$ be a succinct ternary signature of type $\tau_3$.
 Then $\PlHolant(f)$ is $\SHARPP$-hard.
\end{lemma}

\begin{proof}
 Consider the matrix $T = \left[\begin{smallmatrix} 1 & \mathbf{1} \\ \mathbf{1} & T' \end{smallmatrix}\right] \in \mathbb{C}^{\kappa \times \kappa}$,
 where $T' = y J_{\kappa-1} + (x - y) I_{\kappa-1}$
 with $x = -\frac{\kappa + \sqrt{\kappa} - 1}{\sqrt{\kappa} + 1}$
 and $y = \frac{1}{\sqrt{\kappa} + 1}$.
 After scaling by $\frac{1}{\sqrt{\kappa}}$,
 we claim that $T$ is an orthogonal matrix.
 
 Let $r_i$ be the $i$th row of $\frac{1}{\sqrt{\kappa}} T$.
 First we compute the diagonal entries of $\frac{1}{\kappa} T \transpose{T}$.
 Clearly $r_1 \transpose{r_1} = 1$.
 For $2 \le i \le \kappa$,
 we have
 \[
  r_i \transpose{r_i}
  = \frac{1}{\kappa} \left[1 + x^2 + (\kappa - 2) y^2\right]
  = \frac{1}{\kappa} \left[1 + \frac{(\kappa + \sqrt{\kappa} - 1)^2}{(\sqrt{\kappa} + 1)^2} + \frac{\kappa - 2}{(\sqrt{k} + 1)^2}\right]
  = 1.
 \]
 Now we compute the off-diagonal entries.
 For $2 \le i \le \kappa$,
 we have
 \[
  r_1 \transpose{r_i}
  = \frac{1}{\kappa} \left[1 + x + (\kappa - 2) y\right]
  = \frac{1}{\kappa} \left[1 - \frac{\kappa + \sqrt{\kappa} - 1}{\sqrt{\kappa} + 1} + \frac{\kappa - 2}{\sqrt{\kappa} + 1}\right]
  = 0.
 \]
 For $2 \le i < j \le \kappa$,
 we have
 \[
  r_i \transpose{r_j}
  = \frac{1}{\kappa} \left[1 + 2 x y + (\kappa - 3) y^2\right]
  = \frac{1}{\kappa} \left[1 - \frac{2 (\kappa + \sqrt{\kappa} - 1)}{(\sqrt{\kappa} + 1)^2} + \frac{\kappa - 3}{(\sqrt{k} + 1)^2}\right]
  = 0.
 \]
 This proves the claim.

 We apply a holographic transformation by $T$ to the signature $f$ to obtain $\widehat{f} = T^{\otimes 3} f$,
 which does not change the complexity of the problem by Theorem~\ref{thm:ortho_holo_trans}.
 Since the first row of $T$ is a row of all $1$'s,
 the output of $\widehat{f}$ on any input containing the first domain element is~$0$.
 When restricted to the remaining $\kappa - 1$ domain elements,
 $\widehat{f}$ is domain invariant and symmetric,
 so it can be expressed as a succinct ternary signature of type $\tau_3$.
 
 Up to a nonzero factor of $\frac{\kappa^3}{(\sqrt{\kappa} + 1)^2}$,
 it can be verified that $\widehat{f} = \langle -(\kappa - 2) (2 + \sqrt{\kappa}), 2 + \sqrt{\kappa}, 1 \rangle$.
 One way to do this is as follows.
 We write $f = \langle a,b,2 \rangle$ and
 $T = \left[\begin{smallmatrix} 1 & \mathbf{1} \\ \mathbf{1} & T' \end{smallmatrix}\right] \in \mathbb{C}^{\kappa \times \kappa}$,
 where $T' = y J_{\kappa-1} + (x - y) I_{\kappa-1}$.
 The entries of $\widehat{f}$ are polynomials in $\kappa$ with coefficients from $\Z[a,b,x,y]$.
 The degree of these polynomials is at most~$3$ since the arity of $f$ is~$3$.
 After computing the entries of $\widehat{f}$ for domain sizes $3 \le \kappa \le 6$ as elements in $\Z[a,b,x,y]$,
 we interpolate the entries of $\widehat{f}$ as elements in $(\Z[a,b,x,y])[\kappa]$.
 Then replacing $a,b,x,y$ with their actual values gives the claimed expression for the signature.
 
 Since $\kappa \ge 4$, $\widehat{f}$ is $\SHARPP$-hard by Lemma~\ref{lem:unary:dichotomy:AD-like},
 which completes the proof.
\end{proof}

At this point,
we have achieved the broader goal of this section.
For any $a,b,c \in \mathbb{C}$ and domain size $\kappa \ge 3$,
either $\PlHolant(\langle a,b,c \rangle)$ is computable in polynomial time,
or $\PlHolant(\langle a,b,c \rangle)$ is $\SHARPP$-hard,
or we can use $\langle a,b,c \rangle$ to construct $\langle 1 \rangle$
(i.e.~the reduction $\PlHolant(\{\langle a,b,c, \rangle, \langle 1 \rangle\} \le_T \PlHolant(\langle a,b,c \rangle)$ holds).
However, Lemma~\ref{lem:unary:<(k-1)(k-2),-(k-2),2>} is easily generalized and this generalization turns out to be necessary to obtain our dichotomy.

Recall that connecting $f = \langle (\kappa - 1) (\kappa - 2), -(\kappa - 2), 2 \rangle$ to $\langle 1 \rangle$ results in an identically~$0$ signature.
This suggests that we consider the more general signature $\widetilde{f} = \alpha \langle 1 \rangle^{\otimes 3} + \beta f$
for any $\alpha \in \mathbb{C}$ and any nonzero $\beta \in \mathbb{C}$ since this does not change the complexity (as we argue in Corollary~\ref{cor:unary:dichotomy:a+(k-3)b-(k-2)c=0}).
For any $a,b,c \in \mathbb{C}$ satisfying $\mathfrak{B} = 0$ (cf.~(\ref{eqn:ternary:frakB})),
if $\alpha = \frac{2 b + (\kappa - 2) c}{\kappa}$ and $\beta = \frac{-b + c}{\kappa}$,
then $\widetilde{f} = \langle a,b,c \rangle$.
We note that the condition $\mathfrak{B} = 0$ can also be written as $(\kappa - 2) (b - c) = b - a$.
We now prove a dichotomy for the signature $\widetilde{f}$.

\begin{corollary} \label{cor:unary:dichotomy:a+(k-3)b-(k-2)c=0}
 Suppose $\kappa \ge 3$ is the domain size and $a,b,c \in \mathbb{C}$.
 Let $\langle a,b,c \rangle$ be a succinct ternary signature of type $\tau_3$.
 If $\mathfrak{B} = 0$,
 then $\PlHolant(\langle a,b,c \rangle)$ is $\SHARPP$-hard unless $b = c$ or $\kappa = 3$,
 in which case, the problem is computable in polynomial time.
\end{corollary}

\begin{proof}
 If $b = c$,
 then by $\mathfrak{B} = 0$ we have $a = b = c$,
 which means the signature is degenerate and the problem trivially tractable.
 If $\kappa = 3$,
 then $a = c$ and the problem is tractable by case~\ref{case:cor:tractable:holant-star:aba} of Corollary~\ref{cor:tractable:holant-star}.
 Otherwise $b \ne c$ and $\kappa \ge 4$.

 Since $\mathfrak{B} = 0$,
 it can be verified that $\langle a,b,c \rangle = \frac{2 b + (\kappa - 2) c}{\kappa} \langle 1 \rangle^{\otimes 3} + \frac{-b + c}{\kappa} f$,
 where $f = \langle (\kappa - 1) (\kappa - 2), -(\kappa - 2), 2 \rangle$.
 We show that $\PlHolant(\langle a,b,c \rangle)$ is $\SHARPP$-hard iff $\PlHolant(f)$ is.
 Since $\PlHolant(f)$ is $\SHARPP$-hard by Lemma~\ref{lem:unary:<(k-1)(k-2),-(k-2),2>},
 this proves the result.
 
 Let $G = (V,E)$ be a connected planar 3-regular graph with $n = |V|$ and $m = |E|$.
 We can view $\PlHolant(G; \langle a,b,c \rangle)$ as a sum of $2^n$ Holant computations using the signatures $\alpha \langle 1 \rangle^{\otimes 3}$ and $\beta f$.
 Each of these Holant computations considers a different assignment of either $\alpha \langle 1 \rangle^{\otimes 3}$ or $\beta f$ to each vertex.
 Since connecting $f$ to $\langle 1 \rangle$ gives an identically~$0$ signature,
 if any connected signature grid contains both $\alpha \langle 1 \rangle^{\otimes 3}$ and $\beta f$,
 then that particular Holant computation is~$0$.
 This is because a vertex of degree three assigned $\langle 1 \rangle^{\otimes 3}$ is equivalent to three vertices of degree one
 connected to the same three neighboring vertices and each assigned $\langle 1 \rangle$.
 There are only two possible assignments that could be nonzero.
 If all vertices are assigned $\alpha \langle 1 \rangle^{\otimes 3}$,
 then the Holant is $\alpha^n \kappa^m$.
 Otherwise, all vertices are assigned $\beta f$ and the Holant is $\beta^n \PlHolant(G; f)$.
 Thus, $\PlHolant(G; \alpha \langle 1 \rangle^{\otimes 3} + \beta f) = \alpha^n \kappa^m + \beta^n \PlHolant(G; f)$.
 Since $\beta \ne 0$,
 one can solve for either Holant value given the other.
\end{proof}

\section{Interpolating All Binary Signatures of Type \texorpdfstring{$\tau_2$}{tau-2}} \label{sec:binary}

In this section,
we show how to interpolate all binary succinct signatures of type $\tau_2$ in most settings.
We use two general techniques to achieve this goal.
In the first subsection,
we use a generalization of the anti-gadget technique that creates a multitude of gadgets.
They are so numerous that one is most likely to succeed.
In the second subsection,
we introduce a new technique called \emph{Eigenvalue Shifted Triples} (EST).
These generalize the technique of Eigenvalue Shifted Pairs from~\cite{KC10},
and we use EST to interpolate binary succinct signatures in cases where the anti-gadget technique cannot handle.
There are a few isolated problems for which neither technique works.
However, these problems are easily handled separately in Lemma~\ref{lem:appendix:binary} in Appendix~\ref{sec:appendix:binary}.

From Section~\ref{sec:unary},
every problem fits into one of three cases: either
(1) the problem is tractable,
(2) the problem is $\SHARPP$-hard, or
(3) we can construct the succinct unary signature $\langle 1 \rangle$ of type $\tau_1$.
Thus, many results in this section assume that $\langle 1 \rangle$ is available.

\subsection{E Pluribus Unum}

We will use Lemma~\ref{lem:k>r:binary:interpolate} to prove our binary interpolation.
The main technical difficulty is to satisfy the third condition of Lemma~\ref{lem:k>r:binary:interpolate},
which is to prove that some recurrence matrix (that defines a sequence of gadgets) has infinite order up to a scalar.
When the matrix has a finite order up to a scalar,
we can utilize this failure condition to our advantage by constructing an anti-gadget~\cite{CKW12},
which is the “last” gadget with a distinct signature (up to a scalar) in the infinite sequence of gadgets.
To make sure that we construct a multitude of nontrivial gadgets without cancellation,
we put the anti-gadget inside another gadget
(contrast the gadget in Figure~\ref{fig:gadget:binary:anti-gadget} with the gadget in Figure~\ref{subfig:gadget:binary:parallel_edges}).
From among this plethora of gadgets,
at least one must succeed under the right conditions.

Although this idea works quite well in that some gadget among those constructed does succeed,
we still must prove that one such gadget succeeds in every setting.
We aim to exhibit a recurrence matrix whose ratio of eigenvalues is not a root of unity.
We consider three related recurrence matrices at once.
The next two lemmas consider two similar situations involving the eigenvalues of three such matrices.
When applied,
these lemmas show that some recurrence matrix must have eigenvalues with distinct complex norms,
even though exactly which one among them succeeds may depend on the parameters in a complicated way.

\begin{lemma} \label{lem:binary:norm:same_argument}
 Let $d_0, d_1, d_2, \Psi \in \mathbb{C}$.
 If $d_0$, $d_1$, and $d_2$ have the same argument but are distinct,
 then for all $\rho \in \R$, there exists $i \in \{0,1,2\}$ such that $|\Psi + d_i| \ne \rho$.
\end{lemma}

\begin{proof}
 Assume to the contrary that there exists $\rho \in \R$ such that $|\Psi + d_i| = \rho$ for every $i \in \{0,1,2\}$.
 In the complex plane,
 consider the circle centered at the origin of radius $\rho$.
 Each $\Psi + d_i$ is a distinct point on this circle as well as a distinct point on a common line through $\Psi$.
 However, the line intersects the circle in at most two points, a contradiction.
\end{proof}

\begin{lemma} \label{lem:binary:norm:same_norm}
 Let $d_0, d_1, d_2, \Psi \in \mathbb{C}$.
 If $d_0$, $d_1$, and $d_2$ have the same complex norm but are distinct and $\Psi \ne 0$,
 then for all $\rho \in \R$, there exists $i \in \{0,1,2\}$ such that $|\Psi + d_i| \ne \rho$.
\end{lemma}

\begin{proof}
 Let $\ell = |d_0|$.
 Assume to the contrary that there exists $\rho \in \R$ such that $|\Psi + d_i| = \rho$ for every $i \in \{0,1,2\}$.
 In the complex plane,
 consider the circle centered at the origin of radius $\rho$ and the circle centered at $\Psi$ of radius $\ell$.
 Since $\Psi \ne 0$, these circles are distinct.
 Each $\Psi + d_i$ is a distinct point on both circles.
 However, these circles intersect in at most two points, a contradiction.
\end{proof}

Now we use Lemma~\ref{lem:binary:norm:same_argument} and Lemma~\ref{lem:binary:norm:same_norm}
as well as our generalization of the anti-gadget technique to establish a crucial lemma.

\begin{lemma} \label{lem:binary:general:root_of_unity}
 Suppose $\kappa \ge 3$ is the domain size and $a,b,c,\omega \in \mathbb{C}$.
 Let $\mathcal{F}$ be a set of signatures containing the succinct binary signature $\langle \omega + \kappa - 1, \omega - 1 \rangle$ of type $\tau_2$
 and the succinct ternary signature $\langle a,b,c \rangle$ of type $\tau_3$.
 If the following three conditions are satisfied:
 \begin{enumerate}
  \item $\omega \not\in \{0, \pm 1\}$, \label{condition:lem:binary:general:omega}
  \item $\mathfrak{B} \ne 0$, and \label{condition:lem:binary:general:a+(k-3)b-(k-2)c!=0}
  \item at least one of the following holds: \label{condition:lem:binary:general:two_cases}
  \begin{enumerate}[label=(\roman*)]
   \item $\mathfrak{C} = 0$ or \label{case:lem:binary:general:frakC=0}
   \item $\mathfrak{C}^2 = \omega^{2 \ell} \mathfrak{B}^2$ for some $\ell \in \{0,1\}$ but either $\mathfrak{C}^2 \ne \mathfrak{A}^2$ or $\kappa \ne 3$, \label{case:lem:binary:general:frakC^2=w^(2ell)frakB^2}
  \end{enumerate}
 \end{enumerate}
 then
 \[
  \PlHolant(\mathcal{F} \union \{\langle x,y \rangle\}) \le_T \PlHolant(\mathcal{F})
 \]
 for any $x,y\in \mathbb{C}$,
 where $\langle x,y \rangle$ is a succinct binary signature of type $\tau_2$.
\end{lemma}

We use this lemma to establish that various $2$-by-$2$ recurrence matrices have infinite order modulo scalars.
When applied,
$\omega$ will be the ratio of two eigenvalues,
one of which is a multiple of $\mathfrak{B}$ or $\mathfrak{B}^2$ by a nonzero function of $\kappa$.

\begin{figure}[t]
 \centering
 \begin{tikzpicture}[scale=\scale,transform shape,node distance=\nodeDist,semithick]
  \node[external] (0)                    {};
  \node[internal] (1) [right       of=0] {};
  \node[square]   (2) [above right of=1] {};
  \node[triangle] (3) [below right of=1] {};
  \node[internal] (4) [below right of=2] {};
  \node[external] (5) [right       of=4] {};
  \path (0) edge             (1)
        (1) edge[bend left]  (2)
            edge[bend right] (3)
        (2) edge[bend left]  (4)
        (3) edge[bend right] (4)
        (4) edge             (5);
  \begin{pgfonlayer}{background}
   \node[draw=\borderColor,thick,rounded corners,inner xsep=12pt,inner ysep=8pt,fit = (1) (2) (3) (4)] {};
  \end{pgfonlayer}
 \end{tikzpicture}
 \caption{Binary gadget that generalizes the anti-gadget technique.
 The circle vertices are assigned $\langle a,b,c \rangle$ while the square and triangle vertices are each assigned the signature of some gadget.}
 \label{fig:gadget:binary:anti-gadget}
\end{figure}
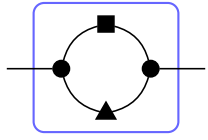

\begin{proof}[Proof of Lemma~\ref{lem:binary:general:root_of_unity}.]
 Let $\Phi = \frac{\mathfrak{C}^2}{\mathfrak{B}^2}$ and $\Psi = \frac{(\kappa - 2) \mathfrak{A}^2}{\mathfrak{B}^2}$.
 Consider the recursive construction in Figure~\ref{fig:gadget:k>r:binary:interpolation}.
 After scaling by a nonzero factor of $\kappa$,
 we assign $f = \frac{1}{\kappa} \langle \omega + \kappa - 1, \omega - 1 \rangle$ to every vertex.
 Let $f_s$ be the succinct binary signature of type $\tau_2$ for the $s$th gadget in this construction.
 We can express $f_s$ as $M^s \left[\begin{smallmatrix} 1 \\ 0 \end{smallmatrix}\right]$,
 where $M = \frac{1}{\kappa} \left[\begin{smallmatrix} \omega + \kappa - 1 & (\kappa - 1) (\omega - 1) \\ \omega - 1 & (\kappa - 1) \omega + 1 \end{smallmatrix}\right]
 = \left[\begin{smallmatrix} 1 & 1 - \kappa \\ 1 & 1 \end{smallmatrix}\right]
 \left[\begin{smallmatrix} \omega & 0 \\ 0 & 1 \end{smallmatrix}\right]
 \left[\begin{smallmatrix} 1 & 1 - \kappa \\ 1 & 1 \end{smallmatrix}\right]^{-1}$
 by Lemma~\ref{lem:k>r:binary:interpolation:eigenvalues}.
 Then $f_s = \frac{1}{\kappa} \langle \omega^s + \kappa - 1, \omega^s - 1 \rangle$.
 The eigenvalues of $M$ are~$1$ and $\omega$,
 so the determinant of $M$ is $\omega \ne 0$.
 If $\omega$ is not a root of unity,
 then we are done by Corollary~\ref{cor:k>r:binary:interpolate}.
 
 Otherwise, suppose $\omega$ is a primitive root of unity of order $n$.
 Since $\omega \not = \pm 1$ by assumption, $n \ge 3$.
 Now consider the gadget in Figure~\ref{fig:gadget:binary:anti-gadget}.
 We assign $\langle a,b,c \rangle$ to the circle vertices,
 $f_r = \frac{1}{\kappa} \langle \omega^r + \kappa - 1, \omega^r - 1 \rangle$ to the square vertex,
 and $f_s = \frac{1}{\kappa} \langle \omega^s + \kappa - 1, \omega^s - 1 \rangle$ to the triangle vertex,
 where $r,s \ge 0$ are parameters of our choice.
 By Lemma~\ref{lem:compute:binary:anti-gadget},
 up to a nonzero factor of $\frac{\mathfrak{B}^2}{\kappa}$,
 this gadget has the succinct binary signature
 \[
  f(r,s) =
  \tfrac{1}{\kappa}
  \langle
   \Phi \omega^{r+s} + (\kappa - 1) (\omega^r + \omega^s + \Psi + 1),\ %
   \Phi \omega^{r+s} -              (\omega^r + \omega^s + \Psi + 1) + \kappa
  \rangle
 \]
 of type $\tau_2$.
 Consider using this gadget in the recursive construction of Figure~\ref{fig:gadget:k>r:binary:interpolation}.
 Let $f_t(r,s)$ be the succinct binary signature of type $\tau_2$ for the $t$th gadget in this recursive construction.
 Then $f_1(r,s) = f(r,s)$ and $f_t(r,s) = (M(r,s))^t \left[\begin{smallmatrix} 1 \\ 0 \end{smallmatrix}\right]$,
 where the eigenvalues of $M(r,s)$ are $\Phi \omega^{r+s} + \kappa - 1$ and $\omega^r + \omega^s + \Psi$ by Lemma~\ref{lem:k>r:binary:interpolation:eigenvalues}.
 Thus, the determinant of $M(r,s)$ is $(\Phi \omega^{r+s} + \kappa - 1) (\omega^r + \omega^s + \Psi)$.
 Since $\Phi$ is either~$0$ or a power of $\omega$ by condition \ref{condition:lem:binary:general:two_cases},
 the first factor is nonzero for any choice of $r$ and $s$.
 However, for some $r$ and $s$,
 it might be that $g(r,s) = \omega^r + \omega^s + \Psi = 0$.
 
 Suppose $\Psi = 0$.
 We consider the two possible cases of $\Phi$ in order to finish the proof under this assumption.
 \begin{enumerate}
  \item Suppose $\Phi = 0$.
  Consider the gadget $M(0, 1)$.
  The determinant of $M(0, 1)$ is nonzero since $g(0, 1) \ne 0$ and the ratio of its eigenvalues is not a root of unity because they have distinct complex norms.
  Thus, we are done by Corollary~\ref{cor:k>r:binary:interpolate}.
  
  \item Suppose $\Phi = \omega^{2 \ell}$ for some $\ell \in \{0,1\}$.
  Consider the gadget $M(n - \ell, n - \ell)$.
  The determinant of $M(n - \ell, n - \ell)$ is nonzero since $g(n - \ell, n - \ell) \ne 0$ and the ratio of its eigenvalues is not a root of unity because they have distinct complex norms.
  Thus, we are done by Corollary~\ref{cor:k>r:binary:interpolate}.
 \end{enumerate}
 
 Otherwise, $\Psi \ne 0$.
 We claim that $g(r,s) = 0$ can hold for at most one choice of $r,s \in \Z_n$ (modulo the swapping of $r$ and $s$).
 To see this, consider $r_1, s_1, r_2, s_2$ such that $g(r_1, s_1) = 0 = g(r_2, s_2)$.
 Then $\omega^{r_1} + \omega^{s_1} = -\Psi = \omega^{r_2} + \omega^{s_2}$.
 By taking complex norms and applying the law of cosines,
 we have $\cos \theta_1 = \cos \theta_2$,
 where $\theta_j = \arg(\omega^{s_j - r_j})$ is the angle from $\omega^{r_j}$ to $\omega^{s_j}$ for $j \in \{1,2\}$.
 Thus, $\theta_1 = \pm \theta_2$.
 Since $\Psi \ne 0$, we have $\theta_1 \ne \pm \pi$.
 If $\theta_1 = \theta_2$,
 then $\omega^{r_1} (1 + e^{i \theta_1}) = \omega^{r_2} (1 + e^{i \theta_1})$.
 Since $\theta_1 \ne \pm \pi$,
 the factor $1 + e^{i \theta_1}$ is nonzero.
 After dividing by this factor,
 we conclude that $r_1 = r_2$ and thus $s_1 = s_2$.
 Otherwise, $\theta_1 = -\theta_2$.
 Then $\omega^{r_1} (1 + e^{i \theta_1}) = \omega^{s_2} (1 + e^{i \theta_1})$ and we conclude that $r_1 = s_2$ and $s_1 = r_2$.
 This proves the claim.
 
 Suppose $n \ge 4$ and let $S_0 = \{(0,0), (1, n-1), (2, n-2)\}$ and $S_1 = \{(1,1), (2, 0), (3, n-1)\}$.
 Then $g(r,s) = 0$ holds for at most one $(r,s) \in S_0 \union S_1$.
 In particular, $g(r,s)$ is either nonzero for all $(r,s) \in S_0$ or nonzero for all $(r,s) \in S_1$.
 Pick $j \in \{0,1\}$ such that $g(r,s)$ is nonzero for all $(r,s) \in S_j$.
 By Lemma~\ref{lem:binary:norm:same_argument} with $d_i = (\omega^i + \omega^{-i}) \omega^j$ and $\rho = |\Phi \omega^{2j} + \kappa - 1|$,
 there exists some $(r,s) \in S_j$ such that the eigenvalues of $M(r, s)$ have distinct complex norms,
 so we are done by Corollary~\ref{cor:k>r:binary:interpolate}.
 
 Otherwise, $n = 3$.
 We consider the two possible cases of $\Phi$ in order to finish the proof.
 \begin{enumerate}
  \item Suppose $\Phi = 0$.
  Let $S_j = \{(0,j), (1,j+1), (2,j+2)\}$.
  Then $g(r,s) = 0$ holds for at most one $(r,s) \in S_0 \union S_1$.
  In particular, $g(r,s)$ is either nonzero for all $(r,s) \in S_0$ or nonzero for all $(r,s) \in S_1$.
  Pick $j \in \{0,1\}$ such that $g(r,s)$ is nonzero for all $(r,s) \in S_j$.
  By Lemma~\ref{lem:binary:norm:same_norm} with $d_i = (1 + \omega^j) \omega^i$ and $\rho = \kappa - 1$,
  there exists some $(r,s) \in S_j$ such that the eigenvalues of $M(r, s)$ have distinct complex norms,
  so we are done by Corollary~\ref{cor:k>r:binary:interpolate}.
  
  \item Suppose $\Phi = \omega^{2 \ell}$ for some $\ell \in \{0,1\}$ but either $\mathfrak{C}^2 \ne \mathfrak{A}^2$ or $\kappa \ne 3$.
  Note that this is equivalent to $\Phi \ne \Psi$ or $\kappa \ne 3$.
  Consider the set $S = \{(0,0), (0,1), (0,2), (1,1), (1,2), (2,2)\}$.
  If there exists some $(r,s) \in S$ such that $g(r,s) \ne 0$ and the eigenvalues of $M(r, s)$ have distinct complex norms,
  then we are done by Corollary~\ref{cor:k>r:binary:interpolate}.
 
  Otherwise, for every $(r,s) \in S$,
  either $g(r,s) = 0$ or the eigenvalues of $M(r,s)$ have the same complex norm.
  If the latter condition were to always hold,
  then we would have
  \begin{alignat*}{2}
   \left|2          + \Psi\right| &= \left|\omega^{2 \ell}     + \kappa - 1\right| &&= \left|{-}1        + \Psi\right|,\\
   \left|2 \omega^2 + \Psi\right| &= \left|\omega^{2 \ell + 1} + \kappa - 1\right| &&= \left|{-}\omega^2 + \Psi\right|, \text{and}\\
   \left|2 \omega   + \Psi\right| &= \left|\omega^{2 \ell + 2} + \kappa - 1\right| &&= \left|{-}\omega   + \Psi\right|,
  \end{alignat*}
  where each equality corresponds to one of the six $M(r,s)$ having eigenvalues of equal complex norm for $(r,s) \in S$.
  Of the six equalities,
  at most one may not hold since $g(r,s) = 0$ for at most one $(r,s) \in S$.
  Since $n = 3$,
  two of the three terms of the form $|\omega^{2 \ell + m} + \kappa - 1|$ must be equal,
  so we can write the stronger condition
  \begin{gather}
   \left|2 \omega^2 + \Psi \omega^\ell\right| = \left|\omega   + \kappa - 1\right| = \left|{-}\omega^2 + \Psi \omega^\ell\right| \notag\\
                            \rotatebox{90}{=} \label{eqn:binary:general}\\
   \left|2 \omega   + \Psi \omega^\ell\right| = \left|\omega^2 + \kappa - 1\right| = \left|{-}\omega   + \Psi \omega^\ell\right|\rlap{.} \notag
  \end{gather}
  As it is, one of the horizontal equalities in~(\ref{eqn:binary:general}) may not hold.
  However, even without one of these equalities, we can still reach a contradiction.
  
  We show that $\Psi \omega^{\ell} \in \R$ even if one of the equalities in~(\ref{eqn:binary:general}) does not hold.
  In fact, either the left or the right half of the equalities in~(\ref{eqn:binary:general}) hold.
  In the first case,
  $|2 \omega^2 + \Psi \omega^\ell| = |2 \omega + \Psi \omega^\ell|$ holds and we get $\Psi \omega^\ell \in \R$.
  Similarly in the second case,
  $|{-}\omega^2 + \Psi \omega^\ell| = |{-}\omega + \Psi \omega^\ell|$ holds and we get $\Psi \omega^\ell \in \R$ as well.
  Next, we use real and imaginary parts to calculate the complex norms even if one of the equalities in~(\ref{eqn:binary:general}) does not hold.
  Either the top half of the equalities hold and thus $|2 \omega^2 + \Psi \omega^\ell| = |{-}\omega^2 + \Psi \omega^\ell|$,
  or the bottom half of the equalities hold and thus $|2 \omega + \Psi \omega^\ell| = |{-}\omega + \Psi \omega^\ell|$.
  In any case,
  it readily follows that $\Psi \omega^\ell = 1$.
  This implies $\Psi = \omega^{2 \ell}$,
  so we can rewrite~(\ref{eqn:binary:general}) as
  \begin{gather*}
   \sqrt{3} = \left|\omega   + \kappa - 1\right| = \sqrt{3}\\
                   \rotatebox{90}{=}\\
   \sqrt{3} = \left|\omega^2 + \kappa - 1\right| = \sqrt{3}\rlap{,}
  \end{gather*}
  where at most one equation may not hold.
  This forces $\kappa = 3$.
  However, $\Phi = \omega^{2 \ell} = \Psi$ and $\kappa = 3$ is a contradiction.
  \qedhere
 \end{enumerate}
\end{proof}

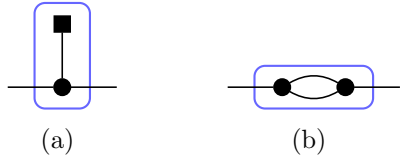
\begin{figure}[t]
 \centering
 \subcaptionbox
 {\label{subfig:gadget:binary:unary}}{
  \begin{tikzpicture}[scale=\scale,transform shape,node distance=\nodeDist,semithick]
   \node[internal] (0)              {};
   \node[external] (1) [left  of=0] {};
   \node[square]   (2) [above of=0] {};
   \node[external] (3) [right of=0] {};
   \path (0) edge (1)
             edge (2)
             edge (3);
   \begin{pgfonlayer}{background}
    \node[draw=\borderColor,thick,rounded corners,inner xsep=12pt,inner ysep=8pt,fit = (0) (2)] {};
   \end{pgfonlayer}
  \end{tikzpicture}}
 \qquad
 \subcaptionbox
 {\label{subfig:gadget:binary:parallel_edges}}{
  \begin{tikzpicture}[scale=\scale,transform shape,node distance=\nodeDist,semithick]
   \node[external] (0)              {};
   \node[internal] (1) [right of=0] {};
   \node[internal] (2) [right of=1] {};
   \node[external] (3) [right of=2] {};
   \path (0) edge             (1)
         (1) edge[bend left]  (2)
             edge[bend right] (2)
         (2) edge             (3);
   \begin{pgfonlayer}{background}
    \node[draw=\borderColor,thick,rounded corners,inner xsep=12pt,inner ysep=8pt,fit = (1) (2)] {};
   \end{pgfonlayer}
  \end{tikzpicture}}
 \caption{Binary gadgets used to interpolate any succinct binary signature of type $\tau_2$.
 The circle vertices are assigned $\langle a,b,c \rangle$ and the square vertex is assigned $\langle 1 \rangle$.}
 \label{fig:gadget:binary}
\end{figure}

The previous lemma is strong enough to handle the typical case.

\begin{lemma} \label{lem:binary:general}
 Suppose $\kappa \ge 3$ is the domain size and $a,b,c \in \mathbb{C}$.
 Let $\mathcal{F}$ be a signature set containing the succinct ternary signature $\langle a,b,c \rangle$ of type $\tau_3$
 and the succinct unary signature $\langle 1 \rangle$ of type $\tau_1$.
 If
 \begin{enumerate}
  \item $\mathfrak{B} \ne 0$,
  \item $\mathfrak{C} \ne 0$,
  \item $\mathfrak{C}^2 \ne \mathfrak{B}^2$, and
  \item either $\mathfrak{C}^2 \ne \mathfrak{A}^2$ or $\kappa \ne 3$,
 \end{enumerate}
 then
 \[
  \PlHolant(\mathcal{F} \union \{\langle x,y \rangle\}) \le_T \PlHolant(\mathcal{F})
 \]
 for any $x,y \in \mathbb{C}$,
 where $\langle x,y \rangle$ is a succinct binary signature of type $\tau_2$.
\end{lemma}

\begin{proof}
 Let $\omega = \frac{\mathfrak{C}}{\mathfrak{B}}$,
 which is well-defined.
 Consider the gadget in Figure~\ref{subfig:gadget:binary:unary}.
 We assign $\langle a,b,c \rangle$ to the circle vertex and $\langle 1 \rangle$ to the square vertex.
 Up to a nonzero factor of $\frac{\mathfrak{B}}{\kappa}$,
 this gadget has the succinct binary signature
 \[
  \frac{\kappa}{\mathfrak{B}} \langle a + (\kappa - 1) b, 2 b + (\kappa - 2) c \rangle = \langle \omega + \kappa - 1, \omega - 1 \rangle
 \]
 of type $\tau_2$.
 Then we are done by Lemma~\ref{lem:binary:general:root_of_unity} with $\ell = 1$
 in case~\ref{case:lem:binary:general:frakC^2=w^(2ell)frakB^2} of condition~\ref{condition:lem:binary:general:two_cases}.
\end{proof}

If $\mathfrak{B} = 0$,
then we already know the complexity by Corollary~\ref{cor:unary:dichotomy:a+(k-3)b-(k-2)c=0}.
The other failure conditions from the previous lemma are:
\begin{align}
 \mathfrak{C} - \mathfrak{B} &= \kappa [2 b + (\kappa - 2) c] = 0; \label{eqn:case:binary:fail:2b+(k-2)c=0}\\
 \mathfrak{C} + \mathfrak{B} &= 2 a + 2 (2 \kappa - 3) b + (\kappa - 2)^2 c = 0; \label{eqn:case:binary:fail:2a+2(2k-3)b+(k-2)^2c=0}\\
 \mathfrak{C} &= 0; \label{eqn:case:binary:fail:a+3(k-1)b+(k-2)(k-1)c=0}\\
 \kappa = 3 \text{ and } \mathfrak{C} - \mathfrak{A} &= 0, \qquad \text{or equivalently} \qquad \kappa = 3 \text{ and } b = 0; \label{eqn:case:binary:fail:k=3_and_b=0}\\
 \kappa = 3 \text{ and } \mathfrak{C} + \mathfrak{A} &= 0, \qquad \text{or equivalently} \qquad \kappa = 3 \text{ and } 2 a + 3 b + 4 c = 0. \label{eqn:case:binary:fail:k=3_and_2a+3b+4c=0}
\end{align}
Notice that these five failure conditions are \emph{linear} in $a,b,c$.

By starting the proof with a different gadget,
Lemma~\ref{lem:binary:general:root_of_unity} can handle the first three failure conditions.
The last two failure conditions require a new idea,
Eigenvalue Shifted Triples,
which we introduce in Section~\ref{subsec:binary:EST}.
In fact, these two cases are equivalent under an orthogonal holographic transformation.

The next lemma considers the failure condition in~(\ref{eqn:case:binary:fail:2b+(k-2)c=0}).
Note that $\mathfrak{C} = \mathfrak{B}$ iff the signature can be written as $\langle 2 a, -(\kappa - 2) c, 2 c \rangle$ up to a factor of~$2$.
The first excluded case in Lemma~\ref{lem:binary:2b+(k-2)c=0} is handled by Corollary~\ref{cor:unary:dichotomy:a+(k-3)b-(k-2)c=0}
and the last two excluded cases are tractable by Corollary~\ref{cor:tractable:100}.

\begin{lemma} \label{lem:binary:2b+(k-2)c=0}
 Suppose $\kappa \ge 3$ is the domain size and $a, c \in \mathbb{C}$.
 Let $\mathcal{F}$ be a signature set containing the succinct ternary signature $\langle 2 a, -(\kappa - 2) c, 2 c \rangle$ of type $\tau_3$
 and the succinct unary signature $\langle 1 \rangle$ of type $\tau_1$.
 If
 \begin{enumerate}
  \item $2 a \ne (\kappa - 1) (\kappa - 2) c$, \label{cond:lem:binary:2b+(k-2)c=0:2a!=(k-1)(k-2)c}
  \item $4 a \ne (\kappa^2 - 6 \kappa + 4) c$, and \label{cond:lem:binary:2b+(k-2)c=0:4a!=(k^2-6k+4)c}
  \item $c \ne 0$, \label{cond:lem:binary:2b+(k-2)c=0:c!=0}
 \end{enumerate}
 then
 \[
  \PlHolant(\mathcal{F} \union \{\langle x,y \rangle\}) \le_T \PlHolant(\mathcal{F})
 \]
 for any $x,y \in \mathbb{C}$,
 where $\langle x,y \rangle$ is a succinct binary signature of type $\tau_2$.
\end{lemma}

\begin{proof}
 Note that when $2 b  = - (\kappa - 2) c$,
 we have $\mathfrak{B} = \mathfrak{C} = 2 a - (\kappa - 1) (\kappa - 2) c$ by~(\ref{eqn:case:binary:fail:2b+(k-2)c=0}),
 which is nonzero by condition~\ref{cond:lem:binary:2b+(k-2)c=0:2a!=(k-1)(k-2)c} of the lemma.
 Let $\omega_0 = 4 a^2 + (\kappa - 2) [4 ac + (2 \kappa^2 + \kappa - 2) c^2]$
 and assume $\omega_0 \ne 0$.
 Then let $\omega = \frac{\mathfrak{B}^2}{\omega_0} \ne 0$.
 By conditions~\ref{cond:lem:binary:2b+(k-2)c=0:4a!=(k^2-6k+4)c} and~\ref{cond:lem:binary:2b+(k-2)c=0:c!=0},
 it follows that $\omega \ne 1$.
 Also we note that when $2 b  = - (\kappa - 2) c$,
 we have $2 \mathfrak{A} = 2 a + (3 \kappa - 2) c$ and $2 \mathfrak{C} = 2 a - (\kappa - 1) (\kappa - 2) c$.
 By the same conditions,~\ref{cond:lem:binary:2b+(k-2)c=0:4a!=(k^2-6k+4)c} and~\ref{cond:lem:binary:2b+(k-2)c=0:c!=0},
 we have $\mathfrak{C}^2 \ne \mathfrak{A}^2$.
 We further assume that $\omega \ne -1$,
 which is equivalent to $8 a^2 - 4 (\kappa - 2)^2 a c + (\kappa - 2) (\kappa^3 - 2 \kappa^2 + 6 \kappa - 4) c^2 \ne 0$.
 
 Consider the gadget in Figure~\ref{subfig:gadget:binary:parallel_edges}.
 We assign $\langle 2 a, -(\kappa - 2) c, 2 c \rangle$ to the vertices.
 Up to a nonzero factor of $\frac{\omega_0}{\kappa}$,
 this gadget has the succinct binary signature
 \[
  \frac{\kappa}{\omega_0}
  \langle
   4 a^2 + (\kappa - 1) (\kappa - 2) (3 \kappa - 2) c^2,
   \quad
   -(\kappa - 2) [4 ac - (\kappa^2 - 6 \kappa + 4) c^2] \rangle
  =
  \langle
   \omega + \kappa - 1,
   \quad
   \omega - 1
  \rangle
 \]
 of type $\tau_2$.
 Then we are done by Lemma~\ref{lem:binary:general:root_of_unity} with $\ell = 0$
 in case~\ref{case:lem:binary:general:frakC^2=w^(2ell)frakB^2} of condition~\ref{condition:lem:binary:general:two_cases}.
 
 Now we deal with the following exceptional cases.
 \begin{enumerate}
  \item If $\omega_0 = 0$,
  then $2 a = -\big[\kappa - 2 \pm i \kappa \sqrt{2 (\kappa - 2)}\big] c$.
  Up to a nonzero factor of $-c$,
  we have $-\frac{1}{c} \langle 2 a, -(\kappa - 2) c, 2 c \rangle = \langle \kappa - 2 \pm i \kappa \sqrt{2 (\kappa - 2)}, \kappa - 2, -2 \rangle$
  and are done by case~\ref{case:lem:appendix:k-2pmik2k-2k-2-2} of Lemma~\ref{lem:appendix:binary}.
  
  \item If $8 a^2 - 4 (\kappa - 2)^2 a c + (\kappa - 2) (\kappa^3 - 2 \kappa^2 + 6 \kappa - 4) c^2 = 0$,
  then $4 a = \big[(\kappa - 2)^2 \pm i \kappa \sqrt{\kappa^2 - 4}\big] c$.
  Up to a nonzero factor of $\frac{c}{2}$,
  we have
  \[
   \frac{2}{c} \langle 2 a, -(\kappa - 2) c, 2 c \rangle
   = \langle (\kappa - 2)^2 \pm i \kappa \sqrt{\kappa^2 - 4}, -2 (\kappa - 2), 4 \rangle
  \]
  and are done by case~\ref{case:lem:appendix:k-22pmikk2-2-2k-24} of Lemma~\ref{lem:appendix:binary}.
  \qedhere
 \end{enumerate}
\end{proof}

The next lemma considers the failure condition in~(\ref{eqn:case:binary:fail:2a+2(2k-3)b+(k-2)^2c=0}).
Note that $\mathfrak{C} = -\mathfrak{B}$ iff the signature can be written as $\langle -2 (2 \kappa - 3) b - (\kappa - 2)^2 c, 2 b, 2 c \rangle$ up to a factor of~$2$.
The first excluded case in Lemma~\ref{lem:binary:2a+2(2k-3)b+(k-2)^2c=0} is handled by Corollary~\ref{cor:unary:dichotomy:a+(k-3)b-(k-2)c=0}
and the last excluded case is tractable by Corollary~\ref{cor:tractable:u^21u_a+5b+2c=5b^2+2bc+c^2=0}.

\begin{lemma} \label{lem:binary:2a+2(2k-3)b+(k-2)^2c=0}
 Suppose $\kappa \ge 3$ is the domain size and $a, b \in \mathbb{C}$.
 Let $\mathcal{F}$ be a signature set containing the succinct ternary signature $\langle -2 (2 \kappa - 3) b - (\kappa - 2)^2 c, 2 b, 2 c \rangle$ of type $\tau_3$
 and the succinct unary signature $\langle 1 \rangle$ of type $\tau_1$.
 If
 \begin{enumerate}
  \item $2 b \ne -(\kappa - 2) c$ and \label{cond:lem:binary:2a+2(2k-3)b+(k-2)^2c=0:2b!=-(k-2)c}
  \item $\kappa \ne 4$ or $5 b^2 + 2 b c + c^2 \ne 0$, \label{cond:lem:binary:2a+2(2k-3)b+(k-2)^2c=0:k!=4_5b^2+2bc+c^2!=0}
 \end{enumerate}
 then
 \[
  \PlHolant(\mathcal{F} \union \{\langle x,y \rangle\}) \le_T \PlHolant(\mathcal{F})
 \]
 for any $x,y \in \mathbb{C}$,
 where $\langle x,y \rangle$ is a succinct binary signature of type $\tau_2$.
\end{lemma}

\begin{proof}
 Note that when $2 a = -2 (2 \kappa - 3) b - (\kappa - 2)^2 c$,
 we have $\mathfrak{B} = -\mathfrak{C}$ by~(\ref{eqn:case:binary:fail:2a+2(2k-3)b+(k-2)^2c=0})
 and $2 \mathfrak{B} = -\kappa [2 b + (\kappa - 2) c]$,
 which is nonzero by condition~\ref{cond:lem:binary:2a+2(2k-3)b+(k-2)^2c=0:2b!=-(k-2)c} of the lemma.
 Let $\omega_0 = 8 (2 \kappa - 3) b^2 + (\kappa - 2) \left[8 (\kappa - 3) b c + (\kappa^2 - 6 \kappa + 12) c^2\right]$
 and assume $\omega_0 \ne 0$.
 Then let $\omega = \frac{\kappa [2 b + (\kappa - 2) c]^2}{\omega_0}$.
 By condition~\ref{cond:lem:binary:2a+2(2k-3)b+(k-2)^2c=0:2b!=-(k-2)c},
 $\omega \ne 0$.
 It can be shown that $\kappa [2 b + (\kappa - 2) c]^2 = \omega_0$ is equivalent to $(b - c) [3 b + (\kappa - 3) c] = 0$.
 Thus, assume $b \ne c$ and $3 b \ne -(\kappa - 3) c$.
 Then $\omega \ne 1$.
 Also we note that when $2 a = -2 (2 \kappa - 3) b - (\kappa - 2)^2 c$,
 we have $2 \mathfrak{A} = -\kappa [4 b + (\kappa - 4) c]$ and $2 \mathfrak{C} = \kappa [2 b + (\kappa - 2) c]$.
 By the same assumptions,
 $b \ne c$ and $3 b \ne -(\kappa - 3) c$,
 we have $\mathfrak{C}^2 \ne \mathfrak{A}^2$.
 Further assume that $\omega \ne -1$,
 which is equivalent to $2 (5 \kappa - 6) b^2 + (\kappa - 2) [6 (\kappa - 2) b c + (\kappa^2 - 4 \kappa + 6) c^2] \ne 0$.
 
 Consider the gadget in Figure~\ref{subfig:gadget:binary:parallel_edges}.
 We assign $\langle -2 (2 \kappa - 3) b - (\kappa - 2)^2 c, 2 b, 2 c \rangle$ to the vertices.
 Up to a nonzero factor of $\frac{\omega_0}{4}$,
 this gadget has the succinct binary signature $\frac{1}{\omega_0} \langle x, y \rangle = \langle \omega + \kappa - 1, \omega - 1 \rangle$ of type $\tau_2$,
 where
 \begin{align*}
  x &= 4 (4 \kappa^2 - 9 \kappa + 6) b^2 + (\kappa - 2) \left[4 (\kappa - 2) (2 \kappa - 3) b c + (\kappa^3 - 6 \kappa^2 + 16 \kappa - 12) c^2\right]
  \qquad \text{and}\\
  y &= -4 (\kappa - 2) \left[3 b^3 + (\kappa - 6) b c - (\kappa - 3) c^2\right].
 \end{align*}
 Then we are done by Lemma~\ref{lem:binary:general:root_of_unity} with $\ell = 0$
 in case~\ref{case:lem:binary:general:frakC^2=w^(2ell)frakB^2} of condition~\ref{condition:lem:binary:general:two_cases}.
 
 Now we deal with the following exceptional cases.
 \begin{enumerate}
  \item If $\omega_0 = 0$,
  then we have $-4 (2 \kappa - 3) b = \big[2 (\kappa - 3) (\kappa - 2) \pm i \kappa \sqrt{2 (\kappa - 2)}\big] c$
  but $\kappa \ne 4$ by condition~\ref{cond:lem:binary:2a+2(2k-3)b+(k-2)^2c=0:k!=4_5b^2+2bc+c^2!=0} since otherwise $\omega_0 = 8 (5 b^2 + 2 b c + c^2) \ne 0$.
  Up to a nonzero factor of $\frac{c}{2 (2 \kappa - 3)}$,
  \begin{multline*}
   \frac{2 (2 \kappa - 3)}{c}
   \langle -2 (2 \kappa - 3) b - (\kappa - 2)^2 c, 2 b, 2 c \rangle\\
   =
   \langle
    -(2 \kappa - 3) \big[2 (\kappa - 2) \mp i \kappa \sqrt{2 (\kappa - 2)}\big],
    \quad
    -2 (\kappa - 3) (\kappa - 2) \mp i \kappa \sqrt{2 (\kappa - 2)},
    \quad
    4 (2 \kappa - 3)
   \rangle
  \end{multline*}
  and are done by case~\ref{case:lem:appendix:-2k-32k-2pmik2k-2-2k-3k-2pmik2k-242k-3} of Lemma~\ref{lem:appendix:binary}.
  
  \item If $b = c$,
  then up to a nonzero factor of $c$,
  we have $\frac{1}{c} \langle -2 (2 \kappa - 3) b - (\kappa - 2)^2 c, 2 b, 2 c \rangle = \langle -\kappa^2 + 2, 2, 2 \rangle$
  and are done by case~\ref{case:lem:appendix:-k2+222} Lemma~\ref{lem:appendix:binary}.
  
  \item If $3 b = -(\kappa - 3) c$,
  then up to a nonzero factor of $\frac{c}{3}$,
  we have $\frac{3}{c} \langle -2 (2 \kappa - 3) b - (\kappa - 2)^2 c, 2 b, 2 c \rangle = \langle \kappa^2 - 6 \kappa + 6, -2 (\kappa - 3), 6 \rangle$
  and are done by case~\ref{case:lem:appendix:k2-6k+6-2k-36} of Lemma~\ref{lem:appendix:binary}.
  
  \item If $2 (5 \kappa - 6) b^2 + (\kappa - 2) [6 (\kappa - 2) b c + (\kappa^2 - 4 \kappa + 6) c^2] = 0$,
  then $-2 (5 \kappa - 6) b = \big[3 (\kappa - 2)^2  \pm i \kappa \sqrt{\kappa^2 - 4}\big] c$.
  Up to a nonzero factor of $\frac{c}{5 \kappa - 6}$,
  \begin{multline*}
   \frac{5 \kappa - 6}{c}
   \langle -2 (2 \kappa - 3) b - (\kappa - 2)^2 c, 2 b, 2 c \rangle\\
   =
   \langle
    (\kappa - 3) (\kappa - 2)^2 \pm i \kappa (2 \kappa - 3) \sqrt{\kappa^2 - 4},
    \quad
    -3 (\kappa - 2)^2 \mp i \kappa \sqrt{\kappa^2 - 4},
    \quad
    2 (5 \kappa - 6)
   \rangle
  \end{multline*}
  and are done by case~\ref{case:lem:appendix:k-3k-22pmik2k-3-k2-4-3k-22mpikk2-425k-6} of Lemma~\ref{lem:appendix:binary}.
  \qedhere
 \end{enumerate}
\end{proof}

The next lemma considers the failure condition in~(\ref{eqn:case:binary:fail:a+3(k-1)b+(k-2)(k-1)c=0}).
Note that $\mathfrak{C} = 0$ iff the signature can be written as $\langle -3 (\kappa - 1) b - (\kappa - 1) (\kappa - 2) c, b, c \rangle$.
The excluded case in Lemma~\ref{lem:binary:a+3(k-1)b+(k-2)(k-1)c=0} is handled by Corollary~\ref{cor:unary:dichotomy:a+(k-3)b-(k-2)c=0}.

\begin{lemma} \label{lem:binary:a+3(k-1)b+(k-2)(k-1)c=0}
 Suppose $\kappa \ge 3$ is the domain size and $b, c \in \mathbb{C}$.
 Let $\mathcal{F}$ be a signature set containing the succinct ternary signature $\langle -3 (\kappa - 1) b - (\kappa - 1) (\kappa - 2) c, b, c \rangle$ of type $\tau_3$
 and the succinct unary signature $\langle 1 \rangle$ of type $\tau_1$.
 If $2 b \ne -(\kappa - 2) c$,
 then
 \[
  \PlHolant(\mathcal{F} \union \{\langle x,y \rangle\}) \le_T \PlHolant(\mathcal{F})
 \]
 for any $x,y \in \mathbb{C}$,
 where $\langle x,y \rangle$ is a succinct binary signature of type $\tau_2$.
\end{lemma}

\begin{proof}
 Note that when $a = -3 (\kappa - 1) b - (\kappa - 2) (\kappa - 1) c$,
 we have $\mathfrak{C} = 0$
 and $2 \mathfrak{B} = -\kappa [2 b + (\kappa - 2) c]$,
 which is nonzero by assumption.
 Let $\omega_0 = (9 \kappa - 10) b^2 + (\kappa - 2) [2 (3 \kappa - 5) b c + (\kappa^2 - 4 \kappa + 5) c^2]$
 and assume $\omega_0 \ne 0$.
 Then let $\omega = \frac{(\kappa - 1) [2 b + (\kappa - 2) c]^2}{\omega_0}$.
 By assumption, $\omega \ne 0$.
 Assume $\omega \ne 1$,
 which is equivalent to $-(5 \kappa - 6) b^2 - (\kappa - 3) (\kappa - 2) (2 b - c) c \ne 0$.
 Further assume $\omega \ne -1$,
 which is equivalent to $(13 \kappa - 14) b^2 + (\kappa - 2) [2 (5 \kappa - 7) b c + (2 \kappa^2 - 7 \kappa + 7) c^2] \ne 0$.
 
 Consider the gadget in Figure~\ref{subfig:gadget:binary:parallel_edges}.
 We assign $\langle -3 (\kappa - 1) b - (\kappa - 1) (\kappa - 2) c, b, c \rangle$ to the vertices.
 Up to a nonzero factor of $\omega_0$,
 this gadget has the succinct binary signature $\frac{1}{\omega_0} \langle x,y \rangle = \langle \omega + \kappa - 1, \omega - 1 \rangle$ of type $\tau_2$,
 where
 \begin{align*}
  x &= (\kappa - 1) \left\{3 (3 \kappa - 2) b^2 + (\kappa - 2) \left[6 b c + (\kappa^2 - 3 \kappa + 3) c^2\right]\right\}
  \qquad \text{and}\\
  y &= -(5 \kappa - 6) b^2 - (\kappa - 3) (\kappa - 2) (2 b - c) c.
 \end{align*}
 Then we are done by Lemma~\ref{lem:binary:general:root_of_unity}
 via case~\ref{case:lem:binary:general:frakC=0} of condition~\ref{condition:lem:binary:general:two_cases}.
 
 Now we deal with the following exceptional cases.
 \begin{enumerate}
  \item If $\omega_0 = 0$,
  then $-(9 \kappa - 10) b = [(\kappa - 2) (3 \kappa - 5) \pm i \kappa \sqrt{2 (\kappa - 2)}] c$.
  Up to a nonzero factor of $\frac{c}{9 \kappa - 10}$,
  we have
  \begin{multline*}
   \frac{9 \kappa - 10}{c}
   \langle -3 (\kappa - 1) b - (\kappa - 1) (\kappa - 2) c, b, c \rangle\\
   =
   \langle
    -(\kappa - 1) \big[5 (\kappa - 2) \mp 3 i \kappa \sqrt{2 (\kappa - 2)}\big],
    \quad
    -(\kappa - 2) (3 \kappa - 5) \mp i \kappa \sqrt{2 (\kappa - 2)},
    \quad
    9 \kappa - 10
   \rangle
  \end{multline*}
  and we are done by case~\ref{case:lem:appendix:-k-15k-2mp3ik2k-2-k-23k-5mpik3k-29k-10} of Lemma~\ref{lem:appendix:binary}.
  
  \item If $-(5 \kappa - 6) b^2 - (\kappa - 3) (\kappa - 2) (2 b - c) c = 0$,
  then $-(5 \kappa - 6) b = \big[(\kappa - 3) (\kappa - 2) \pm \kappa \sqrt{\kappa^2 - 5 \kappa + 6}\big] c$.
  Up to a nonzero factor of $-\frac{c}{5 \kappa - 6}$,
  we have
  \begin{multline*}
   -\frac{5 \kappa - 6}{c}
   \langle -3 (\kappa - 1) b - (\kappa - 1) (\kappa - 2) c, b, c \rangle\\
   =
   \langle
    (\kappa - 1) \left[(\kappa - 2) (2 \kappa + 3) \mp 3 \kappa \sqrt{\kappa^2 - 5 \kappa + 6}\right],
    \quad
    (\kappa - 3) (\kappa - 2) \pm \kappa \sqrt{\kappa^2 - 5 \kappa + 6},
    \quad
    -5 \kappa + 6
   \rangle
  \end{multline*}
  and are done by case~\ref{case:lem:appendix:k-1k-22k+3mp3kk2-5k+6k-3k-2pmkk2-5k+6-5k+6} Lemma~\ref{lem:appendix:binary}.
  
  \item If $(13 \kappa - 14) b^2 + (\kappa - 2) [2 (5 \kappa - 7) b c + (2 \kappa^2 - 7 \kappa + 7) c^2] = 0$,
  then $-(13 \kappa - 14) b = \big[(\kappa - 2) (5 \kappa - 7) \pm i \kappa \sqrt{\kappa^2 - \kappa - 2}\big] c$.
  Up to a nonzero factor of $\frac{c}{13 \kappa - 14}$,
  we have
  \begin{multline*}
   \frac{13 \kappa - 14}{c}
   \langle -3 (\kappa - 1) b - (\kappa - 1) (\kappa - 2) c, b, c \rangle\\
   =
   \langle
    (\kappa - 1) \left[(\kappa - 2) (2 \kappa - 7) \pm 3 i \kappa \sqrt{\kappa^2 - \kappa - 2}\right],
    ~~~
    -(\kappa - 2) (5 \kappa - 7) \mp i \kappa \sqrt{\kappa^2 - \kappa - 2},
    ~~~
    13 \kappa - 14
   \rangle
  \end{multline*}
  and are done by case~\ref{case:lem:appendix:k-1k-22k-7pm3ikk2-k-2-k-25k-7mpikk2-k-213k-14} of Lemma~\ref{lem:appendix:binary}.
  \qedhere
 \end{enumerate}
\end{proof}

\subsection{Eigenvalue Shifted Triples} \label{subsec:binary:EST}

To handle failure conditions~(\ref{eqn:case:binary:fail:k=3_and_b=0}) and~(\ref{eqn:case:binary:fail:k=3_and_2a+3b+4c=0}) from Lemma~\ref{lem:binary:general},
we need another technique.
We introduce an Eigenvalue Shifted Triple,
which extends the concept of an Eigenvalue Shifted Pair.

\begin{definition}[Definition~4.6 in~\cite{KC10}]
 A pair of nonsingular matrices $M, M' \in \mathbb{C}^{2 \times 2}$ is called an \emph{Eigenvalue Shifted Pair}
 if $M' = M + \delta I$ for some nonzero $\delta \in \mathbb{C}$,
 and $M$ has distinct eigenvalues.
\end{definition}

Eigenvalue shifted pairs were used in~\cite{KC10} to show that interpolation succeeds in most cases
since these matrices correspond to some recursive gadget constructions and at least one of them usually has eigenvalues with distinct complex norms.
In~\cite{KC10},
it is shown that the interpolation succeeds unless the variables in question take real values.
Then other techniques were developed to handle the real case.
We use Eigenvalue Shifted Pairs in a stronger way.
We exhibit three matrices such that any two form an Eigenvalue Shifted Pair.
Provided these shifts are linearly independent over $\R$,
this is enough to show that interpolation succeeds for both real and complex settings of the variables.
We call this an Eigenvalue Shifted Triple.

\begin{definition}
 A trio of nonsingular matrices $M_0, M_1, M_2 \in \mathbb{C}^{2 \times 2}$ is called an \emph{Eigenvalue Shifted Triple} (EST)
 if $M_0$ has distinct eigenvalues and there exist nonzero $\delta_1, \delta_2 \in \mathbb{C}$ satisfying $\frac{\delta_1}{\delta_2} \not\in \R$
 such that $M_1 = M_0 + \delta_1 I$, and $M_2 = M_0 + \delta_2 I$.
\end{definition}

We note that if $M_0$, $M_1$, and $M_2$ form an Eigenvalue Shifted Triple,
then any permutation of the matrices is also an Eigenvalue Shifted Triple.

The proof of the next lemma is similar to the proof of Lemma~4.7 in~\cite{KC10_arXiv},
the full version of~\cite{KC10}.

\begin{lemma} \label{lem:binary:EST}
 Suppose $\alpha, \beta, \delta_1, \delta_2 \in \mathbb{C}$.
 If $\alpha \ne \beta$, $\delta_1, \delta_2 \ne 0$, and $\frac{\delta_1}{\delta_2} \not\in \R$,
 then $|\alpha| \ne |\beta|$ or $|\alpha + \delta_1| \ne |\beta + \delta_1|$ or $|\alpha + \delta_2| \ne |\beta + \delta_2|$.
\end{lemma}

\begin{proof}
 Assume for a contradiction that $|\alpha| = |\beta|$, $|\alpha + \delta_1| = |\beta + \delta_1|$, and $|\alpha + \delta_2| = |\beta + \delta_2|$.
 After a rotation in the complex plane,
 we can assume that $\alpha = \overline{\beta}$.
 Note that all of our assumptions are unchanged by this rotation.
 For $i \in \{1,2\}$,
 we have
 \begin{align*}
  (\alpha + \delta_i) \overline{(\alpha + \delta_i)}
  &= |\alpha + \delta_i|^2\\
  &= |\beta + \delta_i|^2\\
  &= (\beta + \delta_i) \overline{(\beta + \delta_i)}
  = (\overline{\alpha} + \delta_i) (\alpha + \overline{\delta_i}).
 \end{align*}
 This implies $(\overline{\alpha} - \alpha) (\overline{\delta_i} - \delta_i) = 0$.
 Since $\alpha \ne \beta = \overline{\alpha}$,
 we have $\delta_i \in \R$.
 Then $\frac{\delta_1}{\delta_2} \in \R$, a contradiction.
\end{proof}

\begin{figure}[t]
 \centering
 \captionsetup[subfigure]{labelformat=empty}
 \subcaptionbox{$N_0$}{
  \begin{tikzpicture}[scale=\scale,transform shape,node distance=\nodeDist,semithick]
   \node[internal] (0)                   {};
   \node[external] (1) [above left of=0] {};
   \node[external] (2) [below left of=0] {};
   \node[external] (3) [left       of=1] {};
   \node[external] (4) [left       of=2] {};
   \node[square] (7) [right      of=0] {};
   \path (0) edge[out= 135, in=0] (3)
             edge[out=-135, in=0] (4)
             edge (7);
   \begin{pgfonlayer}{background}
    \node[draw=\borderColor,thick,rounded corners,inner xsep=12pt,inner ysep=8pt,fit = (1) (2) (7)] {};
   \end{pgfonlayer}
  \end{tikzpicture}}
 \qquad
 \subcaptionbox{$N_1$}{
  \begin{tikzpicture}[scale=\scale,transform shape,node distance=\nodeDist,semithick]
   \node[internal] (0)                   {};
   \node[external] (1) [above left of=0] {};
   \node[external] (2) [below left of=0] {};
   \node[internal] (3) [left       of=1] {};
   \node[internal] (4) [left       of=2] {};
   \node[external] (5) [left       of=3] {};
   \node[external] (6) [left       of=4] {};
   \node[square] (7) [right      of=0] {};
   \path (0) edge[out= 135, in=0] (3)
             edge[out=-135, in=0] (4)
             edge (7)
         (3) edge (4)
             edge (5)
         (4) edge (6);
   \begin{pgfonlayer}{background}
    \node[draw=\borderColor,thick,rounded corners,inner xsep=12pt,inner ysep=8pt,fit = (3) (4) (7)] {};
   \end{pgfonlayer}
  \end{tikzpicture}}
 \qquad
 \subcaptionbox{$N_{s+1}$}{
  \begin{tikzpicture}[scale=\scale,transform shape,node distance=\nodeDist,semithick]
   \node[external]  (0)                   {\Huge $N_s$};
   \node[external]  (1) [above left of=0] {};
   \node[external]  (2) [below left of=0] {};
   \node[internal]  (3) [left       of=1] {};
   \node[internal]  (4) [left       of=2] {};
   \node[external]  (5) [left       of=3] {};
   \node[external]  (6) [left       of=4] {};
   \path (0) edge[out= 135, in=0] (3)
             edge[out=-135, in=0] (4)
         (3) edge (4)
             edge (5)
         (4) edge (6);
   \begin{pgfonlayer}{background}
    \node[draw=\borderColor,thick,densely dashed,rounded corners,fit = (0)] {};
    \node[draw=\borderColor,thick,rounded corners,inner xsep=12pt,inner ysep=8pt,fit = (3) (4) (0)] {};
   \end{pgfonlayer}
  \end{tikzpicture}}
 \caption{Alternative recursive construction to interpolate a binary signature (cf.~Figure~\ref{fig:gadget:k>r:binary:interpolation}).
 The circle vertices are assigned $\langle a,b,c \rangle$ and the square vertex is assigned $\langle 1 \rangle$.}
 \label{fig:gadget:binary:interpolation:alternative}
\end{figure}
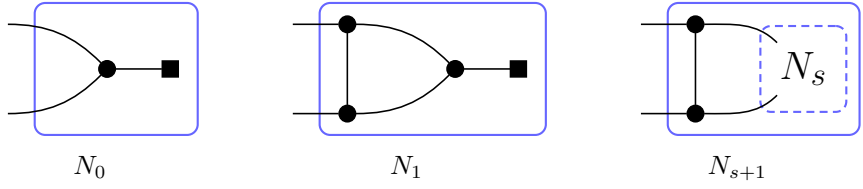

The next lemma considers the failure condition in~(\ref{eqn:case:binary:fail:k=3_and_b=0}),
which is $\kappa = 3$ and $b = 0$,
so the signature has the form $\langle a,0,c \rangle$.
If $a = 0$,
then the problem is already $\SHARPP$-hard by Theorem~\ref{thm:edge_coloring:k=r}.
If $c = 0$,
then the problem is tractable by case~\ref{case:cor:tractable:holant-star:equality} of Corollary~\ref{cor:tractable:holant-star}.
If $a^3 = c^3$,
then the problem is tractable by Corollary~\ref{cor:tractable:k=3_a^3=c^3}.

\begin{lemma} \label{lem:binary:k=3_and_b=0}
 Suppose the domain size is~$3$ and $a, c \in \mathbb{C}$.
 Let $\mathcal{F}$ be a signature set containing the succinct ternary signature $\langle a,0,c \rangle$ of type $\tau_3$
 and the succinct unary signature $\langle 1 \rangle$ of type $\tau_1$.
 If $a c \ne 0$ and $a^3 \ne c^3$,
 then
 \[
  \PlHolant(\mathcal{F} \union \{\langle x,y \rangle\}) \le_T \PlHolant(\mathcal{F})
 \]
 for any $x,y \in \mathbb{C}$,
 where $\langle x,y \rangle$ is a succinct binary signature of type $\tau_2$.
\end{lemma}

\begin{proof}
 Assume $2 a + c \ne 0$ and let $\omega = \frac{a^2 + 2 c^2}{c (2 a + c)}$.
 Assume $a^2 + 2 c^2$ so that $\omega \ne 0$.
 Further assume $a^2 + 2 a c + 3 c^2 \ne 0$ so that $\omega^2 \ne 1$
 as well as $a^2 + a c + 7 c^2 \ne 0$ so that $\omega^3 \ne 1$.
 Note that these conclusions also require $a \ne c$ and $a^3 \ne c^3$ respectively.
 
 Consider using the recursive construction in Figure~\ref{fig:gadget:binary:interpolation:alternative}.
 The circle vertices are assigned $\langle a,0,c \rangle$ and the square vertex is assigned $\langle 1 \rangle$.
 Let $z = \frac{c}{a}$,
 which is well-defined by assumption.
 The succinct signature of type $\tau_2$ for the initial gadget $N_0$ in this construction is $\langle a,c \rangle$.
 Up to a nonzero factor of $a$,
 this signature is $f_0 = \frac{1}{a} \langle a,c \rangle = \langle 1,z \rangle$.
 Then up to a nonzero factor of $c (2 a + c)$,
 the succinct signature of type $\tau_2$ for the $s$th gadget in this construction is $f_s  = \langle \omega^k, z \rangle = M^s f_0$,
 where
 \[
  M =
  \frac{1}{c (2 a + c)}
  \left[
   \begin{smallmatrix}
    a^2 + 2 c^2 & 0 \\
    0           & c (2 a + c)
   \end{smallmatrix}
  \right]
  =
  \left[
   \begin{smallmatrix}
    \omega & 0 \\
    0 & 1
   \end{smallmatrix}
  \right].
 \]
 
 Clearly $M$ is nonsingular.
 The determinant of $[f_0\ M f_0] = \left[\begin{smallmatrix} a & a \omega \\ c & c \end{smallmatrix}\right]$ is $z (1 - \omega) \ne 0$.
 If $\omega$ is not a root of unity,
 then we are done by Lemma~\ref{lem:k>r:binary:interpolate}.
 
 Otherwise, suppose $\omega$ is a primitive root of unity of order $n$.
 By assumption, $n \ge 4$.
 Now consider the recursive construction in Figure~\ref{fig:gadget:k>r:binary:interpolation}.
 We assign $f_s$ to every vertex,
 where $s \ge 0$ is a parameter of our choice.
 Let $g_t(s)$ be the signature of the $t$th gadget in this recursive construction when using $f_s$.
 Then $g_1(s) = f_s$ and $g_t(s) = (N(s))^t \left[\begin{smallmatrix} 1 \\ 0 \end{smallmatrix}\right]$,
 where $N(s) = \left[\begin{smallmatrix} \omega^s & 2 z \\ z & \omega^s + z \end{smallmatrix}\right]$.
 
 By Lemma~\ref{lem:k>r:binary:interpolation:eigenvalues},
 the eigenvalues of $N(s)$ are $\omega^s + 2 z$ and $\omega^s - z$,
 which means the determinant of $N(s)$ is $(\omega^s + 2 z) (\omega^s - z)$.
 Each eigenvalue can vanish for at most one value of $s \in \Z_n$ since both eigenvalues are linear polynomials in $\omega^s$ that are not identically~$0$.
 Furthermore, at least one of the eigenvalues never vanishes for all $s \in \Z_n$ since otherwise $1 = |z| = \frac{1}{2}$.
 
 Thus, at most one matrix among $N(0)$, $N(1)$, $N(2)$, and $N(3)$ can be singular.
 Pick distinct $j, k, \ell \in \{0,1,2,3\}$ such that $N(j)$, $N(k)$, and $N(\ell)$ are nonsingular.
 To finish the proof,
 we show that $N(j)$, $N(k)$, and $N(\ell)$ form an Eigenvalue Shifted Triple.
 Then by Lemma~\ref{lem:binary:EST},
 at least one of the matrices has eigenvalues with distinct complex norms,
 so we are done by Corollary~\ref{cor:k>r:binary:interpolate}.
 
 The eigenvalue shift from $N(j)$ to $N(k)$ is $\delta_{j,k} = \omega^j (\omega^{k-j} - 1)$,
 which is nonzero since $j$ and $k$ are distinct in $\Z_n$.
 Assume for a contradiction that $\frac{\delta_{j,k}}{\delta_{j,\ell}} \in \R$,
 which is equivalent to $\arg(\delta_{j,k}) = \arg(\pm \delta_{j,\ell})$.
 Then we have
 \begin{equation}
  \arg\left(\omega^{k-j} - 1\right) = \arg\left(\pm (\omega^{\ell-j} - 1)\right). \label{eqn:binary:b=0:arg}
 \end{equation}
 In the complex plane,
 any nonzero $x - 1 \in \mathbb{C}$ with $|x| = 1$ lies on the circle of radius~$1$ centered at $(-1,0)$.
 Such $x$ satisfy $\frac{\pi}{2} < \arg(x - 1) < \frac{3 \pi}{2}$.
 Thus, the argument of $x - 1$ is unique, even up to a sign, contradicting~(\ref{eqn:binary:b=0:arg}).
 Therefore, $M_j$, $M_k$, and $M_\ell$ form an Eigenvalue Shifted Triple as claimed.
 
 Now we deal with the following exceptional cases.
 \begin{enumerate}
  \item If $2 a + c = 0$,
  then up to a nonzero factor of $a$,
  we have $\frac{1}{a} \langle a,0,c \rangle = \langle 1, 0, -2 \rangle$
  and are done by case~\ref{case:lem:appendix:10-2} of Lemma~\ref{lem:appendix:binary}.
  
  \item If $a^2 + 2 c^2 = 0$,
  then $a = \pm i \sqrt{2} c$.
  Up to a nonzero factor of $c$,
  we have $\frac{1}{c} \langle a,0,c \rangle = \langle \pm i \sqrt{2}, 0, 1 \rangle$
  and are done by case~\ref{case:lem:appendix:pmi201} of Lemma~\ref{lem:appendix:binary}.
  
  \item If $a^2 + 2 a c + 3 c^2 = 0$,
  then $a = c (-1 \pm i \sqrt{2})$.
  Up to a nonzero factor of $c$,
  we have $\frac{1}{c} \langle a,0,c \rangle = \langle -1 \pm i \sqrt{2}, 0, 1 \rangle$
  and are done by case~\ref{case:lem:appendix:-1pmi201} of Lemma~\ref{lem:appendix:binary}.
  
  \item If $a^2 + a c + 7 c^2 = 0$,
  then $2 a = c (-1 \pm 3 i \sqrt{3})$.
  Up to a nonzero factor of $\frac{c}{2}$,
  we have $\frac{2}{c} \langle a,0,c \rangle = \langle -1 \pm 3 i \sqrt{3}, 0, 2 \rangle$
  and are done by case~\ref{case:lem:appendix:-1pm3i302} of Lemma~\ref{lem:appendix:binary}.
  \qedhere
 \end{enumerate}
\end{proof}

The next lemma considers the failure condition in~(\ref{eqn:case:binary:fail:k=3_and_2a+3b+4c=0}).
Since this failure condition is just a holographic transformation of the failure condition in~(\ref{eqn:case:binary:fail:k=3_and_b=0}),
the excluded cases in this lemma are handled exactly as those preceding Lemma~\ref{lem:binary:k=3_and_b=0}.

\begin{lemma} \label{lem:binary:k=3_and_2a+3b+4c=0}
 Suppose the domain size is~$3$ and $b, c \in \mathbb{C}$.
 Let $\mathcal{F}$ be a signature set containing the succinct ternary signature $\langle -3 b -4 c, 2 b, 2 c \rangle$ of type $\tau_3$
 and the succinct unary signature $\langle 1 \rangle$ of type $\tau_1$.
 Assume $T^{\otimes 3} \langle -3 b -4 c, 2 b, 2 c \rangle = \langle \hat{a}, \hat{b}, \hat{c} \rangle$,
 where $T = \left[\begin{smallmatrix*}[r] 1 & -2 & -2 \\ -2 & 1 & -2 \\ -2 & -2 & 1 \end{smallmatrix*}\right]$.
 If $\hat{a} \hat{c} \ne 0$ and $\hat{a}^3 \ne \hat{c}^3$,
 then
 \[
  \PlHolant(\mathcal{F} \union \{\langle x,y \rangle\}) \le_T \PlHolant(\mathcal{F})
 \]
 for any $x,y \in \mathbb{C}$,
 where $\langle x,y \rangle$ is a succinct binary signature of type $\tau_2$.
\end{lemma}

\begin{proof}
 By Lemma~\ref{lem:compute:ternary:holographic_transformation} with $x = 1$ and $y = -2$,
 we have $\hat{b} = 0$.
 Thus after a holographic transformation by $T$,
 we are in the case covered by Lemma~\ref{lem:binary:k=3_and_b=0}.
 Since $T$ is orthogonal after scaling by $\frac{1}{3}$,
 the complexity of these problems are unchanged by Theorem~\ref{thm:ortho_holo_trans}.
\end{proof}

We summarize this section with the following lemma.

\begin{corollary} \label{cor:binary:interpolate}
 Suppose the domain size is $\kappa \ge 3$ and $a, b, c \in \mathbb{C}$.
 Let $\mathcal{F}$ be a signature set containing the succinct ternary signature $\langle a, b, c \rangle$ of type $\tau_3$
 and the succinct unary signature $\langle 1 \rangle$ of type $\tau_1$.
 Then
 \[
  \PlHolant(\mathcal{F} \union \{\langle x,y \rangle\}) \le_T \PlHolant(\mathcal{F})
 \]
 for any $x,y \in \mathbb{C}$,
 where $\langle x,y \rangle$ is a succinct binary signature of type $\tau_2$,
 unless
 \begin{itemize}
  \item $\mathfrak{B} = 0$ or
  \item there exist $\lambda \in \mathbb{C}$ and $T \in \left\{I_\kappa, \kappa I_\kappa - 2 J_\kappa\right\}$ such that
  \[
   \langle a,b,c \rangle =
   \begin{cases}
    T^{\otimes 3} \lambda \langle 1,0,0 \rangle, \text{ or}\\
    T^{\otimes 3} \lambda \langle 0,0,1 \rangle \text{ and } \kappa = 3, \text{ or}\\
    T^{\otimes 3} \lambda \langle 1,0,\omega  \rangle \text{ and } \kappa = 3 \text{ where } \omega^3 = 1, \text{ or}\\
    T^{\otimes 3} \lambda \langle \mu^2,1,\mu \rangle \text{ and } \kappa = 4 \text{ where } \mu = -1 \pm 2 i.
   \end{cases}
  \]
 \end{itemize}
\end{corollary}

\begin{proof}
 If failure condition~(\ref{eqn:case:binary:fail:2b+(k-2)c=0}),
 (\ref{eqn:case:binary:fail:2a+2(2k-3)b+(k-2)^2c=0}),
 (\ref{eqn:case:binary:fail:a+3(k-1)b+(k-2)(k-1)c=0}),
 (\ref{eqn:case:binary:fail:k=3_and_b=0}),
 or~(\ref{eqn:case:binary:fail:k=3_and_2a+3b+4c=0}) holds,
 then we are done by
 Lemma~\ref{lem:binary:2b+(k-2)c=0},
 Lemma~\ref{lem:binary:2a+2(2k-3)b+(k-2)^2c=0},
 Lemma~\ref{lem:binary:a+3(k-1)b+(k-2)(k-1)c=0},
 Lemma~\ref{lem:binary:k=3_and_b=0}, or
 Lemma~\ref{lem:binary:k=3_and_2a+3b+4c=0} respectively,
 with the various excluded cases listed.
 If none of~(\ref{eqn:case:binary:fail:2b+(k-2)c=0}),
 (\ref{eqn:case:binary:fail:2a+2(2k-3)b+(k-2)^2c=0}),
 (\ref{eqn:case:binary:fail:a+3(k-1)b+(k-2)(k-1)c=0}),
 (\ref{eqn:case:binary:fail:k=3_and_b=0}),
 and~(\ref{eqn:case:binary:fail:k=3_and_2a+3b+4c=0}) hold,
 then we are done by Lemma~\ref{lem:binary:general}.
\end{proof}

\section{The Main Dichotomy} \label{sec:dichotomy}

Now we can prove our main dichotomy theorem.

\begin{theorem} \label{thm:dichotomy}
 Suppose $\kappa \ge 3$ is the domain size and $a,b,c \in \mathbb{C}$.
 Let $\langle a,b,c \rangle$ be a succinct ternary signature of type $\tau_3$.
 Then $\PlHolant(\langle a,b,c \rangle)$ is $\SHARPP$-hard
 unless at least one of the following holds:
 \begin{enumerate}
  \item $a = b = c$; \label{case:thm:dichotomy:degenerate}
  \item $a = c$ and $\kappa = 3$; \label{case:thm:dichotomy:k=3_aba}
 \end{enumerate}
 there exists $\lambda \in \mathbb{C}$ and $T \in \left\{I_\kappa, \kappa I_\kappa - 2 J_\kappa\right\}$ such that
 \begin{enumerate}
  \setcounter{enumi}{2} 
  \item $\langle a,b,c \rangle = T^{\otimes 3} \lambda \langle 1,0,0 \rangle$; \label{case:thm:dichotomy:100}
  \item $\langle a,b,c \rangle = T^{\otimes 3} \lambda \langle 1,0,\omega  \rangle$ and $\kappa = 3$ where $\omega^3 = 1$; \label{case:thm:dichotomy:k=3_10w}
  \item $\langle a,b,c \rangle = T^{\otimes 3} \lambda \langle \mu^2,1,\mu \rangle$ and $\kappa = 4$ where $\mu = -1 \pm 2 i$; \label{case:thm:dichotomy:k=4_u^21u}
 \end{enumerate}
 in which case, the computation can be done in polynomial time.
\end{theorem}

\begin{proof}
 The signature in case~\ref{case:thm:dichotomy:degenerate} is degenerate,
 which is trivially tractable.
 Case~\ref{case:thm:dichotomy:k=3_aba} is tractable by case~\ref{case:cor:tractable:holant-star:aba} of Corollary~\ref{cor:tractable:holant-star}.
 Case~\ref{case:thm:dichotomy:100} is tractable by Corollary~\ref{cor:tractable:100}.
 Case~\ref{case:thm:dichotomy:k=3_10w} is tractable by Corollary~\ref{cor:tractable:k=3_a^3=c^3}.
 Case~\ref{case:thm:dichotomy:k=4_u^21u} is tractable by Lemma~\ref{lem:tractable:u^21u}.
 
 Otherwise, $\langle a,b,c \rangle$ is none of these tractable cases.
 If $\mathfrak{B} = 0$,
 then we are done by Corollary~\ref{cor:unary:dichotomy:a+(k-3)b-(k-2)c=0},
 so assume that $\mathfrak{B} \ne 0$.
 If $a + (\kappa - 1) b = 0$ and $b^2 - 4 b c - (\kappa - 3) c^2 = 0$,
 then we are done by Lemma~\ref{lem:unary:dichotomy:AD-like},
 so assume that $a + (\kappa - 1) b \ne 0$ or $b^2 - 4 b c - (\kappa - 3) c^2 \ne 0$.
 
 If $a + (\kappa - 1) b \ne 0$,
 then we have the succinct unary signature $\langle 1 \rangle$ of type $\tau_1$ by Lemma~\ref{lem:unary:construct_<1>}.
 Otherwise, $a + (\kappa - 1) b = 0$ and $b^2 - 4 b c - (\kappa - 3) c^2 \ne 0$.
 Since $\mathfrak{B} \ne 0$,
 we have $2 b + (\kappa - 2) c \ne 0$.
 Then again we have $\langle 1 \rangle$ by Lemma~\ref{lem:unary:construct_<1>}.
 Thus, in either case, we have $\langle 1 \rangle$.
 
 By Corollary~\ref{cor:binary:interpolate},
 we have all binary succinct signatures $\langle x,y \rangle$ for any $x,y \in \mathbb{C}$.
 Then we are done by Lemma~\ref{lem:ternary}.
\end{proof}

\paragraph{Acknowledgements}

We thank Joanna Ellis-Monaghan for bringing~\cite{Ell04} to our attention.
We are thankful to Mingji Xia who discussed with us an early version of this work.
We are very grateful to Bjorn Poonen and especially Aaron Levin for sharing their expertise on Runge's method,
and in particular for the auxiliary function $g_2(x, y)$ in the proof of Lemma~\ref{lem:ternary:lattice:no_linear}.
We benefited from discussions with William Whistler on a draft of this work,
whom we thank.

\bibliographystyle{plain}
\bibliography{bib}

\appendix

\section{Computing Gadget Signatures} \label{sec:compute}

In this paper,
some of the more difficult claims to verify are those when we say that a particular $\mathcal{F}$-gate (or gadget) has a particular signature.
This is an essential difficultly that cannot be avoided.
We are proving that $\PlHolant(\mathcal{F})$ is $\SHARPP$-hard for various $\mathcal{F}$
(and computing the signature of an $\mathcal{F}$-gate is a generalization of this problem).
Thus, one should not expect to be able to compute these signatures significantly faster in general than what the naive algorithm can do.

This has always been an issue for any dichotomy theorem about counting problems,
but with larger domain sizes,
we seem to be reaching the limit of what can be computed by hand for the signatures of gadget constructions that are presented in our proofs.
To counter this,
the standard techniques are to utilize the smallest gadgets (that suffice)
or an infinite family of related gadgets with a (small) description of finite size,
which we certainly employ.
Additionally, we point out some tricks, where they exist, to save as much work as possible.

Beyond all this, we also face another problem.
We would like to express the signature of a gadget as a function of the domain size.
To compute the signature of a gadget for each domain size is no longer a finite computation.
However, each entry of the gadget's signature is a polynomial in the domain size of degree at most the number of internal edges in the gadget.
To obtain these polynomials,
one can interpolate them by computing the signature for small domain sizes.
It is easy to write a program to do this.

When computing by hand,
there is another possibility that works quite well.
One partitions the internal edge assignments into a limited number of parts such that the assignments in each part contribute the same quantity to the Holant sum.
This is best explained with some examples.

\begin{figure}[ht]
 \centering
 \subcaptionbox
 {\label{subfig:gadget:compute:unary:self-loop}}{
  \begin{tikzpicture}[scale=\scale,transform shape,node distance=\nodeDist,semithick]
   \node[external] (0)              {};
   \node[internal] (1) [above of=0] {};
   \path (0) edge (1)
         (1) edge[out=150, in=30, looseness=20] node[pos=0.3] (e1) {} node[pos=0.7] (e2) {} (1);
   \begin{pgfonlayer}{background}
    \node[draw=\borderColor,thick,rounded corners,inner xsep=6pt,inner ysep=12pt,fit = (1) (e1) (e2)] {};
   \end{pgfonlayer}
  \end{tikzpicture}}%
 \subcaptionbox
 {\label{subfig:gadget:compute:unary_parallel_edges}}{
  \begin{tikzpicture}[scale=\scale,transform shape,node distance=\nodeDist,semithick]
   \node[external] (0)                    {};
   \node[internal] (1) [above       of=0] {};
   \node[internal] (2) [above left  of=1] {};
   \node[internal] (3) [above right of=1] {};
   \path (0) edge (1)
         (1) edge[bend left ] (2)
             edge[bend right] (3)
         (2) edge[bend left ] coordinate (c1) (3)
             edge[bend right] (3);
   \begin{pgfonlayer}{background}
    \node[draw=\borderColor,thick,rounded corners,inner xsep=6pt,inner ysep=12pt,fit = (1) (2) (3) (c1)] {};
   \end{pgfonlayer}
  \end{tikzpicture}}%
 \subcaptionbox
 {\label{subfig:gadget:compute:unary:self-loop_with_binary}}{
  \begin{tikzpicture}[scale=\scale,transform shape,node distance=\nodeDist,semithick]
   \node[external] (0)              {};
   \node[internal] (1) [above of=0] {};
   \path (0) edge (1)
         (1) edge[out=150, in=30, looseness=20] node[pos=0.3] (e1) {} node[pos=0.5,square] (e2) {} node[pos=0.7] (e3) {} (1);
   \begin{pgfonlayer}{background}
    \node[draw=\borderColor,thick,rounded corners,inner xsep=6pt,inner ysep=12pt,fit = (1) (e1) (e2) (e3)] {};
   \end{pgfonlayer}
  \end{tikzpicture}}%
 \subcaptionbox
 {\label{subfig:gadget:compute:binary_edge}}{
  \begin{tikzpicture}[scale=\scale,transform shape,node distance=\nodeDist,semithick]
   \node[external] (0)              {};
   \node[external] (1) [right of=0] {};
   \node[external] (2) [right of=1] {};
   \node[external] (3) [right of=2] {};
   \node[external] (4) [below of=3] {};
   \path (0) edge             (3);
   \begin{pgfonlayer}{background}
    \node[draw=\borderColor,thick,rounded corners,inner xsep=12pt,inner ysep=8pt,fit = (1) (2)] {};
   \end{pgfonlayer}
  \end{tikzpicture}}%
 \subcaptionbox
 {\label{subfig:gadget:compute:binary:parallel_edges}}{
  \begin{tikzpicture}[scale=\scale,transform shape,node distance=\nodeDist,semithick]
   \node[external] (0)              {};
   \node[internal] (1) [right of=0] {};
   \node[internal] (2) [right of=1] {};
   \node[external] (3) [right of=2] {};
   \node[external] (4) [below of=3] {};
   \path (0) edge             (1)
         (1) edge[bend left]  (2)
             edge[bend right] (2)
         (2) edge             (3);
   \begin{pgfonlayer}{background}
    \node[draw=\borderColor,thick,rounded corners,inner xsep=12pt,inner ysep=8pt,fit = (1) (2)] {};
   \end{pgfonlayer}
  \end{tikzpicture}}
 \caption{Gadgets~(\subref{subfig:gadget:compute:unary:self-loop}) and~(\subref{subfig:gadget:compute:unary_parallel_edges}) are used to construct $\langle 1 \rangle$.
  They are special cases of~(\subref{subfig:gadget:compute:unary:self-loop_with_binary})
  and are obtained by replacing the square in~(\subref{subfig:gadget:compute:unary:self-loop_with_binary})
  with either~(\subref{subfig:gadget:compute:binary_edge}) or~(\subref{subfig:gadget:compute:binary:parallel_edges}) respectively.
  All (circle) vertices are assigned $\langle a,b,c \rangle$.}
 \label{fig:gadget:compute:unary_construction}
\end{figure}

\begin{lemma} \label{lem:compute:unary:construct_<1>}
 Suppose $\kappa \ge 3$ is the domain size and $a,b,c,x,y \in \mathbb{C}$.
 Let $\langle a,b,c \rangle$ be a succinct ternary signature of type $\tau_3$ and let $\langle x,y \rangle$ be a succinct signature of type $\tau_2$.
 If we assign $\langle a,b,c \rangle$ to the circle vertex and $\langle x,y \rangle$ to the square vertex of the gadget in Figure~\ref{subfig:gadget:compute:unary:self-loop_with_binary},
 then the succinct unary signature of type $\tau_1$ of the resulting gadget is $\langle x [a + (\kappa - 1) b] + y (\kappa - 1) [2 b + (\kappa - 2) c] \rangle$.
 
 If the square vertex is replaced by Figure~\ref{subfig:gadget:compute:binary_edge},
 then the resulting signature is $\langle a + (\kappa - 1) b \rangle$.
 If the square vertex is replaced by Figure~\ref{subfig:gadget:compute:binary:parallel_edges}, and $a + (\kappa - 1) b = 0$,
 then the resulting signature is
 \begin{equation} \label{eqn:compute:unary:construct_<1>}
  \langle -(\kappa - 1) (\kappa - 2) [2 b + (\kappa - 2) c] [b^2 - 4 b c - (\kappa - 3) c^2] \rangle.
 \end{equation}
\end{lemma}

\begin{proof}
 Since $\langle a,b,c \rangle$ and $\langle x,y \rangle$ are domain invariant,
 the signatures of these gadgets are also domain invariant.
 Any domain invariant unary signature has a succinct signature of type $\tau_1$.

 Let $g \in [\kappa]$ be a possible edge assignment,
 which we call a color.
 Suppose the external edge is assigned $g$ and consider all internal edge assignments that assign the same colors to both edges.
 For such assignments,
 $\langle x,y \rangle$ contributes a factor of $x$.
 Now if this color assigned to both internal edges is also $g$,
 then $\langle a,b,c \rangle$ contributes a factor of $a$.
 Thus, the Holant sum includes one factor of $a x$.
 If the two internal edges are assigned any color different from $g$,
 then $\langle a,b,c \rangle$ contributes a factor of $b$.
 Since there are $\kappa - 1$ such colors,
 this adds $(\kappa - 1) b x$ to the Holant sum.
 
 Now consider all internal assignments that assign different colors to the edges.
 For such assignments,
 $\langle x,y \rangle$ contributes a factor of $y$.
 First, suppose that one of the internal edges is assigned $g$.
 There are two ways this could happen and $\langle a,b,c \rangle$ contributes a factor of $b$.
 Since there are $\kappa - 1$ choices for the remaining edge assignment,
 this adds $2 (\kappa - 1) b y$ to the Holant sum.
 Lastly, suppose that the two internal edges are not assigned $g$.
 Then $\langle a,b,c \rangle$ contributes a factor of $c$.
 Since there are $(\kappa - 1) (\kappa - 2)$ such assignments,
 this adds $(\kappa - 1) (\kappa - 2) c y$ to the Holant sum.
 Thus, the resulting signature is $\langle x [a + (\kappa - 1) b] + y (\kappa - 1) [2 b + (\kappa - 2) c] \rangle$ as claimed.
 
 Replacing the square by Figure~\ref{subfig:gadget:compute:binary_edge} is equivalent to setting $x = 1$ and $y = 0$,
 which gives $\langle a + (\kappa - 1) b \rangle$.
 Replacing the square by Figure~\ref{subfig:gadget:compute:binary:parallel_edges} is equivalent to setting $x$ and $y$ to the values given in Lemma~\ref{lem:compute:binary:parallel_edges}.
 The resulting signature is indeed~(\ref{eqn:compute:unary:construct_<1>}).
\end{proof}

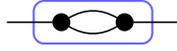
\begin{figure}[ht]
 \centering
 \begin{tikzpicture}[scale=\scale,transform shape,node distance=\nodeDist,semithick]
   \node[external] (0)              {};
   \node[internal] (1) [right of=0] {};
   \node[internal] (2) [right of=1] {};
   \node[external] (3) [right of=2] {};
   \path (0) edge             (1)
         (1) edge[bend left]  (2)
             edge[bend right] (2)
         (2) edge             (3);
   \begin{pgfonlayer}{background}
    \node[draw=\borderColor,thick,rounded corners,inner xsep=12pt,inner ysep=8pt,fit = (1) (2)] {};
  \end{pgfonlayer}
 \end{tikzpicture}
 \caption{A simple binary gadget.}
 \label{fig:gadget:compute:binary:parallel_edges}
\end{figure}

\begin{lemma} \label{lem:compute:binary:parallel_edges}
 Suppose $\kappa \ge 3$ is the domain size and $a,b,c \in \mathbb{C}$.
 Let $\langle a,b,c \rangle$ be a succinct ternary signature of type $\tau_3$.
 If we assign $\langle a,b,c \rangle$ to both vertices of the gadget in Figure~\ref{fig:gadget:compute:binary:parallel_edges},
 then the succinct binary signature of type $\tau_2$ of the resulting gadget is $\langle x, y \rangle$,
 where
 \begin{align*}
  x &=   a^2
       + 3 (\kappa - 1) b^2
       + (\kappa - 1) (\kappa - 2) c^2
       \qquad \text{and}\\
  y &=   2 a b
       + \kappa b^2
       + 4 (\kappa - 2) b c
       + (\kappa - 2) (\kappa - 3) c^2.
 \end{align*}
\end{lemma}

\begin{proof}
 Since $\langle a,b,c \rangle$ is domain invariant,
 the signature of this gadget is also domain invariant.
 Any domain invariant binary signature has a succinct signature of type $\tau_2$.
 
 Let $g,r \in [\kappa]$ be distinct edge assignments.
 We have two entries to compute.
 To compute $x$,
 suppose that both external edges are assigned $g$.
 We begin with the case where both internal edges have the same assignment.
 If this assignment is $g$,
 then $a^2$ is contributed to the sum.
 If this assignment is not $g$,
 then $b^2$ is contributed to the sum for a total contribution of $(\kappa - 1) b^2$.
 Now consider the case that the two internal edges have a different assignment.
 If one of these assignments is $g$,
 then $b^2$ is contributed to the sum for a total contribution of $2 (\kappa - 1) b^2$.
 If neither assignment is $g$,
 then $c^2$ is contributed to the sum for a total contribution of $(\kappa - 1) (\kappa - 2) c^2$.
 These total contributions sum to the value for $x$ given in Lemma~\ref{lem:compute:binary:parallel_edges}.
 
 To compute $y$,
 suppose one external edge is assigned $g$ and the other is assigned $r$.
 We begin with the case where both internal edges have the same assignment.
 If this assignment is $g$ or $r$,
 then $a b$ is contributed to the sum for a total contribution of $2 a b$.
 If this assignment is not $g$ or $r$,
 then $b^2$ is contributed to the sum for a total contribution of $(\kappa - 2) b^2$.
 Now consider the case that the two internal edges have a different assignment.
 If both are assigned $g$ or $r$,
 then $b^2$ is contributed to the sum for a total contribution of $2 b^2$.
 If exactly one is assigned $g$ or $r$,
 then $b c$ is contributed to the sum for a total contribution of $4 (\kappa - 2) b c$.
 If neither is assigned $g$ or $r$,
 then $c^3$ is contributed to the sum for a total contribution of $(\kappa - 2) (\kappa - 3) c^3$.
 These total contributions sum to the value for $y$ given in Lemma~\ref{lem:compute:binary:parallel_edges}.
\end{proof}

When checking these proofs,
a concern is that some assignments might not have been counted.
One sanity check to address this concern is to set $a = b = c = 1$ and inspect the resulting expression.
If computed correctly,
the result will be $\kappa^m$,
where $m$ is the number of internal edges,
which is the number of internal edge assignments.

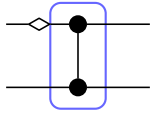
\begin{figure}[ht]
 \centering
 \begin{tikzpicture}[scale=\scale,transform shape,node distance=\nodeDist,semithick]
  \node[external] (0)              {};
  \node[external] (1) [below of=0] {};
  \node[internal] (2) [right of=0] {};
  \node[internal] (3) [right of=1] {};
  \node[external] (4) [right of=2] {};
  \node[external] (5) [right of=3] {};
  \path (0.west) edge[postaction={decorate, decoration={
                                             markings,
                                             mark=at position 0.70  with {\arrow[>=diamond,white] {>}; },
                                             mark=at position 0.70  with {\arrow[>=open diamond]  {>}; } } }] (2)
        (1.west) edge (3)
        (2)      edge (3)
                 edge (4.east)
        (3)      edge (5.east);
  \begin{pgfonlayer}{background}
   \node[draw=\borderColor,thick,rounded corners,inner xsep=12pt,inner ysep=8pt,fit = (2) (3)] {};
  \end{pgfonlayer}
 \end{tikzpicture}
 \caption{A simple quaternary gadget.}
 \label{fig:gadget:compute:quaternary:I}
\end{figure}

\begin{lemma} \label{lem:compute:quaternary:I}
 Suppose $\kappa \ge 3$ is the domain size and $a,b,c \in \mathbb{C}$.
 Let $\langle a,b,c \rangle$ be a succinct ternary signature of type $\tau_3$.
 If we assign $\langle a,b,c \rangle$ to both vertices of the gadget in Figure~\ref{fig:gadget:compute:quaternary:I},
 then the succinct quaternary signature of type $\tau_4$ of the resulting gadget is
 \[
  f =
  \left\langle
   f_{\subMat{1}{1}{1}{1}},
   f_{\subMat{1}{1}{1}{2}},
   f_{\subMat{1}{1}{2}{2}},
   f_{\subMat{1}{1}{2}{3}},
   f_{\subMat{1}{2}{1}{2}},
   f_{\subMat{1}{2}{1}{3}},
   f_{\subMat{1}{2}{2}{1}},
   f_{\subMat{1}{2}{3}{1}},
   f_{\subMat{1}{2}{3}{4}}
  \right\rangle,
 \]
 where
 \begin{align*}
  f_{\subMat{1}{1}{1}{1}} &= a^2 + (\kappa - 1) b^2,\\
  f_{\subMat{1}{1}{1}{2}} &= b [a + b + (\kappa - 2) c],\\
  f_{\subMat{1}{1}{2}{2}} &= 2 b^2 + (\kappa - 2) c^2,\\
  f_{\subMat{1}{1}{2}{3}} &= b^2 + 2 b c + (\kappa - 3) c^2,\\
  f_{\subMat{1}{2}{1}{2}} &= f_{\subMat{1}{1}{2}{2}},\\
  f_{\subMat{1}{2}{1}{3}} &= f_{\subMat{1}{1}{2}{3}},\\
  f_{\subMat{1}{2}{2}{1}} &= b [2 a + (\kappa - 2) b],\\
  f_{\subMat{1}{2}{3}{1}} &= a c + 2 b^2 + (\kappa - 3) b c, \text{ and}\\
  f_{\subMat{1}{2}{3}{4}} &= c [4 b + (\kappa - 4) c].
 \end{align*}
\end{lemma}

\begin{proof}
 Since $\langle a,b,c \rangle$ is domain invariant,
 the signature of this gadget is also domain invariant.
 The vertical and horizontal symmetry of this gadget implies that its signature has a succinct signature of type $\tau_4$.
 
 Let $w,x,y,z \in [\kappa]$ be distinct edge assignments.
 We have nine entries to compute.
 Recall that the edge with the diamond is considered the first input and the rest are ordered counterclockwise.
 \begin{enumerate}
  \item To compute $f_{\subMat{1}{1}{1}{1}}$,
  suppose the external assignment is $(w,w,w,w)$.
  If the internal edge is also assigned $w$,
  then $a^2$ is contributed to the sum.
  If the internal edge is not assigned $w$,
  then $b^2$ is contributed to the sum for a total contribution of $(\kappa - 1) b^2$.
  
  \item To compute $f_{\subMat{1}{1}{1}{2}}$,
  suppose the external assignment is $(w,w,w,x)$.
  If the internal edge is assigned $w$,
  then $a b$ is contributed to the sum.
  If the internal edge is assigned $x$,
  then $b^2$ is contributed to the sum.
  If the internal edge is not assigned $w$ or $x$,
  then $b c$ is contributed to the sum for a total contribution of $(\kappa - 2) b c$.
  
  \item To compute $f_{\subMat{1}{1}{2}{2}}$,
  suppose the external assignment is $(w,w,x,x)$.
  If the internal edge is assigned $w$,
  then $b^2$ is contributed to the sum.
  If the internal edge is assigned $x$,
  then $b^2$ is contributed to the sum.
  If the internal edge is not assigned $w$ or $x$,
  then $c^2$ is contributed to the sum for a total contribution of $(\kappa - 2) c^2$.
  
  \item To compute $f_{\subMat{1}{1}{2}{3}}$,
  suppose the external assignment is $(w,w,x,y)$.
  If the internal edge is assigned $w$,
  then $b^2$ is contributed to the sum.
  If the internal edge is assigned $x$,
  then $b c$ is contributed to the sum.
  If the internal edge is assigned $y$,
  then $b c$ is contributed to the sum.
  If the internal edge is not assigned $w$, $x$ or $y$,
  then $c^2$ is contributed to the sum for a total contribution of $(\kappa - 3) c^2$.
  
  \item To compute $f_{\subMat{1}{1}{2}{3}}$,
  suppose the external assignment is $(w,x,w,x)$.
  This entry is the same as that for $(w,w,x,x)$.
  The reason is that the signature is unchanged if the two external edges of the lower vertex are swapped since $\langle a,b,c \rangle$ is symmetric.
  
  \item To compute $f_{\subMat{1}{2}{1}{3}}$,
  suppose the external assignment is $(w,x,w,y)$.
  This entry is the same as that for $(w,w,x,y)$ for the same reason as the previous entry.
  
  \item To compute $f_{\subMat{1}{2}{2}{1}}$,
  suppose the external assignment is $(w,x,x,w)$.
  If the internal edge is assigned $w$,
  then $a b$ is contributed to the sum.
  If the internal edge is assigned $x$,
  then $a b$ is contributed to the sum.
  If the internal edge is not assigned $w$ or $x$,
  then $b^2$ is contributed to the sum for a total contribution of $(\kappa - 2) b^2$.
  
  \item To compute $f_{\subMat{1}{2}{3}{1}}$,
  suppose the external assignment is $(w,x,y,w)$.
  If the internal edge is assigned $w$,
  then $a c$ is contributed to the sum.
  If the internal edge is assigned $x$,
  then $b^2$ is contributed to the sum.
  If the internal edge is assigned $y$,
  then $b^2$ is contributed to the sum.
  If the internal edge is not assigned $w$, $x$ or $y$,
  then $b c$ is contributed to the sum for a total contribution of $(\kappa - 3) c^2$.
  
  \item To compute $f_{\subMat{1}{2}{3}{4}}$,
  suppose the external assignment is $(w,x,y,z)$.
  If the internal edge is assigned $w$, $x$, $y$, or $z$,
  then $b c$ is contributed to the sum for a total contribution of $4 b c$.
  If the internal edge is not assigned $w$, $x$, $y$ or $z$,
  then $c^2$ is contributed to the sum for a total contribution of $(\kappa - 4) c^2$.
 \end{enumerate}
 These total contributions each sum to their corresponding entry of $f$ given in the statement of Lemma~\ref{lem:compute:quaternary:I}.
\end{proof}

Although possible,
it would be difficult to compute the signature of the gadget in Figure~\ref{subfig:gadget:compute:ternary:triangle}
through partitioning of the internal edge assignments alone.
To simplify matters,
we utilize the calculations from Lemma~\ref{lem:compute:quaternary:I}.
Since composing the gadget in Figure~\ref{subfig:gadget:compute:ternary:triangle_inside}
with the one in Figure~\ref{subfig:gadget:compute:ternary:triangle_outside} gives a symmetric signature,
we refrain from distinguishing the external edges of the gadget in Figure~\ref{subfig:gadget:compute:ternary:triangle_outside}.

\begin{figure}[ht]
 \centering
 \def\capWidth{4.5cm}
 \captionsetup[subfigure]{width=\capWidth}
 \subcaptionbox{\label{subfig:gadget:compute:ternary:triangle_inside}Inner structure}[\capWidth]{
  \begin{tikzpicture}[scale=\scale,transform shape,node distance=\nodeDist,semithick]
  \node[external] (0)              {};
  \node[external] (1) [below of=0] {};
  \node[internal] (2) [right of=0] {};
  \node[internal] (3) [right of=1] {};
  \node[external] (4) [right of=2] {};
  \node[external] (5) [right of=3] {};
  \path (0.west) edge[postaction={decorate, decoration={
                                             markings,
                                             mark=at position 0.70  with {\arrow[>=diamond,white] {>}; },
                                             mark=at position 0.70  with {\arrow[>=open diamond]  {>}; } } }] (2)
        (1.west) edge (3)
        (2)      edge (3)
                 edge (4.east)
        (3)      edge (5.east);
  \begin{pgfonlayer}{background}
   \node[draw=\borderColor,thick,rounded corners,inner xsep=12pt,inner ysep=8pt,fit = (2) (3)] {};
   \end{pgfonlayer}
  \end{tikzpicture}
 }
 \qquad
 \subcaptionbox{\label{subfig:gadget:compute:ternary:triangle_outside}Outer structure}[\capWidth]{
  \begin{tikzpicture}[scale=\scale,transform shape,node distance=\nodeDist,semithick]
   \node[external] (0)                     {};
   \node[external] (1)  [      right of=0] {};
   \node[square]   (2)  [below right of=1] {};
   \node[external] (3)  [below left  of=2] {};
   \node[external] (4)  [      left  of=3] {};
   \node[internal] (5)  [      right of=2] {};
   \node[external] (6)  [      right of=5] {};
   \path (0) edge[out=0, in= 135, postaction={decorate, decoration={
                                                         markings,
                                                         mark=at position 0.99 with {\arrow[>=diamond,white] {>}; },
                                                         mark=at position 0.99 with {\arrow[>=open diamond]  {>}; } } }] (2)
         (4) edge[out=0, in=-135] (2)
         (2) edge[bend left]  (5)
             edge[bend right] (5)
         (5) edge             (6);
   \begin{pgfonlayer}{background}
    \node[draw=\borderColor,thick,rounded corners,inner xsep=12pt,inner ysep=8pt,fit = (1) (3) (5)] {};
   \end{pgfonlayer}
  \end{tikzpicture}
 }
 \qquad
 \subcaptionbox{\label{subfig:gadget:compute:ternary:triangle}Entire binary gadget}[\capWidth]{
  \begin{tikzpicture}[scale=\scale,transform shape,node distance=\nodeDist,semithick]
  \node [external] (0)              {};
  \node [internal] (1) [below of=0] {};
  \path (1) ++(-120:\nodeDist) node [internal] (2) {} ++(-150:\nodeDist) node [external] (3) {};
  \path (1) ++( -60:\nodeDist) node [internal] (4) {} ++( -30:\nodeDist) node [external] (5) {};
  \path (0) edge (1)
        (1) edge (2)
            edge (4)
        (2) edge (3)
            edge (4)
        (4) edge (5);
  \begin{pgfonlayer}{background}
   \node[draw=\borderColor,thick,rounded corners,inner xsep=12pt,inner ysep=12pt,fit = (1) (2) (4)] {};
   \end{pgfonlayer}
  \end{tikzpicture}
 }
 \caption{Decomposition of a ternary gadget.
 All circle vertices are assigned $\langle a,b,c \rangle$
 and the square vertex in~(\subref{subfig:gadget:compute:ternary:triangle_outside})
 is assigned the signature of the gadget in~(\subref{subfig:gadget:compute:ternary:triangle_inside}).}
 \label{fig:gadget:compute:ternary:triangle_construction}
\end{figure}

\begin{lemma} \label{lem:compute:ternary:triangle}
 Suppose $\kappa \ge 3$ is the domain size and $a,b,c \in \mathbb{C}$.
 Let $\langle a,b,c \rangle$ be a succinct ternary signature of type $\tau_3$.
 If we assign $\langle a,b,c \rangle$ to all vertices of the gadget in Figure~\ref{subfig:gadget:compute:ternary:triangle},
 then the succinct ternary signature of type $\tau_3$ of the resulting gadget is $\langle a', b', c' \rangle$,
 where
 \begin{align*}
  a' = {}&  a^3
          + 3 (\kappa - 1) a b^2
          + 4 (\kappa - 1) b^3
          + 3 (\kappa - 1) (\kappa - 2) (b^2 c + b c^2)
          + (\kappa - 1) (\kappa - 2) (\kappa - 3) c^3,\\
  b' = {}&  a^2 b
          + 4 a b^2
          + 2 (\kappa - 2) a b c
          + (\kappa - 2) a c^2
          + (5 \kappa - 7) b^3
          + (\kappa - 2) (\kappa + 5) b^2 c\\
         &+ (\kappa - 2) (7 \kappa - 18) b c^2
          + (\kappa - 2) (\kappa - 3)^2 c^3,
         \qquad \text{and}\\
  c' = {}&  3 a b^2
          + 6 a b c
          + 3 (\kappa - 3) a c^2
          + (\kappa + 5) b^3
          + 3 (7 \kappa - 18) b^2 c
          + 9 (\kappa - 3)^2 b c^2\\
         &+ (\kappa^3 - 9 \kappa^2 + 29 \kappa - 32) c^3.
 \end{align*}
 Furthermore, if $\mathfrak{A} = 0$, then
 \begin{align*}
  a' &= 3 b' - 2 c',\\
  b' &=   (5 \kappa + 14) b^3
        + (\kappa^2 + 9 \kappa - 42) b^2 c
        + (7 \kappa^2 - 33 \kappa + 42) b c^2
        + (\kappa - 2) (\kappa^2 - 6 \kappa + 7) c^3,
        \qquad \text{and}\\
  c' &=   (\kappa + 14) b^3
        + 21 (\kappa - 2) b^2 c
        + 3 (3 \kappa^2 - 15 \kappa + 14) b c^2
        + (\kappa^3 - 9 \kappa^2 + 23 \kappa - 14) c^3.
 \end{align*}
\end{lemma}

\begin{proof}
 Since $\langle a,b,c \rangle$ is domain invariant,
 the signature of this gadget is also domain invariant.
 As a ternary signature,
 the rotational symmetry of this gadget implies the symmetry of the signature.
 Any symmetric domain invariant ternary signature has a succinct signature of type $\tau_3$.
 
 Consider the gadget in Figure~\ref{subfig:gadget:compute:ternary:triangle_inside}.
 We assign $\langle a,b,c \rangle$ to both vertices.
 Then by Lemma~\ref{lem:compute:quaternary:I},
 the succinct quaternary signature of this gadget is the signature $f$ given in Lemma~\ref{lem:compute:quaternary:I}.

 Now consider the gadget in Figure~\ref{subfig:gadget:compute:ternary:triangle_outside}.
 We assign $\langle a,b,c \rangle$ to the circle vertex and $f$ to the square vertex.
 The resulting gadget is the one in Figure~\ref{subfig:gadget:compute:ternary:triangle},
 which is symmetric.
 Thus, there is no need to distinguish the external edges.
 We have three entries to compute.
 
 Let $g,r,y \in [\kappa]$ be distinct edge assignments.
 To compute $a'$,
 suppose that all external edges are assigned $g$.
 We begin with the case where both internal edges have the same assignment.
 If this assignment is $g$,
 then $a f_{\subMat{1}{1}{1}{1}}$ is contributed to the sum.
 If this assignment is not $g$,
 then $b f_{\subMat{1}{1}{2}{2}}$ is contributed to the sum
 for a total contribution of $(\kappa - 1) b f_{\subMat{1}{1}{2}{2}}$.
 Now consider the case that the two internal edges have a different assignment.
 If one of these assignments is $g$,
 then $b f_{\subMat{1}{1}{1}{2}}$ is contributed to the sum
 for a total contribution of $2 (\kappa - 1) b f_{\subMat{1}{1}{1}{2}}$.
 If neither assignment is $g$,
 then $c f_{\subMat{1}{1}{2}{3}}$ is contributed to the sum
 for a total contribution of $(\kappa - 1) (\kappa - 2) c f_{\subMat{1}{1}{2}{3}}$.
 After substituting for the entries of $f$,
 these total contributions sum to the value for $a'$ given in Lemma~\ref{lem:compute:ternary:triangle}.
 
 To compute $b'$,
 suppose the left external edges are assigned $g$ and the right external edge is assigned $r$.
 We begin with the case where both internal edges have the same assignment.
 If this assignment is $g$,
 then $b f_{\subMat{1}{1}{1}{1}}$ is contributed to the sum.
 If this assignment is $r$,
 then $a f_{\subMat{1}{1}{2}{2}}$ is contributed to the sum.
 If this assignment is not $g$ or $r$,
 then $b f_{\subMat{1}{1}{2}{2}}$ is contributed to the sum
 for a total contribution of $(\kappa - 2) b f_{\subMat{1}{1}{2}{2}}$.
 Now consider the case that the two internal edges have a different assignments.
 If both are assigned $g$ or $r$,
 then $b f_{\subMat{1}{1}{1}{2}}$ is contributed to the sum
 for a total contribution of $2 b f_{\subMat{1}{1}{1}{2}}$.
 If one is assigned $g$ and the other is not assigned $r$,
 then $c f_{\subMat{1}{1}{1}{2}}$ is contributed to the sum
 for a total contribution of $2 (\kappa - 2) c f_{\subMat{1}{1}{1}{2}}$.
 If one is assigned $r$ and the other is not assigned $g$,
 then $b f_{\subMat{1}{1}{2}{3}}$ is contributed to the sum
 for a total contribution of $2 (\kappa - 2) b f_{\subMat{1}{1}{2}{3}}$.
 If neither is assigned $g$ or $r$,
 then $c f_{\subMat{1}{1}{2}{3}}$ is contributed to the sum
 for a total contribution of $(\kappa - 2) (\kappa - 3) c f_{\subMat{1}{1}{2}{3}}$.
 After substituting for the entries of $f$,
 these total contributions sum to the value for $b'$ given in Lemma~\ref{lem:compute:ternary:triangle}.
 
 To compute $c'$,
 suppose the upper-left external edge is assigned $g$,
 the lower-left external edge is assigned $r$,
 and the right external edge is assigned $y$.
 We begin with the case where both internal edges have the same assignment.
 If this assignment is $g$,
 then $b f_{\subMat{1}{1}{1}{2}}$ is contributed to the sum.
 If this assignment is $r$,
 then $b f_{\subMat{1}{1}{1}{2}}$ is contributed to the sum.
 If this assignment is $y$,
 then $a f_{\subMat{1}{1}{2}{3}}$ is contributed to the sum.
 If this assignment is not $g$, $r$, or $y$,
 then $b f_{\subMat{1}{1}{2}{3}}$ is contributed to the sum
 for a total contribution of $(\kappa - 3) b f_{\subMat{1}{1}{2}{3}}$.
 Now consider the case that the two internal edges have a different assignments.
 If the top internal edge is assigned $g$ and the bottom one is assigned $r$,
 then $c f_{\subMat{1}{2}{2}{1}}$ is contributed to the sum.
 If the top internal edge is assigned $r$ and the bottom one is assigned $g$,
 then $c f_{\subMat{1}{2}{1}{2}}$ is contributed to the sum.
 If the top internal edge is assigned $g$ and the bottom one is assigned $y$,
 then $b f_{\subMat{1}{2}{3}{1}}$ is contributed to the sum.
 If the top internal edge is assigned $y$ and the bottom one is assigned $g$,
 then $b f_{\subMat{1}{2}{1}{3}}$ is contributed to the sum.
 If the top internal edge is assigned $r$ and the bottom one is assigned $y$,
 then $b f_{\subMat{1}{2}{1}{3}}$ is contributed to the sum.
 If the top internal edge is assigned $y$ and the bottom one is assigned $r$,
 then $b f_{\subMat{1}{2}{3}{1}}$ is contributed to the sum.
 If the top internal edge is assigned $g$ and the bottom one not assigned $r$ or $y$,
 then $c f_{\subMat{1}{2}{3}{1}}$ is contributed to the sum
 for a total contribution of $(\kappa - 3) c f_{\subMat{1}{2}{3}{1}}$.
 If the bottom internal edge is assigned $g$ and the top one not assigned $r$ or $y$,
 then $c f_{\subMat{1}{2}{1}{3}}$ is contributed to the sum
 for a total contribution of $(\kappa - 3) c f_{\subMat{1}{2}{1}{3}}$.
 If the top internal edge is assigned $r$ and the bottom one not assigned $g$ or $y$,
 then $c f_{\subMat{1}{2}{1}{3}}$ is contributed to the sum
 for a total contribution of $(\kappa - 3) c f_{\subMat{1}{2}{1}{3}}$.
 If the bottom internal edge is assigned $r$ and the top one not assigned $g$ or $y$,
 then $c f_{\subMat{1}{2}{3}{1}}$ is contributed to the sum
 for a total contribution of $(\kappa - 3) c f_{\subMat{1}{2}{3}{1}}$.
 If the one internal edge is assigned $y$ and the other is not assigned $g$ or $r$,
 then $b f_{\subMat{1}{2}{3}{4}}$ is contributed to the sum
 for a total contribution of $2 (\kappa - 3) b f_{\subMat{1}{2}{3}{4}}$.
 If neither internal edge is assigned $g$ $r$, or $y$,
 then $c f_{\subMat{1}{2}{3}{4}}$ is contributed to the sum
 for a total contribution of $(\kappa - 3) (\kappa - 4) c f_{\subMat{1}{2}{3}{4}}$.
 After substituting for the entries of $f$,
 these total contributions sum to the value for $c'$ given in Lemma~\ref{lem:compute:ternary:triangle}.
\end{proof}

The signature of the gadget in Figure~\ref{fig:gadget:compute:binary:anti-gadget} is difficult to compute using gadget compositions and partitioning of internal edge assignments as we have been doing.
Instead, we compute this signature using matrix product, trace, and polynomial interpolation.

\begin{figure}[ht]
 \centering
 \begin{tikzpicture}[scale=\scale,transform shape,node distance=\nodeDist,semithick]
  \node[external] (0)                    {};
  \node[internal] (1) [right       of=0] {};
  \node[square]   (2) [above right of=1] {};
  \node[triangle] (3) [below right of=1] {};
  \node[internal] (4) [below right of=2] {};
  \node[external] (5) [right       of=4] {};
  \path (0) edge             (1)
        (1) edge[bend left]  (2)
            edge[bend right] (3)
        (2) edge[bend left]  (4)
        (3) edge[bend right] (4)
        (4) edge             (5);
  \begin{pgfonlayer}{background}
   \node[draw=\borderColor,thick,rounded corners,inner xsep=12pt,inner ysep=8pt,fit = (1) (2) (3) (4)] {};
  \end{pgfonlayer}
 \end{tikzpicture}
 \caption{A more complicated binary gadget.}
 \label{fig:gadget:compute:binary:anti-gadget}
\end{figure}
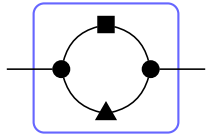

\begin{lemma} \label{lem:compute:binary:anti-gadget}
 Suppose $\kappa \ge 3$ is the domain size and $a,b,c, x_1, y_1, x_2, y_2 \in \mathbb{C}$.
 Let $\langle a,b,c \rangle$ be a succinct ternary signature of type $\tau_3$ and
 $\langle x_1, y_1 \rangle$ and $\langle x_2, y_2 \rangle$ be succinct binary signatures of type $\tau_2$.
 If to the gadget in Figure~\ref{fig:gadget:compute:binary:anti-gadget}
 we assign $\langle a,b,c \rangle$ to the circle vertices,
 $\langle x_1, y_1 \rangle$ to the square vertex, and
 $\langle x_2, y_2 \rangle$ to the triangle vertex,
 then the succinct binary signature of type $\tau_2$ of the resulting gadget is $\langle x, y \rangle$,
 where
 \begin{align*}
  x = {}&  x_1 x_2 a^2
         + 2 (\kappa - 1) (x_1 y_2 + x_2 y_1 + y_1 y_2) a b
         + 2 (\kappa - 1) (\kappa - 2) y_1 y_2 a c\\
        &+ (\kappa - 1) [3 x_1 x_2 + \kappa (x_1 y_2 + x_2 y_1) + (7 \kappa - 12) y_1 y_2] b^2\\
        &+ 2 (\kappa - 1) (\kappa - 2) [2 (x_1 y_2 + x_2 y_1) + (3 \kappa - 7) y_1 y_2] b c\\
        &+ (\kappa - 1) (\kappa - 2) [x_1 x_2 + (\kappa - 3) (x_1 y_2 + x_2 y_1) + (\kappa^2 - 5 \kappa + 7) y_1 y_2] c^2
       \qquad \text{and}\\
  y = {}&  y_1 y_2 a^2
         + 2 [x_1 x_2 + x_1 y_2 + x_2 y_1 + 3 (\kappa - 2) y_1 y_2] a b
         + 2 (\kappa - 2) [x_1 y_2 + x_2 y_1 + (\kappa - 3) y_1 y_2] a c\\
        &+ [\kappa x_1 x_2 + (7 \kappa - 12) (x_1 y_2 + x_2 y_1) + 3 (3 \kappa^2 - 11 \kappa + 11) y_1 y_2] b^2\\
        &+ 2 (\kappa - 2) [2 x_1 x_2 + (3 \kappa - 7) (x_1 y_2 + x_2 y_1) + 3 (\kappa^2 - 4 \kappa + 5) y_1 y_2] b c\\
        &+ (\kappa - 2) [(\kappa - 3) x_1 x_2 + (\kappa^2 - 5 \kappa + 7) (x_1 y_2 + x_2 y_1) + (\kappa^3 - 6 \kappa^2 + 14 \kappa - 13)] c^2.
 \end{align*}
 Furthermore,
 if $\langle x_1, y_1 \rangle = \frac{1}{\kappa} \langle \omega^r + \kappa - 1, \omega^r - 1 \rangle$
 and $\langle x_2, y_2 \rangle = \frac{1}{\kappa} \langle \omega^s + \kappa - 1, \omega^s - 1 \rangle$,
 then
 \begin{align*}
  x = \frac{\mathfrak{B}^2}{\kappa^2} \left[\Phi \omega^{r+s} + (\kappa - 1) (\omega^r + \omega^s + \Psi + 1)\right]
      \quad \text{and} \quad
  y = \frac{\mathfrak{B}^2}{\kappa^2} \left[\Phi \omega^{r+s} -              (\omega^r + \omega^s + \Psi + 1) + \kappa\right],
 \end{align*}
 where $\Phi = \frac{\mathfrak{C}^2}{\mathfrak{B}^2}$ and $\Psi = \frac{(\kappa - 2) \mathfrak{A}^2}{\mathfrak{B}^2}$.
\end{lemma}

\begin{proof}
 Since $\langle a,b,c \rangle$, $\langle x_1, y_1 \rangle$, and $\langle x_2, y_2 \rangle$ are domain invariant,
 the signature of this gadget is also domain invariant.
 Any domain invariant binary signature has a succinct signature of type $\tau_2$.
 
 We compute $a'$, $b'$, and $c'$ using the algorithm for $\Holant(\mathcal{F})$ when every non-degenerate signature in $\mathcal{F}$ is of arity at most~$2$,
 which is to use matrix product and trace.
 Then we finish with polynomial interpolation.
 Let $M_\kappa(t)$ be a $\kappa$-by-$\kappa$ matrix such that
 \[
  (M_\kappa(t))_{i,j} =
  \begin{cases}
   a & i = j = t\\
   b & i = j \ne t\\
   b & i \ne j \text{ and } (i = t \text{ or } j = t)\\
   c & \text{otherwise}.
  \end{cases}
 \]
 For example,
 $M_4(1) = \left[\begin{smallmatrix} a & b & b & b \\ b & b & c & c \\ b & c & b & c \\ b & c & c & b \end{smallmatrix}\right]$.
 If we fix an input of $\langle a,b,c \rangle$ to $t \in [\kappa]$,
 then the resulting binary signature (which is no longer domain invariant) has the signature matrix $M_\kappa(t)$.
 
 Consider $x$ and $y$ as polynomials in $\kappa$ with coefficients in $\Z[a,b,c,x_1,y_1,x_2,y_2]$.
 Then
 \begin{align*}
  x(\kappa) &= \tr\big(M_\kappa(1) [y_1 J_\kappa + (x_1 - y_1) I_\kappa] M_\kappa(1) [y_2 J_\kappa + (x_2 - y_2) I_\kappa]\big)
  \qquad \text{and}\\
  y(\kappa) &= \tr\big(M_\kappa(1) [y_1 J_\kappa + (x_1 - y_1) I_\kappa] M_\kappa(2) [y_2 J_\kappa + (x_2 - y_2) I_\kappa]\big).
 \end{align*}
 Since there are just four internal edges in this gadget,
 both of $x(\kappa)$ and $y(\kappa)$ are of degree at most~$4$ in $\kappa$.
 Therefore, we interpolate each of these polynomials using their evaluations at $3 \le \kappa \le 7$
 and obtain the expressions for $x$ and $y$ given in Lemma~\ref{lem:compute:binary:anti-gadget}.
\end{proof}

\begin{remark}
 Lemma~\ref{lem:compute:binary:parallel_edges} is the special case of Lemma~\ref{lem:compute:binary:anti-gadget}
 with $\langle x_1, y_1 \rangle = \langle x_2, y_2 \rangle = \langle 1,0 \rangle$.
\end{remark}

In order to apply a holographic transformation on a particular signature,
it is convenient to express the signature as a sum of degenerate signatures.
Let $e_{\kappa,i}$ be the standard basis vector of length $\kappa$ with a~$1$ at location $i$ and~$0$ elsewhere.
Also let $\mathbf{1}_\kappa$ be the all $1$'s vector of length $\kappa$.
Then the succinct ternary signature $\langle a,b,c \rangle$ on domain size $\kappa$ can be expressed as
\begin{align}
 \langle a,b,c \rangle
 &=
 c \mathbf{1}_\kappa^{\otimes 3}
 + (a - c) \sum_{i=1}^\kappa e_{\kappa,i}^{\otimes 3}
 + (b - c) \sum_{\substack{i,j \in [\kappa]\\i \ne j}}
 \left(
 \begin{array}{rl}
    & e_{\kappa, i} \otimes e_{\kappa, i} \otimes e_{\kappa, j}\\
  + & e_{\kappa, i} \otimes e_{\kappa, j} \otimes e_{\kappa, i}\\
  + & e_{\kappa, j} \otimes e_{\kappa, i} \otimes e_{\kappa, i}
 \end{array}
 \right) \label{eqn:compute:tensor_rank_c}\\
 &=
 b \mathbf{1}_\kappa^{\otimes 3}
 + (a - b) \sum_{i=1}^\kappa e_{\kappa,i}^{\otimes 3}
 + (c - b) \sum_{\substack{\sigma : {1,2,3} \to [\kappa]\\\sigma \text{ injective}}} e_{\kappa, \sigma(1)} \otimes e_{\kappa, \sigma(2)} \otimes e_{\kappa, \sigma(3)}. \label{eqn:compute:tensor_rank_b}
\end{align}
The expression in~(\ref{eqn:compute:tensor_rank_c}) contains $1 + \kappa + 3 \kappa (\kappa - 1) = 3 \kappa^2 - 2 \kappa + 1$ summands.
In general, this is smaller than the one in~(\ref{eqn:compute:tensor_rank_b}),
which contains $1 + \kappa + \kappa (\kappa - 1) (\kappa - 2) = \kappa^3 - 3 \kappa^2 + 3 \kappa + 1$ summands.
It is advantageous find an expression that minimizes the number of summands.
This leads to less computation in the proof of Lemma~\ref{lem:compute:ternary:holographic_transformation}.
However, determining the fewest number of summands for a given signature is exactly the problem of determining tensor rank,
which is a problem well-known to be difficult~\cite{Has90}.

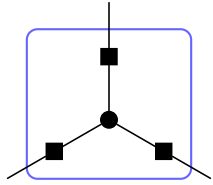
\begin{figure}[ht]
 \centering
 \begin{tikzpicture}[scale=\scale,transform shape,node distance=\nodeDist,semithick]
  \node [external] (0)              {};
  \node [square]   (1) [below of=0] {};
  \node [internal] (2) [below of=1] {};
  \path (2) ++(-150:\nodeDist) node [square] (3) {} ++(-150:\nodeDist) node [external] (4) {};
  \path (2) ++( -30:\nodeDist) node [square] (5) {} ++( -30:\nodeDist) node [external] (6) {};
  \path (0) edge (1)
        (1) edge (2)
        (2) edge (3)
            edge (5)
        (3) edge (4)
        (5) edge (6);
  \begin{pgfonlayer}{background}
   \node[draw=\borderColor,thick,rounded corners,inner xsep=12pt,inner ysep=12pt,fit = (1) (3) (5)] {};
  \end{pgfonlayer}
 \end{tikzpicture}
 \caption{Local holographic transformation gadget construction for a ternary signature.}
 \label{fig:gadget:compute:ternary:local_holographic_transformation}
\end{figure}

There is a gadget construction that mimics the behavior of a holographic transformation.
This construction is called a local holographic transformation~\cite{CLX14}.
For $x,y \in \mathbb{C}$,
let $\langle x,y \rangle$ be a succinct binary signature of type $\tau_2$.
Consider the gadget in Figure~\ref{fig:gadget:compute:ternary:local_holographic_transformation}.
If we assign $\langle a,b,c \rangle$ to the circle vertex and $\langle x,y \rangle$ to the square vertex,
then the resulting signature of this gadget is the same as applying a holographic transformation on $\langle a,b,c \rangle$ with basis
$T = y J_\kappa + (x - y) I_\kappa$.
We use this fact in the following proof.

\begin{lemma} \label{lem:compute:ternary:holographic_transformation}\
 Suppose $\kappa \ge 3$ is the domain size and $a,b,c,x,y \in \mathbb{C}$.
 Let $\langle a,b,c \rangle$ be a succinct signature of type $\tau_3$ and let $T = y J_\kappa + (x - y) I_\kappa$.
 Then $T^{\otimes 3} \langle a,b,c \rangle = \langle a',b',c' \rangle$,
 where
 \begin{align*}
  a' = {}&a \left[x^3 + (\kappa - 1) y^3\right]\\
       {}&+ 3 b (\kappa - 1) \left[x^2 y + x y^2 + (\kappa - 2) y^3\right]\\
       {}&+ c (\kappa - 1) (\kappa - 2) \left[3 x y^2 + (\kappa - 3) y^3\right]\\
  b' = {}&a \left[x^2 y + x y^2 + (\kappa - 2) y^3\right]\\
       {}&+ b \left[x^3 + \kappa x^2 y + (7 \kappa - 12) x y^2 + (3 \kappa^2 - 11 \kappa + 11) y^3\right]\\
       {}&+ c (\kappa - 2) \left[2 x^2 y + (3 \kappa - 7) x y^2 + (\kappa^2 - 4 \kappa + 5) y^3\right], \text{ and}\\
  c' = {}&a \left[3 x y^2 + (\kappa - 3) y^3\right]\\
       {}&+ 3 b \left[2 x^2 y + (3 \kappa - 7) x y^2 + (\kappa^2 - 4 \kappa + 5) y^3\right]\\
       {}&+ c \left[x^3 + 3 (\kappa - 3) x^2 y + 3 (\kappa^2 - 5 \kappa + 7) x y^2 + (\kappa^3 - 6 \kappa^2 + 14 \kappa - 13) y^3\right].
 \end{align*}
 In particular,
 \begin{align*}
  a' - b' = (x - y)^2 [2 \mathfrak{D} + \mathfrak{A} (x - y)]
            \qquad \text{and} \qquad
  b' - c' = (x - y)^2 \mathfrak{D},
 \end{align*}
 where $\mathfrak{D} = (b - c) (x - y) + \mathfrak{B} y$.
 Furthermore, if $\mathfrak{A} = 0$, then
 \begin{align*}
  a' &= 3 b' - 2 c',\\
  b' &= [x + (\kappa - 1) y]
        \left\{
           b x^2
         + 2 [2 b + (\kappa - 3) c] x y
         + [(3 \kappa - 5) b + (\kappa^2 - 5 \kappa + 6) c] y^2
        \right\}
        \qquad \text{and}\\
  c' &= [x + (\kappa - 1) y]
        \left\{
           c x^2
         + 2 [3 b + (\kappa - 4) c] x y
         + [(3 \kappa - 6) b + (\kappa^2 - 5 \kappa + 7) c] y^2
        \right\}.
 \end{align*}
 If $\kappa = 3$, $x = -1$, and $y = 2$, then
 \begin{align*}
  a' =  -3 (5 a + 12 b - 8 c),
        \qquad
  b' =  -3 (2 a +  3 b + 4 c),
        \qquad
        \text{and}
        \qquad
  c' = 3 (4 a - 12 b - c).
 \end{align*}
\end{lemma}

\begin{proof}
 Let $\widehat{f} = T^{\otimes 3} \langle a,b,c \rangle$.
 Since $\langle a,b,c \rangle$ and $\langle x,y \rangle$ are domain invariant,
 the signature of the gadget in Figure~\ref{fig:gadget:compute:ternary:local_holographic_transformation},
 which is the same signature $\widehat{f}$,
 is also domain invariant.
 As a ternary signature,
 the rotational symmetry of this gadget implies the symmetry of the signature.
 Any symmetric domain invariant ternary signature has a succinct signature of type $\tau_3$.
 
 The entries of $\widehat{f}$ are polynomials in $\kappa$ with coefficients from $\Z[a,b,c,x,y]$.
 The degree of these polynomials is at most~$3$ since the arity of $\langle a,b,c \rangle$ is~$3$.
 We compute the entries of $\widehat{f} = T^{\otimes 3} \langle a,b,c \rangle$ as elements in $\Z[a,b,c,x,y]$
 for domain sizes $3 \le \kappa \le 6$ by replacing $\langle a,b,c \rangle$ with an equivalent expression from either~(\ref{eqn:compute:tensor_rank_c}) or~(\ref{eqn:compute:tensor_rank_b}).
 Then we interpolate the entries of $\widehat{f}$ as elements in $(\Z[a,b,c,x,y])[\kappa]$.
 The resulting expressions for the signature entries are as given in the statement of Lemma~\ref{lem:compute:ternary:holographic_transformation}.
 
 It is straightforward to verify the expressions for $a' - b'$ and $b' - c'$ given those for $a'$, $b'$, and $c'$.
 Recall that $\mathfrak{A} = a - 3 b + 2 c$.
 If $\mathfrak{A} = 0$,
 then it follows that $a' - 3 b' + 2 c' = 0$ as well since
 \begin{align*}
  a' - 3 b' + 2 c'
  &= a' - b' - 2 (b' - c')\\
  &= (x - y)^2 [2 \mathfrak{D} + \mathfrak{A} (x - y)] - 2 (x - y)^2 \mathfrak{D}\\
  &= \mathfrak{A} (x - y)^3
  = 0.
 \end{align*}
 The expressions for $b'$ and $c'$ when $\mathfrak{A} = 0$ directly follow from their general expressions above.
\end{proof}

By composing smaller gadgets,
we can easily compute the signatures of rather large gadgets.

\begin{figure}[ht]
 \centering
 \def\capWidth{4cm}
 \captionsetup[subfigure]{width=\capWidth}
 \subcaptionbox{\label{subfig:gadget:compute:binary:diamond_inside}Inner structure}[\capWidth]{
  \begin{tikzpicture}[scale=\scale,transform shape,node distance=\nodeDist,semithick]
   \node[external] (0)              {};
   \node[external] (1) [below of=0] {};
   \node[internal] (2) [right of=0] {};
   \node[internal] (3) [right of=1] {};
   \node[external] (4) [right of=2] {};
   \node[external] (5) [right of=3] {};
   \path (0.west) edge[postaction={decorate, decoration={
                                              markings,
                                              mark=at position 0.70  with {\arrow[>=diamond,white] {>}; },
                                              mark=at position 0.70  with {\arrow[>=open diamond]  {>}; } } }] (2)
         (1.west) edge (3)
         (2)      edge (3)
                  edge (4.east)
         (3)      edge (5.east);
   \begin{pgfonlayer}{background}
    \node[draw=\borderColor,thick,rounded corners,inner xsep=12pt,inner ysep=8pt,fit = (2) (3)] {};
   \end{pgfonlayer}
  \end{tikzpicture}
 }
 \qquad
 \subcaptionbox{\label{subfig:gadget:compute:binary:diamond_outside}Outer structure}[\capWidth]{
  \begin{tikzpicture}[scale=\scale,transform shape,node distance=\nodeDist,semithick]
   \node[external] (0)              {};
   \node[internal] (1) [right of=0] {};
   \node[square]   (2) [right of=1] {};
   \node[internal] (3) [right of=2] {};
   \node[external] (4) [right of=3] {};
   \path (0) edge             (1)
         (1) edge[bend left, postaction={decorate, decoration={
                                                    markings,
                                                    mark=at position 0.999 with {\arrow[>=diamond,white] {>}; },
                                                    mark=at position 0.999 with {\arrow[>=open diamond]  {>}; } } }] (2)
             edge[bend right] (2)
         (2) edge[bend left]  (3)
             edge[bend right] (3)
         (3) edge             (4);
   \begin{pgfonlayer}{background}
    \node[draw=\borderColor,thick,rounded corners,inner xsep=12pt,inner ysep=12pt,fit = (1) (3)] {};
   \end{pgfonlayer}
  \end{tikzpicture}
 }
 \qquad
 \subcaptionbox{\label{subfig:gadget:compute:binary:diamond}Entire binary gadget}[\capWidth]{
  \begin{tikzpicture}[scale=\scale,transform shape,node distance=\nodeDist,semithick]
   \node[external] (0)                    {};
   \node[internal] (1) [      right of=0] {};
   \node[internal] (2) [above right of=1] {};
   \node[internal] (3) [below right of=1] {};
   \node[internal] (4) [below right of=2] {};
   \node[external] (5) [      right of=4] {};
   \path (0) edge (1)
         (1) edge (2)
             edge (3)
         (2) edge (3)
             edge (4)
         (3) edge (4)
         (4) edge (5);
   \begin{pgfonlayer}{background}
    \node[draw=\borderColor,thick,rounded corners,inner xsep=12pt,inner ysep=8pt,fit = (1) (2) (3) (4)] {};
   \end{pgfonlayer}
  \end{tikzpicture}
 }
 \caption{Decomposition of a binary gadget.
 All circle vertices are assigned $\langle a,b,c \rangle$
 and the square vertex in~(\subref{subfig:gadget:compute:binary:diamond_outside})
 is assigned the signature of the gadget in~(\subref{subfig:gadget:compute:binary:diamond_inside}).}
 \label{fig:gadget:compute:binary:diamond_construction}
\end{figure}

\begin{lemma} \label{lem:compute:binary:diamond:construction}
 Suppose $\kappa \ge 3$ is the domain size and $a,b,c \in \mathbb{C}$.
 Let $\langle a,b,c \rangle$ be a succinct ternary signature of type $\tau_3$.
 If $\langle a,b,c \rangle$ is assigned to every vertex of the gadget in Figure~\ref{subfig:gadget:compute:binary:diamond},
 then the resulting signature is the succinct binary signature $\langle x,y \rangle$ of type $\tau_2$,
 where
 \begin{align*}
  x = {}&  a^4 + 6 (\kappa - 1) a^2 b^2 + 16 (\kappa - 1) a b^3 + 12 (\kappa - 1) (\kappa - 2) a b^2 c + 12 (\kappa - 1) (\kappa - 2) a b c^2\\
        &+ 4 (\kappa - 1) (\kappa - 2) (\kappa - 3) a c^3 + 3 (\kappa - 1) (5 \kappa - 7) b^4 + 4 (\kappa - 1) (\kappa - 2) (\kappa+5) b^3 c\\
        &+ 6 (\kappa - 1) (\kappa - 2) (7 \kappa - 18) b^2 c^2 + 12 (\kappa - 3)^2 (\kappa - 1) (\kappa - 2) b c^3\\
        &+ (\kappa - 1) (\kappa - 2) (\kappa^3 - 9 \kappa^2 + 29 \kappa - 32) c^4
        \qquad \text{and}\\
  y = {}&  2 a^3 b + (\kappa + 4) a^2 b^2 + 4 (\kappa - 2) a^2 b c + (\kappa - 2) a^2 c^2 + 2 (9 \kappa - 11) a b^3 + 2 (\kappa - 2) (3 \kappa + 8) a b^2 c\\
        &+ 2 (\kappa - 2) (12 \kappa - 31) a b c^2 + 2 (\kappa - 2) (2 \kappa^2 - 11 \kappa + 16) a c^3 + (7 \kappa^2 + 3 \kappa - 24) b^4\\
        &+ 2 (\kappa - 2) (\kappa^2 + 31 \kappa - 70) b^3 c + (\kappa - 2) (48 \kappa^2 - 234 \kappa + 301) b^2 c^2\\
        &+ 2 (\kappa - 2) (6 \kappa^3 - 45 \kappa^2 + 121 \kappa - 116) b c^3 + (\kappa - 2) (\kappa - 3) (\kappa^3 - 7 \kappa^2 + 19 \kappa - 20) c^4.
 \end{align*}
\end{lemma}

\begin{proof}
 Since $\langle a,b,c \rangle$ is domain invariant,
 the signature of this gadget is also domain invariant.
 Any domain invariant binary signature has a succinct signature of type $\tau_2$.

\begin{table}[p]
 \begin{align*}
  x = {}&  a^6 + 9 (\kappa - 1) a^4 b^2
         + 32 (\kappa - 1) a^3 b^3
         + 18 (\kappa - 1) (\kappa - 2) a^3 b^2 c
         + 12 (\kappa - 1) (\kappa - 2) a^3 b c^2\\
        &+ 2 (\kappa - 1) (\kappa - 2) (\kappa - 3) a^3 c^3
         + 3 (\kappa - 1) (16 \kappa - 7) a^2 b^4
         + 6 (\kappa - 1) (\kappa - 2) (\kappa + 19) a^2 b^3 c\\
        &+ 18 (\kappa - 1) (\kappa - 2) (4 \kappa - 7) a^2 b^2 c^2
         + 6 (\kappa - 1) (\kappa - 2) (\kappa^2 + 2 \kappa - 13) a^2 b c^3\\
        &+ 3 (\kappa - 1) (\kappa - 2) (3 \kappa^2 - 17 \kappa + 25) a^2 c^4
         + 6 (\kappa - 1) (\kappa^2 + 27 \kappa - 42) a b^5\\
        &+ 6 (\kappa - 1) (\kappa - 2) (40 \kappa - 41) a b^4 c
         + 24 (\kappa - 1) (\kappa - 2) (3 \kappa^2 + 8 \kappa - 36) a b^3 c^2\\
        &+ 6 (\kappa - 1) (\kappa - 2) (\kappa^3 + 50 \kappa^2 - 285 \kappa + 393) a b^2 c^3\\
        &+ 6 (\kappa - 1) (\kappa - 2) (13 \kappa^3 - 108 \kappa^2 + 311 \kappa - 307) a b c^4\\
        &+ 6 (\kappa - 1) (\kappa - 2) (\kappa - 3) (\kappa^3 - 8 \kappa^2 + 24 \kappa - 26) a c^5\\
        &+ (\kappa - 1) (\kappa^3 + 83 \kappa^2 - 189 \kappa + 81) b^6
         + 18 (\kappa - 1) (\kappa - 2) (4 \kappa^2 + 13 \kappa - 43) b^5 c\\
        &+ 3 (\kappa - 1) (\kappa - 2) (7 \kappa^3 + 222 \kappa^2 - 1156 \kappa + 1442) b^4 c^2\\
        &+ 2 (\kappa - 1) (\kappa - 2) (\kappa^4 + 221 \kappa^3 - 1725 \kappa^2 + 4576 \kappa - 4153) b^3 c^3\\
        &+ 3 (\kappa - 1) (\kappa - 2) (43 \kappa^4 - 441 \kappa^3 + 1791 \kappa^2 - 3393 \kappa + 2505) b^2 c^4\\
        &+ 6 (\kappa - 1) (\kappa - 2) (\kappa - 3) (3 \kappa^4 - 29 \kappa^3 + 116 \kappa^2 - 228 \kappa + 182) b c^5\\
        &+ (\kappa - 1) (\kappa - 2) (\kappa^6 - 15 \kappa^5 + 98 \kappa^4 - 361 \kappa^3 + 798 \kappa^2 - 1004 \kappa + 556) c^6\\
  \text{and}\\
  y = {}&  2 a^5 b
         + (\kappa + 8) a^4 b^2
         + 4 (\kappa - 2) a^4 b c
         + 2 (\kappa - 2) a^4 c^2
         + 4 (9 \kappa - 11) a^3 b^3
         + 2 (\kappa - 2) (3 \kappa + 17) a^3 b^2 c\\
        &+ 4 (\kappa - 2) (7 \kappa - 18) a^3 b c^2
         + 2 (\kappa - 3)^2 (\kappa - 2) a^3 c^3
         + (23 \kappa^2 + 49 \kappa - 114) a^2 b^4\\
        &+ 2 (\kappa - 2) (\kappa^2 + 94 \kappa - 147) a^2 b^3 c
         + 6 (\kappa - 2) (12 \kappa^2 - 34 \kappa + 17) a^2 b^2 c^2\\
        &+ 2 (\kappa - 2) (3 \kappa^3 + 9 \kappa^2 - 97 \kappa + 149) a^2 b c^3
         + (\kappa - 2) (9 \kappa^3 - 68 \kappa^2 + 181 \kappa - 171) a^2 c^4\\
        &+ 2 (3 \kappa^3 + 73 \kappa^2 - 183 \kappa + 99) a b^5
         + 2 (\kappa - 2) (96 \kappa^2 - 43 \kappa - 255) a b^4 c\\
        &+ 4 (\kappa - 2) (16 \kappa^3 + 94 \kappa^2 - 655 \kappa + 855) a b^3 c^2\\
        &+ 2 (\kappa - 2) (3 \kappa^4 + 159 \kappa^3 - 1233 \kappa^2 + 3164 \kappa - 2809) a b^2 c^3\\
        &+ 2 (\kappa - 2) (39 \kappa^4 - 375 \kappa^3 + 1425 \kappa^2 - 2555 \kappa + 1825) a b c^4\\
        &+ 2 (\kappa - 2) (3 \kappa^5 - 36 \kappa^4 + 181 \kappa^3 - 482 \kappa^2 + 686 \kappa - 418) a c^5\\
        &+ (\kappa^4 + 50 \kappa^3 - 17 \kappa^2 - 396 \kappa + 486) b^6\\
        &+ 2 (\kappa - 2) (28 \kappa^3 + 251 \kappa^2 - 1302 \kappa + 1467) b^5 c\\
        &+ (\kappa - 2) (19 \kappa^4 + 745 \kappa^3 - 5374 \kappa^2 + 12664 \kappa - 10320) b^4 c^2\\
        &+ 2 (\kappa - 2) (\kappa^5 + 224 \kappa^4 - 2062 \kappa^3 + 7371 \kappa^2 - 12357 \kappa + 8227) b^3 c^3\\
        &+ (\kappa - 2) (129 \kappa^5 - 1464 \kappa^4 + 6952 \kappa^3 - 17464 \kappa^2 + 23397 \kappa - 13387) b^2 c^4\\
        &+ 2 (\kappa - 2) (9 \kappa^6 - 123 \kappa^5 + 727 \kappa^4 - 2405 \kappa^3 + 4754 \kappa^2 - 5374 \kappa + 2718) b c^5\\
        &+ (\kappa - 3) (\kappa - 2) (\kappa^6 - 13 \kappa^5 + 74 \kappa^4 - 239 \kappa^3 + 470 \kappa^2 - 544 \kappa + 292) c^6.
 \end{align*}
 \caption{The signature of the gadget in Figure~\ref{subfig:gadget:compute:binary:square} is $\langle x,y \rangle$ for the $x$ and $y$ above.}
 \label{tbl:compute:xy}
 \thisfloatpagestyle{empty}
\end{table}
 
 Consider the gadget in Figure~\ref{subfig:gadget:compute:binary:diamond_inside}.
 We assign $\langle a,b,c \rangle$ to both vertices.
 By Lemma~\ref{lem:compute:quaternary:I},
 this gadget has the succinct quaternary signature $f$ of type $\tau_4$,
 where $f$ is given in Lemma~\ref{lem:compute:quaternary:I}.
 
 Now consider the gadget in Figure~\ref{subfig:gadget:compute:binary:diamond_outside}.
 We assign $\langle a,b,c \rangle$ the circle vertices and $f$ to the square vertex.
 By partitioning the internal edge assignments into parts with the same contribution to the sum,
 one can verify that this gadget has the succinct binary signature $\langle x,y \rangle$ of type $\tau_2$,
 where
 \begin{alignat*}{2}
  x = &&{}                                          &f_{\subMat{1}{1}{1}{1}} \left[a^2 + (\kappa - 1) b^2\right]\\
      &&{}+                          4 (\kappa - 1) &f_{\subMat{1}{1}{1}{2}} \left[a b + b^2 + (\kappa - 2) b c\right]\\
      &&{}+                            (\kappa - 1) &f_{\subMat{1}{1}{2}{2}} \left[2 a b + (\kappa - 2) b^2\right]\\
      &&{}+             2 (\kappa^2 - 3 \kappa + 2) &f_{\subMat{1}{1}{2}{3}} \left[a c + 2 b^2 + (\kappa - 3) b c\right]\\
      &&{}+                            (\kappa - 1) &f_{\subMat{1}{2}{1}{2}} \left[2 b^2 + (\kappa - 2) c^2\right]\\
      &&{}+             2 (\kappa^2 - 3 \kappa + 2) &f_{\subMat{1}{2}{1}{3}} \left[b^2 + 2 b c + (\kappa - 3) c^2\right]\\
      &&{}+                            (\kappa - 1) &f_{\subMat{1}{2}{2}{1}} \left[2 b^2 + (\kappa - 2) c^2\right]\\
      &&{}+             2 (\kappa^2 - 3 \kappa + 2) &f_{\subMat{1}{2}{3}{1}} \left[b^2 + 2 b c + (\kappa - 3) c^2\right]\\
      &&{}+ (\kappa^3 - 6 \kappa^2 + 11 \kappa - 6) &f_{\subMat{1}{2}{3}{4}} \left[4 b c + (\kappa - 4) c^2\right]
      \qquad \qquad \text{and} \qquad\\
  y = &&{}                            &f_{\subMat{1}{1}{1}{1}} \left[2 a b + (\kappa - 2) b^2\right]\\
      &&{}+                         4 &f_{\subMat{1}{1}{1}{2}} \left[a b + (\kappa - 2) a c + (2 \kappa - 3) b^2 + (\kappa - 2)^2 b c\right]\\
      &&{}+                           &f_{\subMat{1}{1}{2}{2}} \left[a^2 + 2 (\kappa - 2) a b + (\kappa^2 - 3 \kappa + 3) b^2\right]\\
      &&{}+            2 (\kappa - 2) &f_{\subMat{1}{1}{2}{3}} \left[2 a b + (\kappa - 3) a c + 2 (\kappa - 2) b^2 + (\kappa^2 - 4 \kappa + 5) b c\right]\\
      &&{}+                           &f_{\subMat{1}{2}{1}{2}} \left[2 b^2 + 4 (\kappa - 2) b c + (\kappa^2 - 5 \kappa + 6) c^2\right]\\
      &&{}+            2 (\kappa - 2) &f_{\subMat{1}{2}{1}{3}} \left[3 b^2 + 2 (2 \kappa - 5) b c + (\kappa^2 - 5 \kappa + 7) c^2\right]\\
      &&{}+                           &f_{\subMat{1}{2}{2}{1}} \left[2 b^2 + 4 (\kappa - 2) b c + (\kappa^2 - 5 \kappa + 6) c^2\right]\\
      &&{}+            2 (\kappa - 2) &f_{\subMat{1}{2}{3}{1}} \left[3 b^2 + 2 (2 \kappa - 5) b c + (\kappa^2 - 5 \kappa + 7) c^2\right]\\
      &&{}+ (\kappa^2 - 5 \kappa + 6) &f_{\subMat{1}{2}{3}{4}} \left[4 b^2 + 4 (\kappa - 3) b c + (\kappa^2 - 5 \kappa + 8) c^2\right].
 \end{alignat*}
 Substituting for the entries of $f$ gives the result stated in Lemma~\ref{lem:compute:binary:diamond:construction}.
\end{proof}

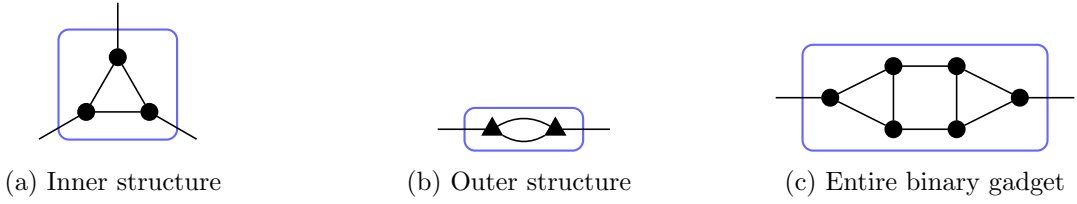
\begin{figure}[ht]
 \centering
 \def\capWidth{4.5cm}
 \captionsetup[subfigure]{width=\capWidth}
 \subcaptionbox{\label{subfig:gadget:compute:binary:square_inside}Inner structure}[\capWidth]{
  \begin{tikzpicture}[scale=\scale,transform shape,node distance=\nodeDist,semithick]
  \node [external] (0)              {};
  \node [internal] (1) [below of=0] {};
  \path (1) ++(-120:\nodeDist) node [internal] (2) {} ++(-150:\nodeDist) node [external] (3) {};
  \path (1) ++( -60:\nodeDist) node [internal] (4) {} ++( -30:\nodeDist) node [external] (5) {};
  \path (0) edge (1)
        (1) edge (2)
            edge (4)
        (2) edge (3)
            edge (4)
        (4) edge (5);
  \begin{pgfonlayer}{background}
   \node[draw=\borderColor,thick,rounded corners,inner xsep=12pt,inner ysep=12pt,fit = (1) (2) (4)] {};
   \end{pgfonlayer}
  \end{tikzpicture}
 }
 \qquad
 \subcaptionbox{\label{subfig:gadget:compute:binary:square_outside}Outer structure}[\capWidth]{
  \begin{tikzpicture}[scale=\scale,transform shape,node distance=\nodeDist,semithick]
   \node[external] (0)               {};
   \node[triangle] (1)  [right of=0] {};
   \node[external] (n1) [right of=0] {};
   \node[triangle] (2)  [right of=1] {};
   \node[external] (n2) [right of=1] {};
   \node[external] (3)  [right of=2] {};
   \path (0) edge             (1)
         (1) edge[bend left]  (2)
             edge[bend right] (2)
         (2) edge             (3);
   \begin{pgfonlayer}{background}
    \node[draw=\borderColor,thick,rounded corners,inner xsep=12pt,inner ysep=8pt,fit = (n1) (n2)] {};
   \end{pgfonlayer}
  \end{tikzpicture}
 }
 \qquad
 \subcaptionbox{\label{subfig:gadget:compute:binary:square}Entire binary gadget}[\capWidth]{
  \begin{tikzpicture}[scale=\scale,transform shape,node distance=\nodeDist,semithick]
   \node[internal] (0)                    {};
   \node[internal] (1) [      right of=0] {};
   \node[internal] (2) [below       of=0] {};
   \node[internal] (3) [below       of=1] {};
   \node[external] (4) [      left  of=0] {};
   \node[external] (5) [      right of=1] {};
   \path let
          \p0 = (0),
          \p2 = (2),
          \p4 = (4)
         in
          node[internal] (6) at (\x4, \y0 / 2 + \y2 / 2) {};
   \path let
          \p1 = (1),
          \p3 = (3),
          \p5 = (5)
         in
          node[internal] (7) at (\x5, \y1 / 2 + \y3 / 2) {};
   \node[external] (8) [      left  of=6] {};
   \node[external] (9) [      right of=7] {};
   \path (0) edge (1)
             edge (2)
             edge (6)
         (1) edge (3)
             edge (7)
         (2) edge (3)
             edge (6)
         (3) edge (7)
         (6) edge (8)
         (7) edge (9);
   \begin{pgfonlayer}{background}
    \node[draw=\borderColor,thick,rounded corners,inner xsep=12pt,inner ysep=8pt,fit = (0) (3) (6) (7)] {};
   \end{pgfonlayer}
  \end{tikzpicture}
 }
 \caption{Decomposition of a binary gadget.
 All circle vertices are assigned $\langle a,b,c \rangle$
 and the triangle vertices in~(\subref{subfig:gadget:compute:binary:square_outside})
 is assigned the signature of the gadget in~(\subref{subfig:gadget:compute:binary:square_inside}).}
 \label{fig:gadget:compute:binary:square_construction}
\end{figure}

\begin{lemma} \label{lem:compute:binary:square:construction}
 Suppose $\kappa \ge 3$ is the domain size and $a,b,c \in \mathbb{C}$.
 Let $\langle a,b,c \rangle$ be a succinct ternary signature of type $\tau_3$.
 If $\langle a,b,c \rangle$ is assigned to every vertex of the gadget in Figure~\ref{subfig:gadget:compute:binary:square},
 then the resulting signature is the binary succinct signature $\langle x,y \rangle$ of type $\tau_2$,
 where $x$ and $y$ are given in Table~\ref{tbl:compute:xy}.
\end{lemma}

\begin{proof}
 Since $\langle a,b,c \rangle$ is domain invariant,
 the signature of this gadget is also domain invariant.
 Any domain invariant binary signature has a succinct signature of type $\tau_2$.
 
 Consider the gadget in Figure~\ref{subfig:gadget:compute:binary:square_inside}.
 We assign $\langle a,b,c \rangle$ to all vertices.
 By Lemma~\ref{lem:compute:ternary:triangle},
 this gadget has the succinct ternary signature $f = \langle a_0, b_0, c_0 \rangle$ of type $\tau_4$,
 where $a_0$, $b_0$, and $c_0$ are given in the statement of Lemma~\ref{lem:compute:ternary:triangle} as $a'$, $b'$, and $c'$ respectively.
 
 Now consider the gadget in Figure~\ref{subfig:gadget:compute:binary:square_outside}.
 We assign $f$ to the vertices.
 By Lemma~\ref{lem:compute:binary:parallel_edges},
 the resulting gadget has the binary succinct signature $\langle x,y \rangle$ of type $\tau_2$,
 where
 \begin{align*}
  x &=   a_0^2
       + 3 (\kappa - 1) b_0^2
       + (\kappa - 1) (\kappa - 2) c_0^2
       \qquad \text{and}\\
  y &=   2 a_0 b_0
       + \kappa b_0^2
       + 4 (\kappa - 2) b_0 c_0
       + (\kappa - 2) (\kappa - 3) c_0^2.
 \end{align*}
 Substituting for $a_0$, $b_0$, and $c_0$ gives the result in Table~\ref{tbl:compute:xy}.
\end{proof}

Beyond the gadgets in this section,
there are two $9$-by-$9$ recurrence matrices that appear in our proofs (see Table~\ref{tbl:ternary:symmetric_weave_interpolation} and Table~\ref{tbl:unary:matrix:recurrence}).
No entry in those recurrence matrices is any harder to compute than any signature entry appearing in this section.
The difficulty with these recurrence matrices is the sheer number of terms that must be computed.

\section{More Binary Interpolation} \label{sec:appendix:binary}

For some settings of $a,b,c \in \mathbb{C}$,
Lemma~\ref{lem:binary:general:root_of_unity} and Lemma~\ref{lem:binary:k=3_and_b=0} do not apply.
However, these settings are easily handled on a case-by-case basis.

\begin{lemma} \label{lem:appendix:binary}
 Suppose $\kappa \ge 3$ is the domain size.
 Let $\mathcal{F}$ be a signature set containing the succinct unary signature $\langle 1 \rangle$ of type $\tau_1$
 and any of the following succinct ternary signatures of type $\tau_3$:
 \begin{enumerate}
  \item $\langle \kappa - 2 \pm i \kappa \sqrt{2 (\kappa - 2)}, \kappa - 2, -2 \rangle$;%
  \label{case:lem:appendix:k-2pmik2k-2k-2-2}
  
  \item $\langle (\kappa - 2)^2 \pm i \kappa \sqrt{\kappa^2 - 4}, -2 (\kappa - 2), 4 \rangle$;%
  \label{case:lem:appendix:k-22pmikk2-2-2k-24}
  
  \item $\langle -(2 \kappa - 3) \big[2 (\kappa - 2) \pm i \kappa \sqrt{2 (\kappa - 2)}\big], -2 (\kappa - 3) (\kappa - 2) \pm i \kappa \sqrt{2 (\kappa - 2)}, 4 (2 \kappa - 3) \rangle$
  with $\kappa \ne 4$;%
  \label{case:lem:appendix:-2k-32k-2pmik2k-2-2k-3k-2pmik2k-242k-3}
  
  \item $\langle -\kappa^2 + 2, 2, 2 \rangle$;%
  \label{case:lem:appendix:-k2+222}
  
  \item $\langle \kappa^2 - 6 \kappa + 6, -2 (\kappa - 3), 6 \rangle$;%
  \label{case:lem:appendix:k2-6k+6-2k-36}
  
  \item $\langle (\kappa - 3) (\kappa - 2)^2 \pm i \kappa (2 \kappa - 3) \sqrt{\kappa^2 - 4}, -3 (\kappa - 2)^2 \mp i \kappa \sqrt{\kappa^2 - 4}, 2 (5 \kappa - 6) \rangle$;%
  \label{case:lem:appendix:k-3k-22pmik2k-3-k2-4-3k-22mpikk2-425k-6}
  
  \item $\langle -(\kappa - 1) \big[5 (\kappa - 2) \pm 3 i \kappa \sqrt{2 (\kappa - 2)}\big], -(\kappa - 2) (3 \kappa - 5) \pm i \kappa \sqrt{2 (\kappa - 2)}, 9 \kappa - 10 \rangle$;%
  \label{case:lem:appendix:-k-15k-2mp3ik2k-2-k-23k-5mpik3k-29k-10}
  
  \item $\langle (\kappa - 1) \big[(\kappa - 2) (2 \kappa + 3) \pm 3 \kappa \sqrt{\kappa^2 - 5 \kappa + 6}\big], (\kappa - 3) (\kappa - 2) \mp \kappa \sqrt{\kappa^2 - 5 \kappa + 6}, -5 \kappa + 6 \rangle$;%
  \label{case:lem:appendix:k-1k-22k+3mp3kk2-5k+6k-3k-2pmkk2-5k+6-5k+6}
  
  \item $\langle (\kappa - 1) \big[(\kappa - 2) (2 \kappa - 7) \pm 3 i \kappa \sqrt{\kappa^2 - \kappa - 2}\big], -(\kappa - 2) (5 \kappa - 7) \mp i \kappa \sqrt{\kappa^2 - \kappa - 2}, 13 \kappa - 14 \rangle$;%
  \label{case:lem:appendix:k-1k-22k-7pm3ikk2-k-2-k-25k-7mpikk2-k-213k-14}
  
  \item $\langle 1,0,-2 \rangle$ with $\kappa = 3$;
  \label{case:lem:appendix:10-2}
  
  \item $\langle \pm i \sqrt{2}, 0, 1 \rangle$ with $\kappa = 3$;
  \label{case:lem:appendix:pmi201}
  
  \item $\langle -1 \pm i \sqrt{2}, 0, 1 \rangle$ with $\kappa = 3$;
  \label{case:lem:appendix:-1pmi201}
  
  \item $\langle -1 \pm 3 i \sqrt{3}, 0, 2 \rangle$ with $\kappa = 3$;
  \label{case:lem:appendix:-1pm3i302}
  
 \end{enumerate}
 Then
 \[
  \PlHolant(\mathcal{F} \union \{\langle x,y \rangle\}) \le_T \PlHolant(\mathcal{F})
 \]
 for any $x,y \in \mathbb{C}$,
 where $\langle x,y \rangle$ is a succinct binary signature of type $\tau_2$.
\end{lemma}

\begin{proof}
 In each case,
 we use the recursive construction in Figure~\ref{fig:gadget:k>r:binary:interpolation}.
 We simply state which gadget we use,
 the signature of that gadget,
 and the eigenvalues of its associated recurrence matrix (cf.~Lemma~\ref{lem:k>r:binary:interpolation:eigenvalues}).
 Then the result easily follows from Corollary~\ref{cor:k>r:binary:interpolate} as the eigenvalues have distinct complex norms.
 
 We use four possible gadgets,
 which are in Figure~\ref{subfig:gadget:binary:unary},
 Figure~\ref{subfig:gadget:compute:binary:diamond},
 and Figure~\ref{subfig:gadget:compute:binary:square}.
 The signatures for the last two gadgets are given by
 Lemma~\ref{lem:compute:binary:diamond:construction} and Lemma~\ref{lem:compute:binary:square:construction} respectively.
 
 \begin{enumerate}
  \item For $\langle \kappa - 2 \pm i \kappa \sqrt{2 (\kappa - 2)}, \kappa - 2, -2 \rangle$,
  we first use the gadget in Figure~\ref{subfig:gadget:compute:binary:diamond}.
  Let $\gamma = \pm i \sqrt{2 (\kappa - 2)}$.
  Up to a nonzero factor of $\frac{(\gamma - 2)^7 \gamma^2 (\gamma + 2)^3}{64}$,
  the signature of the gadget is $\langle -1, 1 \rangle$,
  which means the eigenvalues are $\kappa -2$ and $-2$.
  If $\kappa \ne 4$,
  then these eigenvalues have distinct complex norms.
  Otherwise, $\kappa = 4$ and we use the gadget in Figure~\ref{subfig:gadget:compute:binary:square}.
  Up to a factor of $\pm 65536 i$,
  the signature of this gadget is $\langle 1, -3 \rangle$,
  which means the eigenvalues are $-8$ and $4$.
  
  \item For $\langle (\kappa - 2)^2 \pm i \kappa \sqrt{\kappa^2 - 4}, -2 (\kappa - 2), 4 \rangle$,
  we first use the gadget in Figure~\ref{subfig:gadget:compute:binary:diamond}.
  Let $\gamma = \pm i \sqrt{\kappa^2 - 4}$.
  Up to a nonzero factor of $-4 (\kappa - 2) \kappa^3 (\kappa^2 - 4 \gamma - 8)$,
  the signature of this gadget is $\langle \kappa^2 - 6 \kappa + 4, -2 (\kappa - 4) \rangle$,
  which means the eigenvalues are $-(\kappa - 2)^2$ and $\kappa^2 - 4 k - 4$.
  If $\kappa \ge 5$,
  then these eigenvalues have opposite signs but cannot be the negative of each other.
  Thus, they have distinct complex norms.
  The same conclusion holds for $\kappa = 3$ by direct inspection.
  Otherwise, $\kappa = 4$ and we use the gadget in Figure~\ref{subfig:gadget:compute:binary:square}.
  Up to a factor of $2097152$,
  the signature of this gadget is $\langle 5,1 \rangle$,
  which means the eigenvalues are $8$ and $4$.
  
  \item For $\langle -(2 \kappa - 3) \big[2 (\kappa - 2) \pm i \kappa \sqrt{2 (\kappa - 2)}\big], -2 (\kappa - 3) (\kappa - 2) \pm i \kappa \sqrt{2 (\kappa - 2)}, 4 (2 \kappa - 3) \rangle$,
  we have $\kappa \ne 4$.
  We use the gadget in Figure~\ref{subfig:gadget:compute:binary:diamond}.
  Let $\gamma = \pm i \sqrt{2 (\kappa - 2)}$.
  Up to a nonzero factor of $-4 (\kappa - 2) \kappa^6 (3 \kappa - 4) (4 \kappa^2 - 28 \kappa + 41 - 4 \gamma (2 \kappa - 5))$,
  the signature of the gadget is $\frac{1}{\kappa} \langle 3 \kappa - 4, \kappa - 4 \rangle$,
  which means the eigenvalues are $\kappa - 2$ and~$2$.
  
  \item For $\langle -\kappa^2 + 2, 2, 2 \rangle$,
  we use the gadget in Figure~\ref{subfig:gadget:compute:binary:diamond}.
  Up to a nonzero factor of $(\kappa - 2) \kappa^5$,
  the signature for this gadget is $\langle \kappa^2 + 2 \kappa - 4, -4 \rangle$,
  which means the eigenvalues are $(\kappa - 2) \kappa$ and $\kappa (\kappa + 2)$.
  
  \item For $\langle \kappa^2 - 6 \kappa + 6, -2 (\kappa - 3), 6 \rangle$,
  we use the gadget in Figure~\ref{subfig:gadget:compute:binary:diamond}.
  Up to a nonzero factor of $(\kappa - 2) \kappa^5$,
  the signature for this gadget is $\langle \kappa^2 + 2 \kappa - 4, -4 \rangle$,
  which means the eigenvalues are $(\kappa - 2) \kappa$ and $\kappa (\kappa + 2)$.
  
  \item For $\langle (\kappa - 3) (\kappa - 2)^2 \pm i \kappa (2 \kappa - 3) \sqrt{\kappa^2 - 4}, -3 (\kappa - 2)^2 \mp i \kappa \sqrt{\kappa^2 - 4}, 2 (5 \kappa - 6) \rangle$,
  we use the gadget in Figure~\ref{subfig:gadget:compute:binary:diamond}.
  Let $\gamma = \pm i \sqrt{\kappa^2 - 4}$.
  Up to a nonzero factor of $(\gamma - 2)^2 (\gamma + 2)^2 (\kappa - 2) \kappa [7 \kappa^2 + 60 \kappa - 164 + 8 \gamma (3 \kappa - 10)]$,
  the signature of the gadget is $\langle -\kappa^4 + 6 \kappa^3 + 4 \kappa^2 - 24 \kappa + 16, 2 (\kappa^3 - 2 \kappa^2 - 8 \kappa + 8)\rangle$,
  which means the eigenvalues are $\lambda_1 = (\kappa - 2) \kappa (\kappa^2 + 2 \kappa - 4)$ and $\lambda_2 = -\kappa (\kappa + 2) (\kappa^2 - 6 \kappa + 4)$.
  For $3 \le \kappa \le 5$,
  one can directly check that these eigenvalues have distinct complex norms.
  For $\kappa \ge 6$,
  we have $\lambda_2 < 0$,
  so these eigenvalues have the same complex norm preciously when $\lambda_1 = -\lambda_2$.
  However, $\lambda_1 + \lambda_2 = 4 \kappa^3 \ne 0$,
  so the eigenvalues have distinct complex norms.
  
  \item For $\langle -(\kappa - 1) \big[5 (\kappa - 2) \pm 3 i \kappa \sqrt{2 (\kappa - 2)}\big], -(\kappa - 2) (3 \kappa - 5) \pm i \kappa \sqrt{2 (\kappa - 2)}, 9 \kappa - 10 \rangle$,
  we first use the gadget in Figure~\ref{subfig:gadget:compute:binary:diamond}.
  Let $\gamma = \pm i \sqrt{2 (\kappa - 2)}$.
  Up to a nonzero factor of $-(\kappa - 2) (\kappa - 1) \kappa^5 [81 \kappa^2 - 756 \kappa + 1252 - 24 (9 \kappa - 26) \gamma]$,
  the signature of this gadget is $\langle 5 \kappa - 6, \kappa - 6 \rangle$,
  which means the eigenvalues are $\kappa - 2$ and~$4$.
  If $\kappa \ne 6$,
  then these eigenvalues have distinct complex norms.
  Otherwise, $\kappa = 6$ and we use the gadget in Figure~\ref{subfig:gadget:compute:binary:square}.
  Up to a factor of $-17199267840 (1169 \pm 450 i \sqrt{2})$,
  the signature of this gadget is $\langle 7, 13\rangle$,
  which means the eigenvalues are~$72$ and~$-6$.
  
  \item For $\langle (\kappa - 1) \big[(\kappa - 2) (2 \kappa + 3) \pm 3 \kappa \sqrt{\kappa^2 - 5 \kappa + 6}\big], (\kappa - 3) (\kappa - 2) \mp \kappa \sqrt{\kappa^2 - 5 \kappa + 6}, -5 \kappa + 6 \rangle$,
  we first use the gadget in Figure~\ref{subfig:gadget:compute:binary:diamond}.
  Let $\gamma = \pm \sqrt{\kappa^2 - 5 \kappa + 6}$.
  Up to a factor of $(\kappa - 2) (\kappa - 1) \kappa^5 [313 \kappa^2 - 1500 \kappa + 1764 -24 (13 \kappa - 30) \gamma]$,
  the signature of this gadget is $\langle \kappa^3 - 3 \kappa^2 + 3, -\kappa + 3\rangle$,
  which means the eigenvalues are $\lambda_1 = (\kappa - 2)^2 \kappa$ and $\lambda_2 = \kappa (\kappa^2 - 3 \kappa + 1)$.
  If $\kappa \ge 4$,
  these eigenvalues are positive,
  so they have the same complex norm preciously when $\lambda_1 = \lambda_2$.
  However, $\lambda_1 - \lambda_2 = - (\kappa - 3) \kappa \ne 0$,
  so the eigenvalues have distinct complex norms.
  Otherwise, $\kappa = 3$ and we use the gadget in Figure~\ref{subfig:gadget:compute:binary:square}.
  Up to a factor of~$9565938$,
  the signature of this gadget is $\langle 5,2 \rangle$,
  which means the eigenvalues are~$9$ and~$3$.
  
  \item
  For $\langle (\kappa - 1) \big[(\kappa - 2) (2 \kappa - 7) \pm 3 i \kappa \sqrt{\kappa^2 - \kappa - 2}\big], -(\kappa - 2) (5 \kappa - 7) \mp i \kappa \sqrt{\kappa^2 - \kappa - 2}, 13 \kappa - 14 \rangle$,
  we use the gadget in Figure~\ref{subfig:gadget:compute:binary:diamond}.
  Let $\gamma = \pm i \sqrt{\kappa^2 - \kappa - 2}$.
  Up to a nonzero factor of $(\kappa - 2) (\kappa - 1) \kappa^5 [119 \kappa^2 + 76 \kappa - 772 + 24 (5 \kappa - 22) \gamma]$,
  the signature of this gadget is $\langle -\kappa^3 + 7 \kappa^2 - 4 \kappa - 3, 2 \kappa^2 - 7 \kappa - 3 \rangle$,
  which means the eigenvalues are $\lambda_1 = (\kappa - 2) \kappa^2$ and $\lambda_2 = -\kappa (k^2 - 5 \kappa - 3)$.
  For $3 \le \kappa \le 5$,
  one can directly check that these eigenvalues have distinct complex norms.
  For $\kappa \ge 6$,
  we have $\lambda_2 < 0$,
  so these eigenvalues have the same complex norm preciously when $\lambda_1 = -\lambda_2$.
  However, $\lambda_1 + \lambda_2 = 3 \kappa (\kappa + 1) \ne 0$,
  so the eigenvalues have distinct complex norms.
  
  \item For $\langle 1,0,-2 \rangle$ with $\kappa = 3$,
  we use the gadget in Figure~\ref{subfig:gadget:compute:binary:diamond}.
  Up to a factor of~$3$,
  the signature of this gadget is $\langle 11,-4 \rangle$,
  which means the eigenvalues are~$3$ and~$15$.
  
  \item For $\langle \pm i \sqrt{2}, 0, 1 \rangle$ with $\kappa = 3$,
  we use the gadget in Figure~\ref{subfig:gadget:binary:unary}.
  The signature of this gadget is $\langle \pm i \sqrt{2}, 1 \rangle$,
  which means the eigenvalues are $2 \pm i \sqrt{2}$ and $-1 \pm i \sqrt{2}$.
  
  \item For $\langle -1 \pm i \sqrt{2}, 0, 1 \rangle$ with $\kappa = 3$,
  we use the gadget in Figure~\ref{subfig:gadget:binary:unary}.
  The signature of this gadget is $\langle -1 \pm i \sqrt{2}, 1 \rangle$,
  which means the eigenvalues are $1 \pm i \sqrt{2}$ and $-2 \pm i \sqrt{2}$.
  
  \item For $\langle -1 \pm 3 i \sqrt{3}, 0, 2 \rangle$ with $\kappa = 3$,
  we use the gadget in Figure~\ref{subfig:gadget:compute:binary:diamond}.
  Up to a factor of~$72$,
  the signature of this gadget is $\frac{1}{3} \langle 25 \pm 13 \sqrt{3}, -5 \pm i \sqrt{3} \rangle$,
  which means the eigenvalues are $5 (1 \pm i \sqrt{3})$ and $2 (5 \pm 2 \sqrt{3})$.
  \qedhere
 \end{enumerate}
\end{proof}

\section{Invariance Properties from Row Eigenvectors} \label{sec:appendix:invariant}

The purpose of this section is to show how a recursive construction in an interpolation proof
can be used to form a hypothesis about possible invariance properties.
We often find that no matter what constructions one considers,
all signatures they produce satisfy certain invariance.
Instead of defining this notion formally,
we prove the following lemma as an example.
After this lemma and its proof,
we explain that this invariance can be suggested by certain recursive constructions in an alternative proof of Theorem~\ref{thm:edge_coloring:k=r},
that it is $\SHARPP$-hard to count edge $\kappa$-coloring over planar $\kappa$-regular graphs for all $\kappa \ge 3$.
This alternative proof uses the interpolation techniques that we developed in Section~\ref{sec:interpolation}.

\begin{lemma} \label{lem:k=r:invariant}
 Suppose $\kappa \ge 3$ is the domain size.
 If $F$ is a planar $\{\AD_{\kappa,\kappa}\}$-gate with succinct quaternary signature $\langle a,b,c,d,e \rangle$ of type $\tau_\text{color}$,
 then $a + c = b + d$.
\end{lemma}

\begin{proof}
 Fix two distinct colors $g,y \in [\kappa]$.
 We define the \emph{swap} of an edge colored $g$ or $y$ to be the opposite of these two colors.
 That is, swapping the color of an edge colored $g$ (resp.~$y$) gives the same edge colored $y$ (resp.~$g$).
 The $i$th external edge of $F$ is the external edge that corresponds to the $i$th input of $F$.
 Recall that the input edges of $F$ are ordered cyclically.
 
 For $1 \le i \le 4$,
 let $S_i$ (resp.~$S_i'$) be the set of colorings of the edges (both internal and external)
 of $F$ with an external coloring in the partition $P_i$ of the succinct signature type $\tau_\text{color}$
 such that the first external edge of $F$ is colored $g$ (resp.~$y$) and the remaining external edges are either colored $g$ or $y$ (as dictated by $P_i$).
 Note that $|S_i| = |S_i'|$ for $1 \le i \le 4$.
 Furthermore, the sizes of these sets do not depend on the choice of $g,y \in [\kappa].$
 Thus, it suffices to show that
 \begin{equation} \label{eqn:k=r:invariant}
  |S_1 \union S_1' \union S_3 \union S_3'|
  =
  |S_2 \union S_2' \union S_4 \union S_4'|.
 \end{equation}
 
 Let $\sigma \in S_1 \union S_1' \union S_3 \union S_3'$ be a coloring of $F$.
 Starting at the first external edge of $F$,
 there is a unique path $\pi_1$ that alternates in edge colors between~$g$ and~$y$ and terminates at another external edge of $F$.
 Suppose for a contradiction that this path terminates at the third external edge of $F$.
 Also consider the unique path $\pi_2$ that starts at the second external edge of $F$,
 alternates in edge colors between~$g$ and~$y$,
 and must terminate at the fourth external edge of $F$.
 These two paths must cross somewhere since their ends are crossed.
 By planarity, they must cross at a vertex, and yet they must be vertex disjoint.
 This is a contradiction.
 Therefore, the path $\pi_1$ either terminates at the second or fourth external edge of $F$.

 Suppose $\pi_1$ terminates at the second external edge of $F$.
 If $\sigma \in S_1$ (resp.~$\sigma \in S_1'$),
 then swapping the colors of every edge in $\pi_1$ gives a new coloring $\pi_1' \in S_2'$ (resp.~$\pi_1' \in S_2$).
 Similarly,
 if $\sigma \in S_3$ (resp.~$\sigma \in S_3'$),
 then swapping the colors of every edge in $\pi_1$ gives a new coloring $\pi_1' \in S_4'$ (resp.~$\pi_1' \in S_4$).
 
 Otherwise, $\pi_1$ terminates at the fourth external edge of $F$.
 If $\sigma \in S_1$ (resp.~$\sigma \in S_1'$),
 then swapping the colors of every edge in $\pi_1$ gives a new coloring $\pi_1' \in S_4'$ (resp.~$\pi_1' \in S_4$).
 Similarly,
 if $\sigma \in S_3$ (resp.~$\sigma \in S_3'$),
 then swapping the colors of every edge in $\pi_1$ gives a new coloring $\pi_1' \in S_2'$ (resp.~$\pi_1' \in S_2$).
 
 Furthermore, this mapping from $S_1 \union S_1' \union S_3 \union S_3'$ to $S_2 \union S_2' \union S_4 \union S_4'$ is invertible.
 Therefore, we have established~(\ref{eqn:k=r:invariant}), as desired.
\end{proof}

Now we give an alternative proof of Theorem~\ref{thm:edge_coloring:k=r}.
The recursive construction in this proof will suggest the invariance in Lemma~\ref{lem:k=r:invariant}.

Let $q(x,\kappa) = x^3 - x^2 + x - (\kappa - 1)$.
First we determine the nature of the roots of $q(x,\kappa)$.
 
\begin{lemma} \label{lem:invariant:roots_nature}
 For all $\kappa \in \Z$,
 the polynomial $q(x,\kappa)$ in $x$ has one real root $r \in \R$ and two nonreal complex conjugate roots $\alpha, \overline{\alpha} \in \mathbb{C}$,
 such that $\alpha + \overline{\alpha} = 1 - r$ and $\alpha \overline{\alpha} = r^2 - r + 1$.
 
 Furthermore,
 if $q(x,\kappa)$ is reducible in $\Q[x]$ and $\kappa \ge 3$,
 then $r \ge 2$ is an integer.
\end{lemma}

\begin{proof}
 The discriminant of $q(x,\kappa)$ with respect to $x$ is $\operatorname{disc}_x(q(x,\kappa)) = -27 \kappa^2 + 68 \kappa - 44 < 0$,
 so $q(x,\kappa)$ has one real root $r \in \R$ and two nonreal complex conjugate roots $\alpha, \overline{\alpha} \in \mathbb{C}$.
 We have
 \begin{align*}
  \alpha + \overline{\alpha} + r& = 1\\
  \alpha \overline{\alpha} + (\alpha + \overline{\alpha}) r &= 1\\
  \alpha \overline{\alpha} r & = \kappa - 1.
 \end{align*}
 It follows that $\alpha + \overline{\alpha} = 1 - r$,
 $\alpha \overline{\alpha} = r^2 - r + 1$,
 and
 \begin{align}
  \kappa &= r^3 - r^2 + r + 1. \label{eqn:appendix:invariant:roots_nature}
 \end{align}
 
 If $q(x,\kappa)$ is reducible in $\Q[x]$ with $\kappa \ge 3$,
 then $r \in \Z$ by Gauss's Lemma
 and so $r \ge 2$ by~(\ref{eqn:appendix:invariant:roots_nature}).
\end{proof}
 
\begin{lemma} \label{lem:appendix:satisfy_LC}
 If $\kappa \ge 3$ is an integer,
 then the roots of $x^3 - x^2 + x - (\kappa - 1)$ satisfy the lattice condition.
\end{lemma}

\begin{proof}
 If $q(x,\kappa)$ is irreducible in $\Q[x]$,
 then its Galois group is $S_3$ or $A_3$ and
 so its roots satisfy the lattice condition by Lemma~\ref{lem:interpolation:lattice_condition:cubic}.
 
 Otherwise,
 $q(x,\kappa)$ is reducible in $\Q[x]$.
 By Lemma~\ref{lem:invariant:roots_nature},
 $q(x,\kappa)$ has one real root $r \in \Z$ satisfying $r \ge 2$
 and two nonreal complex conjugate roots $\alpha, \overline{\alpha} \in \mathbb{C}$
 satisfying $\alpha + \overline{\alpha} = 1 - r$ and $\alpha \overline{\alpha} = r^2 - r + 1$.
 Suppose there exist $i,j,k \in \Z$ such that $\alpha^i \overline{\alpha}^j = r^k$ and $i + j = k$.
 We want to show that $i=j=k=0$.

 There is an element in the Galois group of $q(x,\kappa)$ that fixes $\mathbb{Q}$ pointwise and swaps $\alpha$ and $\overline{\alpha}$.
 Thus $\alpha^j \overline{\alpha}^i = r^k$.
 Dividing these two equations gives $(\alpha / \overline{\alpha})^{i-j} = 1$.
 We claim that $\omega = \alpha / \overline{\alpha}$ cannot be a root of unity and hence $i=j$.
 For a contradiction,
 suppose $\omega$ is a $d$th primitive root of unity.
 Let $f(x) = (x - \alpha) (x - \overline{\alpha}) = x^2 + (r-1) x + (r^2-r+1) \in \mathbb{Z}[x]$.
 Then $\omega$ belongs to the splitting field of $f$ over $\mathbb{Q}$,
 which is a degree~$2$ extension over $\mathbb{Q}$.
 This implies that the Euler totient function $\phi(d) \divides 2$.
 Therefore $d \in \{1,2,3,4,6\}$.
 Let $\rho = \frac{\alpha + \overline{\alpha}}{\alpha \overline{\alpha}} = \frac{1 + \omega}{\omega \overline{\alpha}} = \frac{1-r}{r^2-r+1} \in \mathbb{Q}$.
 Since $r \ge 2$,
 we have $\rho \neq 0$ and hence $d \neq 2$.
 Moreover, $f(x) = x^2 - (2 + \omega + \omega^{-1}) \rho^{-1} x + (2 + \omega + \omega^{-1}) \rho^{-2}$.
 Notice that the quantity $2 + \omega + \omega^{-1}$ is $4,1,2,3$ respectively, when $d = 1,3,4,6$.
 As $(2 + \omega + \omega^{-1}) \rho^{-2} \in \mathbb{Z}$,
 we get that $\rho^{-1}$ must be an integer when $d = 3,4,6$ and half an integer when $d = 1$.
 However $\rho^{-1} = -r + \frac{1}{r-1}$.
 The only possibility is $r = 3$ and $d = 1$;
 yet it is easy to check that $\omega \neq 1$ when this holds.
 This proves the claim.

 From $\alpha \overline{\alpha} = r^2 - r + 1$,
 we have $(r^2 - r + 1)^i = (\alpha \overline{\alpha})^i = r^k$.
 Since $r$ and $r^2 - r + 1$ are relatively prime and $r \ge 2$,
 we must have $i = k = 0$.
\end{proof}

\begin{figure}[t]
 \centering
 \captionsetup[subfigure]{labelformat=empty}
 \subcaptionbox{$N_0$}{
  \begin{tikzpicture}[scale=\scale,transform shape,node distance=\nodeDist,semithick]
   \node[external] (0)                    {};
   \node[external] (1) [right       of=0] {};
   \node[internal] (2) [below right of=1] {};
   \node[external] (3) [below left  of=2] {};
   \node[external] (4) [left        of=3] {};
   \node[external] (5) [above right of=2] {};
   \node[external] (6) [right       of=5] {};
   \node[external] (7) [below right of=2] {};
   \node[external] (8) [right       of=7] {};
    \path (0) edge[out=   0, in=135, postaction={decorate, decoration={
                                                            markings,
                                                            mark=at position 0.4   with {\arrow[>=diamond,white] {>}; },
                                                            mark=at position 0.4   with {\arrow[>=open diamond]  {>}; },
                                                            mark=at position 0.999 with {\arrow[>=diamond,white] {>}; },
                                                            mark=at position 1     with {\arrow[>=open diamond]  {>}; } } }] (2)
          (2) edge[out=-135, in=  0] (4)
              edge[out=  45, in=180] (6)
              edge[out= -45, in=180] (8);
   \begin{pgfonlayer}{background}
    \node[draw=\borderColor,thick,rounded corners,fit = (1) (3) (5) (7)] {};
   \end{pgfonlayer}
  \end{tikzpicture}}
 \qquad
 \subcaptionbox{$N_1$}{
  \begin{tikzpicture}[scale=\scale,transform shape,node distance=\nodeDist,semithick]
   \node[external]  (0)                    {};
   \node[external]  (1) [right of=0]       {};
   \node[internal]  (2) [below right of=1] {};
   \node[external]  (3) [below left  of=2] {};
   \node[external]  (4) [left        of=3] {};
   \node[external]  (5) [right       of=2] {};
   \node[internal]  (6) [right       of=5] {};
   \node[external]  (7) [above right of=6] {};
   \node[external]  (8) [right       of=7] {};
   \node[external]  (9) [below right of=6] {};
   \node[external] (10) [right       of=9] {};
   \path (0) edge[out=  0, in= 135, postaction={decorate, decoration={
                                                           markings,
                                                           mark=at position 0.4   with {\arrow[>=diamond,white] {>}; },
                                                           mark=at position 0.4   with {\arrow[>=open diamond]  {>}; } } }] (2)
         (4) edge[out=  0, in=-135, postaction={decorate, decoration={
                                                           markings,
                                                           mark=at position 0.999 with {\arrow[>=diamond,white] {>}; },
                                                           mark=at position 1     with {\arrow[>=open diamond]  {>}; } } }] (2)
         (2) edge[bend left,        postaction={decorate, decoration={
                                                           markings,
                                                           mark=at position 0.999 with {\arrow[>=diamond,white] {>}; },
                                                           mark=at position 1     with {\arrow[>=open diamond]  {>}; } } }] (6)
             edge[bend right]        (6)
         (6) edge[out= 45, in= 180]  (8)
             edge[out=-45, in= 180] (10);
   \begin{pgfonlayer}{background}
    \node[draw=\borderColor,thick,rounded corners,fit = (1) (3) (7) (9)] {};
   \end{pgfonlayer}
  \end{tikzpicture}}
 \qquad
 \subcaptionbox{$N_{k+1}$}{
  \begin{tikzpicture}[scale=\scale,transform shape,node distance=\nodeDist,semithick]
   \node[external] (0)              {};
   \node[external] (1) [above of=0] {};
   \node[external] (2) [below of=0] {};
   \node[external] (3) [right of=0] {};
   \node[external] (4) [above of=3] {};
   \node[external] (5) [below of=3] {};
   \path let
          \p1 = (0),
          \p2 = (3)
         in
          node[external] at (\x1 / 2 + \x2 / 2, \y1) {\Huge \begin{sideways}$N_k$\end{sideways}};
   \path let
          \p1 = (1),
          \p2 = (4)
         in
          node[external] (6) at (3 * \x1 / 4 + \x2 / 4, \y1) {};
   \path let
          \p1 = (1),
          \p2 = (4)
         in
          node[external] (7) at (\x1 / 4 + 3 * \x2 / 4, \y1) {};
   \path let
          \p1 = (2),
          \p2 = (5)
         in
          node[external] (8) at (3 * \x1 / 4 + \x2 / 4, \y1) {};
   \path let
          \p1 = (2),
          \p2 = (5)
         in
          node[external] (9) at (\x1 / 4 + 3 * \x2 / 4, \y1) {};
   \node[external] (10) [above left  of=6]  {};
   \node[external] (11) [below left  of=8]  {};
   \node[external] (12) [below left  of=10] {};
   \node[external] (13) [above left  of=11] {};
   \node[external] (14) [left        of=12] {};
   \node[external] (15) [left        of=13] {};
   \node[external] (16) [above right of=7]  {};
   \node[external] (17) [below right of=9]  {};
   \node[external] (18) [below right of=16] {};
   \node[external] (19) [above right of=17] {};
   \node[external] (n1) [right       of=18] {};
   \path let
          \p1 = (n1),
          \p2 = (18),
          \p3 = (19)
         in
          node[internal] (20) at (\x1, \y2 / 2 + \y3 / 2) {};
   \node[external] (n2) [right of=20] {};
   \node[external] (n3) [right of=n2] {};
   \path let
          \p1 = (n3),
          \p2 = (14)
         in
          node[external] (21) at (\x1, \y2) {};
   \path let
          \p1 = (n3),
          \p2 = (15)
         in
          node[external] (22) at (\x1, \y2) {};
   \path (6)         edge[out=  90, in=   0]     (10.center)
         (11.center) edge[out=   0, in=-110, postaction={decorate, decoration={
                                                                    markings,
                                                                    mark=at position 0.999 with {\arrow[>=diamond,white] {>}; },
                                                                    mark=at position 1     with {\arrow[>=open diamond]  {>}; } } }] (8)
         (14)        edge[out=   0, in= 180, postaction={decorate, decoration={
                                                                    markings,
                                                                    mark=at position 0.4   with {\arrow[>=diamond,white] {>}; },
                                                                    mark=at position 0.4   with {\arrow[>=open diamond]  {>}; } } }] (10.center)
         (11.center) edge[out= 180, in=   0]     (15)
         (7)         edge[out=  90, in= 180]     (16.center)
         (9)         edge[out= -90, in= 180]     (17.center)
         (16.center) edge[out=   0, in= 135, postaction={decorate, decoration={
                                                                    markings,
                                                                    mark=at position 0.999 with {\arrow[>=diamond,white] {>}; },
                                                                    mark=at position 0.999 with {\arrow[>=open diamond]  {>}; } } }] (20)
         (17.center) edge[out=   0, in=-135]     (20)
         (20)        edge[out=  45, in= 180]     (21)
                     edge[out= -45, in= 180]     (22);
   \begin{pgfonlayer}{background}
    \node[draw=\borderColor,thick,densely dashed,rounded corners,fit = (1) (2) (4) (5)] {};
    \node[draw=\borderColor,thick,rounded corners,fit = (10) (11) (12) (13) (16) (17) (n2)] {};
   \end{pgfonlayer}
  \end{tikzpicture}}
 \caption{Alternate recursive construction to interpolate $\langle 2,1,0,1,0 \rangle$.
 The vertices are assigned the signature of the gadget in Figure~\ref{fig:gadget:k=r:arity_reduction}.}
 \label{fig:gadget:invariant:weave}
\end{figure}
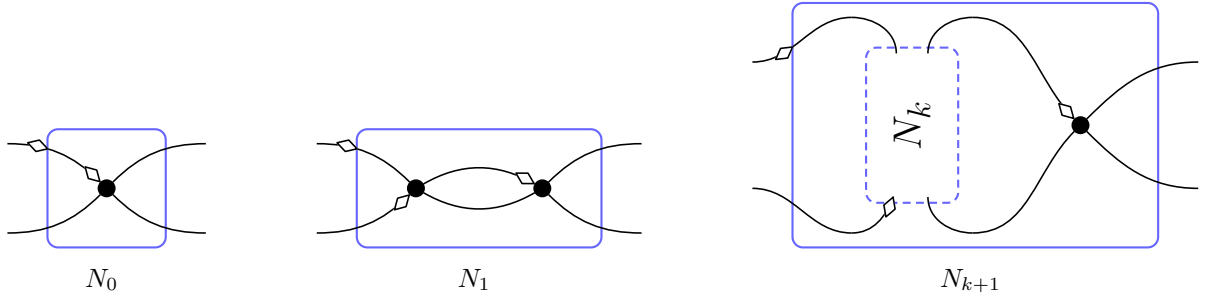

\begin{proof}[Alternative proof of Theorem~\ref{thm:edge_coloring:k=r}]
 As before,
 let $\kappa$ be the domain size of all Holant problems in this proof
 and let $\langle 2,1,0,1,0 \rangle$ be a succinct quaternary signature of type $\tau_\text{color}$.
 We reduce from $\PlHolant(\langle 2,1,0,1,0 \rangle)$ to $\PlHolant(\AD_{\kappa,\kappa})$,
 which denotes the problem of counting edge $\kappa$-colorings in planar $\kappa$-regular graphs as a Holant problem.
 Then by Corollary~\ref{cor:k=r:21010_hard},
 we conclude that $\PlHolant(\AD_{\kappa,\kappa})$ is $\SHARPP$-hard.

 Consider the gadget in Figure~\ref{fig:gadget:k=r:arity_reduction},
 where the bold edge represents $\kappa - 2$ parallel edges.
 We assign $\AD_{\kappa,\kappa}$ to both vertices.
 Up to a nonzero factor of $(\kappa-2)!$,
 this gadget has the succinct quaternary signature $f = \langle 0,1,1,0,0 \rangle$ of type $\tau_\text{color}$.
 Now consider the recursive construction in Figure~\ref{fig:gadget:invariant:weave}.
 All vertices are assigned the signature $f$.
 Let $f_s$ be the succinct quaternary signature of type $\tau_\text{color}$ for the $s$th gadget of the recursive construction.
 Then $f_0 = f$ and $f_s = M^s f_0$,
 where
 \[
  M =
  \begin{bmatrix}
   0 & 0 & 0 & \kappa - 1 & 0 \\
   1 & 0 & 0 & \kappa - 2 & 0 \\
   0 & 1 & 0 &          0 & 0 \\
   0 & 0 & 1 &          0 & 0 \\
   0 & 0 & 0 &          0 & 1
  \end{bmatrix}.
 \]
 The row vectors
 \[
  (1, -1, 1, -1, 0) 
  \qquad \text{and} \qquad
  (0,  0, 0,  0, 1)
 \]
 are linearly independent row eigenvectors of $M$,
 with eigenvalues~$-1$ and~$1$ respectively,
 that are orthogonal to the initial signature $f_0$.
 Note that our target signature $\langle 2,1,0,1,0 \rangle$ is also orthogonal to these two row eigenvectors.
 
 Up to a factor of $(x - 1) (x + 1)$,
 the characteristic polynomial of $M$ is $x^3 - x^2 + x - (\kappa - 1)$.
 The roots of this polynomial satisfy the lattice condition by Lemma~\ref{lem:appendix:satisfy_LC}.
 In particular,
 these three roots are distinct.
 By Lemma~\ref{lem:invariant:roots_nature},
 the only real root is at least~$2$.
 Thus, all five eigenvalues of $M$ are distinct,
 so $M$ is diagonalizable.
 
 The $3$-by-$3$ matrix in the upper-left corner of $[f_0\ M f_0\ \dots\ M^4 f_0]$ is
 $\left[\begin{smallmatrix} 0 & 0 & \kappa - 1 \\ 1 & 0 & \kappa - 2 \\ 1 & 1 & 0 \end{smallmatrix}\right]$.
 Its determinant is $\kappa - 1 \ne 0$.
 Thus, $[f_0\ M f_0\ \dots\ M^4 f_0]$ has rank at least~$3$,
 so by Lemma~\ref{lem:2nd_condition_implication},
 $f_0$ is not orthogonal to the three remaining row eigenvectors of $M$.
 
 Therefore, by Lemma~\ref{lem:interpolate_all_not_orthogonal},
 we can interpolate $\langle 2,1,0,1,0 \rangle$,
 which completes the proof.
\end{proof}

Notice that the row eigenvector $(1,-1,1,-1,0)$ suggests that $a - b + c - d = 0$ is an invariance shared by all signatures of symmetric ternary constructions.
Some row eigenvectors,
like $(0,0,0,0,1)$,
only indicate an invariance present in some recursive constructions.
(When $\kappa = 4$,
there are recursive constructions for which $(0,0,0,0,1)$ is not a row eigenvector of the recurrence matrix.)
The row eigenvector
$(1,-1,1,-1,0)$ is more intrinsic;
it must appear because of the invariance present in all constructions as shown in Lemma~\ref{lem:k=r:invariant}.

This suggests an approach to discover new invariance properties.
Given a set $\mathcal{F}$ of signatures,
create some recursive construction and inspect the row eigenvectors of the resulting recurrence matrix.
For example,
consider the set $\mathcal{F}_{\mathfrak{A}} = \{\langle a,b,c \rangle \st a,b,c \in \mathbb{C} \text{ and } \mathfrak{A} = 0\}$,
where $\mathfrak{A} = a - 3 b + 2 c$.
It seems that $\mathcal{F}_{\mathfrak{A}}$ is closed under symmetric ternary constructions,
such as those in Section~\ref{subsec:ternary:construct}.
In particular,
$(1,-3,2)$ is a row eigenvector of the recurrence matrix for every recursive ternary construction with symmetric signatures that we tried.
However, we do not know how to prove this closure property.

\end{document}